\documentclass[11pt,a4paper]{article}
\usepackage[T1]{fontenc}
\IfFileExists{lmodern.sty}{\usepackage{lmodern}}{}
\usepackage[margin=2.5cm]{geometry}
\usepackage{amsmath,amssymb,amsthm,bm,booktabs,array}
\usepackage{graphicx}
\usepackage[font=small,labelfont=bf]{caption}
\usepackage[hidelinks]{hyperref}
\newtheorem{corollary}{Corollary}
\newtheorem{proposition}{Proposition}
\newtheorem{lemma}{Lemma}
\begin{document}

\title{Revisiting the $r_s$ expansion of the uniform electron gas in $d$ dimensions:\texorpdfstring{\\}{ } \texorpdfstring{$\log r_s$}{log rs} terms from coincident Mellin poles}

\author{Fumihiro Imoto}
\date{}
\maketitle

\begin{abstract}
The high-density expansion of the correlation energy of the uniform
electron gas contains $\log r_s$ terms, where $r_s$ is the Wigner--Seitz radius. Their
form depends on the dimension $d$. Adapting the Mellin-transform analysis
of Feynman integrals in relativistic field theory, we write the energy of
the ring series and of diagrams linear in one screened interaction line in $d$
dimensions as Mellin convolutions of a radial kernel, fixed by the diagram,
with a scaling function. The scaling function depends on the momentum
transfer only through its ratio to a density scale. In the Mellin variable $s$, screening places
poles at $s=N(d-1)$ with $N=1,2,\dots$, and the diagram places a pole
at $s=d+1$ for the rings and at $s=2$ for a polarization insertion. Once the factors $\log(1/q)$, with $q$ the momentum transfer, cancel between the vertex and self-energy corrections, as the Ward identity requires, a $\log r_s$ term
appears where two poles coincide. This rule reproduces the known $\log r_s$
terms in two and three dimensions. In two dimensions it gives the
coefficient of the next $\log r_s$ term, $r_s^2\log r_s$, in closed form,
$7/\pi+4\log2/\pi-160/(9\pi^2)-16G/\pi^2=-0.17546$~Ry with $G$ Catalan's
constant. In three dimensions the screened
second-order exchange diagram alone has an $r_s\log^2r_s$ term at high
density and an $r_s^{-1/4}$ term at low density, and the Fock
self-energy corrections cancel both.
\end{abstract}

\section{Introduction}\label{sec:intro}
The uniform electron gas (UEG) is the reference system for the
exchange--correlation functionals of density-functional theory
\cite{HK1964,KS1965,LoosGill}. At high density in three dimensions its
correlation energy has the expansion
\begin{equation}
 \epsilon_c(r_s)=A\log r_s+B+C\,r_s\log r_s+D\,r_s+R(r_s),
 \label{eq:jellium-expansion}
\end{equation}
where $\epsilon_c$ is the correlation energy per electron in Ry, $\log$
is to base $e$, $r_s$ is the Wigner--Seitz radius in Bohr
units and $R(r_s)$ is the remainder. The coefficients, reviewed in
Ref.~\cite[Eq.~(34)]{LoosGill}, are
$A=2(1-\log2)/\pi^2\approx0.0622$ \cite{Macke1950,GellMannBrueckner},
$B\approx-0.0938$ \cite{GellMannBrueckner,OMS} and \cite[Eq.~(8)]{Endo1999},
$C=\alpha_L(7\pi^2/6-12\log2-1)/(2\pi^3)\approx0.0185$
\cite{DuBois1959,CarrMaradudin,LoosGill2011} with
$\alpha_L=[4/(9\pi)]^{1/3}$, and $D\approx-0.020$
\cite{CarrMaradudin} and \cite[Eq.~(8)]{Endo1999}. In two dimensions, the case of the
electron layers of semiconductor heterostructures \cite{AndoFowlerStern1982},
the expansion reads
\begin{equation}
 \epsilon_c^{\rm 2D}(r_s)=a_0+a_\ell\,r_s\log r_s+a_1\,r_s+O(r_s^2\log r_s),
 \label{eq:2d-expansion}
\end{equation}
with $a_0=2(\log2-1)+2G-16\beta(4)/\pi^2=-0.38499$~Ry, where $G=\beta(2)$
is Catalan's constant and $\beta$ the Dirichlet beta function
\cite{LoosGill2DPRB}, and
$a_\ell=-2\sqrt2\,[10/(3\pi)-1]=-0.17263$~Ry
\cite{RajagopalKimball1977,LoosGill2DPRB}. The coefficient $a_1$ is not
known \cite[Table~I]{LoosGill}, and neither is the next $\log r_s$ term.
The three-dimensional expansion has been tested against quantum Monte
Carlo energies \cite{Ruggeri2018}. In the analogous mean-field
(semiclassical) scaling limit, its leading term has been proved by matching upper and lower bounds
\cite{ChristiansenHainzlNam2023,ChristiansenHainzlNam2026}. The expansion
has also been applied to exchange contributions to the self-energy and
the momentum distribution \cite{Ziesche2007,Ziesche2010}.

Perturbation theory appears to give a single $\log r_s$ at each order in
three dimensions \cite[Eq.~(21)]{LoosGill}, \cite[Sec.~I]{Endo1999} and
\cite{Macke1950}, but this form extrapolates the known orders. Which
diagrams contribute at the next orders has been determined order by
order, each order with its own resummation. The dimension changes the
form of the $\log r_s$ terms.
For the ring diagrams, Schlesier, Benavides-Riveros and Marques found that
the $\log r_s$ term of three dimensions becomes a power of $r_s$ for $d>3$
\cite{Schlesier2020}. In two dimensions, the first $\log r_s$ term of the ring
diagrams moves to $r_s\log r_s$ \cite{RajagopalKimball1977}. The mechanism behind a single
$\log r_s$ in three dimensions is well known
\cite{GellMannBrueckner,ZiescheCioslowski2005,Endo1999}: screening cuts
off the small-momentum Coulomb divergence at the Thomas--Fermi wave number
$k_{\rm TF}\propto r_s^{1/2}k_F$, with $k_F$ the Fermi wave number. The
integral $\int_{k_{\rm TF}}^{k_F}\mathrm dq/q$ over the momentum
transfer $q$ thus leaves $\log r_s$. What is missing is a criterion that
predicts, before a diagram is evaluated and in any dimension, at which
order a $\log r_s$ term appears, how many factors of $\log r_s$ it carries, and
when such terms cancel between diagrams.

A second question concerns correlation-energy methods based, like the
$GW$ approximation \cite{Hedin1965}, on the screened interaction of the
random-phase approximation (RPA). Their exchange corrections, such as
second-order screened exchange (SOSEX), are widely used in quantum
chemistry and materials science \cite{Gruneis2009,Hummel2019,Forster2022}
and are built from the diagrams studied here. The screened
second-order exchange diagram is called the \emph{kite} after its shape.
Benites, Rosado and Manousakis (BRM) distinguish the kite from SOSEX
(Sec.~\ref{sec:diagram-coefficients}). They computed the energy of the kite over a wide range
of densities and fitted the high-density part with terms linear in
$\log r_s$ \cite{Benites2024}.
Numerical values of a single diagram cannot show whether a screened vertex correction evaluated
alone produces terms that the accompanying self-energy corrections would
remove.

In this paper we give the missing criterion for the ring series and for
diagrams linear in one screened line, in $d$ dimensions with the
three-dimensional Coulomb interaction $e^2/r$. With $q$ the momentum transfer of the line in
units of $k_F$, the energy has at high density the form
$\int_0^\infty k(q)\,g(q/\varepsilon)\,\mathrm dq$. This integral is a
Mellin convolution of a \emph{radial kernel} $k$, the contribution of the
diagram per unit $q$, with a \emph{scaling function} $g$, which depends on $q$
only through $q/\varepsilon$. At high density, for a diagram linear in the
line, the scaling function is the relative change of the line caused by screening. The scale $\varepsilon=k_{\rm TF}/k_F$ depends on
the density.
Each term of the small-$q$ and large-$q$ expansions of $k$ and $g$, an
\emph{endpoint term}, gives a pole of the Mellin transform. The density
expansion is read off these poles
\cite{ParisKaminski,BleisteinHandelsman,FlajoletGourdonDumas1995,Lopez2008}
(Sec.~\ref{sec:general}). The approach adapts the Mellin analysis of
Feynman integrals in relativistic field theory
\cite{BergereLam1974,FriotGreynatDeRafael2005} to many-body diagrams. No
step requires a resummation tailored to a particular diagram.

In the Mellin variable $s$, long-range screening places the poles of the
Mellin transform of the scaling function at the
integer multiples $s=N(d-1)$, $N=1,2,\dots$, and the phase space of the
diagram places the kernel pole at $s=d+1$ for the rings and at $s=2$ for
a polarization insertion (Sec.~\ref{sec:resonances} and
figure~\ref{fig:poles}). Once the Fermi-surface factors $\log(1/q)$ of the kernel
(Sec.~\ref{sec:overview}) cancel between the vertex and self-energy
corrections, as the Ward identity requires, a $\log r_s$ term appears only
where two poles coincide, a \emph{resonance}. This rule reproduces the
Gell-Mann--Brueckner ($A$) and
Carr--Maradudin ($C$) terms in three dimensions, the Rajagopal--Kimball term
($a_\ell$) in
two, and the power law of the rings for $d>3$ found by Schlesier
et al.\ \cite{Schlesier2020}.

The new result is the coefficient of the next $\log r_s$ term in two
dimensions (Sec.~\ref{sec:twod}). There the kernel pole of the polarization
insertion at $s=2$ meets the
second screening pole, and the conserving first-order exchange
polarization $\Pi^{(1)}$ contributes $a_{2\ell}\,r_s^2\log r_s$ to
\eqref{eq:2d-expansion}. We obtain the long-wavelength limit of this
polarization in closed form. It interpolates between the static and
high-frequency limits known from earlier work
\cite{Rajagopal1977,CHSS1982,Hameeuw2003EPJB,DasSarma2026}. Diagrams in
which three or more screened lines of small momentum transfer meet at one
closed fermion loop enter at order $r_s^4$ in two dimensions by power counting, because such
loops, summed over the orderings of their vertices, remain bounded at small
momenta $Q$ and frequencies of order $v_FQ$
\cite{NeumayrMetzner1998,KopperMagnen2001} (Sec.~\ref{sec:completeness}).

In Secs.~\ref{sec:kite}--\ref{sec:se-cancellation} we follow the Ward-identity cancellation for the
screened kite in three dimensions. The kite alone acquires a term
$r_s\log^2r_s$ at high density and a term $r_s^{-1/4}$ at low density.
The fit of BRM absorbs the first of these into its $r_s\log r_s$
coefficient. The Fock self-energy corrections
\cite{GeldartTaylor1970,EngelVosko1990,EngelVoskoErratum,Holas1979,Qian2015}
cancel both terms, with remainder bounds. This cancellation is consistent
with the Ward identity
\cite{Rice1965,MahanSernelius1989} and with conserving approximations
\cite{BaymKadanoff1961,Baym1962,Takada1995,MaebashiTakada2011,RichardsonAshcroft1994}.

The spin dependence
of the isolated kite is the subject of a companion paper
\cite{Imoto2026b}. The paramagnetic leading coefficients of that paper
agree with the present ones. Low-density limits of single diagrams with
free propagators are formal, since the physical ground state at low
density is the Wigner crystal \cite{Wigner1934,CarrCHF1961}. They
nevertheless test the method, and their exponents can be used to check
interpolating
parametrizations of density functionals
\cite{VWN,PerdewZunger,PerdewWang92} fitted to quantum Monte Carlo
energies \cite{CeperleyAlder}.

The method is developed in Sec.~\ref{sec:general}. It derives the rule
that a $\log r_s$ term appears only at a resonance and ends with the
check of this rule against the known $\log r_s$ terms in two and three
dimensions and against the ring series for $d>3$
(Secs.~\ref{sec:ring-d} and \ref{sec:Pi1-d}). Section~\ref{sec:results}
gives the results: the $r_s^2\log r_s$ term of the two-dimensional gas
(Sec.~\ref{sec:twod}) and the Ward-identity cancellation for the kite
(Secs.~\ref{sec:kite}--\ref{sec:se-cancellation}). Proofs, remainder bounds and numerical
details are given in the Supporting Information (SI). SI
table~\ref{S-sm-tab:guide} lists, for each result, the SI sections that
derive it. SI table~\ref{S-sm-tab:symbols} lists the main symbols and
where each is defined.

\section{Theory: screened diagrams as Mellin convolutions}\label{sec:setting}\label{sec:general}\label{sec:spectral}\label{sec:mellin}\label{sec:density}
\subsection{Overview}\label{sec:overview}
The method has three steps.
\begin{enumerate}
\item The first step is the spectral representation (Sec.~\ref{sec:spectral-rep}). The
screened line $W$ is written as the bare line minus a superposition of
Lorentzian lines. This step is needed at low density, where
the line reduces to a single Lorentzian, and for the remainder bounds.
At high density the next step follows directly from $W$ at fixed
$\nu/q$.
\item The second step is the Mellin convolution (Secs.~\ref{sec:radial-general} and
\ref{sec:screening-functions}). After the frequency integration the
energy of a diagram linear in $W$ becomes
$I(\varepsilon)=\int_0^\infty k(q)\,g(q/\varepsilon)\,\mathrm dq$, the
Mellin convolution of a radial kernel $k$, fixed by the diagram, with a
scaling function $g$, fixed by the line (at low density together with
the frequency factor of the diagram at the width of the line,
Sec.~\ref{sec:screening-functions}). This form is exact for the
limiting forms of the line in each density limit, for example with the
Lindhard function replaced by its small-$q$ limit at high density. The
corrections to these limiting forms are bounded separately. Apart from the
factor $r_s^{\,n-2}$ of a diagram with $n$ bare interactions
(Sec.~\ref{sec:linear-diagrams}), the density enters only through the
scale $\varepsilon$. At high density this scale is the scaled
Thomas--Fermi wave number $q_{\rm TF}=k_{\rm TF}/k_F$, with
$q_{\rm TF}^{d-1}\propto r_s$ (Sec.~\ref{sec:units}). At low density in
three dimensions it is the scale $q_*\propto r_s^{1/4}$
(Sec.~\ref{sec:screening-functions}). The
ring series is a nonlinear functional of the line, but its energy has
the same form, with a scalar function of the screening ratio $c$ of \eqref{eq:cdef} as the
scaling function (Sec.~\ref{sec:ring-d}).
\item The third step is the endpoint--residue correspondence
(Sec.~\ref{sec:mellin-convolution}). Each endpoint term of $k$ or of
$g$ at $q\to0$ or $q\to\infty$ gives one pole of the Mellin transform
in a complex variable $s$. The position of the pole gives the power of
$r_s$, and its order minus one gives the highest power of $\log r_s$.
\end{enumerate}
Two poles at the same position add their orders. Diagrams whose
kernels have equal and opposite endpoint terms cancel pole by pole.

In a Fermi
sea the kernel of a polarization insertion has the small-$q$ endpoint term
$k\simeq q\log^m(1/q)$. Its factor $q$ comes from phase space, and each
factor $\log(1/q)$, a \emph{Fermi-surface factor}, comes from the vertex
or self-energy at the Fermi surface. The kernel of the ring series has the
small-$q$ endpoint term $k\propto q^d$.

The RPA line contributes
\emph{universal} poles, which do not depend on the diagram. At high density these poles come from its
\emph{Thomas--Fermi tail}, the decay of the screening correction in
powers $(q_{\rm TF}/q)^{N(d-1)}$, $N=1,2,\dots$, for $q\gg q_{\rm TF}$.
They lie at $s=N(d-1)$. In three dimensions the first of them is at $s=2$. At
low density in three dimensions, the pole at $s=-2$ comes from the
\emph{large-width term}, the first term of the expansion of the low-density
scaling function in $1/\Lambda_\lambda$, of order $q/\Lambda_\lambda\propto t^2$
(Sec.~\ref{sec:screening-functions}). At low density the screened line is
strongly reduced below the plasma frequency and close to the bare line above it.
The energy of the large-width term comes from this high-frequency window, where
only the tail $\nu^{-2}$ of the diagram in the imaginary frequency $\nu$
(Sec.~\ref{sec:units}) enters. This energy is the leading
low-density term only when the kernel decays faster than $q^{-3}$ at
large $q$ (Corollary~\ref{cor:exponents}).

Sections~\ref{sec:ring-d}, \ref{sec:Pi1-d} and \ref{sec:results} apply the three steps to
particular diagrams. In each case the scaling functions are those of
Sec.~\ref{sec:screening-functions}, and the work consists of finding the
endpoint terms of the radial kernel and the poles they produce.

\subsection{Units and RPA screening}\label{sec:units}
We consider the paramagnetic UEG in $d$ dimensions at zero temperature,
with free Green functions $G_0$ for all propagators and the Coulomb
interaction $e^2/r$ of three-dimensional electrostatics. The
$d$-dimensional Fourier transform of this interaction is $\mathcal C_de^2/q^{d-1}$ with
$\mathcal C_d=(4\pi)^{(d-1)/2}\Gamma(\tfrac{d-1}2)$ \cite{Schlesier2020},
so that $\mathcal C_3=4\pi$ and $\mathcal C_2=2\pi$. For $d=2$ the model is the
strictly two-dimensional electron gas
\cite{RajagopalKimball1977,LoosGill2DPRB}. Noninteger $d$ is used only
to follow how the poles move with the dimension
(Sec.~\ref{sec:resonances}). With $n_e$ the electron density, $a_B$ the
Bohr radius and $V_d=\pi^{d/2}/\Gamma(\tfrac d2+1)$ the volume of the unit
ball, the density parameter is defined by $n_eV_d(r_sa_B)^d=1$. With two
spin states, $n_e=2V_dk_F^d/(2\pi)^d$, so that
\begin{equation}
 k_Fa_B=\frac1{\alpha_dr_s},\qquad
 \alpha_d=\bigl[2^{d-1}\Gamma(\tfrac d2+1)^2\bigr]^{-1/d},\qquad
 \alpha_3=\alpha_L=\Bigl(\frac4{9\pi}\Bigr)^{1/3},\quad
 \alpha_2=\frac1{\sqrt2}.
 \label{eq:alpha-d}
\end{equation}
The coupling constant of the gas is $e^2/(\hbar v_F)=1/(k_Fa_B)=\alpha_dr_s$,
with $v_F$ the Fermi velocity.

With $m_e$ the electron mass and $e$ the
elementary charge, we use the units (called Fermi units below)
\begin{equation}
 E_0=\frac{\hbar^2k_F^2}{m_e}=2E_F,\qquad
 \omega_0=E_0/\hbar,\qquad v_0=\frac{\mathcal C_de^2}{k_F^{d-1}}.
 \label{eq:unit-scales}
\end{equation}
Here $E_F$ is the Fermi energy, $\bm q$ the momentum transfer with
magnitude $q$, and $\mathrm i\nu$, $\nu>0$, the imaginary frequency.
Wave vectors are measured in $k_F$, frequencies in $\omega_0$,
interactions $v$ and $W$ in $v_0$, responses $\Pi$ and $\chi$ in
$v_0^{-1}$ [all introduced with \eqref{eq:response-conventions} below],
and Green functions and self-energies in $E_0^{-1}$ and $E_0$. In these
units $v(q)=q^{1-d}$ and the free dispersion is $k^2/2$. At small momentum
transfer the excitation energy $\bm p\cdot\bm q+q^2/2$ of a pair with
hole momentum $\bm p$ at the Fermi surface is of order $q$, since the
Fermi velocity is $1$ in these units. We therefore use the scaled
frequency $z=\nu/q$. Diagrams are evaluated in these units, and the energies are
converted at the end to Rydbergs per electron, with
$E_0=2(k_Fa_B)^2~{\rm Ry}=2/(\alpha_dr_s)^2~{\rm Ry}$.

The correlation energy is the integral of the interaction energy over
the coupling strength. The energy expressions of Sec.~\ref{sec:linear-diagrams} differ in
how they weight this integral. We therefore introduce the
interaction-strength parameter $0\le\lambda\le1$ at fixed density through $\widehat H_\lambda=\widehat T+\lambda\widehat V$,
with $\widehat T$ the kinetic energy and $\widehat V$ the Coulomb interaction including the neutralizing
background. The bare and screened lines at coupling
$\lambda$ are denoted by $\lambda v$ and $W_\lambda$. The density
response, the dielectric function $\epsilon_\lambda$ and the screened
line are given by
\begin{equation}
 \chi_\lambda=\frac{\Pi_\lambda}{1-\lambda v\Pi_\lambda},\qquad
 \epsilon_\lambda=1-\lambda v\Pi_\lambda,\qquad
 W_\lambda=\lambda v+(\lambda v)\chi_\lambda(\lambda v).
 \label{eq:response-conventions}
\end{equation}
Here $\Pi$ is the proper (one-interaction-irreducible) response, written
$\Pi^*$ in Ref.~\cite[Chap.~4, Sec.~12]{FetterWalecka1971}. These
relations fix its sign and normalization. The full density response is
$\chi$, and $\Sigma$ denotes the proper self-energy.

Summing chains
of free polarization bubbles joined by Coulomb lines gives the RPA
screened interaction \cite{Hedin1965,Lindhard1954}, in which
$\Pi_\lambda=\Pi_0=\chi_0$ is the free response:
\begin{equation}
 \begin{aligned}
 W_\lambda&=\lambda v+\lambda v\Pi_0W_\lambda,\\
 W_\lambda(q,\mathrm i\nu)
 &=\frac{\lambda v(q)}{1-\lambda v(q)\Pi_0(q,\mathrm i\nu)}
 =\frac{\lambda}{q^{d-1}+\lambda q_{\rm TF}^{d-1}\mathcal L(q,\nu)}.
 \end{aligned}
 \label{eq:Dlambda}
\end{equation}
The last expression of \eqref{eq:Dlambda} uses
$\Pi_0=-q_{\rm TF}^{d-1}\mathcal L$ of Eq.~\eqref{eq:lindhard-convention}
below. With $k_{\rm TF}$ the Thomas--Fermi wave number and
$\mathcal A_d=2\pi^{d/2}/\Gamma(\tfrac d2)$ the area of the unit sphere,
the free response and the scaled Thomas--Fermi wave number are
\begin{equation}
 \Pi_0=\chi_0=-q_{\rm TF}^{d-1}\mathcal L(q,\nu),\qquad
 q_{\rm TF}=\frac{k_{\rm TF}}{k_F},\qquad
 q_{\rm TF}^{d-1}=\frac{\mathcal B_d}{k_Fa_B}=a_{\rm TF}r_s,\qquad
 \mathcal B_d=\frac{2\mathcal A_d\mathcal C_d}{(2\pi)^d},
 \label{eq:lindhard-convention}
\end{equation}
where $a_{\rm TF}=\mathcal B_d\alpha_d$.
In three
dimensions $\mathcal B_3=4/\pi$, $q_{\rm TF}^2=4/(\pi a_Bk_F)$ and
$a_{\rm TF}=4\alpha_L/\pi$. In two dimensions $\mathcal B_2=2$ and
$q_{\rm TF}=\sqrt2\,r_s$. The nonnegative Lindhard function $\mathcal L$ is
written with the real argument $\nu$ although $\Pi_0$ is evaluated at
$\mathrm i\nu$ (the same convention is used for $Y_D$ below). It
satisfies $\mathcal L(0,0)=1$. Its closed form is given in
Ref.~\cite{FetterWalecka1971} for $d=3$ and in Ref.~\cite{Stern1967} for
$d=2$.

We write the screening ratio and the \emph{long-wavelength Lindhard shape},
the small-$q$ limit of $\mathcal L$ at fixed $z=\nu/q$, as
\begin{align}
 c(q,\nu)&=-v(q)\Pi_0(q,\mathrm i\nu)=\frac{q_{\rm TF}^{d-1}\mathcal L(q,\nu)}{q^{d-1}},
 \label{eq:cdef}\\
 h_d(z)&=\lim_{q\to0}\mathcal L(q,qz)=\Bigl\langle\frac{\mu^2}{\mu^2+z^2}\Bigr\rangle_{S^{d-1}}
 =1-{}_2F_1\bigl(1,\tfrac12;\tfrac d2;-z^{-2}\bigr),
 \label{eq:hz}
\end{align}
so that $h_d$ gives the dependence of the long-wavelength response on the
scaled frequency. Here $\mu=\hat{\bm p}\cdot\hat{\bm q}$ and the average
runs over the unit sphere $S^{d-1}$ of hole directions $\hat{\bm p}$. The
limit holds because at small $q$ the occupation difference confines the
hole momentum to a shell of width $q|\mu|$ at the Fermi surface, where
the excitation energy is $q\mu$ (SI Sec.~\ref{S-dsec:lindhard}). In
particular
\begin{equation}
 h(z)\equiv h_3(z)=1-z\arctan\frac1z,\qquad
 h_2(z)=1-\frac{z}{\sqrt{1+z^2}},\qquad
 h_4(z)=\bigl(\sqrt{1+z^2}-z\bigr)^2,
 \label{eq:hd-examples}
\end{equation}
for $z\ge0$, with $h_d$ even in $z$, $h_d(0)=1$ and $h_d(z)\simeq1/(dz^2)$
for $z\to\infty$.
In three dimensions the limit
also follows from the integral representation of $\mathcal L(q,qz)$ in
SI Eq.~\eqref{S-app-app-eq:lindhard-shell}, which gives
$h(z)=\int_0^1a^2\,\mathrm da/(a^2+z^2)$ with the error bound of SI
Eq.~\eqref{S-app-eq:highdensity-u3}. For $d=3$ and $q\gg1$,
$\mathcal L\simeq4/[3q^2(1+4\nu^2/q^4)]$
(Sec.~\ref{sec:screening-functions}).

\subsection{Diagrams linear in one screened line}\label{sec:linear-diagrams}
We write the formulas of this and the next two subsections for $d=3$,
as needed for the kite (Secs.~\ref{sec:kite}--\ref{sec:se-cancellation}). In $d$ dimensions
$\mathrm d^3q$ is replaced by $\mathrm d^dq$ and the angular factor
$4\pi$ by $\mathcal A_d$. The $d$-dimensional forms are given where
they are used (Secs.~\ref{sec:ring-d}, \ref{sec:Pi1-d} and \ref{sec:twod}).

Consider a connected closed diagram $D$ built from free propagators
$G_0$ and bare Coulomb lines, and select one interaction line. Replacing
it by a line $f$ and carrying out all other internal integrations gives
an energy per electron, in Ry, that is linear in $f$:
\begin{equation}
 \mathcal J_D[f]=\int\mathrm d^3q\int_0^\infty\frac{\mathrm d\nu}{\pi}\,
 f(q,\mathrm i\nu)\,Y_D(q,\nu),
 \label{eq:general-functional}
\end{equation}
where the bare diagram is $\mathcal J_D[v]$. Cutting the selected line
leaves a polarization insertion $\Pi_D(q,\mathrm i\nu)$ with two
density vertices. We say that
$\Pi_D$ is \emph{closed} with the line $f$. For second-order exchange,
$\Pi_D$ is the first-order exchange vertex correction
$\Pi^{(1)}_{\rm vtx}$ (Sec.~\ref{sec:se-example}). The integrand is
$Y_D=-12\,\Pi_D/(S_D\pi^3q_{\rm TF}^4)=r_s^{\,n-2}\hat Y_D$, with $n$ the
number of bare interactions of $D$. Here $1/S_D$
is the closed-diagram coefficient ($1/4$ for the kite,
Sec.~\ref{sec:diagram-coefficients}). When $\Pi_D$ is the sum of $m$
insertions that close to the same diagram, $1/S_D$ is that coefficient
divided by $m$, so that $S_D=4$ also for the self-energy pair of
Sec.~\ref{sec:se-cancellation}. The insertion $\Pi_D$ carries the fermion signs
and spin sums and is normalized like the responses of SI
Sec.~\ref{S-app-sec:response}. The factor $\hat Y_D$ is independent of density
in Fermi units. For the
kite this form is derived in SI Secs.~\ref{S-app-sec:units-kite} and
\ref{S-app-sec:response}
[Eqs.~\eqref{S-app-eq:GM}--\eqref{S-app-eq:A0-dimensional} and
\eqref{S-app-eq:Pix-norm}], and the end
of SI Sec.~\ref{S-app-sec:units-kite} states the conditions under which the same steps hold
for any $D$. The factor
$r_s^{\,n-2}$ arises because each bare interaction carries
$v_0k_F^d/E_0=\mathcal C_d\alpha_dr_s$ relative to the kinetic scale $E_0$
($4\pi\alpha_Lr_s$ for $d=3$), while $E_0$ itself is proportional to
$r_s^{-2}$ in Rydbergs, in every dimension.

A single diagram with a screened line has no unique energy, because
the chain diagrams that build the screened line can be weighted in
different ways. We use three \emph{energy expressions}, which differ in the
weights they assign to the diagrams that contain the screened line. The density expansions of the full-coupling energy and of the
coupling averages below have the same exponents, apart from a term of
order $r_s^{-2}$ at low density. Their
coefficients are related by the Mellin multipliers of
Sec.~\ref{sec:mellin-convolution}. The full-coupling energy is
$F_D(r_s)=\mathcal J_D[W_{\lambda=1}(r_s)]$, where the argument $r_s$ of
a line records the density at which it is evaluated. Coupling-constant averages
take the form
\begin{equation}
 e_D(r_s)=\int_0^1w_D(\lambda)\,F_D(\lambda r_s)\,\mathrm d\lambda,
 \label{eq:general-average}
\end{equation}
because coupling and density enter the screening denominator only
through $\lambda q_{\rm TF}^2$, so that
$W_\lambda(r_s)=\lambda W_{\lambda=1}(\lambda r_s)$ by \eqref{eq:Dlambda}
and $q_{\rm TF}^2=a_{\rm TF}r_s$. The weight $w_D$
depends on the powers of $\lambda$ carried by the bare lines and on which
interactions are \emph{marked}, that is, differentiated with
respect to $\lambda$ in the coupling-constant integration. For $n\ne2$
these powers include the factor $\lambda^{2-n}$ from $Y_D=r_s^{\,n-2}\hat Y_D$ at
density $\lambda r_s$. For the
kite, the functional $\mathcal J_D$ is written $\mathcal K$ and the
full-coupling energy $F_D$ is written $F$ in
Sec.~\ref{sec:diagram-coefficients}. Since $\mathcal K$ is linear in the line
and $n=2$, we have
$\mathcal K[W_\lambda(r_s)]=\lambda F(\lambda r_s)$. The energy
$e_{\rm kite}=2\int_0^1\mathcal K[W_\lambda]\,\mathrm d\lambda$ of
Eq.~\eqref{eq:energy-definitions} is therefore \eqref{eq:general-average} with
$w_D=2\lambda$, which satisfies $\int_0^1w_D\,\mathrm d\lambda=1$. The
third energy expression sums closed diagrams with polarization chains on
the selected line and is treated in Sec.~\ref{sec:diagram-coefficients}. Diagrams with two or more screened
lines are outside the present treatment, except for the order counting of
Sec.~\ref{sec:completeness}.

\subsection{Spectral representation}\label{sec:spectral-rep}
Let $W^R_\lambda$ be the retarded
screened interaction, obtained from $W_\lambda$ of
\eqref{eq:Dlambda} by the continuation $\mathrm i\nu\to\omega+\mathrm i0^+$
(likewise $\epsilon^R_\lambda$ and $\Pi_0^R$ below from
$\epsilon_\lambda$ and $\Pi_0$ of \eqref{eq:response-conventions} and
\eqref{eq:lindhard-convention}). For $q>0$ and real $\omega>0$, its
spectral density (loss spectrum)
\begin{equation}
 B_{W,\lambda}(q,\omega)
 =-\frac1\pi\operatorname{Im}W_\lambda^R(q,\omega)\ge0,
 \label{eq:loss-spectrum}
\end{equation}
is nonnegative because it is proportional to the power absorbed from
a weak external potential, which is nonnegative in a stable ground
state (passivity). For the RPA line this follows directly from
$B_{W,\lambda}=\lambda^2v^2|\epsilon^R_\lambda|^{-2}
[-\pi^{-1}\operatorname{Im}\Pi_0^R(q,\omega)]$ and
$\operatorname{Im}\Pi_0^R(q,\omega)\le0$ for $\omega>0$, as the
closed form of the Lindhard function shows \cite{FetterWalecka1971}.
Where an undamped plasmon makes $\epsilon^R_\lambda$ vanish,
$B_{W,\lambda}$ has the delta-function weight
$\lambda v/\partial_\omega\epsilon^R_\lambda$. This weight is positive because
$\epsilon^R_\lambda$ increases with $\omega$ above the particle--hole
continuum. It increases there because the Kramers--Kronig relation and
$\operatorname{Im}\Pi_0^R\le0$ imply that $\operatorname{Re}\Pi_0^R$ decreases with $\omega$.

With the high-frequency limit $W_\lambda\to\lambda v$
[SI Eq.~\eqref{S-app-eq:herglotz}], the Kramers--Kronig relation
\cite{Nussenzveig1972} on
the imaginary axis reads
$W_\lambda(q,\mathrm i\nu)-\lambda v(q)=-\int_0^\infty
2\omega B_{W,\lambda}(q,\omega)/(\omega^2+\nu^2)\,\mathrm d\omega$.
The substitution $\omega=q\Lambda$ expresses $W_\lambda$ as the bare
line minus a superposition of Lorentzian lines, each equal to the bare Coulomb factor
times a Lorentzian in $z$ of width $\Lambda>0$ (SI
Sec.~\ref{S-app-sec:loss-spectrum} and figure~\ref{fig:screened-spectral}):
\begin{equation}
 \begin{aligned}
 W_\lambda(q,\mathrm i\nu)
 &=\lambda v(q)-\int_0^\infty a_\lambda(q,\Lambda)
                         b_\Lambda(q,\mathrm i\nu)\,\mathrm d\Lambda,\\
 b_\Lambda(q,\mathrm i\nu)
 &=\frac1{q^2}\frac{\Lambda^2}{\Lambda^2+(\nu/q)^2},\qquad
 a_\lambda(q,\Lambda)=\frac{2q^2}{\Lambda}B_{W,\lambda}(q,q\Lambda).
 \end{aligned}
 \label{eq:block-superposition}
\end{equation}
Here $\Lambda=\omega/q$ labels the spectral frequency $q\Lambda$. It
is the phase velocity of a spectral component in units of the Fermi
velocity, not a damping rate. At small $q$ the particle--hole
continuum has $\Lambda\lesssim1$, while the plasmon has
$\Lambda\simeq\omega_p/q\to\infty$, with $\omega_p$ the plasma frequency
of \eqref{eq:qstar} below. Each $b_\Lambda$ is the interaction
mediated by a single mode of frequency $q\Lambda$, normalized to the
bare line at $\nu=0$, and $b_\Lambda\to v$ as $\Lambda\to\infty$.
The
representation holds for any passive screened line with the same
high-frequency limit. Spectral structure beyond RPA
\cite{TakadaYasuhara2002,Takada2016} can therefore be included through its loss
spectrum. Unlike plasmon-pole
models \cite{HybertsenLouie1986,GodbyNeeds1989}, the representation integrates over all
$\Lambda$ and retains the complete spectrum.

\begin{figure}[ht]
\centering

\includegraphics[width=0.9\textwidth]{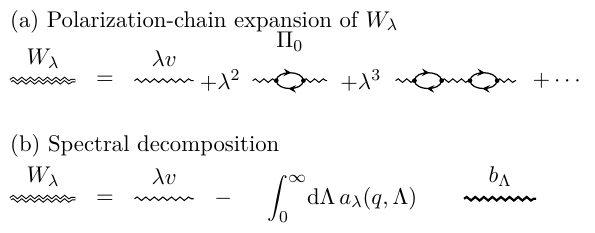}

\caption{The RPA screened interaction $W_\lambda$ at coupling $\lambda$ (double wavy line) as (a) a
chain of polarization bubbles and (b) the bare line $\lambda v$ (thin)
minus a superposition of Lorentzian lines $b_\Lambda$ (thick) with
weight $a_\lambda(q,\Lambda)\ge0$.}
\label{fig:screened-spectral}
\end{figure}

Inserting the spectral representation into the linear functional gives
\begin{equation}
 \mathcal J_D[W_\lambda]=\lambda\mathcal J_D[v]
 -\int_0^\infty
       \mathcal J_D[a_\lambda(q,\Lambda)\,b_\Lambda(q,\mathrm i\nu)]\,\mathrm d\Lambda.
 \label{eq:common-blocks}
\end{equation}
Here the functional acts on the product $a_\lambda(q,\Lambda)b_\Lambda(q,\mathrm i\nu)$
as a line, a function of $q$ and $\nu$ at fixed $\Lambda$, so that
$\mathcal J_D[a_\lambda(q,\Lambda)b_\Lambda(q,\mathrm i\nu)]
=\int_0^\infty a_\lambda(q,\Lambda)\,e_{D,\Lambda}(q)\,\mathrm dq$ with
the radial kernel $e_{D,\Lambda}$ of \eqref{eq:radial-kernel} below.
The interchange of the spectral and energy integrals is justified whenever
$\int\mathrm d^3q\,\mathrm d\nu\,v(q)|Y_D(q,\nu)|$ is finite
(SI Sec.~\ref{S-app-app:block-generality}). Since the weight
$a_\lambda(q,\Lambda)$ depends on $q$, the energy is not
$\mathcal J_D[b_\Lambda]$ times a spectral weight.

\subsection{Radial kernels and frequency integrals}\label{sec:radial-general}
For a Lorentzian line the frequency integral can be done first. The
energy per unit momentum transfer,
\begin{equation}
 \mathcal J_D[b_\Lambda]=\int_0^\infty e_{D,\Lambda}(q)\,\mathrm dq,\qquad
 e_{D,\Lambda}(q)=4\pi\int_0^\infty\frac{\mathrm d\nu}{\pi}
 \frac{\Lambda^2}{\Lambda^2+\nu^2/q^2}\,Y_D(q,\nu),
 \label{eq:radial-kernel}
\end{equation}
defines the radial kernel of $D$, with bare-line limit $e_{D,\infty}$
($\Lambda\to\infty$, where $b_\Lambda\to v$). Here $e_{D,\Lambda}(q)$,
an energy per unit $q$, is distinct from the coupling average
$e_D(r_s)$ of \eqref{eq:general-average}.
When $Y_D$ is built from electron--hole factors
$(\mathrm i\nu-\delta)^{-1}$ with signed excitation energies $\delta$,
the frequency integrals are elementary. The integrals below take the real part
of the product of these factors because
the frequency integral runs over $\nu>0$ only and the integrand at $-\nu$
is the complex conjugate. With $\Omega=q\Lambda$ the
Lorentzian is $\Omega^2/(\Omega^2+\nu^2)$, and for a double pole in $\nu$
\begin{equation}
 \int_0^\infty\frac{\mathrm d\nu}{\pi}
 \frac{\Omega^2}{\Omega^2+\nu^2}
 \operatorname{Re}(\mathrm i\nu-\delta)^{-2}
 =\frac{\Omega}{2(\Omega+|\delta|)^2},
 \qquad \Omega>0,\quad\delta\ne0,
 \label{eq:sigma-frequency-main}
\end{equation}
which follows by differentiating the single-pole integral in $\delta$
(SI Sec.~\ref{S-app-app:sigma-reduction}). SI
Sec.~\ref{S-app-app:Jkernel} gives the general closed form for two
factors. For a frequency-independent line the frequency
integral acts on the fermion denominators alone,
\begin{equation}
 \int_0^\infty\frac{\mathrm d\nu}{\pi}
 \operatorname{Re}\frac1{(\mathrm i\nu-a)(\mathrm i\nu-b)}
 =\frac{\operatorname{sgn}(ab)-1}{2(|a|+|b|)},
 \label{eq:Wdef}
\end{equation}
which is the limit of infinite width of the Lorentzian integral (SI
Sec.~\ref{S-app-app:Jkernel}). Here $a$ and $b$ are
the signed excitation energies of the two electron--hole factors, and
equal signs mean that both pairs make transitions in the same
direction.

The infinite-width limit vanishes when $a$ and $b$ have the same sign, whereas the
Lorentzian integral of the same pair is positive. In imaginary time each factor $(\mathrm i\nu-\delta)^{-1}$
is nonzero for only one sign of the time argument, fixed by the sign
of $\delta$, and the frequency integral of a product is the time
integral of one factor against the other at reversed time. With a
static line the two factors therefore give a nonzero integral only if
their excitation energies have opposite signs. A line
with its own frequency dependence inserts a time dependence between
them and removes this restriction. Transitions in the
same direction thus contribute to a dynamically screened diagram
but not to a statically screened one. This distinction fixes the
low-density exponents of Sec.~\ref{sec:lowdensity} and SI Sec.~\ref{S-app-sec:lowdensity}.

\subsection{Scaling functions}\label{sec:screening-functions}
Table~\ref{tab:screening} lists the scaling functions of the RPA line
as functions of the scaled momentum transfer $t=q/\varepsilon$, with the
density scale $\varepsilon$ of Sec.~\ref{sec:overview}. At high density $\varepsilon=q_{\rm TF}$.
At small momentum transfer and fixed $z$ the Lindhard function then tends
to $h_d(z)$ of \eqref{eq:hz}, so that $c\simeq h_d(z)/t^{d-1}$ with
$t=q/q_{\rm TF}$. The full-coupling line satisfies
$W_{\lambda=1}-v=-vc/(1+c)$, whose Thomas--Fermi tail
$\sum_{N\ge1}(-h_d)^Nt^{-N(d-1)}$ gives simple poles at $s=N(d-1)$
(Sec.~\ref{sec:mellin-convolution}). In three dimensions the first of them
comes from the $t^{-2}$ term $-h/t^2$ of this tail. For the coupling average of BRM
(Sec.~\ref{sec:diagram-coefficients}) the line enters as $2\overline W$,
with $\overline W=\int_0^1W_\lambda\,\mathrm d\lambda$.
With
$W_\lambda=\lambda v/(1+\lambda c)$, from \eqref{eq:Dlambda} and
\eqref{eq:cdef}, the integral gives [SI Eq.~\eqref{S-app-eq:highdensity-u1}]
\begin{equation}
 2\overline W-v=v[\rho(c)-1],\qquad
 \rho(c)=\frac{2[c-\log(1+c)]}{c^2},
 \label{eq:rho}
\end{equation}
with $\rho(c)\simeq1-\tfrac23c$ as $c\to0$.

The low-density analysis is restricted to $d=3$.
At low density the scale is
\begin{equation}
 \varepsilon=q_*=\Bigl(\frac{4q_{\rm TF}^2}{3}\Bigr)^{1/4},\qquad
 \omega_p=\Bigl(\frac{q_{\rm TF}^2}{3}\Bigr)^{1/2}=\frac{q_*^2}{2},
 \label{eq:qstar}
\end{equation}
which grows as $r_s^{1/4}$ in Fermi units, whereas $k_Fq_*$ decreases as
$r_s^{-3/4}$. The scale $q_*$ is the momentum transfer at which the
recoil energy $q^2/2$ equals the plasma frequency $\omega_p$. The plasma
frequency $(4\pi n_ee^2/m_e)^{1/2}$ is $(q_{\rm TF}^2/3)^{1/2}$ in units of
$\omega_0$.
Since at low density $\omega_p\gg1/2$,
the Fermi energy in these units, the relevant momentum
transfers are large compared with $k_F$. For $q\gg1$ the Fermi momenta
are negligible, every pair has the energy $q^2/2$, and the Lindhard
function is $\mathcal L(q,q^2u)\simeq4/[3q^2(1+4u^2)]$ with
$u=\nu/q^2$ [SI Eq.~\eqref{S-app-eq:Glargeq}]. This form has a single pole in
$\nu$, and the screening and bare terms balance at $q\sim q_*$. The RPA
line then also has a single pole, at the frequency
$q\Lambda_\lambda=\tfrac12(q^4+\lambda q_*^4)^{1/2}$, which
interpolates between $\sqrt\lambda\,\omega_p$ for $q\ll q_*$ and
$q^2/2$ for $q\gg q_*$. The line is therefore one Lorentzian, whose
weight $\theta_\lambda$ and scaled width $\Lambda_\lambda/q$ depend on
$q$ only through $q/q_*$:
\begin{equation}
 \begin{aligned}
 W_\lambda&=\lambda\bigl[v-\theta_\lambda(q)b_{\Lambda_\lambda(q)}+v\varsigma_\lambda\bigr],\\
 \theta_\lambda(q)&=\frac{\lambda q_*^4}{q^4+\lambda q_*^4},\qquad
 \Lambda_\lambda(q)=\frac q2\sqrt{1+\lambda(q_*/q)^4},\\
 \varsigma_\lambda
 &=O\!\left(q^{-2}
       \frac{\lambda\hat c}{(1+\lambda\hat c)^2}\right),\qquad
 \hat c=\frac{(q_*/q)^4}{1+4\nu^2/q^4}.
 \end{aligned}
 \label{eq:pole-dispersion}
\end{equation}
Here $b_{\Lambda_\lambda(q)}$ is the Lorentzian line of
\eqref{eq:block-superposition} with width $\Lambda_\lambda(q)$, $\hat c$
is the leading large-$q$ form of the screening ratio $c$ of
\eqref{eq:cdef} and $\varsigma_\lambda$ the remainder [SI Eqs.~\eqref{S-app-eq:pole-dispersion}
and \eqref{S-app-eq:lowdensity-u4}].
This remainder is bounded on the imaginary axis, so no undamped plasmon is
assumed. The corrections to the large-$q$ form of the Lindhard
function enter $\varsigma_\lambda$ at relative order $q^{-2}$, uniformly
on the imaginary axis [SI Eq.~\eqref{S-app-eq:Glargeq}].

\begin{table}[tb]
\centering\footnotesize
\caption{Scaling functions of the RPA line, with
$t=q/\varepsilon$ and $h_d(z)$ of \eqref{eq:hz} ($h=h_3$). The first row
holds in $d$ dimensions, the others for $d=3$. At low density the table
gives the weight and the scaled width of the single Lorentzian. The
scaling function of a diagram is then $\theta_\lambda$ times its
frequency factor evaluated at $\Lambda_\lambda/q$, e.g.\
$\theta_\lambda p_{\rm loc}$ for the kite, with the frequency factor
$p_{\rm loc}$ of \eqref{eq:largeq-block}. At high density the Thomas--Fermi tails give the
universal poles $s=N(d-1)$ ($s=2$ for $d=3$). At low density the large-width term
$q/\Lambda_\lambda\propto t^2$ gives the universal pole $s=-2$,
because the radial kernel $e_{D,\Lambda}$ of a diagram whose
$Y_D$ has the high-frequency tail $A_D(q)/\nu^2$
(Corollary~\ref{cor:exponents}) differs from $e_{D,\infty}$ by
$-2\pi A_D/(q\Lambda)$ at large width [SI
Eq.~\eqref{S-sm-slim-eq:large-width}].}
\label{tab:screening}
\setlength{\tabcolsep}{4pt}
\begin{tabular}{@{}llll@{}}
\toprule
Limit & $\varepsilon$ & Scaling function & Endpoint behavior\\
\midrule
$r_s\to0$, full coupling & $q_{\rm TF}$ & $-c/(1+c)$, $c=h_d/t^{d-1}$ & $\sum_{N\ge1}(-h_d)^Nt^{-N(d-1)}$, $t\gg1$\\
$r_s\to0$, BRM average & $q_{\rm TF}$ & $\rho(h/t^2)-1$ & $-2h/(3t^2)$, $t\gg1$\\
$r_s\to\infty$ & $q_*$ & $\theta_\lambda=\lambda/(t^4+\lambda)$, $\Lambda_\lambda/q=\tfrac12\sqrt{1+\lambda t^{-4}}$ & $q/\Lambda_\lambda\simeq2t^2/\sqrt\lambda$, $t\ll1$\\
\bottomrule
\end{tabular}
\end{table}

\subsection{Mellin convolutions and the endpoint--residue correspondence}\label{sec:mellin-convolution}
After the frequency integration, the energies considered here are built
from integrals of the form
\begin{equation}
 I(\varepsilon)=\int_0^\infty k(q)\,g(q/\varepsilon)\,\mathrm dq,
 \label{eq:mellin-convolution}
\end{equation}
with the radial kernel $k$ of the diagram, the scaling function $g$
of the line (at low density together with a frequency factor of the
diagram), the scaled momentum transfer $t=q/\varepsilon$, and the scale
$\varepsilon=q_{\rm TF}$ at high density or $\varepsilon=q_*$ at low
density (Sec.~\ref{sec:screening-functions}). In the density limits $k$ and $g$
depend in addition on $z$, which is integrated last. With the Mellin transforms
$\widetilde k(s)=\int_0^\infty q^{s-1}k(q)\,\mathrm dq$ and
$\widetilde g(s)$, defined in the same way, the Mellin--Parseval formula \cite[Sec.~3.1.3]{ParisKaminski} gives
\begin{equation}
 I(\varepsilon)=\frac1{2\pi\mathrm i}\int_{\gamma-\mathrm i\infty}^{\gamma+\mathrm i\infty}
 \widetilde k(1-s)\,\widetilde g(s)\,\varepsilon^{s}\,\mathrm ds
 \label{eq:mellin-parseval}
\end{equation}
on a vertical line ${\rm Re}\,s=\gamma$ in the strip where both
transforms converge. The argument $1-s$ follows from inserting
$g(t)=\frac1{2\pi\mathrm i}\int_{\gamma-\mathrm i\infty}^{\gamma+\mathrm i\infty}\widetilde g(s)\,t^{-s}\,\mathrm ds$ into
\eqref{eq:mellin-convolution}.

The asymptotic expansion of $I(\varepsilon)$ for small and for large
$\varepsilon$ follows from the \emph{endpoint--residue correspondence}. This correspondence
is the converse mapping between the singular expansion of a Mellin transform (its
poles with their principal parts) and the
asymptotic expansion of the function
\cite{ParisKaminski,BleisteinHandelsman,FlajoletGourdonDumas1995}. The
approach follows Friot, Greynat and de Rafael
\cite{FriotGreynatDeRafael2005,AguilarDeRafaelGreynat2008}, who read
expansions of relativistic Feynman integrals in a ratio of mass scales
off the poles of a Mellin transform obtained from a Mellin--Barnes
representation. That representation writes a power of a sum as a contour integral over
gamma functions. Here the energy is already a Mellin convolution, so no
such representation is needed. At high density the small- and
large-momentum terms play the roles of the regions of momentum transfer
of order $k_{\rm TF}$ and $k_F$, respectively, in the expansion by regions
\cite{BenekeSmirnov1998,Smirnov2002}. The usual contour argument
requires the analytic continuation of the transforms and their decay
along vertical lines. L\'opez showed that the same expansion, with a
bound on the remainder, follows from the endpoint expansions of the two
factors and the convergence of the integral alone \cite{Lopez2008}.
We use the correspondence in this form. Its conditions therefore concern
only the endpoint behavior of $k$ and $g$ and the convergence of
\eqref{eq:mellin-convolution}. SI Sec.~\ref{S-sm-slim-sec:prop1} states
these conditions in full and proves the correspondence as SI
Proposition~\ref{S-prop:endpoint-residue}. The proof extends that of
Ref.~\cite{Lopez2008} to endpoint terms with a factor $\log^m(1/q)$.

The correspondence holds when \eqref{eq:mellin-convolution} converges
absolutely, $\widetilde k(1-s)$ and $\widetilde g(s)$ have a common strip
of absolute convergence, $k$ admits finite endpoint expansions in terms
of the form $x^{a}\log^{m}x$ at $0$ and at $\infty$, and $g$ admits one
in terms of the form $x^{a}$ at the endpoint that generates the expansion
($t\to\infty$ for $\varepsilon\to0$ and $t\to0$ for
$\varepsilon\to\infty$), with remainders of strictly lower order and
bounded as in SI Proposition~\ref{S-prop:endpoint-residue}, the bounds
being integrable in a parameter such as $z$ that is integrated after
the residues are taken. A small-$q$ term $q^{\alpha}\log^m(1/q)$
of $k$ gives $\widetilde k(1-s)$ a pole of order $m+1$ at $s=1+\alpha$,
a large-$q$ term $q^{-\alpha'}$ a simple pole at $s=1-\alpha'$. A
large-$t$ term $t^{-\beta}$ of $g$ gives $\widetilde g$ a simple pole at
$s=\beta$, a small-$t$ term $t^{\beta'}$ one at $s=-\beta'$. As
$\varepsilon\to0$, $I(\varepsilon)$ is minus the sum of the residues of
$\widetilde k(1-s)\widetilde g(s)\varepsilon^s$ at the poles to the
right of the strip, in increasing order, and as $\varepsilon\to\infty$
the sum of the residues at the poles to its left, in decreasing order.
A pole of total order $N$ at $s_0$ contributes $\varepsilon^{s_0}$
times a polynomial of degree $N-1$ in $\log\varepsilon$. The remainder
after the poles up to the line ${\rm Re}\,s=\sigma$ is bounded by
$\varepsilon^{\sigma}$ times a power of $1+|\log\varepsilon|$, and by a
constant times $\varepsilon^{\sigma}$ if the endpoint bounds extend
beyond that line in the direction of the expansion with no further pole.
SI Sec.~\ref{S-sm-slim-sec:residues} writes out the residue at a pole of
higher order, a formula used repeatedly below.

For the first-order exchange diagrams the remainder is also bounded
explicitly (SI Secs.~\ref{S-app-sec:highdensity},
\ref{S-app-sec:dispersion} and \ref{S-sm-sec:lowdensity-remainders}).
At low density the width $\Lambda_\lambda(q)$ of the single Lorentzian
that represents the screened line (Sec.~\ref{sec:screening-functions})
enters the kernel. The energy therefore has the form \eqref{eq:mellin-convolution}
only after the large-momentum kernel is separated. The low-density
statements proved here rest on these direct bounds.

By the endpoint--residue correspondence, each term of a density
expansion is the residue of one pole, generated by one endpoint term of
the kernel or of the scaling function, or by endpoint terms of both
when their poles coincide. The leading exponent therefore
follows from a comparison of pole positions before any coefficient is
computed. Corollaries~\ref{cor:exponents} and \ref{cor:logs} below
follow from the correspondence by such comparisons.
Corollary~\ref{cor:exponents} makes the comparison for the
RPA scaling functions of Sec.~\ref{sec:screening-functions}.

For a line $f$ with
relative change $g=(f-v)/v$, inserting $\mathrm d^3q=4\pi q^2\,\mathrm dq$
and $\mathrm d\nu=q\,\mathrm dz$ into \eqref{eq:general-functional} gives the energy change in terms of a radial kernel,
\begin{equation}
 \mathcal J_D[f]-\mathcal J_D[v]=\int_0^\infty\frac{\mathrm dz}{\pi}\int_0^\infty k_D(q;z)\,g(q,\mathrm iqz)\,\mathrm dq,
 \qquad k_D(q;z)=4\pi q\,Y_D(q,qz).
 \label{eq:kD}
\end{equation}
In the density limits $g(q,\mathrm iqz)$ depends on
$t=q/\varepsilon$ and $z$ alone (table~\ref{tab:screening}), and at high density it is the scaling function of
\eqref{eq:mellin-convolution}.
At high density the kernel $k$ of
\eqref{eq:mellin-convolution} is this $z$-resolved kernel, whose
Lorentzian average
$\int_0^\infty(\mathrm dz/\pi)\,\Lambda^2k_D(q;z)/(\Lambda^2+z^2)$ is
$e_{D,\Lambda}(q)$ of \eqref{eq:radial-kernel}. At low density $k$ is
$e_{D,\Lambda}(q)$ at fixed scaled width $w_\Lambda=\Lambda/q$, since
$\Lambda_\lambda/q$ depends on $q$ only through $q/q_*$
(Sec.~\ref{sec:screening-functions}). The exponents $\alpha_0$, $m$ and
$\alpha$ in Corollaries~\ref{cor:exponents} and \ref{cor:logs} refer to
these kernels, not to the bare-line kernel $e_{D,\infty}$. For the kite,
with $A_0$ and $X$ of \eqref{eq:KD} and \eqref{eq:Xdef} in
Sec.~\ref{sec:diagram-coefficients},
$k_D(q;z)=4\pi A_0\,qX(q,qz)\simeq\kappa_D(z)\,q\log(1/q)$ with
$\kappa_D=4\pi A_0a_x$ by \eqref{eq:Az}, and $e_{D,\Lambda}\propto q^{-2}$
at fixed $w_\Lambda$ by \eqref{eq:largeq-block}. The bare-line kernel
$e_{D,\infty}$ has no term $q\log(1/q)$, because
$\int_0^\infty a_x\,\mathrm dz=0$ [$\kappa(\Lambda)$, the coefficient of
$q\log(1/q)$ in $e_{D,\Lambda}$ for the kite, vanishes at
$\Lambda=\infty$, SI Eq.~\eqref{S-app-eq:kappa}]. The bare-line kernel decays as $q^{-4}$ at large momentum transfer. For $q>2$ the
regions in which both loop momenta lie in the same one of the two balls
$|\bm p|<1$ and $|\bm p+\bm q|<1$ give no contribution for the bare line by
\eqref{eq:Wdef}. The other two regions are of order $q^{-4}$ (SI
Sec.~\ref{S-app-sec:largeq-kernel}).

\begin{corollary}[Exponent rules for RPA screening]\label{cor:exponents}
Let the radial kernel of $D$ in $d$ dimensions behave as
$\kappa_0q^{\alpha_0}\log^m(1/q)$ for $q\to0$ and as
$\kappa_\infty q^{-\alpha}$ for $q\to\infty$, let the integrand $Y_D$
have the high-frequency tail $Y_D\simeq A_D(q)/\nu^2$, and assume that the
residues named below are nonzero. Take the RPA scaling
functions of Sec.~\ref{sec:screening-functions}, whose high-density
tail starts with $t^{-(d-1)}$ and whose low-density expansion for $d=3$
proceeds in powers of $t^2$, starting with the large-width term. With the overall factor
$r_s^{\,n-2}$ of Sec.~\ref{sec:linear-diagrams} multiplying all orders stated below,
the following holds. Here and in Corollary~\ref{cor:logs}, $W$ denotes
the screened line at the coupling of the energy expression,
$W_{\lambda=1}$ unless stated otherwise.
\begin{enumerate}
\item At high density the screening correction
$\mathcal J_D[W]-\mathcal J_D[v]$ scales as
$q_{\rm TF}^{\min(1+\alpha_0,\,d-1)}$, up to powers of $\log q_{\rm TF}$. For
$d=3$ and $\alpha_0=1$ the two poles coincide at $s=2$ and the correction
is $q_{\rm TF}^2\log^{m+1}q_{\rm TF}\propto r_s\log^{m+1}r_s$.
\item At low density, for $d=3$, $\mathcal J_D[W]$ scales as
$q_*^{-\min(\alpha-1,\,2)}$, with $\log q_*$ when $\alpha=3$. For
$\alpha>3$ the leading term comes from the screening pole at $s=-2$,
\begin{equation}
 \mathcal J_D[W_{\lambda=1}]\simeq\frac{4\pi}{q_*^2}\int_0^\infty A_D(q)\,\mathrm dq,
 \label{eq:fsum-pole}
\end{equation}
which depends on the diagram only through its high-frequency tail,
provided $\int_0^\infty A_D\,\mathrm dq$ converges and the remainder of
the expansion of the radial kernel in $1/\Lambda$ is integrable uniformly
in $q_*$ (SI Sec.~\ref{S-sm-slim-sec:largewidth}).
\end{enumerate}
\end{corollary}

The corollary follows from the endpoint--residue correspondence (SI
Proposition~\ref{S-prop:endpoint-residue}) by
comparing the positions of the first poles of the kernel and of the
scaling function on either side of the strip (proof in SI
Sec.~\ref{S-sm-slim-sec:corollaries}). The corollary assumes the
dynamically screened RPA line. A static line has infinite width, so its
expansion has no term of order $q/\Lambda$. The low-density energy of the kite with a static line is at most of order $r_s^{-3/4}$ (SI Sec.~\ref{S-app-sec:lowdensity}).

Corollary~\ref{cor:logs} applies the first statement of
Corollary~\ref{cor:exponents} to a kernel with the small-$q$
endpoint term $q\log^m(1/q)$, that is, $\alpha_0=1$. With the overall
factor $r_s^{\,n-2}$ it turns that statement into a counting rule for the powers of
$\log r_s$.

\begin{corollary}[Powers of $\log r_s$ at high density in $d$ dimensions]\label{cor:logs}
Let a diagram $D$ with $n$ bare interactions have a $z$-resolved radial
kernel $k_D(q;z)$ of Eq.~\eqref{eq:kD} behaving as
$r_s^{\,n-2}\kappa_D(z)\,q\log^m(1/q)$ at small momentum transfer, with
$\kappa_D$ independent of the density, and let the scaling function
have the Thomas--Fermi tail $\sum_{N\ge1}\gamma_N(z)\,t^{-N(d-1)}$
[$\gamma_N=(-h_d)^N$ at full coupling]. Then the kernel pole at $s=2$
contributes a term
$r_s^{\,n-2+2/(d-1)}\log^mr_s$ to
$\mathcal J_D[W]-\mathcal J_D[v]$ if $2/(d-1)$ is not an integer. This
difference is the single-chain resummation, the sum of all diagrams with
a chain of free bubbles inserted into the selected line, each with the
coefficient $1/S_D$ of $D$. At the
\emph{resonant dimensions} $d=1+2/N_*$, $N_*=1,2,3,\dots$, the kernel pole coincides
with the screening pole $N=N_*$ and contributes
$r_s^{\,n-2+N_*}\log^{m+1}r_s$ with a coefficient proportional to
$\int_0^\infty\kappa_D(z)\gamma_{N_*}(z)\,\mathrm dz/\pi$. The screening
poles $N(d-1)<2$ contribute $r_s^{\,n-2+N}$ without $\log r_s$. In
three dimensions $N_*=1$, the term is $r_s^{\,n-1}\log^{m+1}r_s$, and
$\gamma_1=-2h/3$ for the coupling average and $-h$ at full coupling
(table~\ref{tab:screening}). For a set of diagrams whose factors $\log(1/q)$ cancel, such as
vertex and self-energy corrections related by the Ward identity, the
corollary applies to the summed kernel, and the power of $\log r_s$ is
set by the factors $\log(1/q)$ that survive in the sum.
\end{corollary}

The proof is given in SI Sec.~\ref{S-sm-slim-sec:corollaries}.
Corollary~\ref{cor:logs} determines the leading power of $\log r_s$ in the screening
correction of a diagram, but it does not classify all contributions at
the same order of the full high-density expansion.

Both kinds of universal pole of the
scaling function are tied to conservation laws. The order of a
high-density pole at a resonance depends on the Ward identity, which
fixes the endpoint of the kernel that the pole meets. The low-density pole is
constrained by the $f$-sum rule. Section~\ref{sec:se-example} shows how each
law enters the correlation energy.

The high-density poles at $s=N(d-1)$ come from the Thomas--Fermi
tail $c\simeq q_{\rm TF}^{d-1}h_d(z)/q^{d-1}$ of the long-range screening,
with $c$ and $h_d$ of \eqref{eq:cdef} and \eqref{eq:hz}
(Sec.~\ref{sec:screening-functions}). The kernel of any polarization
insertion has the small-$q$ endpoint $k\propto q$. At fixed $z$ the
electron--hole factors stay finite as $q\to0$, since both the shell
volume and the excitation energies are of order $q$. The measure
$\mathrm d\nu=q\,\mathrm dz$ supplies one factor $q$, and the Coulomb line $q^{1-d}$
compensates the factor $q^{d-1}$ of $\mathrm d^dq$.
The pole of this endpoint at $s=2$ coincides with a screening pole when $N(d-1)=2$.
This coincidence produces the terms of Corollary~\ref{cor:logs}, such as
the Carr--Maradudin part $0.01304\,r_s\log r_s$~Ry in three dimensions
(Sec.~\ref{sec:Pi1-d}) and the $r_s^2\log r_s$ term in two
(Sec.~\ref{sec:twod}). The ring kernel carries no Coulomb line and has
the endpoint $q^d$, whose pole at $s=d+1$ coincides with a screening pole
when $N(d-1)=d+1$. This coincidence gives the Gell-Mann--Brueckner term in three
dimensions and the Rajagopal--Kimball term in two
(Sec.~\ref{sec:ring-d}). The Ward identity of charge conservation
constrains the density and current vertices jointly through the inverse
propagator. For the first-order exchange polarization considered here, the
explicit vertex and self-energy expressions give equal and opposite terms with the Fermi-surface factor $\log(1/q)$, which cancel in the sum
(Sec.~\ref{sec:se-example}). At higher orders the Ward identity still
constrains vertex and self-energy jointly, but the surviving power of
$\log r_s$ must be determined order by order.

In three dimensions the low-density pole at $s=-2$ comes from the
large-width term of the scaling function. In Corollary~\ref{cor:exponents} it competes with the pole at
$s=1-\alpha$ from the large-$q$ term $q^{-\alpha}$ of the kernel, and
the leading term comes from whichever lies closer to the strip. To trace the origin of
\eqref{eq:fsum-pole}, we use the low-density
line of Sec.~\ref{sec:screening-functions}, which for $q\ll q_*$ is
$W\simeq v\,\nu^2/(\nu^2+\omega_p^2)$, with $\omega_p=q_*^2/2$ the
plasma frequency [\eqref{eq:qstar}]. The line is small below $\omega_p$
and close to $v$ above it, so the diagram is probed only where
$Y_D\simeq A_D/\nu^2$.
The energy per unit $q$ is then
$4\pi\int_0^\infty\frac{\mathrm d\nu}{\pi}\frac{A_D(q)}{\nu^2+\omega_p^2}
=4\pi A_D(q)/q_*^2$. The residue at $s=-2$ is thus fixed by the
high-frequency tail $A_D$, the contribution of $D$ to the first frequency moment of the
density response. At large $\nu$ the exact response decays as $-n_eq^2/(m_e\nu^2)$ in
physical units, that is, as $-q_{\rm TF}^2q^2/(3\nu^2)$ in the units of
\eqref{eq:unit-scales}. The particle number fixes this coefficient
for every interaction strength [the $f$-sum rule, \eqref{eq:fsum-exact}
below]. The proper response $\Pi$ has the same tail,
since $\chi-\Pi=O(\nu^{-4})$. In the sum of all corrections to $\Pi$ of a
given order the $\nu^{-2}$ terms should therefore cancel.

Under an additional remainder condition, this cancellation of the tails
also removes the residue at $s=-2$ (SI Lemma~\ref{S-lem:fsum}, proved
in SI Sec.~\ref{S-sm-slim-sec:lemma}).
For a set $\mathcal S$ of diagrams closed with the same screened line and the
same coupling weight, suppose that the summed integrand
$Y_{\mathcal S}=\sum_{D\in\mathcal S}Y_D$ has the high-frequency tail
$Y_{\mathcal S}\simeq A_{\mathcal S}(q)/\nu^2$. If $A_{\mathcal S}(q)=0$ for every $q$ and the summed
radial kernel satisfies the remainder condition of
Corollary~\ref{cor:exponents}, the pole at $s=-2$ is absent from
$\sum_{D\in\mathcal S}\mathcal J_D[W]$. The condition $A_{\mathcal S}=0$ holds for the
first-order exchange corrections to the proper polarization at fixed
density. These conditions concern the summed integrand and radial kernel, and the individual diagrams
need not satisfy them. For the
kite alone $A_D(q)$ tends to the nonzero constant $-3/(4\pi^3)$ at
large $q$ (SI Sec.~\ref{S-sm-slim-sec:largewidth}), so
$\int_0^\infty A_D\,\mathrm dq$ diverges, consistent with the
$q^{-2}$ endpoint of its kernel (Sec.~\ref{sec:radial-kernels}). The
leading low-density term of the kite is therefore the kernel pole at
$s=-1$ ($\alpha=2$ in Corollary~\ref{cor:exponents}), not the universal
pole. The high-frequency tail vanishes only for the sum with the Fock self-energy corrections. The
$f$-sum-rule argument for complete sets of corrections of higher order
is given in SI Sec.~\ref{S-sm-slim-sec:lemma}, after the proof of SI
Lemma~\ref{S-lem:fsum}. At higher
order this argument assumes that the coupling expansion of the exact
fixed-density response may be taken term by term at large frequency.
Individual diagrams with bare internal lines need not have finite
tails. Poles generated by the kernel at large
momentum transfer obey no such general rule. They depend on the
occupation restrictions and Coulomb factors of the loop momenta, as in
the kite, where they come from transitions with both loop momenta in the same
Fermi ball (Sec.~\ref{sec:radial-kernels}).

The coupling averages \eqref{eq:general-average} are
Mellin convolutions in $r_s$. A term $\varphi\,r_s^p$ of $F_D$, with
constant $\varphi$, appears in $e_D$
multiplied by the \emph{Mellin multiplier} $M(p)=\int_0^1w_D(\lambda)\lambda^p\,\mathrm d\lambda$,
provided this moment converges ($p>-2$ for $w_D=2\lambda$). The same
multiplier applies to the highest power of $\log r_s$ in $\varphi\,r_s^p\log^mr_s$ (SI
Sec.~\ref{S-sm-slim-sec:moment} gives the lower powers). At low density the
couplings $\lambda\lesssim1/r_s$, for which $\lambda r_s\lesssim1$ and
$F_D$ is not in its low-density form, add a term of order $r_s^{-2}$ for
$w_D=2\lambda$. This term is smaller than all low-density terms obtained here.

\label{sec:cancellation-general}%
Diagrams that enter one physical quantity together, such as the vertex
and self-energy corrections of a response, are closed with the same
line and weight. If their radial kernels have equal and opposite
endpoint terms, the corresponding poles cancel in the sum for every
common scaling function, and the sum starts at the next pole. At first
order the endpoint terms cancel at both endpoints of the kernel
(Sec.~\ref{sec:se-example}).

\subsection{Ring series in \texorpdfstring{$d$}{d} dimensions}\label{sec:resonances}\label{sec:benchmarks}\label{sec:ring-d}\label{sec:GB}
The known $\log r_s$ terms of \eqref{eq:jellium-expansion} and
\eqref{eq:2d-expansion} are residues of coincident poles in the sense of
the endpoint--residue correspondence (Sec.~\ref{sec:mellin-convolution},
SI Proposition~\ref{S-prop:endpoint-residue}). They were originally obtained by
dedicated partial summations. Here each follows from two endpoint terms,
one of the kernel and one of the scaling function, and the dimension
enters only through the positions of the poles
(figure~\ref{fig:poles}). The statements of Secs.~\ref{sec:ring-d} and
\ref{sec:Pi1-d} concern kernels
without Fermi-surface factors $\log(1/q)$, either because the diagram has none,
as for the rings, or because the factors $\log(1/q)$ of the vertex and self-energy corrections
cancel, as the Ward identity requires
(Corollary~\ref{cor:logs}). The screening poles lie at the integer multiples of
$d-1$, and the kernel poles at $s=d+1$ for the rings and at $s=2$ for a
polarization insertion. As $d$ varies, the lines $s=N(d-1)$ cross the
kernel poles at isolated dimensions, and only there does a $\log r_s$ term
appear.

\begin{figure}[tb]
\centering
\includegraphics[width=0.56\textwidth]{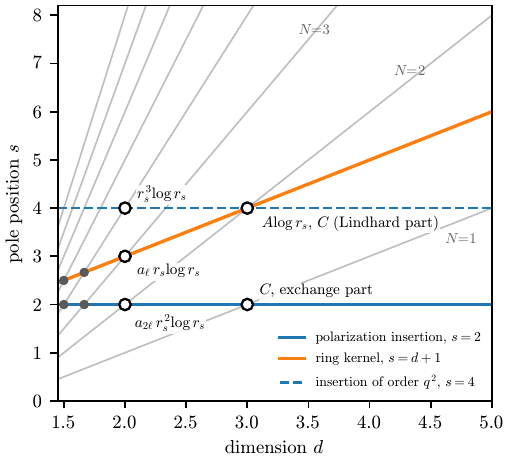}
\caption{Positions $s$ of the Mellin poles against the dimension $d$.
Gray lines: poles $s=N(d-1)$ of the Thomas--Fermi tail of the RPA
scaling function, labeled by $N$. Colored lines: kernel poles of a polarization
insertion ($s=2$, Secs.~\ref{sec:Pi1-d} and \ref{sec:twod}), of the ring
series ($s=d+1$, Sec.~\ref{sec:ring-d}) and of a polarization insertion of
order $q^2$, such as the $q^2$ term of the Lindhard function ($s=4$, dashed). Open circles mark the
coincidences that give the $\log r_s$ terms of
\eqref{eq:jellium-expansion} and \eqref{eq:2d-expansion}, the term
$a_{2\ell}\,r_s^2\log r_s$ obtained in Sec.~\ref{sec:twod}, and the term
$r_s^3\log r_s$ from the $q^2$ term of the two-dimensional Lindhard
function. Filled dots mark coincidences at noninteger $d$.}
\label{fig:poles}
\end{figure}

The RPA correlation energy depends nonlinearly on the line, through the scalar
function $\log(1+c)-c$ of the screening ratio $c$ of \eqref{eq:cdef}. It is
still an integral of a density-independent kernel against a function
of the scaled momentum transfer. This integral follows from the $d$-dimensional form of the adiabatic-connection fluctuation--dissipation formula
\eqref{eq:acfd} of Sec.~\ref{sec:se-example} with the RPA response
$\chi_\lambda=\Pi_0/\epsilon_\lambda$, whose coupling integral gives
$-[\log(1+c)-c]$. In the units of Sec.~\ref{sec:units}, with the volume
$(2\pi)^d/(2V_d)$ per electron, $\mathcal A_d=dV_d$ and
$\int_0^\infty\mathrm d\nu=\tfrac q2\int_{-\infty}^\infty\mathrm dz$
(SI Sec.~\ref{S-dsec:ring}), the RPA correlation energy per electron is
\begin{equation}
 e_{\rm RPA}=\frac{d}{8\pi}\,E_0\int_{-\infty}^{\infty}\mathrm dz
 \int_0^\infty\mathrm dq\,q^d\,\bigl[\log(1+c)-c\bigr],
 \label{eq:ring-energy}
\end{equation}
with $c$ of \eqref{eq:cdef} at $\nu=qz$ and $E_0$ in Ry. For $d=3$,
$E_0=32/(\pi^2q_{\rm TF}^4)$~Ry and \eqref{eq:ring-energy} is Eq.~(20)
of Ref.~\cite{Endo1997}. With $\nu=qz$, one factor $q$ in
\eqref{eq:ring-energy} is the Jacobian of the frequency integral. Together
with $\mathrm d^dq=\mathcal A_dq^{d-1}\mathrm dq$, this Jacobian makes the kernel at small momentum
transfer the phase-space factor $q^d$. This small-$q$ endpoint term gives
the Mellin transform $\widetilde k(1-s)$ a simple pole at $s=d+1$. With
$c\simeq h_d(z)/t^{d-1}$ and $t=q/q_{\rm TF}$ (table~\ref{tab:screening}),
the scaling function $\log(1+h_dt^{1-d})-h_dt^{1-d}$ has the Thomas--Fermi tail
$\sum_{N\ge2}(-1)^{N+1}h_d^Nt^{-N(d-1)}/N$, which gives simple poles at
$s=N(d-1)$. The kernel pole coincides with one of them when
\begin{equation}
 N_{\rm r}=\frac{d+1}{d-1}\in\{2,3,4,\dots\},\qquad
 d=\frac{N_{\rm r}+1}{N_{\rm r}-1}=3,\ 2,\ \tfrac53,\ \tfrac32,\ \dots .
 \label{eq:ring-resonance}
\end{equation}
At these resonant dimensions, the same set as in
Corollary~\ref{cor:logs}, the double pole gives, with
$\log q_{\rm TF}=\log r_s/(d-1)+O(1)$ from \eqref{eq:lindhard-convention}
(SI Sec.~\ref{S-dsec:ring}),
\begin{equation}
 e_{\rm RPA}\supset\frac{d}{8\pi}\,E_0\,q_{\rm TF}^{d+1}\,
 \frac{(-1)^{N_{\rm r}}}{N_{\rm r}(d-1)}\int_{-\infty}^{\infty}h_d(z)^{N_{\rm r}}\,\mathrm dz\;\log r_s,
 \qquad E_0q_{\rm TF}^{d+1}\propto r_s^{(3-d)/(d-1)}.
 \label{eq:ring-resonant}
\end{equation}
Here and below, $A\supset B$ means that the density expansion of $A$
contains the term $B$. With $h_d$ in place of the Lindhard function, the $q$ integral diverges at large $q$, and its resonant term has an integrand $\propto1/q$. The Lindhard function cuts it off at
$q\sim1$ and changes only the term without $\log r_s$. The coefficient
of $\log r_s$ therefore comes from the window $q_{\rm TF}\ll q\ll1$ (SI
Sec.~\ref{S-sm-slim-sec:residues}).

In three dimensions $N_{\rm r}=2$, the prefactor is $3/\pi^3$, and as
$r_s\to0$ [SI Eq.~\eqref{S-sm-slim-eq:GB}]
\begin{equation}
 e_{\rm RPA}=\frac{3}{\pi^3}\int_{-\infty}^{\infty}h(z)^2\,\mathrm dz\;\log r_s+O(1)
 =\frac{2}{\pi^2}(1-\log2)\log r_s+O(1),
 \label{eq:GB}
\end{equation}
which is the Gell-Mann--Brueckner coefficient $A=0.06218$ of
\eqref{eq:jellium-expansion} \cite{GellMannBrueckner}. In two dimensions
$N_{\rm r}=3$, $E_0q_{\rm TF}^3=8\sqrt2\,r_s$~Ry, and
\begin{equation}
 e_{\rm RPA}\supset-\frac{2\sqrt2}{3\pi}\int_{-\infty}^{\infty}h_2(z)^3\,\mathrm dz\;r_s\log r_s
 =-2\sqrt2\Bigl(\frac{10}{3\pi}-1\Bigr)r_s\log r_s,
 \label{eq:RK}
\end{equation}
with $\int_{-\infty}^\infty h_2^3\,\mathrm dz=10-3\pi$ (SI
Sec.~\ref{S-dsec:lindhard}). This is the Rajagopal--Kimball coefficient
$a_\ell$ of \eqref{eq:2d-expansion}
\cite{RajagopalKimball1977,LoosGill2DPRB}. In two dimensions the
screening pole $N=2$ lies at $s=2$, to the left of the kernel pole at
$s=3$. The leading term of the ring series is therefore the constant of the
second-order ring diagram, which is part of $a_0$. The $\log r_s$ term appears one
order later. The Gell-Mann--Brueckner and Rajagopal--Kimball terms both
contain a single $\log r_s$ because the ring
kernel has no Fermi-surface factor.

For $d>3$, $1<N_{\rm r}<2$. No screening pole
then coincides with the kernel pole, and the integral over
$t=q/q_{\rm TF}$ of the scaling function against $t^d$ converges at
both ends. The leading term is then a pure power
\cite{Schlesier2020},
\begin{equation}
 e_{\rm RPA}\simeq A_d\,r_s^{(3-d)/(d-1)},\qquad
 A_d=\frac{d\,\alpha_d^{N_{\rm r}-2}\mathcal B_d^{N_{\rm r}}}{4(d+1)\sin(\pi N_{\rm r})}
 \int_{-\infty}^{\infty}h_d(z)^{N_{\rm r}}\,\mathrm dz\ \ {\rm Ry},
 \label{eq:Ad}
\end{equation}
which follows from
$\int_0^\infty t^d[\log(1+ht^{1-d})-ht^{1-d}]\,\mathrm dt
=h^{N_{\rm r}}\pi/[(d-1)N_{\rm r}\sin(\pi N_{\rm r})]$ (SI Sec.~\ref{S-dsec:ring}).
We check the coefficient in three cases. (i) For $d=4$, where $N_{\rm r}=5/3$,
the long-wavelength Lindhard shape $h_4=(\sqrt{1+z^2}-z)^2$ gives $\int h_4^{5/3}\,\mathrm dz=60/91$ and
$A_4=-24\cdot2^{5/12}/(91\sqrt3)=-0.20325$~Ry. (ii) As $d\to3^+$,
$(d-3)A_d\to-2A$. Since
$r_s^{(3-d)/(d-1)}=1-\tfrac12(d-3)\log r_s+O((d-3)^2)$, this pole in the
dimension turns into the term $A\log r_s$ of \eqref{eq:GB}. The constant $-2A/(d-3)$
left over cancels against the opposite pole of the explicit $r_s^0$ term
of Eq.~(21) of Ref.~\cite{Schlesier2020}, so that the energy stays finite
as $d\to3$. The $\log r_s$ term of three dimensions is thus the
confluence of two powers of $r_s$. (iii) As $d\to\infty$,
$A_d/d\to-1/(2e\pi)=-0.05855$~Ry. Equation~\eqref{eq:Ad} is a simplified
form of Eq.~(21) of Ref.~\cite{Schlesier2020}: the series in that equation sums to the
factor $\pi/[N_{\rm r}\sin(\pi N_{\rm r})]$. The two forms agree numerically to machine
precision (SI Sec.~\ref{S-dsec:ring}). Table~\ref{tab:Ad} lists $A_d$ for
$d=4$--$9$. These values differ from those in Table~I of the same work,
in which a numerical error was later noted \cite[Ref.~68]{DucPham2023}.

\begin{table}[tb]
\centering\small
\caption{Coefficient $A_d$ of the leading term $A_dr_s^{(3-d)/(d-1)}$
of the RPA correlation energy for $d>3$, \eqref{eq:Ad}, in Ry. The third
column gives Table~I of Ref.~\cite{Schlesier2020}, converted from
Hartree to Ry, and the last column the ratio of the second column to the third.}
\label{tab:Ad}
\begin{tabular}{@{}cccc@{}}
\toprule
$d$ & $A_d$ (Ry), \eqref{eq:Ad} & Table~I of Ref.~\cite{Schlesier2020} (Ry) & ratio\\
\midrule
4 & $-0.20325$ & $-0.0392$ & 5.19\\
5 & $-0.17975$ & $-0.0570$ & 3.15\\
6 & $-0.19753$ & $-0.0782$ & 2.53\\
7 & $-0.22709$ & $-0.1018$ & 2.23\\
8 & $-0.26238$ & $-0.1276$ & 2.06\\
9 & $-0.30123$ & $-0.1546$ & 1.95\\
\bottomrule
\end{tabular}
\end{table}

\subsection{First-order exchange polarization in \texorpdfstring{$d$}{d} dimensions}\label{sec:Pi1-d}\label{sec:DuBois}
After the rings, the bare diagrams with the next-strongest small-$q$
divergence form a second series. It inserts the
first-order exchange polarization
$\Pi^{(1)}=\Pi^{(1)}_{\rm vtx}+\Pi^{(1)}_{\rm SE}$, the sum of the exchange vertex
correction and the Fock self-energy insertions of
figure~\ref{fig:diagrams}(a)--(c), once into the ring
chain \cite{DuBois1959,CarrMaradudin}. Replacing $\Pi_0$ by
$\Pi_0+\Pi^{(1)}$ in \eqref{eq:ring-energy} changes $c$ by
$-v\Pi^{(1)}$, and the first-order change of $\log(1+c)-c$ is
$-v\Pi^{(1)}[1/(1+c)-1]$. A ring of $j$ lines with one insertion
$\Pi^{(1)}$ receives $j$ times the coefficient $1/(2j)$ of the ring of
$j$ free bubbles, that is, $\frac12$. This is the closed-diagram
coefficient per insertion, so the expansion reproduces the closed-diagram series of
Refs.~\cite{DuBois1959,CarrMaradudin} (SI
Sec.~\ref{S-sm-slim-sec:residues}). The energy of the series is
\begin{equation}
 e_{\Pi^{(1)}}=\frac{d}{8\pi}\,E_0\int_{-\infty}^{\infty}\mathrm dz
 \int_0^\infty\mathrm dq\,q^d\,\bigl[-v\Pi^{(1)}(q,\mathrm iqz)\bigr]
 \Bigl[\frac1{1+c}-1\Bigr].
 \label{eq:Pi1-energy}
\end{equation}
The scaling function $1/(1+c)-1$ has the tail
$\sum_{N\ge1}(-h_d)^Nt^{-N(d-1)}$, the same as that of the full-coupling
line. This tail gives poles at $s=N(d-1)$.

By the Ward identity, shown explicitly
in \eqref{eq:ward} of Sec.~\ref{sec:se-example}, the Fermi-surface factors $\log(1/q)$
of the vertex and self-energy kernels cancel, and $\Pi^{(1)}$ tends to a finite limit
at fixed $z$,
\begin{equation}
 \Pi^{(1)}(q,\mathrm iqz)\to-q_{\rm TF}^{d-1}\,\alpha_dr_s\,R_d(z),\qquad
 R_d(z)=-\frac{\mathcal B_d}4\Bigl\langle\!\Bigl\langle
 \frac{\mu\mu'\operatorname{Re}\bigl[(b-b')^2\bigr]}{|\hat{\bm p}-\hat{\bm p}'|^{d-1}}
 \Bigr\rangle\!\Bigr\rangle,
 \label{eq:Rd}
\end{equation}
with $\alpha_dr_s=e^2/(\hbar v_F)$ the coupling constant of
\eqref{eq:alpha-d}, $b=(\mathrm iz-\mu)^{-1}$,
$b'=(\mathrm iz-\mu')^{-1}$, $\mu=\hat{\bm p}\cdot\hat{\bm q}$,
$\mu'=\hat{\bm p}'\cdot\hat{\bm q}$, and the double average running over
pairs of directions on the unit sphere. Equation~\eqref{eq:Rd} is the
long-wavelength limit of the first-order exchange polarization
(SI Sec.~\ref{S-dsec:response}). In this limit each occupation difference becomes a
factor $q\mu$ on the Fermi surface, and each electron--hole factor becomes
$b/q$. The factor $(b-b')^2=(\mu-\mu')^2b^2b'^2$ vanishes
quadratically for coincident directions and thereby removes the forward
singularity of the Coulomb factor $|\hat{\bm p}-\hat{\bm p}'|^{1-d}$.
The kernel of \eqref{eq:Pi1-energy} therefore has the endpoint term
$\propto q$ without a factor $\log(1/q)$. Its pole at $s=2$ coincides with the
screening pole $N=N_*$ when $N_*(d-1)=2$, that is, at the resonant dimensions
$d=1+2/N_*=3,2,\frac53,\dots$, and there, by SI Eq.~\eqref{S-sm-slim-eq:leading-residue} with $m=0$,
\begin{equation}
 e_{\Pi^{(1)}}\supset\frac{d}{8\pi}\,E_0\,\alpha_dr_s\,q_{\rm TF}^{d+1}\,
 \frac{(-1)^{N_*+1}}{d-1}\int_{-\infty}^\infty h_d(z)^{N_*}R_d(z)\,\mathrm dz\;\log r_s,
 \qquad E_0\,r_s\,q_{\rm TF}^{d+1}\propto r_s^{N_*}.
 \label{eq:Pi1-resonant}
\end{equation}

In three dimensions the average over the relative azimuth reduces
\eqref{eq:Rd} to a double integral over $\mu$ and $\mu'$. The result is
$R_3=R^{(1)}/(4\pi)$, with the function of DuBois
\cite{DuBois1959} and \cite[Eqs.~(40) and (41)]{Endo1997},
\begin{equation}
 R^{(1)}(z)=\frac{4}{1+z^2}\Bigl[(1+3z^2)-(2+3z^2)\,z\arctan\frac1z\Bigr],
 \label{eq:R1}
\end{equation}
so that $\Pi^{(1)}(q,\mathrm iqz)\to-(q_{\rm TF}^4/16)R^{(1)}(z)$.
(Eq.~(41) of Ref.~\cite{Endo1997} prints the argument of the arctangent
as $\mu$, to be read as $1/\mu$.) We checked the reduction
by direct quadrature of the reduced double integral. The quadrature agrees with the closed form to within
$2.5\times10^{-17}$ (SI Sec.~\ref{S-dsec:response}). With $N_*=1$, $q_{\rm TF}^2=4\alpha_Lr_s/\pi$
and $E_0q_{\rm TF}^4=32/\pi^2$~Ry, \eqref{eq:Pi1-resonant} gives the
$r_s\log r_s$ term of the series in Ry (SI Sec.~\ref{S-sm-slim-sec:residues}),
\begin{equation}
 e_{\Pi^{(1)}}=\frac{3\alpha_L}{2\pi^4}\int_{-\infty}^{\infty}h(z)R^{(1)}(z)\,\mathrm dz\;
 r_s\log r_s+O(r_s)=\frac{\alpha_L(\pi^2-12\log2)}{2\pi^3}\,r_s\log r_s+O(r_s),
 \label{eq:DuBois}
\end{equation}
that is, $0.01304\,r_s\log r_s$, the exchange part of the coefficient of Carr and Maradudin
\cite{CarrMaradudin} and \cite[Eq.~(5)]{Endo1999}. The residue integral agrees with the closed form of
\cite[Table~I]{LoosGill2011} and \cite[Table~1]{LoosGill} to 40 digits in multiprecision arithmetic
(SI Sec.~\ref{S-sm-sec:verification}). SI Sec.~\ref{S-sm-slim-sec:residues} bounds the remainder in \eqref{eq:DuBois} under the assumption that the correction to the limit \eqref{eq:Rd} vanishes as a positive power of $q$.

The rest of the jellium
$r_s\log r_s$ term,
$\alpha_L(\pi^2-6)/(12\pi^3)\,r_s\log r_s=0.00542\,r_s\log r_s$,
arises in the ring series. It comes from the cross term of
$h(z)$ and the next term of the small-$q$ expansion
$\mathcal L(q,qz)=h(z)+q^2g_3(z)+O(q^4)$, with
$g_3(z)=-1/[12(1+z^2)^2]$ from SI
Eq.~\eqref{S-app-app-eq:lindhard-shell}. To first order the change
$q_{\rm TF}^2g_3$ of $c$ adds $q^3\,q_{\rm TF}^2g_3\,[-c/(1+c)]$ to
the integrand of \eqref{eq:ring-energy}. This addition is the ring series with one
insertion of the polarization $-q_{\rm TF}^2q^2g_3$. Its
term $q^3q_{\rm TF}^2g_3\,h^2/t^4$ has the kernel $q^3$, with its pole
at $s=4$. This is the pole $s=2+m$ of a polarization insertion whose
small-$q$ form is $q^m$ times a function of $z$, here with $m=2$
(the dashed line of figure~\ref{fig:poles} and the counting of Sec.~\ref{sec:completeness}). It
coincides with the second screening pole $s=2(d-1)=4$, the resonance of
the ring series itself. The double pole gives $q_{\rm TF}^6\log q_{\rm TF}$ \cite[Eqs.~(33) and (36)]{Endo1997} and
\cite[Eq.~(3)]{Endo1999} (SI Sec.~\ref{S-sm-slim-sec:residues}).

Together the two parts give the
coefficient $C$ of \eqref{eq:jellium-expansion}. Loos and
Gill call the part from the ring series the RPA contribution and the part
from the first-order exchange polarization the second-order exchange contribution
\cite[Eqs.~(29) and (30)]{LoosGill}. If the vertex
correction were kept alone, its kernel would carry $\log(1/q)$, the pole
would be triple, and the term would be $r_s\log^2r_s$, as found for the
kite in Sec.~\ref{sec:highdensity}. The $r_s\log r_s$ term of jellium
thus carries a single $\log r_s$ because the factors $\log(1/q)$ of the vertex and self-energy
corrections cancel, as the Ward identity requires. This cancellation removes
the higher-order pole, in accordance with Corollary~\ref{cor:logs}.

\section{Results}\label{sec:results}
Section~\ref{sec:twod} derives the new result, the coefficient $a_{2\ell}$
of $r_s^2\log r_s$ in two dimensions, by the rule of
Sec.~\ref{sec:general}: the kernel pole of the polarization insertion at
$s=2$ meets the second screening pole. It also states the hypotheses
under which $a_{2\ell}$ is the complete coefficient.
Sections~\ref{sec:kite}--\ref{sec:se-cancellation} apply the rule to the
screened second-order exchange diagram, the kite, in three dimensions.
They give its energy expressions and the endpoints of its radial kernel,
its leading terms at high and low density, and the cancellation of these
terms by the Fock self-energy corrections. Section~\ref{sec:highdensity}
compares the high- and low-density terms with the fit of BRM.

\subsection{The \texorpdfstring{$r_s^2\log r_s$}{rs2 log rs} term of the two-dimensional gas}\label{sec:twod}
In two dimensions the kernel pole of the polarization insertion at $s=2$ coincides with the
second screening pole, $N_*=2$ (figure~\ref{fig:poles}).
Equation~\eqref{eq:Pi1-resonant} therefore predicts a term $r_s^2\log r_s$ from the
first-order exchange polarization. The coefficient of this term requires the
long-wavelength form $R_2(z)$ of $\Pi^{(1)}$ at every $z$, which we now obtain in
closed form.

\label{sec:R2}For $d=2$ the average in \eqref{eq:Rd} runs over pairs of angles on the
Fermi circle, $|\hat{\bm p}-\hat{\bm p}'|=2|\sin(\phi/2)|$ with $\phi$
the angle between the two directions, and $\mathcal B_2=2$. We expand the
three factors in Fourier series on the circle. Because the numerator vanishes at $\phi=0$, the Coulomb factor
$1/|2\sin(\phi/2)|$ enters only through the differences of its Fourier
coefficients, $(2/\pi)\sum_{j=1}^{|m|}1/(2j-1)$ for the harmonic $m$. The Fourier coefficients of
$b=(\mathrm iz-\cos\theta)^{-1}$ are geometric in
$r=\sqrt{1+z^2}-z$, and the sum over $m$ is elementary (SI
Sec.~\ref{S-dsec:R2}). With $s_z=\sqrt{1+z^2}$, $r=s_z-z$ and $a=z/s_z$,
\begin{equation}
 R_2(z)=\frac2\pi\Bigl[h_2^2+2a^2\,\frac{r\arctan r-r^2}{1+r^2}\Bigr]
 =\frac2\pi\Bigl[h_2+zh_2'+\frac{2a^2\,r\arctan r}{1+r^2}\Bigr],
 \label{eq:R2}
\end{equation}
where $h_2=1-a$ is the long-wavelength Lindhard shape of \eqref{eq:hd-examples} and
$h_2'=\mathrm dh_2/\mathrm dz$. Both forms hold for $z\ge0$, and $R_2$ is even in $z$. A quadrature of
the angular double integral \eqref{eq:Rd}, independent of the Fourier sum,
agrees with \eqref{eq:R2} to $10^{-17}$ at $z=1$ and to $10^{-13}$ at
$z=0.3$ and $3$ (SI Sec.~\ref{S-dsec:R2}). The function is
shown in figure~\ref{fig:R2}. It has the following properties.

(i) In the static limit, $R_2(0)=2/\pi$, so that
$\Pi_0+\Pi^{(1)}\to-q_{\rm TF}(1+\sqrt2\,r_s/\pi)$ for $q\to0$ at $\nu=0$.
This limit gives the Hartree--Fock compressibility of the two-dimensional gas.
The ratio of this compressibility to the free value, $1/(1-\sqrt2\,r_s/\pi)$ to this order, was given by Rajagopal \cite{Rajagopal1977}, by
Czachor, Holas, Sharma and Singwi (CHSS) \cite[Eq.~(5.29)]{CHSS1982}, by
Hameeuw, Brosens and Devreese \cite{Hameeuw2003EPJB} and recently, with
the static wave-number dependence, by Das Sarma \cite{DasSarma2026}.

(ii) In the high-frequency limit,
$R_2(z)=5/(12\pi z^4)-11/(20\pi z^6)+O(z^{-8})$. The absence of a $z^{-2}$
term expresses the $f$-sum rule: the free response already exhausts the first
frequency moment (SI Lemma~\ref{S-lem:fsum}). With $\Pi^{(1)}=-q_{\rm TF}\,
(r_s/\sqrt2)\,R_2$ and $z\to-\mathrm i\omega/q$, the leading term is
$-v\Pi^{(1)}=5r_s^2q^3/(12\pi\omega^4)$ at real frequency $\omega$ in
units of $\omega_0$, which is Eq.~(5.12) of CHSS \cite{CHSS1982} and
Eq.~(B.10) of Ref.~\cite{Hameeuw2003EPJB}. This term fixes the exchange
correction to the plasmon dispersion obtained by Rajagopal
\cite{Rajagopal1977}. The second term predicts
$+11r_s^2q^5/(20\pi\omega^6)$ in the high-frequency expansion, which these
works do not give.

(iii) In accordance with the Ward identity, the vertex and self-energy parts of
$\Pi^{(1)}$ carry $\mp q_{\rm TF}(r_s/\sqrt2)\,\pi^{-1}[h_2+zh_2']\log(1/q)$
at small $q$ and fixed $z$, the two-dimensional form of \eqref{eq:ward}
(SI Sec.~\ref{S-dsec:R2}). In both two and three dimensions the coefficient is
$q_{\rm TF}^{d-1}\alpha_dr_s\,\pi^{-1}[h_d+zh_d']$, and its static value
agrees with the separate terms in $\log(1/q)$ found by Das Sarma
\cite[Eqs.~(21) and (22)]{DasSarma2026}. In the Fourier evaluation, the first two terms
of the second form of \eqref{eq:R2} together give twice this coefficient. They come from
the self-energy terms $b^2+b'^2$ of $(b-b')^2$ through the harmonics
$m=\pm1$. The last term comes from the vertex term $-2bb'$ summed over
all harmonics (SI Sec.~\ref{S-dsec:R2}).

(iv) At finite wave number, the corrections to the limit
\eqref{eq:R2} at fixed $z$ are of order $q^2\log q$, without an $O(q)$
term, at the values of $z$ examined. Das Sarma proved this for the static response and obtained the
coefficient of $q^2\log q$ \cite{DasSarma2026}. Our evaluation of the
finite-$q$ response agrees with his static values within the accuracy of
our quadrature. It shows the same behavior at $z=0.3$ and is consistent
with this behavior at $z=1$ (SI Sec.~\ref{S-dsec:finiteq}).

The first-order exchange polarization of the two-dimensional gas has been
evaluated numerically for arbitrary wave number and frequency
\cite{CHSS1982,Hameeuw2003SSC,Hameeuw2003EPJB}. The limits (i) and (ii)
have been obtained in closed form \cite{Rajagopal1977,CHSS1982,Hameeuw2003EPJB}. The residue
at the coincident pole requires the closed form \eqref{eq:R2} at fixed $z$, which interpolates
between these limits.

\begin{figure}[tb]
\centering
\includegraphics[width=0.95\textwidth]{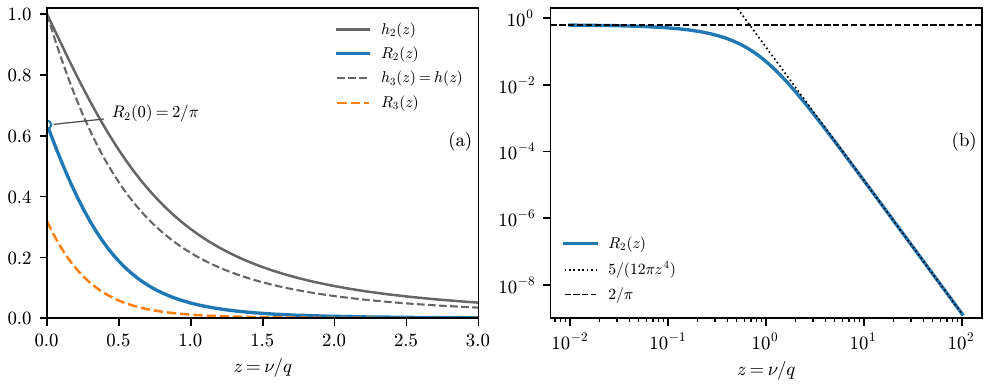}
\caption{(a) Long-wavelength form $R_2(z)$ of the first-order exchange
polarization of the two-dimensional gas, \eqref{eq:R2}, as a function of the
scaled frequency $z$, and the long-wavelength Lindhard shape $h_2(z)$, with
the three-dimensional counterparts $R_3=R^{(1)}/(4\pi)$ of \eqref{eq:R1} and $h_3=h$
(dashed). (b) $R_2(z)$ on log-log axes with its static value $2/\pi$
(Hartree--Fock compressibility) and its high-frequency tail
$5/(12\pi z^4)$.}
\label{fig:R2}
\end{figure}

\label{sec:a2l}For $d=2$ we have $N_*=2$, and the scaling function $1/(1+c)-1$ of
\eqref{eq:Pi1-energy} has the tail $-h_2/t+h_2^2/t^2-\dots$. The first
term gives the simple pole at $s=1$, an $O(r_s)$ contribution to the term with coefficient $a_1$
of \eqref{eq:2d-expansion}. The pole of the second term coincides with the kernel pole at
$s=2$. With $E_0\,\alpha_2r_s\,q_{\rm TF}^3=8r_s^2$~Ry,
\eqref{eq:Pi1-resonant} gives (SI Sec.~\ref{S-dsec:a2l})
\begin{equation}
 e_{\Pi^{(1)}}\supset a_{2\ell}\,r_s^2\log r_s,\qquad
 a_{2\ell}=-\frac2\pi\int_{-\infty}^\infty h_2(z)^2R_2(z)\,\mathrm dz .
 \label{eq:a2l-integral}
\end{equation}
The substitution $r=\sqrt{1+z^2}-z$ maps $z>0$ onto $0<r<1$, with
$h_2=2r^2/(1+r^2)$ and $a=(1-r^2)/(1+r^2)$. It reduces the integral to
three elementary integrals. One of them contains $\arctan r$ and
produces Catalan's constant (SI Sec.~\ref{S-dsec:a2l}). The result is
\begin{equation}
 a_{2\ell}=\frac7\pi+\frac{4\log2}\pi-\frac{160}{9\pi^2}-\frac{16G}{\pi^2}
 =-0.1754613739208\ldots\ {\rm Ry},
 \label{eq:a2l}
\end{equation}
that is, $-0.0877306869604$~Hartree, in agreement with multiprecision
quadrature of \eqref{eq:a2l-integral} to 100 digits. Catalan's constant
appears, as in the second-order exchange part of the constant $a_0$ of
\eqref{eq:2d-expansion}, through the Fermi-circle geometry. At
$r_s=1$ the term vanishes, and at $r_s=0.1$ it amounts to
$4.0\times10^{-3}$~Ry, about a tenth of the term $a_\ell r_s\log r_s$.

For a fully spin-polarized gas the same steps, with spin-resolved
$h_2$ and $R_2$, give $a_{2\ell}/8$, the analog of the factor
$1/(4\sqrt2)$ of the Rajagopal--Kimball term
\cite{RajagopalKimball1977,ChesiGiuliani2007}. At partial polarization
$\zeta$ the coefficient depends weakly on $\zeta$, but it does not tend
to $a_{2\ell}/8$ as $\zeta\to1$. The first-order exchange polarization of the minority
spins grows as $1/k_{F\downarrow}$, while their density of states stays
constant in two dimensions. The expansion is therefore uniform only for
$r_s\ll(1-\zeta)^{1/2}$ (SI Sec.~\ref{S-dsec:spin}).

\label{sec:completeness}Equation~\eqref{eq:a2l} gives the coefficient of the series that inserts
$\Pi^{(1)}$ once into the ring chain. Whether this value is the complete
coefficient of $r_s^2\log r_s$ in $\epsilon_c^{\rm 2D}$ depends on which
other diagrams can produce a coincident pole at this order. Consider the
class of diagrams with one \emph{soft line}, that is, a screened line whose
momentum transfer is of order $k_{\rm TF}$ and whose frequency is of order
$v_F$ times this momentum transfer. Within this class, an insertion of order $r_s^j$
whose small-$q$ expansion at fixed $z$ contains a term $q^mF(z)$ gives,
at a resonance, the term $r_s^{\,j+(3-d+m)/(d-1)}\log r_s$
(SI Sec.~\ref{S-dsec:counting}). In three dimensions this is
$r_s^{\,j+m/2}\log r_s$. Within this class, the order $r_s\log r_s$ receives exactly the
two contributions of Sec.~\ref{sec:Pi1-d}, $\Pi^{(1)}$ ($j=1$, $m=0$) and
the $q^2$ term of the Lindhard function ($j=0$, $m=2$). Their sum is the
established coefficient $C$. In two dimensions the term is
$r_s^{\,j+1+m}\log r_s$. The order $r_s^2\log r_s$ could therefore
receive, besides $\Pi^{(1)}$, only a term linear in $q$ of the Lindhard
function. This term is absent because $\mathcal L(q,qz)$ is even in $q$ at
fixed $z$ (SI Sec.~\ref{S-dsec:lindhard}). The $q^2$ term of the Lindhard
function, $g_2(z)=-z/[8(1+z^2)^{5/2}]$, enters at $r_s^3\log r_s$, as do the
second-order polarization and two insertions of $\Pi^{(1)}$. The
corrections of order $q^2\log q$ to $\Pi^{(1)}$ enter at
$r_s^4\log^2r_s$.

Diagrams with several soft lines are counted by the number $L$ of
independent soft loops, the loops formed by the soft lines when each fermion
loop is shrunk to a vertex. Lines whose momentum transfer is of order $k_F$
are \emph{hard}. The nonlinear vertices that couple soft density
fluctuations are the closed fermion loops with three or more density
vertices. In the sum over diagrams, each such loop appears symmetrized
over the orderings of its vertices. The symmetrized loops are bounded at
soft momenta $Q$ and imaginary frequencies of order $v_FQ$. Neumayr and
Metzner showed this in two dimensions \cite[Eq.~(73)]{NeumayrMetzner1998},
and Kopper and Magnen proved it with Ward identities for $d\le3$
\cite{KopperMagnen2001}. The bound rests on a cancellation among the terms
of the symmetrized loop, the loop cancellation. Without it, the
unsymmetrized three-point loops would each be larger by a factor
$k_F/k_{\rm TF}$. The bound holds when no partial sum of the external
frequencies of a loop is small compared with $v_FQ$, and its constant
grows as such a sum tends to zero. A diagram with $L$ soft loops then
scales as $E_0q_{\rm TF}^{L(d+1)}$, which is $r_s^{3L-2}$ in two
dimensions and $r_s^{2L-2}$ in three (SI Sec.~\ref{S-dsec:counting}). The
ring series is $L=1$, and $L=2$ starts at $r_s^4$ in two dimensions.
Without the loop cancellation, two-loop diagrams would enter at $r_s^2$, the
order considered here.

The completeness of $a_{2\ell}$ rests on two hypotheses. (H1) The bound on the symmetrized loops holds
uniformly on the soft region away from its degenerate parts (small partial
sums of the soft frequencies, collinear or canceling soft momenta). The
degenerate parts and the mixed regions with hard subdiagrams admit
integrable bounds compatible with this bound, and the hard regions contribute
powers of $r_s$ without $\log r_s$. (H2) The proper polarization at fixed
$z$ is
$-q_{\rm TF}[\mathcal L(q,qz)+(r_s/\sqrt2)R_2+r_s^2P_2(z)]+O(r_sq^2\log q)+O(r_s^3)$,
with the free response $\mathcal L$ kept exactly, remainders integrable in $z$ and $P_2$ bounded and decaying as
$z^{-4}$, where the internal lines of the second-order terms are screened. Under (H1) and (H2), the coefficient of $r_s^2\log r_s$
in the correlation energy of the paramagnetic two-dimensional electron gas
is $a_{2\ell}$ of \eqref{eq:a2l} (SI Proposition~\ref{S-prop:complete}).

SI Sec.~\ref{S-dsec:counting} gives the counting behind this statement
and two checks. In three dimensions the same counting places the two-loop
diagrams at $r_s^2$, consistent with the established coefficient $C$, which
contains only one-line terms. Replacing the bare
exchange line inside $\Pi^{(1)}$ by a Thomas--Fermi screened line changes
$R_2$ at fixed $z$ by an amount linear in the screening wave number $\kappa_{\rm s}$,
without $\log\kappa_{\rm s}$. Screening of the internal line therefore does not shift
the coefficient. Without (H1) and (H2), \eqref{eq:a2l} still gives the
coefficient of the $\Pi^{(1)}$ series, the two-dimensional counterpart of
the exchange part of the Carr--Maradudin coefficient.

\subsection{The screened kite: energy expressions and endpoints of the radial kernel}\label{sec:kite}
Sections~\ref{sec:kite}--\ref{sec:se-cancellation} treat the three-dimensional gas, $d=3$, except for the last
paragraph of Sec.~\ref{sec:se-cancellation}. The screened second-order exchange diagram, the kite, is evaluated
here in the energy expression of BRM [Eqs.~(21)--(24) of Ref.~\cite{Benites2024}].
BRM call it the kite-diagram series. They distinguish it from the
adiabatic-connection SOSEX (AC-SOSEX) and from the
coupled-cluster SOSEX of Freeman \cite{Freeman1977}, whose values lie closer to a
quasistatic approximation of the kite [Eq.~(61) of Ref.~\cite{Benites2024}]. The radial kernel of the kite has an endpoint term
with a Fermi-surface factor $\log(1/q)$ at small momentum transfer and a $q^{-2}$ endpoint term at large
momentum transfer. Together with the scaling functions of
table~\ref{tab:screening}, these endpoint terms give the leading terms of the kite energy
(Sec.~\ref{sec:highdensity}). The Fock self-energy insertions that
complete the first-order exchange polarization $\Pi^{(1)}$ have the same endpoint terms
with opposite coefficients and cancel these leading terms (Sec.~\ref{sec:se-example}).

\begin{figure}[tb]
\centering
\includegraphics[width=0.9\textwidth]{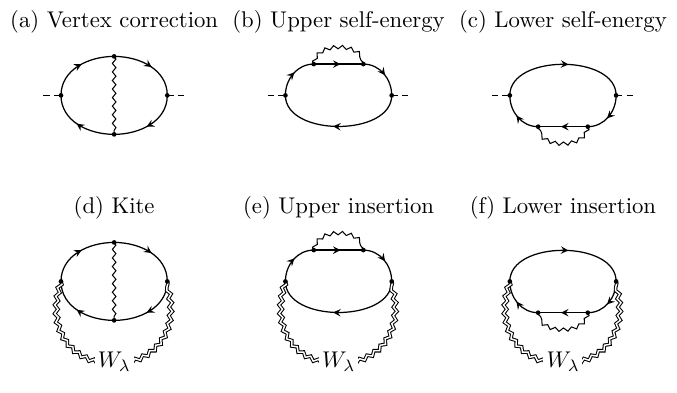}
\caption{(a) Exchange vertex correction, (b,c) Fock self-energy insertions
in the polarization, and (d--f) the energies obtained by closing them with
the same screened line $W_\lambda$. Panel (d) is the kite. Arrows denote $G_0$, and dashed stubs
denote density probes.}
\label{fig:diagrams}
\end{figure}
\label{sec:diagram-coefficients}In the kite (figure~\ref{fig:diagrams}d), one interaction line is
screened and the other remains bare. We denote the functional \eqref{eq:general-functional}
of the kite by $\mathcal K[f]$, normalized so that $\mathcal K[v]$ is the bare kite energy. The bare kite is a single four-vertex fermion loop with two
Coulomb lines. Two Wick contractions produce this topology. They give the
closed-diagram coefficient $2/(2!\,2^2)=1/4$, where $1/2!$ is the
second-order factor and $1/2^2$ comes from the factor $\tfrac12$ of
each Coulomb interaction. The spin sum over the
single loop gives a factor $2$. With the Fourier measures, Coulomb
units and conversion to Ry per electron (SI Sec.~\ref{S-app-sec:setting}),
the integrand $Y_D$ of \eqref{eq:general-functional} is $A_0X$ with
$S_D=4$ and $n=2$, so that
\begin{align}
 \mathcal K[f]&=A_0\int\mathrm d^3q\int_0^\infty
 \frac{\mathrm d\nu}{\pi}\,f(q,\mathrm i\nu)X(q,\nu),
 \qquad A_0=\frac{3}{32\pi^5},\label{eq:KD}\\
 X(q,\nu)&=\int\frac{\mathrm d^3p\,\mathrm d^3k}{|\bm p-\bm k|^2}
 \operatorname{Re}\frac{\Delta_{\bm p}\Delta_{\bm k}}
 {(\mathrm i\nu-\delta_{\bm p})(\mathrm i\nu-\delta_{\bm k})},\label{eq:Xdef}\\
 \Delta_{\bm p}&=\Theta(1-|\bm p|)-\Theta(1-|\bm p+\bm q|),
 \qquad\delta_{\bm p}=\bm p\cdot\bm q+q^2/2.\nonumber
\end{align}
The occupation difference $\Delta_{\bm p}$ is nonzero only when one of
$\bm p$ and $\bm p+\bm q$ is occupied and the other is empty, and
$\delta_{\bm p}$ is the signed excitation energy of the free electron--hole pair. Each
electron--hole factor $\Delta_{\bm p}/(\mathrm i\nu-\delta_{\bm p})$
arises from the internal fermion frequency integral of two propagators,
and the internal Coulomb factor $|\bm p-\bm k|^{-2}$ couples the two factors.
The real part is taken because the frequency integral runs over
$\nu>0$ only (Sec.~\ref{sec:radial-general}).

As a check on the normalization, the bare limit gives, in Ry,
\begin{equation}
 \mathcal K[v]=e_{2x}
 =\frac13\log2-\frac{3\zeta(3)}{2\pi^2},
 \label{eq:bare-value}
\end{equation}
the Onsager--Mittag--Stephen second-order exchange energy
\cite{OMS,Glasser}. A direct Monte Carlo evaluation of \eqref{eq:KD}
with $f=v$ gives $0.04836(2)$~Ry, in agreement with this value (SI
Sec.~\ref{S-sm-sec:verification}).

The full-coupling kite energy and the BRM coupling-averaged energy,
which is \eqref{eq:general-average} with $w_D=2\lambda$, are
\cite{Benites2024}
\begin{equation}
 F(r_s)=\mathcal K[W_{\lambda=1}(r_s)],\qquad
 e_{\rm kite}(r_s)=2\int_0^1\mathrm d\lambda\,
             \mathcal K[W_\lambda(r_s)]=2\mathcal K[\overline W],
 \label{eq:energy-definitions}
\end{equation}
where $\overline W=\int_0^1W_\lambda\,\mathrm d\lambda$. SI
Eq.~\eqref{S-app-eq:Kweight} derives $e_{\rm kite}$ from the
self-energy trace, and SI Sec.~\ref{S-sm-app-app:brm} gives the
correspondence with the BRM representation. Without
screening, $W_\lambda=\lambda v$ and $\overline W=v/2$, so both
energies reduce to $e_{2x}$. With $c$ of \eqref{eq:cdef}, $2\overline W=v\rho(c)$
by \eqref{eq:rho}. Both
$F$ and $e_{\rm kite}$ describe the isolated kite, and $e_{\rm kite}$ at
zero spin polarization is the energy studied in the companion paper \cite{Imoto2026b}.

A third energy expression sums closed diagrams with their closed-diagram coefficients.
The \emph{chain family} consists of the bare kite and every diagram with a
nonempty chain of free polarization bubbles on one kite line. The closed-diagram
coefficient is $1/4$ for the bare kite and $1/2$ for every nonempty
chain (SI Proposition~\ref{S-prop:kite-chain}). The coefficient doubles
because the two kite lines are equivalent and a chain on one of them
breaks the symmetry that exchanges them. The family therefore sums to
\begin{equation}
 e_{\rm closed}=\mathcal K[v]+2\mathcal K[W_{\lambda=1}-v]=2F(r_s)-e_{2x}.
 \label{eq:closed-family}
\end{equation}
For a chain of $j\ge1$ bubbles the BRM average assigns $1/(j+2)$ of
the chain-family coefficient, because it marks only the two original kite
interactions and retains the bare coefficient $1/4$ (SI
Sec.~\ref{S-sm-sec:marking}). Equivalently, the $j$-bubble term of
$W_\lambda$ carries $\lambda^{j+1}$, and
$2\int_0^1\lambda^{j+1}\,\mathrm d\lambda=2/(j+2)$.

The three energy expressions are related by
\begin{equation}
 \begin{aligned}
 e_{\rm kite}(r_s)&=2\int_0^1\lambda F(\lambda r_s)\,\mathrm d\lambda
       =\frac2{r_s^2}\int_0^{r_s} \varrho F(\varrho)\,\mathrm d\varrho,\\
 F(r_s)&=e_{\rm kite}(r_s)+\frac{r_s}{2}e_{\rm kite}'(r_s),\qquad
 e_{\rm closed}(r_s)=2e_{\rm kite}(r_s)+r_s e_{\rm kite}'(r_s)-e_{2x},
 \end{aligned}
 \label{eq:density-transfer}
\end{equation}
where the prime denotes $\mathrm d/\mathrm dr_s$. By the first relation, a term
$\varphi\,r_s^p$ of $F$ appears in $e_{\rm kite}$ multiplied by the
Mellin multiplier $M(p)=2/(p+2)$ of Sec.~\ref{sec:mellin-convolution}.
The first relation is
\eqref{eq:general-average} with $w_D=2\lambda$ and $\varrho=\lambda r_s$.
Multiplying it by $r_s^2$ and differentiating with respect to $r_s$
gives $2r_se_{\rm kite}+r_s^2e_{\rm kite}'=2r_sF$, which is the second
relation. Inserting the second relation into \eqref{eq:closed-family}
gives the third. SI Sec.~\ref{S-sm-sec:coupling-averages} explains why
asymptotic forms are carried from $F$ to $e_{\rm kite}$ and not in the
reverse direction.

\label{sec:radial-kernels}\label{sec:static-ref}\label{sec:block}For the kite, $Y_D=A_0X$ by \eqref{eq:KD}, and we drop the subscript
$D$: the radial kernel $e_{D,\Lambda}$ of \eqref{eq:radial-kernel} becomes
\begin{equation}
 \begin{aligned}
 \mathcal K[b_\Lambda]&=\int_0^\infty e_\Lambda(q)\,\mathrm dq,\\
 e_\Lambda(q)&=4\pi A_0\int_0^\infty\frac{\mathrm d\nu}{\pi}
           \frac{\Lambda^2}{\Lambda^2+\nu^2/q^2}\,X(q,\nu),
 \end{aligned}
 \label{eq:XLambda}
\end{equation}
with bare-line limit $e_\infty(q)$, whose integral is $e_{2x}$. The
two endpoints of $e_\Lambda(q)$ are as follows.

At small momentum transfer, $q\to0$, the occupation difference confines
the loop momenta to thin shells of width $O(q)$ at the Fermi surface, and the Coulomb
interaction between nearly coincident momenta produces a factor $\log(1/q)$.
Inside the two shells the separation $\bm p-\bm k$ has a component of
order $q$ normal to the Fermi surface and a tangential component
$\bm r$, so the Coulomb factor integrates schematically as
$\int\mathrm d^2r/(r^2+q^2)\simeq2\pi\log(1/q)$ over
$q\lesssim r\lesssim1$. This factor $\log(1/q)$ is the Fermi-surface factor
of the radial kernel of the kite. At
fixed $z=\nu/q$, with $h(z)$ of \eqref{eq:hz}, the shell expansion gives
\begin{equation}
 \begin{aligned}
 X(q,qz)&=a_x(z)\log(1/q)+b_x(z)+o(1),\\
 a_x(z)&=8\pi^2[h(z)+zh'(z)],\qquad
 \int_0^\infty\frac{\mathrm dz}{\pi}\,h(z)a_x(z)
       =\frac{4\pi^2}{3}(1-\log2),
 \end{aligned}
 \label{eq:Az}
\end{equation}
where $h'=\mathrm dh/\mathrm dz$ and $o(1)$ refers to $q\to0$ at fixed
$z$. The finite part $b_x(z)$ of the shell expansion enters the energy
only through the $r_s\log r_s$ term
(Sec.~\ref{sec:highdensity}). The angular integrals that give
$a_x$ and the $z$ integral are in SI Sec.~\ref{S-app-sec:shell-origin},
Eqs.~\eqref{S-app-eq:Az} and \eqref{S-app-eq:calA}.
Integrating \eqref{eq:Az} against the Lorentzian factor, with $\mathrm d\nu=q\,\mathrm dz$
supplying the factor $q$, gives
$e_\Lambda(q)=\kappa(\Lambda)q\log(1/q)+O(q)$ with
$\kappa(\Lambda)=4\pi A_0\int_0^\infty(\mathrm dz/\pi)\,\Lambda^2a_x(z)/(\Lambda^2+z^2)$
[closed form in SI Eq.~\eqref{S-app-eq:kappa}, where the coefficient of
the $O(q)$ term also follows from the pointwise expansion]. The
frequency integration is performed
first to control the integrated shell remainder
(SI Sec.~\ref{S-app-app:shell-integrated}).

At large momentum transfer, $q>2$, the unit Fermi ball
$\mathbb B$ and its translate $\mathbb B_{-\bm q}$ by $-\bm q$ do not
overlap, and $\Delta_{\bm p}$ is $+1$ on the former and $-1$ on the
latter (figure~\ref{fig:ball-regions}). The product
$\Delta_{\bm p}\Delta_{\bm k}$ gives two \emph{same-ball regions}, in which both loop momenta lie inside the same ball,
and two \emph{different-ball regions}, in which they lie in different balls.
Each same-ball region has the internal Coulomb integral $4\pi^2$ and excitation energies
$\pm q^2/2+O(q)$ of common sign. The different-ball regions are
suppressed by $q^{-2}$. In a same-ball region
$\delta_{\bm p}\simeq\delta_{\bm k}\simeq\pm q^2/2$, so
\eqref{eq:sigma-frequency-main} with $\Omega=q\Lambda$ gives the
frequency integral of the two electron--hole factors,
$2w_\Lambda/[q^2(1+2w_\Lambda)^2]$ with $w_\Lambda=\Lambda/q$. Because
both loop momenta stay in one ball, the Coulomb integral does not
decrease with $q$. The kernel therefore decays only as $q^{-2}$. By
\eqref{eq:Wdef} this same-ball term is absent for a static line, since both
excitation energies have the same sign. For $w_\Lambda$ bounded away
from zero, the radial kernel at large momentum transfer is
\begin{equation}
 e_\Lambda(q)=\frac{3}{\pi^3q^2}\,p_{\rm loc}(w_\Lambda)
        +O\!\left(\frac1{q^3(1+w_\Lambda)}+\frac1{q^4}\right),\qquad
 p_{\rm loc}(w)=\frac{2\pi w}{(1+2w)^2}.
 \label{eq:largeq-block}
\end{equation}
The frequency integral and the bounds on the
different-ball regions are in SI Sec.~\ref{S-app-sec:largeq-kernel}.
The factor $p_{\rm loc}$ is maximal at $w_\Lambda=1/2$, where the spectral
frequency $q\Lambda$ equals the excitation frequency $q^2/2$.

\begin{figure}[tbp]
\centering
\includegraphics[width=0.5\textwidth]{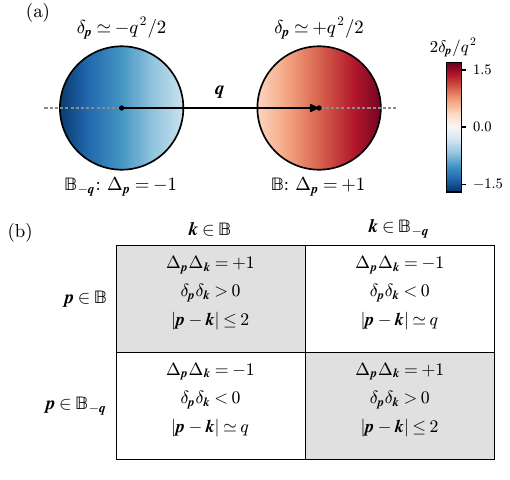}
\caption{Momentum regions of the kite loop for $q>2$ (in units of $k_F$), where the Fermi
ball $\mathbb B$ and its translate $\mathbb B_{-\bm q}$ do not
overlap. (a) Cross section through the axis of $\bm q$, drawn for
$q=3.2$. The occupation difference $\Delta_{\bm p}$ is $+1$ on
$\mathbb B$ and $-1$ on $\mathbb B_{-\bm q}$, and the color gives
the excitation energy $\delta_{\bm p}=\bm p\cdot\bm q+q^2/2$ in units of
$q^2/2$. (b) The four regions of the two loop momenta. In the same-ball
regions (shaded) the two excitation energies have the same sign and
$|\bm p-\bm k|\le2$. These regions
contribute only when the line depends on frequency
[\eqref{eq:Wdef}]. In the different-ball regions the signs are
opposite and $|\bm p-\bm k|\simeq q$, so the Coulomb factor suppresses
them by $q^{-2}$. The signs are exact for $q>2$, whereas
$\delta_{\bm p}\simeq\pm q^2/2$ and $|\bm p-\bm k|\simeq q$ denote the
leading large-$q$ behavior.}
\label{fig:ball-regions}
\end{figure}

\subsection{The kite at high and low density}\label{sec:highdensity}\label{sec:cancel}
At high density, $r_s\to0$, the scale is $\varepsilon=q_{\rm TF}$. We consider the
coupling-averaged line, with $c$ of \eqref{eq:cdef}. At small
momentum transfer we set $\nu=qz$, and $\mathcal L(q,qz)\to h(z)$. Let
$f=2\overline W-v=v[\rho(c)-1]$, so that
$e_{\rm kite}-e_{2x}=\mathcal K[f]$ by \eqref{eq:bare-value} and
\eqref{eq:energy-definitions}. With $\mathrm d^3q=4\pi q^2\mathrm dq$ and
$\mathrm d\nu=q\,\mathrm dz$, \eqref{eq:KD} gives [cf.\ SI
Eq.~\eqref{S-app-eq:window}]
\begin{equation}
 \begin{aligned}
 e_{\rm kite}-e_{2x}&\simeq4\pi A_0\int_0^\infty\frac{\mathrm dz}{\pi}
 \int_0^\infty\mathrm dq\,k_z(q)\,g_z(q/q_{\rm TF}),\\
 k_z(q)&=qX(q,qz),\qquad g_z(t)=\rho\bigl(h(z)/t^2\bigr)-1.
 \end{aligned}
 \label{eq:hd-convolution}
\end{equation}
Here $k_z=k_D/(4\pi A_0)$ with $k_D$ of \eqref{eq:kD}, and $\simeq$
denotes equality up to the $O(r_s|\log r_s|)$ corrections discussed
after \eqref{eq:Clog2}.
The full denominator of the screened line is retained for $q\lesssim q_{\rm TF}$, as in the
ring resummation of Gell-Mann and Brueckner \cite{GellMannBrueckner}.
Expanding $\rho(c)$ in powers of $c\propto q_{\rm TF}^2/q^2$ would
generate bare chain diagrams whose $q$ integrals diverge at small $q$.

The scaling function (table~\ref{tab:screening}) has the tail
$g_z\simeq-2h(z)/(3t^2)$, so that
$\widetilde g_z(s)\simeq2h(z)/[3(s-2)]$ near $s=2$, up to terms regular
there. By \eqref{eq:Az} the radial
kernel has the endpoint term $k_z\simeq a_x(z)q\log(1/q)$, so
that $\widetilde k_z(1-s)\simeq a_x(z)/(s-2)^2$. The two poles coincide
at $s=2$. The resulting triple pole contributes
$-\tfrac13a_x(z)h(z)\,q_{\rm TF}^2\log^2q_{\rm TF}$ to the $q$ integral
of \eqref{eq:hd-convolution} at fixed $z$ (SI
Sec.~\ref{S-sm-slim-sec:residues}), with
$\log q_{\rm TF}=\tfrac12\log r_s+O(1)$. Integrating over $z$
with $\int_0^\infty h\,a_x\,\mathrm dz/\pi=4\pi^2(1-\log2)/3$ and using
$q_{\rm TF}^2=4\alpha_Lr_s/\pi$ from \eqref{eq:lindhard-convention}, we
obtain the energies in Ry,
\begin{equation}
 \begin{aligned}
 e_{\rm kite}(r_s)&=e_{2x}+c_{\rm hd}\,r_s\log^2r_s+O(r_s|\log r_s|),\\
 F(r_s)&=e_{2x}+\tfrac32c_{\rm hd}\,r_s\log^2r_s+O(r_s|\log r_s|),\\
 e_{\rm closed}(r_s)&=e_{2x}+3c_{\rm hd}\,r_s\log^2r_s+O(r_s|\log r_s|),\\
 c_{\rm hd}&=-\frac{\alpha_L(1-\log2)}{6\pi^3}
       =-0.0008594457\ldots .
 \end{aligned}
 \label{eq:Clog2}
\end{equation}
SI
Eqs.~\eqref{S-app-eq:window}--\eqref{S-app-eq:highdensity-u2} give the window form of this calculation.
Replacing $h$ by $\mathcal L(q,qz)$ and controlling the regions
$q\lesssim q_{\rm TF}$ and $q\gtrsim1$ changes the energy only by $O(r_s|\log r_s|)$
(SI Secs.~\ref{S-app-sec:shell-origin} and
\ref{S-app-app:shell-integrated}). The coefficients of $F$ and
$e_{\rm closed}$ are derived in SI Secs.~\ref{S-sm-slim-sec:residues}
and \ref{S-sm-sec:coupling-averages}. At full coupling the tail is $-h/t^2$
instead of $-2h/(3t^2)$, so $F$ follows directly from the same triple
pole. The ratio of the $e_{\rm kite}$ coefficient to the $F$ coefficient is the Mellin multiplier $2/3$
at $p=1$, and $e_{\rm closed}=2F-e_{2x}$ doubles the $r_s\log^2r_s$ correction of $F$. The
finite part $b_x(z)$ of the kernel gives $\widetilde k_z(1-s)$ a simple
pole at the same point. This simple pole and the pole of $\widetilde g_z$ form a
double pole, which gives an $r_s\log r_s$ term. The coefficient of this term,
$c_1=0.01150$~Ry for $e_{\rm kite}$, requires the finite shell parts
and has been evaluated in the companion paper \cite{Imoto2026b}. The
$\Pi^{(1)}$ series of Sec.~\ref{sec:DuBois},
$0.0130\,r_s\log r_s-0.021\,r_s$~Ry \cite{CarrMaradudin} and \cite[Eq.~(5)]{Endo1999},
is a different quantity: it contains $\Pi^{(1)}_{\rm SE}$ as well and
closes $\Pi^{(1)}$ with the line $v/(1+c)$ of the ring series instead
of $2\overline W$.

Direct numerical
integration of $e_{\rm kite}$ at $r_s=10^{-3}$--$10^{-1}$
\cite{Imoto2026b,ImotoCode} shows the $r_s\log^2r_s$ term in this range. The values
are listed in SI table~\ref{S-sm-tab:hd-direct}. A weighted
least-squares fit of the five values of $e_{\rm kite}-e_{2x}$
(weights $1/\sigma^2$, with $\sigma$ the quoted error of each value)
to $c_{\rm fit}r_s\log^2r_s+c_{1,\rm fit}r_s\log r_s+c_2r_s$,
with all three coefficients free, gives $c_{\rm fit}=-0.00091(5)$~Ry. The same fit
to the four values with $r_s\le0.03$ gives $c_{\rm fit}=-0.00092(12)$~Ry. Both agree
with $c_{\rm hd}$. The weighted sum of squared residuals is $0.02$ for
two degrees of freedom. Without the
$r_s\log^2r_s$ term it rises to $380$ for three degrees of freedom. With $c_{\rm fit}$ fixed at $c_{\rm hd}$, the
fitted coefficient of $r_s\log r_s$ is $0.01125(6)$ for $r_s\le0.1$
and approaches $c_1$ as the fit range shrinks (within $1.4\sigma$ for
$r_s\le0.03$). This drift is expected from the terms of order $r_s^2\log^kr_s$
omitted in the fit.

BRM fitted their Monte Carlo values for $r_s\le1$ with
$e_{2x}+C_1r_s\log r_s+C_2r_s+C_3r_s^2\log r_s$ and obtained
$C_1=0.01212$~Ry [Eq.~(39) and Table~I of Ref.~\cite{Benites2024}]. They
regarded the closeness of this coefficient to $0.01304$~Ry, the exchange part of the Carr--Maradudin coefficient
of \eqref{eq:DuBois}, as a strong test of their calculation. For $r_s\le1$ their values agree with
the direct integration, for example $43.53(18)$ against $43.433(6)$~mRy
at $r_s=0.1$. With these error bars, however, a fit without the $r_s\log^2r_s$
term absorbs that term into $C_1$. The closeness of $C_1$ to $0.01304$ is therefore
not a test of the kite. This exchange part belongs to the
complete first-order exchange polarization $\Pi^{(1)}$ (Sec.~\ref{sec:DuBois}). The
isolated kite has $c_{\rm hd}r_s\log^2r_s+c_1r_s\log r_s$ instead.

\label{sec:dispersion}\label{sec:lowdensity}At low density, with the form \eqref{eq:pole-dispersion} of the line, the
scale is $\varepsilon=q_*$. Inserting \eqref{eq:pole-dispersion} into the kite functional gives
$\mathcal K[W_\lambda]=\lambda e_{2x}-\lambda\int_0^\infty\theta_\lambda
e_{\Lambda_\lambda}\,\mathrm dq+\dots$. Here $e_{\Lambda_\lambda}$ is
the radial kernel $e_\Lambda$ of \eqref{eq:XLambda} at the
$q$-dependent width $\Lambda_\lambda(q)$ of \eqref{eq:pole-dispersion}
and the dots denote the contribution of the remainder $\varsigma_\lambda$.
With $e_{2x}=\int_0^\infty e_\infty\,\mathrm dq$, the right-hand side becomes
$\lambda\int_0^\infty[(1-\theta_\lambda)e_\infty
+\theta_\lambda(e_\infty-e_{\Lambda_\lambda})]\,\mathrm dq+\dots$. The
bare-line part $(1-\theta_\lambda)e_\infty$ is concentrated at $q\gtrsim q_*$,
where $e_\infty=O(q^{-4})$, and contributes only a remainder $O(q_*^{-3})$ (SI
Sec.~\ref{S-sm-sec:kite-remainders}). In the second part the relevant
momentum transfers $q\sim q_*\gg1$ lie where \eqref{eq:largeq-block} holds. Since
$p_{\rm loc}\to0$ for $w\to\infty$, the integrand of the second part is
$-3\Psi_\lambda(t)/(\pi^3q^2)$ with
$\Psi_\lambda(t)=\theta_\lambda p_{\rm loc}(w_{\Lambda,\lambda})$,
$w_{\Lambda,\lambda}=\Lambda_\lambda/q$ and $t=q/q_*$. The function
$\Psi_\lambda$ depends on $q$ only through $t$ and vanishes as $t^2$
for $t\to0$ and as $t^{-4}$ for $t\to\infty$. The substitution
$q=q_*x$ therefore gives the factor $1/q_*$ directly. In the language of
the endpoint--residue correspondence (SI
Proposition~\ref{S-prop:endpoint-residue}), this leading term is the residue of the
simple pole at $s=-1$ that the $q^{-2}$ endpoint term of the kernel gives
$\widetilde k(1-s)$. With
$e_{\rm kite}=2\int_0^1\mathcal K[W_\lambda]\,\mathrm d\lambda$ of
\eqref{eq:energy-definitions}, $\theta_\lambda=\lambda/(x^4+\lambda)$ and
$p_{\rm loc}$ of \eqref{eq:largeq-block} written out, the leading term is
\begin{equation}
 \begin{aligned}
 e_{\rm kite}(r_s)&=\frac{A_1}{q_*}+O(q_*^{-3/2}),\\
 A_1&=-\frac{12}{\pi^2}\int_0^1\!\lambda\,\mathrm d\lambda
       \int_0^\infty\!\frac{\mathrm dx}{x^2}
       \frac{\lambda}{x^4+\lambda}
       \frac{w_{\Lambda,\lambda}(x)}{[1+2w_{\Lambda,\lambda}(x)]^2},\\
 w_{\Lambda,\lambda}(x)&=\tfrac12\sqrt{1+\lambda x^{-4}},\qquad
 A_1=-\frac{4\,\Gamma(1/4)}{7\sqrt2\,\pi^{3/2}\,\Gamma(3/4)}
                  =-0.2146937609\ldots.
 \end{aligned}
 \label{eq:A1-dispersion}
\end{equation}
With
$q_*^{-1}=[3\pi/(16\alpha_L)]^{1/4}r_s^{-1/4}=1.0311348\ldots\,r_s^{-1/4}$,
from \eqref{eq:qstar} and $q_{\rm TF}^2=4\alpha_Lr_s/\pi$,
the leading term is $-0.2213782\ldots\,r_s^{-1/4}$
(table~\ref{app-tab:static-dynamic}). The
substitution $x=\lambda^{1/4}y$ separates the coupling integral
$\int_0^1\lambda^{3/4}\,\mathrm d\lambda=4/7$ and leaves the beta function
$\mathrm B(5/4,1/2)$ (SI Sec.~\ref{S-app-sec:lowdensity},
Eqs.~\eqref{S-app-eq:A1-y} and \eqref{S-app-eq:A1-closed}, with the
split \eqref{S-app-eq:A1-split} written out). The coefficient agrees with the matched expansion
of the companion paper \cite{Imoto2026b}. That paper also gives the further
terms and a comparison with numerical integration up to $r_s=10^8$.

BRM inferred a decay as $r_s^{-1/2}$ from their Monte Carlo values and
built it into their fit form [Eq.~(58) of Ref.~\cite{Benites2024}].
The leading term derived here, however, is of order $r_s^{-1/4}$. Their limit of a
vanishing kite as $r_s\to\infty$ [their Eq.~(38)] agrees with
\eqref{eq:A1-dispersion}. At $r_s=10^8$ the direct value, $-2.22(21)$~mRy, is consistent with
$A_1/q_*=-2.214$~mRy \cite{Imoto2026b,ImotoCode}. The Monte
Carlo values of BRM at $r_s=2\times10^5$ and $10^6$, $-7.13$ and
$-3.89$~mRy \cite{Benites2024}, lie above the direct values $-9.13$ and
$-6.35$~mRy. BRM sampled their Eqs.~(28) and (29) with $10^9$ points.
The integration over the momentum transfer in these equations extends
over all space. Their paper does not describe how this integration was mapped or
truncated. The paper alone therefore does not determine the origin of the
difference. The integrand extends to momentum transfers of order
$q_*k_F\approx20k_F$ and $31k_F$ at these densities. A momentum
cutoff near $12k_F$ in the direct integration reproduces the difference
\cite{Imoto2026b}. Such a cutoff is one possible origin of the difference, not an identified
error.

The remainder bound $O(q_*^{-3/2})$ is proved in SI
Sec.~\ref{S-sm-sec:kite-remainders}. The bounds of that section hold at
each fixed coupling. The
$x$ integral of \eqref{eq:A1-dispersion} at $\lambda=1$ therefore gives $F(r_s)=\tfrac78A_1/q_*+O(q_*^{-3/2})$, consistent with
the Mellin multiplier $2/(p+2)=8/7$ at $p=-1/4$. In the chain family, $e_{\rm closed}=2F-e_{2x}$
tends to $-e_{2x}$ as $F\to0$ (table~\ref{app-tab:static-dynamic}). This constant
belongs to the chain family and is expected to be removed once
both kite lines are screened (SI Sec.~\ref{S-sm-sec:two-line}).

\begin{table}[!ht]
\centering\small
\caption{Leading energy terms in Ry per electron for the kite alone, in
the three energy expressions of Sec.~\ref{sec:diagram-coefficients}
(full-coupling line, BRM coupling average, chain family). High density:
correction to the bare value $e_{2x}$ of \eqref{eq:bare-value}, with $c_{\rm hd}$ of \eqref{eq:Clog2}. Low density: the energy, a formal limit
(Sec.~\ref{sec:intro}). The $r_s^{-1/4}$ terms are
$A_1/q_*=-0.2213782\ldots\,r_s^{-1/4}$ of \eqref{eq:A1-dispersion} for
$e_{\rm kite}$, $7/8$ of it for $F$ and $7/4$ of it for
$e_{\rm closed}$, by \eqref{eq:density-transfer}. Derivations are in SI
Sec.~\ref{S-sm-sec:coupling-averages}.}
\label{app-tab:static-dynamic}
\setlength{\tabcolsep}{4pt}
\begin{tabular}{@{}lll@{}}
\toprule
Contribution & High density: correction & Low density: energy\\
\midrule
Kite, full coupling $F$ & $\tfrac32c_{\rm hd}r_s\log^2r_s$ & $-0.1937059506\,r_s^{-1/4}$\\
Kite, averaged $e_{\rm kite}$ & $c_{\rm hd}r_s\log^2r_s$ & $-0.2213782293\,r_s^{-1/4}$\\
Chain family $e_{\rm closed}$ & $3c_{\rm hd}r_s\log^2r_s$ & $-e_{2x}-0.3874119013\,r_s^{-1/4}$\\
\bottomrule
\end{tabular}
\end{table}

\subsection{Cancellation by the Fock self-energy corrections}\label{sec:se-example}\label{sec:exchange-response}\label{sec:se-partner}\label{sec:se-cancellation}
The kite is obtained by closing the vertex part $\Pi_{\rm vtx}^{(1)}$ of
the first-order exchange polarization $\Pi^{(1)}$ with a
screened line. This polarization also
contains a Fock self-energy insertion on each propagator of the bubble
\cite{Holas1979,Qian2015} (figure~\ref{fig:diagrams}). At fixed density
the insertion is $G_0(\Sigma_x-\delta\mu_x)G_0$, where $\Sigma_x(k)$ is
the Fock self-energy at momentum $k$ [SI
Eq.~\eqref{S-app-eq:definitions-u6}]. The chemical-potential shift
$\delta\mu_x=\Sigma_x(1)$ is the self-energy at the Fermi surface $k=1$
and keeps the Fermi surface at its free position. The sum $\Pi_{\rm SE}^{(1)}$ of
the two Fock self-energy insertions [explicit form in SI Eq.~\eqref{S-app-eq:PiSigma}],
the \emph{self-energy pair}, is closed with the line $f$ in the kite normalization,
\begin{equation}
 \mathcal K_{\Pi_{\rm SE}^{(1)}}[f]
 =-\frac3{\pi^3q_{\rm TF}^4}\int\mathrm d^3q
   \int_0^\infty\frac{\mathrm d\nu}{\pi}\,
   f(q,\mathrm i\nu)\Pi_{\rm SE}^{(1)}(q,\mathrm i\nu),
 \label{eq:KT-sigma}
\end{equation}
which defines $\mathcal K_{\Pi}[f]$ for any insertion $\Pi$ closed with
$f$ in the kite normalization. Its integrand is $Y_D$ of
Sec.~\ref{sec:linear-diagrams} with $S_D=4$. The closed
diagram with one Fock self-energy insertion has the coefficient $1/2$. Its two
symmetries, in the sense of SI Sec.~\ref{S-sm-sec:kite-chain}, are the
identity and the rotation of the loop that exchanges its two interaction
lines. Since
$\Pi_{\rm SE}^{(1)}$ contains two insertions that both close to this
diagram, the coefficient for the pair is $1/4$. In this normalization
$\mathcal K_{\Pi_{\rm vtx}^{(1)}}=\mathcal K$ (SI
Sec.~\ref{S-app-sec:response}).

In energy expressions of the RPA-with-exchange (RPAx) type
\cite{HesselmannGorling2010} the first-order exchange polarization $\Pi^{(1)}$
enters closed with $\lambda v/\epsilon_\lambda^2$ instead of
$W_\lambda$, that is, through the factor
$\lambda\Pi^{(1)}/\epsilon_\lambda^2$. Since $\epsilon_\lambda=1+\lambda c$
by \eqref{eq:response-conventions} and \eqref{eq:cdef}, this line has
the scaling function $(1+\lambda c)^{-2}-1$, which also has a $t^{-2}$ tail. The
vertex and self-energy kernels have opposite endpoint terms, so the
corresponding poles cancel for any common line (Sec.~\ref{sec:cancellation-general}).
The leading cancellations derived below therefore also hold for this line.
The self-consistency of a conserving approximation \cite{Baym1962} and the
terms from the Legendre transform between fixed chemical potential and
fixed particle number are not considered.

To see why the self-energy corrections are needed, note that the correlation
energy is determined by the density response. By the adiabatic-connection
fluctuation--dissipation theorem \cite{LangrethPerdew1977,GunnarssonLundqvist1976},
the correlation energy $E_c$ per unit volume, in the units of
Sec.~\ref{sec:units}, is
\begin{equation}
 E_c=-\frac12\int_0^1\mathrm d\lambda\int\frac{\mathrm d^3q}{(2\pi)^3}\,v(q)
 \int_0^\infty\frac{\mathrm d\nu}{\pi}\,
 \bigl[\chi_\lambda(q,\mathrm i\nu)-\chi_0(q,\mathrm i\nu)\bigr].
 \label{eq:acfd}
\end{equation}
Multiplying this energy density
by the volume $3\pi^2$ per electron and by $E_0=32/(\pi^2q_{\rm TF}^4)$~Ry
converts it to Ry per electron.
An approximation to $E_c$ is therefore an approximation to $\chi_\lambda$.
By \eqref{eq:response-conventions}, a change $\delta\Pi_\lambda$ of the
proper response changes the full response, to first order, by
$\delta\chi_\lambda=\delta\Pi_\lambda/\epsilon_\lambda^2$. The
resulting change of the correlation energy is
\begin{equation}
 \delta E_c=-\frac12\int_0^1\mathrm d\lambda\int\frac{\mathrm d^3q}{(2\pi)^3}
 \int_0^\infty\frac{\mathrm d\nu}{\pi}\,
 \frac{v(q)}{\epsilon_\lambda(q,\mathrm i\nu)^2}\,
 \delta\Pi_\lambda(q,\mathrm i\nu).
 \label{eq:acfd-linear}
\end{equation}
For first-order exchange, $\delta\Pi_\lambda=\lambda\Pi^{(1)}$ with
$\Pi^{(1)}=\Pi^{(1)}_{\rm vtx}+\Pi^{(1)}_{\rm SE}$, and
\eqref{eq:acfd-linear} is the RPAx form above. The energy expressions
of Sec.~\ref{sec:linear-diagrams} close $\Pi^{(1)}$ with other screened
lines and coupling weights. These lines share a $t^{-2}$ Thomas--Fermi tail at high density
and the suppression below the plasma frequency at low density.
Evaluating the kite alone amounts to $\delta\Pi_\lambda=\lambda\Pi^{(1)}_{\rm vtx}$.
The exact response obeys two constraints that this choice violates.

(i) The Ward identity, which expresses charge conservation, ties the vertex
correction to the self-energy corrections. At small momentum
transfer their terms in $\log(1/q)$ are equal and opposite, as shown in \eqref{eq:ward} below.
In \eqref{eq:acfd-linear} the measure is
$\mathrm d^3q\,\mathrm d\nu=4\pi q^3\,\mathrm dq\,\mathrm dz$, the bare
line is $q^{-2}$, and for $q\gg q_{\rm TF}$ the screening changes it by
a relative amount of order $q_{\rm TF}^2h(z)/q^2$. The vertex part behaves as
$\Pi^{(1)}_{\rm vtx}\propto q_{\rm TF}^4\log(1/q)$, and the conversion to
Ry per electron (Sec.~\ref{sec:GB}) supplies the factor
$q_{\rm TF}^{-4}$ of $E_0=32/(\pi^2q_{\rm TF}^4)$~Ry. Up to the $z$
integral and in the window $q_{\rm TF}\ll q\ll1$ of Sec.~\ref{sec:GB},
the screening correction to the energy per electron therefore behaves
as
\begin{equation}
 q_{\rm TF}^2\int_{q_{\rm TF}}^1\frac{\log(1/q)}{q}\,\mathrm dq
 =\tfrac12\,q_{\rm TF}^2\log^2q_{\rm TF}\propto r_s\log^2r_s.
 \label{eq:acfd-highdensity}
\end{equation}
For the conserving sum, $\Pi^{(1)}$ tends to the finite limit
$-(q_{\rm TF}^4/16)R^{(1)}(z)$. The same integral without the factor
$\log(1/q)$ then gives $q_{\rm TF}^2\log(1/q_{\rm TF})\propto r_s\log r_s$, the
form of Sec.~\ref{sec:DuBois}. The $r_s\log^2r_s$ term of the kite is thus
the signature in $E_c$ of a response that violates charge conservation.

(ii) In the $f$-sum rule \cite[Sec.~2.2]{PinesNozieres1966}, the
double commutator $[[\hat\rho_{\bm q},\widehat H_\lambda],\hat\rho_{-\bm q}]$,
with $\widehat H_\lambda=\widehat T+\lambda\widehat V$ of
Sec.~\ref{sec:units} and the density operator $\hat\rho_{\bm q}=\sum_j{\rm e}^{-\mathrm i\bm q\cdot\bm r_j}$,
contains only the kinetic energy, because the interaction commutes with
the density. Through the spectral representation of $\chi_\lambda$, the
double commutator fixes the coefficient of $-\nu^{-2}$ in the
large-frequency expansion of the response at $n_eq^2/m_e$ (SI
Sec.~\ref{S-sm-slim-sec:lemma}). With
$n_e=k_F^3/(3\pi^2)$ and $q_{\rm TF}^2=4/(\pi a_Bk_F)$, in the units of
Sec.~\ref{sec:units} this gives, for every $\lambda$,
\begin{equation}
 \chi_\lambda(q,\mathrm i\nu)\simeq-\frac{q_{\rm TF}^2q^2}{3\nu^2},
 \qquad\nu\to\infty,
 \label{eq:fsum-exact}
\end{equation}
which is the free value. Since $\chi_\lambda-\Pi_\lambda=O(\nu^{-4})$,
the $\nu^{-2}$ tail of $\delta\Pi_\lambda$ must vanish in a complete set
of corrections. At low density and $q\ll q_*$ the factor
$1/\epsilon_\lambda$ can be replaced by $\nu^2/(\nu^2+\lambda\omega_p^2)$, with
$\omega_p$ of \eqref{eq:qstar}. This replacement holds because
$\epsilon_\lambda=1+\lambda c$ is
the RPA dielectric function and
$c=q_{\rm TF}^2\mathcal L/q^2\simeq q_{\rm TF}^2/(3\nu^2)=\omega_p^2/\nu^2$ for
$\nu\gg q+q^2/2$, the largest particle--hole energy at momentum
transfer $q$. The relative error of this high-frequency form is at most
$(q+q^2/2)^2/\nu^2$. The high-frequency form is the free
value of \eqref{eq:fsum-exact}, and it follows from the representation
of $\mathcal L$ and the $f$-sum identity in SI
Sec.~\ref{S-app-sec:sigma-lowdensity}. For $q\ll q_*$ its range of validity
contains the frequencies $\nu\sim\omega_p=q_*^2/2$ at which the factor
changes. For $\nu\lesssim q$, both $1/\epsilon_\lambda$ and its
replacement are $O(q^2/(\lambda q_{\rm TF}^2))$ [SI
Eq.~\eqref{S-app-eq:lowdensity-u13}]. The line $v/\epsilon_\lambda^2$ in
\eqref{eq:acfd-linear} thus suppresses the frequencies $\nu\lesssim\omega_p$. A tail
$A/\nu^2$ of $\delta\Pi_\lambda$, with $A$ independent of $\nu$, then
contributes to the frequency integral of \eqref{eq:acfd-linear}
\begin{equation}
 \int_0^\infty\frac{\mathrm d\nu}{\pi}\,
 \frac{\nu^4}{(\nu^2+\lambda\omega_p^2)^2}\,\frac{A}{\nu^2}
 =\frac{A}{4\sqrt\lambda\,\omega_p}\propto\frac{A}{q_*^2},
 \label{eq:acfd-lowdensity}
\end{equation}
at fixed $q$. When the tail is integrable in $q$, this contribution gives the
universal pole at $s=-2$ [cf.~\eqref{eq:fsum-pole}]. For the kite alone
the high-frequency tail does not vanish. It tends to a constant at large $q$ and is
therefore not integrable, so the leading term is instead the kernel pole at
$s=-1$, the $r_s^{-1/4}$ term of Sec.~\ref{sec:lowdensity}. In the sum
with $\Pi^{(1)}_{\rm SE}$ the high-frequency tails cancel for every $q$
[SI Eq.~\eqref{S-app-eq:highdensity-u6}]. The cancellation of the
$r_s^{-1/4}$ term is established below from the combined kernel.

An energy computed from a response that violates one of these two
constraints thus contains a term that the constraint removes once the
accompanying diagrams of the same order are included. Closed with the
bare line, the self-energy insertions give zero, as
shown below. The second-order exchange energy is therefore
complete by itself, and the insertions are needed only when
the line is screened.

In terms of reference states, the Fock self-energy insertions are the first-order
change of the orbital energies. A Kohn--Sham reference with a local
potential shifts all orbital energies by a constant and lacks the
insertions. The dynamically screened one-line kite evaluated on such a
reference therefore retains the
$r_s\log^2r_s$ and $r_s^{-1/4}$ terms, whereas Hartree--Fock orbital
energies contain the insertions to all orders. SI
Sec.~\ref{S-sm-sec:references} discusses these references, other SOSEX
variants
\cite{Jansen2010,Angyan2011,BatesFurche2013,Ren2013,HesselmannGorling2010,Gorling2019,Hellgren2018,MaggioKresse2016},
and a related cancellation in the structure factor observed by Yasuhara
\emph{et al.} \cite{Yasuhara1988}.

The cancellation follows from the frequency denominators. With
$g_p=(\mathrm i\nu-\delta_{\bm p})^{-1}$ and
$g_k=(\mathrm i\nu-\delta_{\bm k})^{-1}$, the self-energy shift
$\Delta\Sigma_x$ of the excitation energy, an integral over the second
loop momentum, turns $g_p$ into $g_p+\Delta\Sigma_xg_p^2$ to first order (SI
Sec.~\ref{S-app-sec:response}). Since
$\Delta\Sigma_x(\bm p)=\Sigma_x(\bm p+\bm q)-\Sigma_x(\bm p)$ is itself
an integral over $\bm k$ of $\Delta_{\bm k}|\bm p-\bm k|^{-2}$, the
self-energy and vertex terms share the factor
$\Delta_{\bm p}\Delta_{\bm k}|\bm p-\bm k|^{-2}$ and differ only in
their frequency factors. This common factor is the form the Ward
identity takes here. After symmetrization in $\bm p\leftrightarrow\bm k$,
the self-energy terms give $g_p^2+g_k^2$ and the vertex term gives
$-2g_pg_k$:
\begin{equation}
 g_p^2+g_k^2-2g_pg_k=(g_p-g_k)^2
 =\frac{[\bm q\cdot(\bm p-\bm k)]^2}
 {(\mathrm i\nu-\delta_{\bm p})^2(\mathrm i\nu-\delta_{\bm k})^2}.
 \label{eq:response-square-reminder}
\end{equation}
The numerator vanishes quadratically at coincident loop momenta. This
quadratic zero removes the Coulomb singularity $|\bm p-\bm k|^{-2}$ from the combined
integrand. The numerator also suppresses the same-ball regions by $q^{-2}$. The
self-energy radial kernel therefore has the two endpoint terms of the kite
with opposite coefficients. At small momentum transfer the cusp $\propto u\log|u|$ of
the Fock self-energy at the Fermi surface, where $|\bm p|=1+u$, integrated across a shell of
width $O(q)$, gives
\begin{equation}
 \Pi_{\rm vtx}^{(1)}(q,\mathrm iqz)
 =-\frac{q_{\rm TF}^4}{4}[h+zh']\log\frac1q+O(q_{\rm TF}^4),\qquad
 \Pi_{\rm SE}^{(1)}(q,\mathrm iqz)
 =+\frac{q_{\rm TF}^4}{4}[h+zh']\log\frac1q+O(q_{\rm TF}^4).
 \label{eq:ward}
\end{equation}
Here $h=h(z)$ and $h'=\mathrm dh/\mathrm dz$ as in \eqref{eq:Az}, so
that the vertex term is $-q_{\rm TF}^4a_x(z)\log(1/q)/(32\pi^2)$. The
remainders refer to $q\to0$ at fixed $z$.
The two terms in $\log(1/q)$ are derived in SI Secs.~\ref{S-app-sec:shell-origin}
and \ref{S-app-sec:cancel} [Eqs.~\eqref{S-app-eq:Pix-norm},
\eqref{S-app-eq:PiSigma}, \eqref{S-app-eq:Az}, \eqref{S-app-eq:fock-shell} and
\eqref{S-app-eq:highdensity-u4}]. Their
cancellation persists after the frequency integration, direction
by direction [SI Eq.~\eqref{S-app-eq:highdensity-u5}]. At large momentum transfer the
self-energy radial kernel, defined as in \eqref{eq:radial-kernel}, is $e_{\Sigma,\Lambda}(q)=-3p_{\rm loc}(w_\Lambda)/(\pi^3q^2)+O(q^{-4})$,
the term of \eqref{eq:largeq-block} with opposite sign (SI
Secs.~\ref{S-app-sec:sigma-lowdensity} and \ref{S-sm-sec:verification}).

The same balance between vertex and self-energy terms underlies the $GW\Gamma$ scheme of Takada
\cite{Takada2001,MaebashiTakada2011}, which adds a vertex function $\Gamma$ to the $GW$ approximation. This scheme constructs the density
vertex as a ratio of inverse propagators consistent with the Ward
identity and describes the remaining part by an electron--hole
interaction. At first order, the \emph{Ward part} of this construction, the part fixed by the
Ward identity, is the summand $-(g_p^2+g_k^2)$ of the vertex term $-2g_pg_k$ in
\eqref{eq:response-square-reminder}. This part cancels the self-energy
terms $g_p^2+g_k^2$, and the remainder $(g_p-g_k)^2$ gives the
finite limit $R^{(1)}$ of Sec.~\ref{sec:Pi1-d} (SI
Sec.~\ref{S-sm-sec:gwgamma}).

The bare-line value $\mathcal K_{\Pi_{\rm SE}^{(1)}}[v]$ vanishes,
consistent with second-order perturbation theory, whose energy consists
of ring and exchange diagrams only. With one line screened, the same
topology is the second diagram in
Fig.~6(b) of Ref.~\cite{Benites2024}, which BRM set to zero because
it contains two free propagators with the same four-momentum. For a
frequency-independent insertion the frequency integral of $G_0(k)^2$
leaves only a term localized at the Fermi surface. This term vanishes for the fixed-density
insertion $\Sigma_x-\delta\mu_x$, so the closure with the bare line
vanishes [SI Eq.~\eqref{S-app-eq:selfenergyreduction-u7}]. With
$W_\lambda$ the insertion depends on frequency, and its closure does not
vanish in general. The diagram then has the topology of the self-energy
pair closed with the screened line.

By the criterion of Sec.~\ref{sec:cancellation-general}, the opposite
endpoint terms remove the triple pole at $s=2$ and the pole at $s=-1$ that
produce the leading terms of the kite. With
$e_{\rm SE}=2\int_0^1\mathcal K_{\Pi_{\rm SE}^{(1)}}[W_\lambda]\,\mathrm d\lambda$,
\begin{equation}
 e_{\rm SE}(r_s)=-\frac{A_1}{q_*}+O(q_*^{-3/2}),\qquad
 e_{\rm kite}(r_s)+e_{\rm SE}(r_s)=O(q_*^{-3/2})=O(r_s^{-3/8})
 \label{eq:low-density-response-cancellation}
\end{equation}
(SI Sec.~\ref{S-app-sec:sigma-lowdensity} bounds the remainders). At
high density the two $r_s\log^2r_s$ terms cancel in the same way. Since
$\mathcal K_{\Pi_{\rm SE}^{(1)}}[v]=0$, the
high-density bounds of Sec.~\ref{sec:highdensity} and SI
Sec.~\ref{S-app-sec:cancel} leave
$e_{\rm kite}+e_{\rm SE}=e_{2x}+O(r_s|\log r_s|)$. The high-frequency tail of the
sum vanishes for every $q$ [SI Eq.~\eqref{S-app-eq:highdensity-u6}],
whereas the tails of the two corrections are not separately
integrable. The vanishing of the summed tail is the tail condition of
SI Lemma~\ref{S-lem:fsum} (Sec.~\ref{sec:mellin-convolution}). However, the remainder condition of the lemma is
not verified here for the summed kernel. The
absence of the universal pole at $s=-2$ therefore follows from the lemma only
under that condition. The proved bound $O(q_*^{-3/2})$ cannot
resolve a term of order $q_*^{-2}$. By
\eqref{eq:response-square-reminder} the combined kernel decays as
$q^{-4}$, so its first possible kernel pole is at $s=-3$. A nonzero
residue there would give a sum of order $q_*^{-3}\propto r_s^{-3/4}$,
whereas the proved bound is $O(r_s^{-3/8})$. The
self-energy pair thus cancels the $r_s\log^2r_s$ term and the $r_s^{-1/4}$
term of every row of table~\ref{app-tab:static-dynamic}, but not the
constant $-e_{2x}$ of $e_{\rm closed}$.

In two dimensions, the same cancellation operates at the resonance of
Sec.~\ref{sec:twod}. By property (iii) of Sec.~\ref{sec:R2}, the vertex
part alone has the kernel endpoint $q\log(1/q)$, and the coincident pole at
$s=2$ becomes triple. The vertex correction inserted alone into the
ring chain, the two-dimensional counterpart of the chain family of the
kite, would contribute
\begin{equation}
 e_{\rm vtx}^{\rm 2D}\supset\frac{2}{3\pi^2}\,(10-3\pi)\,r_s^2\log^2r_s
 =0.03885\,r_s^2\log^2r_s\ {\rm Ry}.
 \label{eq:vtx2d}
\end{equation}
The coefficient follows from \eqref{eq:Pi1-energy} with $R_d$ replaced by
$\pi^{-1}[h_d+zh_d']\log(1/q)$, which turns the double pole of
\eqref{eq:Pi1-resonant} into a triple pole. The derivation also uses the identity
$\int h_d^{N_*}(h_d+zh_d')\,\mathrm dz=[N_*/(N_*+1)]\int h_d^{N_*+1}\,\mathrm dz$,
obtained by integrating $zh_d^{N_*}h_d'$ by parts (SI
Sec.~\ref{S-dsec:R2}). Since $N_*+1=(d+1)/(d-1)=N_{\rm r}$, the integral on
the right is that of the $\log r_s$ term \eqref{eq:ring-resonant} of the rings at the same
dimension. In three dimensions the same steps give $3c_{\rm hd}\,r_s\log^2r_s$,
the chain-family coefficient of \eqref{eq:Clog2}. In both dimensions,
therefore, the $\log^2r_s$ term of the isolated vertex correction is fixed by the integral that
gives the $\log r_s$ term of the rings, $\int h^2\,\mathrm dz$ in three dimensions and
$\int h_2^3\,\mathrm dz$ in two. The self-energy pair removes this $\log^2r_s$ term,
and the conserving sum leaves the single $\log r_s$ of \eqref{eq:a2l}.

\section{Conclusions}\label{sec:summary}
For the ring series and for diagrams linear in one screened interaction
line in $d$ dimensions, the powers of $r_s$ and of $\log r_s$ arise from
residues of a Mellin convolution (table~\ref{tab:catalogue}). Screening
places its poles at the multiples of $d-1$. The phase space places the
kernel pole at $s=d+1$ for the rings and at $s=2$ for a polarization
insertion. Once the Fermi-surface factors $\log(1/q)$ of the vertex and
self-energy corrections cancel, as the Ward identity requires, a $\log r_s$
term appears where two poles coincide. This rule reproduces
the known $\log r_s$ terms in two and three dimensions and the ring
power law for $d>3$.

In two dimensions the insertion pole at $s=2$ meets the second screening
pole, and the residue gives
\begin{equation}
 \epsilon_c^{\rm 2D}(r_s)=a_0+a_\ell\,r_s\log r_s+a_1\,r_s
 +a_{2\ell}\,r_s^2\log r_s+O(r_s^2),\qquad
 a_{2\ell}=\frac7\pi+\frac{4\log2}\pi-\frac{160}{9\pi^2}-\frac{16G}{\pi^2},
 \label{eq:2d-result}
\end{equation}
$a_{2\ell}=-0.17546$~Ry, with $G$ Catalan's constant.
This value is the coefficient of the series with one insertion of the
first-order exchange polarization $\Pi^{(1)}$ into the ring chain. By
power counting with the loop cancellation in symmetrized fermion loops,
diagrams with two or more soft loops start at $r_s^4$. Under hypotheses
(H1) and (H2) of Sec.~\ref{sec:completeness}, $a_{2\ell}$ is the
complete coefficient. It constrains parametrizations of the two-dimensional correlation energy fitted to
quantum Monte Carlo data, such as that of Ref.~\cite{Attaccalite2002}.

In three dimensions the screened kite alone acquires $r_s\log^2r_s$ and
$r_s^{-1/4}$ terms. The Fock self-energy corrections remove both terms,
and the remainders of the sum are bounded. Screened-exchange corrections to RPA-type
correlation energies should therefore be evaluated together with the
self-energy corrections. This conclusion holds at least for the leading terms of energy
expressions with one screened line, without the Legendre-transform terms of Sec.~\ref{sec:se-cancellation} or
self-consistency.

\begin{table}[tb]
\centering\footnotesize
\caption{Asymptotic terms from the endpoint--residue correspondence
(Sec.~\ref{sec:mellin-convolution}, SI Proposition~\ref{S-prop:endpoint-residue}), in Ry per electron. The columns ``Kernel'' and
``$g$'' give the endpoint terms whose poles produce the term, with
$t=q/\varepsilon$. SE denotes the Fock self-energy corrections. In the
Kernel column, ``(Ward)'' marks the kernel of the sum of vertex and self-energy
parts, in which their Fermi-surface factors $\log(1/q)$ cancel (Sec.~\ref{sec:cancellation-general}).
The entry ``opposite'' marks the equal and opposite endpoint terms of the kite and SE
(Sec.~\ref{sec:se-cancellation}). The entry ``in rings'' denotes one insertion into the
ring chain. The
entries of the $g$ column are shown without the factors $h_d(z)^N$ of
Sec.~\ref{sec:screening-functions} (for the kite, $-h$ at full coupling
and $-2h/3$ for the coupling average, table~\ref{tab:screening}). The last
column gives the reference for a known coefficient, notes
a proved remainder bound, or refers to the completeness statement under
hypotheses (H1) and (H2) of Sec.~\ref{sec:completeness} (SI
Proposition~\ref{S-prop:complete}). The two kite coefficients agree with those
of the companion paper \cite{Imoto2026b}.}
\label{tab:catalogue}
\setlength{\tabcolsep}{1.8pt}
\begin{tabular}{@{}lllllll@{}}
\toprule
Quantity (Sec.) & $d$ & Kernel & $g$ & Pole & Term & Remark\\
\midrule
RPA ring (\ref{sec:ring-d}) & 3 & $q^3$ & $t^{-4}$ & double, $s=4$ & $\frac{2}{\pi^2}(1-\log2)\log r_s$ & Ref.~\cite{GellMannBrueckner}\\
RPA ring (\ref{sec:ring-d}) & 2 & $q^2$ & $t^{-3}$ & double, $s=3$ & $-2\sqrt2\bigl(\frac{10}{3\pi}-1\bigr)r_s\log r_s$ & Ref.~\cite{RajagopalKimball1977}\\
RPA ring (\ref{sec:ring-d}) & $>3$ & $q^d$ & $t^{-N(d-1)}$ & simple, $s=d+1$ & $A_d\,r_s^{(3-d)/(d-1)}$ & \eqref{eq:Ad}, cf.~\cite{Schlesier2020}\\
$\Pi^{(1)}$ in rings (\ref{sec:Pi1-d}) & 3 & $qR_3$ (Ward) & $t^{-2}$ & double, $s=2$ & $0.0130\,r_s\log r_s$ & Ref.~\cite{CarrMaradudin}\\
$\Pi^{(1)}$ in rings (\ref{sec:a2l}) & 2 & $qR_2$ (Ward) & $t^{-2}$ & double, $s=2$ & $a_{2\ell}\,r_s^2\log r_s$ & SI Prop.~\ref{S-prop:complete}\\
$\Pi^{(1)}_{\rm vtx}$ in rings (\ref{sec:se-cancellation}) & 2 & $q\log(1/q)$ & $t^{-2}$ & triple, $s=2$ & $0.0389\,r_s^2\log^2r_s$ & canceled by SE\\
Kite, $r_s\to0$ (\ref{sec:highdensity}) & 3 & $q\log(1/q)$ & $t^{-2}$ & triple, $s=2$ & $c_{\rm hd}r_s\log^2r_s$ & remainder bound\\
Kite, $r_s\to\infty$ (\ref{sec:lowdensity}) & 3 & $q^{-2}$ & $\theta_\lambda p_{\rm loc}$ & simple, $s=-1$ & $A_1/q_*\propto r_s^{-1/4}$ & remainder bound\\
Kite $+$ SE, $r_s\to0$ (\ref{sec:se-cancellation}) & 3 & opposite & --- & canceled & $e_{2x}+O(r_s|\log r_s|)$ & remainder bound\\
Kite $+$ SE, $r_s\to\infty$ (\ref{sec:se-cancellation}) & 3 & opposite & --- & canceled & $O(r_s^{-3/8})$ & remainder bound\\
\bottomrule
\end{tabular}
\end{table}

Open problems are a proof of hypotheses (H1) and (H2) of
Sec.~\ref{sec:completeness}, the constant $a_1$ of the
two-dimensional gas, the next $\log r_s$ terms in three dimensions,
which belong to the complete second-order polarization in the ring chain
(Sec.~\ref{sec:completeness}), and a proof of the
$f$-sum-rule cancellation of SI Lemma~\ref{S-lem:fsum} beyond first order.

\section*{Acknowledgments}
This work received no external funding. The author used Claude
(Anthropic) and ChatGPT (OpenAI) for assistance with analytic
manipulations and numerical code development, reviewed the resulting
material, and takes responsibility for the content.

\section*{Data availability statement}
The Supporting Information provides the derivations and numerical data
for the reported checks, including the $d$-dimensional coefficients and the
closed forms of $R_2$ and $a_{2\ell}$. The Python scripts and raw output are
provided as a supplementary code archive. The code for the direct numerical
integration of the kite used in Secs.~\ref{sec:kite} and \ref{sec:highdensity} and the scripts for
most checks of SI Sec.~\ref{S-sm-sec:verification} are available in a
public repository \cite{ImotoCode}. Other data are available from the
author upon reasonable request.

\section*{Conflict of interest}
The author has no conflicts to disclose.

\section*{Supporting Information}
The Supporting Information (SI) contains the $d$-dimensional
derivations behind Secs.~\ref{sec:ring-d}, \ref{sec:Pi1-d} and \ref{sec:twod}: the
long-wavelength Lindhard shape $h_d$, the ring coefficients, the long-wavelength
limit $R_d$ of the first-order exchange polarization, the closed form of $R_2$, the coefficient
$a_{2\ell}$, the finite-$q$ checks, the order counting behind the
completeness statement of Sec.~\ref{sec:completeness} (SI
Proposition~\ref{S-prop:complete}) and the spin dependence. It also contains the
normalization of the kite functional and the derivations behind Secs.~\ref{sec:overview}--\ref{sec:mellin-convolution} and
\ref{sec:kite}--\ref{sec:se-cancellation}, the statement and proof of the endpoint--residue
correspondence (SI Proposition~\ref{S-prop:endpoint-residue}),
the proof of the cancellation of the universal low-density residue (SI
Lemma~\ref{S-lem:fsum}) at first order and the argument for
higher orders, the proof of SI Proposition~\ref{S-prop:kite-chain}, and
the numerical checks. Its sections are cited as SI Sec.~S$n$, and its equations,
figures and tables carry the prefix S.

\clearpage
\makeatletter\@ifundefined{sectionname}{}{\let\sectionname\@empty}\makeatother
\ifcsname c@section\endcsname\setcounter{section}{0}\fi
\ifcsname c@equation\endcsname\setcounter{equation}{0}\fi
\ifcsname c@figure\endcsname\setcounter{figure}{0}\fi
\ifcsname c@table\endcsname\setcounter{table}{0}\fi
\ifcsname c@lemma\endcsname\setcounter{lemma}{0}\fi
\ifcsname c@proposition\endcsname\setcounter{proposition}{0}\fi
\renewcommand{\thelemma}{S\arabic{lemma}}
\renewcommand{\theproposition}{S\arabic{proposition}}
\renewcommand{\thesection}{S\arabic{section}}
\renewcommand{\theequation}{S\arabic{equation}}
\renewcommand{\thefigure}{S\arabic{figure}}
\renewcommand{\thetable}{S\arabic{table}}
\def\theHsection{S.\arabic{section}}
\def\theHequation{S.\arabic{equation}}
\def\theHfigure{S.\arabic{figure}}
\def\theHtable{S.\arabic{table}}
\def\theHlemma{S.\arabic{lemma}}
\def\theHproposition{S.\arabic{proposition}}
\begin{center}
{\Large\bfseries Supplementary Material\\[1ex]\large Supporting Information for ``Revisiting the $r_s$ expansion of the uniform electron gas in $d$ dimensions: $\log r_s$ terms from coincident Mellin poles''}
\end{center}

\noindent This document contains the normalization, the proofs of the
remainder bounds and the numerical details behind the main article.
Table~\ref{S-sm-tab:guide} lists the sections that support each result
of the main text. Equation, section, figure and
table numbers without the prefix S refer to the main article, and
numbers with the prefix S refer to this document.

\begin{table}[h]
\centering\small
\caption{Where each result of the main article is supported.}
\label{S-sm-tab:guide}
\setlength{\tabcolsep}{3.8pt}
\begin{tabular}{@{}lll@{}}
\toprule
Result & Main text & SI sections\\
\midrule
$d$-dimensional constants, long-wavelength Lindhard shape $h_d$ &Sec.~\ref{sec:units} & \ref{S-dsec:lindhard}\\
Ring series in $d$ dimensions, $A$, $a_\ell$, $A_d$ & Sec.~\ref{sec:ring-d} & \ref{S-dsec:ring}, \ref{S-sm-slim-sec:residues}\\
Fermi-surface form of $\Pi^{(1)}$, $R_d$, Carr--Maradudin check & Sec.~\ref{sec:Pi1-d} & \ref{S-dsec:response}, \ref{S-sm-slim-sec:residues}, \ref{S-sm-sec:verification}\\
Closed form of $R_2$, vertex and self-energy $\log(1/q)$ terms & Secs.~\ref{sec:R2}, \ref{sec:se-cancellation} & \ref{S-dsec:R2}\\
Coefficient $a_{2\ell}$ & Sec.~\ref{sec:a2l} & \ref{S-dsec:a2l}\\
Finite wave number in two dimensions & Sec.~\ref{sec:twod} & \ref{S-dsec:lindhard}, \ref{S-dsec:finiteq}\\
Completeness (Proposition~\ref{S-prop:complete}), order counting & Sec.~\ref{sec:completeness} & \ref{S-dsec:counting}\\
Spin polarization in two dimensions & Sec.~\ref{sec:a2l} & \ref{S-dsec:spin}\\
Normalization, kite functional, energy expressions & Secs.~\ref{sec:units}, \ref{sec:diagram-coefficients} & \ref{S-app-sec:setting}, \ref{S-sm-app-app:brm}, \ref{S-sm-sec:diagram-weights}\\
Spectral representation and its convergence & Secs.~\ref{sec:spectral-rep}, \ref{sec:radial-general} & \ref{S-app-sec:spectral}\\
Endpoint--residue correspondence, Corollaries~\ref{cor:exponents}, \ref{cor:logs}, Lemma~\ref{S-lem:fsum} & Secs.~\ref{sec:screening-functions}, \ref{sec:mellin-convolution} & \ref{S-sm-slim-sec2}, \ref{S-app-sec:largeq-screening}\\
Kite, $c_{\rm hd}\,r_s\log^2r_s$, with remainder & Secs.~\ref{sec:radial-kernels}, \ref{sec:highdensity} & \ref{S-app-sec:highdensity}, \ref{S-sm-sec:diagram-weights}, \ref{S-app-app:shell-integrated}\\
Kite, $A_1/q_*\propto r_s^{-1/4}$, with remainder & Secs.~\ref{sec:radial-kernels}, \ref{sec:lowdensity} & \ref{S-sm-sec:diagram-weights}, \ref{S-app-sec:dispersion}, \ref{S-sm-sec:lowdensity-remainders}\\
Cancellation by the Fock self-energy corrections & Sec.~\ref{sec:se-cancellation} & \ref{S-app-sec:response}, \ref{S-app-sec:cancel}, \ref{S-sm-slim-sec2}, \ref{S-app-sec:sigma-reduction}, \ref{S-app-sec:dispersion}\\
Checks of identities and coefficients & --- & \ref{S-sm-sec:verification}\\
\bottomrule
\end{tabular}
\end{table}

\section*{Notation}
The dimension of space is $d$. Sections~\ref{S-dsec:lindhard}--\ref{S-dsec:spin}
treat general $d$, with the notation of main-text Secs.~\ref{sec:units}
and \ref{sec:ring-d}, \ref{sec:Pi1-d} and \ref{sec:twod}: $h_d$ is the long-wavelength
Lindhard shape [$h\equiv h_3$], $g_d$ its $q^2$ correction, $R_d$ the
long-wavelength first-order exchange polarization with $R_3=R^{(1)}/(4\pi)$,
$\mathcal A_d$, $\mathcal B_d$, $\mathcal C_d$ and $\alpha_d$ the constants
of table~\ref{S-dtab:constants}, $N_*=2/(d-1)$ and $N_{\rm r}=(d+1)/(d-1)$
the indices of the resonant screening poles for a polarization
insertion and for the rings, respectively, $A_d$ the ring coefficient for $d>3$ and
$a_{2\ell}$ the coefficient of $r_s^2\log r_s$ in two dimensions. The
proofs of Sec.~\ref{S-sm-slim-sec:corollaries} also hold in $d$ dimensions.
All other sections treat $d=3$. As in the main text, $A\supset B$ means
that the density expansion of $A$ contains the term $B$.
For $d=3$ the conventions are those of main-text Sec.~\ref{sec:setting}
(rederived from dimensional quantities in Sec.~\ref{S-app-sec:setting}): momenta in
units of $k_F$, frequencies in units of $\hbar k_F^2/m_e$, energies per
electron in Ry, the imaginary frequency $\mathrm i\nu$ with $\nu>0$, and
\[
 \alpha_L=[4/(9\pi)]^{1/3},\qquad q_{\rm TF}^2=4\alpha_Lr_s/\pi,\qquad
 A_0=3/(32\pi^5),\qquad v(q)=q^{-2}.
\]
The screened line at coupling $\lambda$ and its coupling average are
\[
 W_\lambda=\lambda v/\epsilon_\lambda,\qquad
 \epsilon_\lambda=1-\lambda v\Pi_0,\qquad \Pi_0=-q_{\rm TF}^2\mathcal L,\qquad
 \overline W=\int_0^1W_\lambda\,\mathrm d\lambda .
\]
Main-text Sec.~\ref{sec:diagram-coefficients} defines the kite functional $\mathcal K[f]$ and
\[
 e_{2x}=\mathcal K[v],\qquad e_{\rm kite}=2\mathcal K[\overline W],\qquad
 e_{\rm closed}=2\mathcal K[W_{\lambda=1}]-e_{2x}.
\]
In remainder estimates, $C_{\rm b}$ is a positive constant that may change
from one estimate to the next and does not depend on the small or large
parameter of that estimate ($q$, $r_s$, $\Lambda$ or $\varepsilon$, as stated there).
The elementary functions used throughout are
\begin{equation}
 \begin{aligned}
 h(z)&=1-z\arctan(1/z),&\quad \bar a_x(z)&=h(z)+zh'(z),\\
 f_L(k)&=\frac12+\frac{1-k^2}{4k}
                   \log\left|\frac{1+k}{1-k}\right|.
 \end{aligned}
 \label{S-sm-eq:elementary-functions}
\end{equation}
Here $z=\nu/q>0$, $h'=\mathrm dh/\mathrm dz$, $\bar a_x=a_x/(8\pi^2)$
with $a_x$ of main-text Eq.~\eqref{eq:Az}, and the endpoint values
of $f_L$ are defined by real-axis limits. The two uses of $f_L$ are
\[
 \Sigma_x(k)=-q_{\rm TF}^2 f_L(k)/2\ \ \text{(Fock self-energy)},\qquad
 \mathcal L(q,0)=f_L(q/2)\ \ \text{(static Lindhard response)}.
\]
The first is derived after Eq.~\eqref{S-app-eq:definitions-u6}. The second is the static
Lindhard function of Ref.~\cite{SI-FetterWalecka1971}, and for $0<q<1$ it
also follows from Eq.~\eqref{S-app-app-eq:lindhard-shell} at $z=0$.
Abbreviations: RPA (random-phase approximation), RPAx (RPA with
exchange), BRM (Benites, Rosado and Manousakis \cite{SI-Benites2024}),
SOSEX (second-order screened exchange), CHSS (Czachor, Holas, Sharma and
Singwi \cite{SI-CHSS1982}), TF (Thomas--Fermi) and $GW\Gamma$ (the $GW$ approximation with
a vertex function $\Gamma$).
Table~\ref{S-sm-tab:symbols} lists the main symbols and terms with the
place where each is defined.

\begin{table}[htbp]
\centering\footnotesize
\caption{Main symbols and terms, with the place where each is defined in the main
article (M) or in this document.}
\label{S-sm-tab:symbols}
\scriptsize
\setlength{\tabcolsep}{4pt}
\renewcommand{\arraystretch}{0.95}
\begin{tabular}{@{}lp{0.54\textwidth}l@{}}
\toprule
Symbol or term & Meaning & Defined in\\
\midrule
$r_s$, $k_F$, $\alpha_L$ & density parameter, Fermi wave number, $\alpha_L=[4/(9\pi)]^{1/3}$ & M Sec.~\ref{sec:units}\\
$q$, $\nu$, $z$ & momentum transfer of the screened line (units of $k_F$), imaginary frequency, scaled frequency $z=\nu/q$ & M Sec.~\ref{sec:units}\\
$\lambda$, $v$, $W_\lambda$, $\overline W$ & coupling strength, $v=q^{-2}$, RPA line at coupling $\lambda$, $\int_0^1W_\lambda\,\mathrm d\lambda$ & M Eqs.~\eqref{eq:Dlambda}, \eqref{eq:energy-definitions}\\
$\Pi$, $\chi$, $\Pi_0$, $\epsilon_\lambda$ & proper and full density response, free response, dielectric function & M Eq.~\eqref{eq:response-conventions}\\
$q_{\rm TF}$, $\mathcal L$, $h(z)$ & Thomas--Fermi wave number, Lindhard function, its small-$q$ limit & M Eqs.~\eqref{eq:lindhard-convention}, \eqref{eq:hz}\\
$h_d(z)$, $g_d(z)$ & long-wavelength Lindhard shape in $d$ dimensions ($h=h_3$), its $O(q^2)$ correction & M Eq.~\eqref{eq:hz}, Sec.~\ref{S-dsec:lindhard}\\
$R_d(z)$, $R_2(z)$ & long-wavelength first-order exchange polarization, its closed form for $d=2$ & M Eqs.~\eqref{eq:Rd}, \eqref{eq:R2}\\
$\mathcal A_d$, $\mathcal B_d$, $\mathcal C_d$, $\alpha_d$, $N_*$, $N_{\rm r}$ & $d$-dimensional constants, indices of the resonant screening poles (insertion, rings) & M Eqs.~\eqref{eq:alpha-d}, \eqref{eq:lindhard-convention}, \eqref{eq:ring-resonance}, \eqref{eq:Pi1-resonant}\\
$A_d$, $a_\ell$, $a_{2\ell}$ & ring coefficient for $d>3$, two-dimensional coefficients of $r_s\log r_s$ and $r_s^2\log r_s$ & M Eqs.~\eqref{eq:Ad}, \eqref{eq:2d-expansion}, \eqref{eq:a2l}\\
$e_{\rm RPA}$, $e_{\Pi^{(1)}}$, $P_2$ & ring energy, energy of the series with one insertion of $\Pi^{(1)}$, second-order coefficient of the proper polarization at fixed $z$ & M Eqs.~\eqref{eq:ring-energy}, \eqref{eq:Pi1-energy}, Prop.~\ref{S-prop:complete}\\
$c$, $\rho(c)$ & screening ratio $c=-v\Pi_0$, $2\overline W=v\rho(c)$ & M Eqs.~\eqref{eq:cdef}, \eqref{eq:rho}\\
$q_*$, $\omega_p$ & low-density scale, plasma frequency & M Eq.~\eqref{eq:qstar}\\
$\theta_\lambda$, $\Lambda_\lambda$, $w_\Lambda$ & weight and width of the single Lorentzian, $w_\Lambda=\Lambda/q$ & M Eq.~\eqref{eq:pole-dispersion}, Sec.~\ref{sec:mellin-convolution}\\
$\hat c$, $\varsigma_\lambda$ & leading large-$q$ form of the screening ratio, remainder of the single-Lorentzian form & M Eq.~\eqref{eq:pole-dispersion}\\
$B_{W,\lambda}$, $a_\lambda$, $b_\Lambda$, $\Lambda$ & loss spectrum, spectral weight, Lorentzian line, phase velocity (Lorentzian width in $z$) & M Eqs.~\eqref{eq:loss-spectrum}, \eqref{eq:block-superposition}\\
$\Omega$ & $\Omega=q\Lambda$, Lorentzian width in $\nu$ & M Eq.~\eqref{eq:sigma-frequency-main}\\
$D$, $n$, $\mathcal J_D[f]$, $Y_D$, $\Pi_D$ & diagram, number of bare interactions, functional linear in the line $f$, its integrand, polarization insertion & M Eq.~\eqref{eq:general-functional}\\
$S_D$ & reciprocal of the closed-diagram coefficient & M Sec.~\ref{sec:linear-diagrams}\\
$w_D$, $e_D$, $F_D$, $M(p)$ & coupling weight, coupling-averaged and full-coupling energies of $D$, Mellin multiplier & M Eq.~\eqref{eq:general-average}, Sec.~\ref{sec:mellin-convolution}\\
$e_{D,\Lambda}$, $e_{D,\infty}$, $e_\Lambda$, $e_\infty$ & radial kernel of $D$ for a Lorentzian line and for the bare line, the same for the kite & M Eqs.~\eqref{eq:radial-kernel}, \eqref{eq:XLambda}\\
$k$, $g$, $\varepsilon$, $t$, $s$ & radial kernel, scaling function, scale, $t=q/\varepsilon$, Mellin variable & M Eqs.~\eqref{eq:mellin-convolution}, \eqref{eq:mellin-parseval}\\
$k_D$, $\kappa_D$, $\gamma_N$ & $z$-resolved radial kernel, its coefficient of $q\log^m(1/q)$, coefficient of the $t^{-N(d-1)}$ tail & M Eq.~\eqref{eq:kD}, Cor.~\ref{cor:logs}\\
$A_D$, $\mathcal S$ & coefficient of the high-frequency tail $Y_D\simeq A_D/\nu^2$, set of diagrams of Lemma~\ref{S-lem:fsum} & M Corollary~\ref{cor:exponents}, Lemma~\ref{S-lem:fsum}\\
$\mathcal K[f]$, $X$, $A_0$, $\Delta_{\bm p}$, $\delta_{\bm p}$ & kite functional, kite integrand, prefactor, occupation difference, excitation energy & M Eqs.~\eqref{eq:KD}, \eqref{eq:Xdef}\\
$e_{2x}$ & bare second-order exchange energy & M Eq.~\eqref{eq:bare-value}\\
$F$, $e_{\rm kite}$, $e_{\rm closed}$ & full-coupling, coupling-averaged and chain-family energies of the kite & M Eqs.~\eqref{eq:energy-definitions}, \eqref{eq:closed-family}\\
$a_x$, $b_x$, $\kappa$, $\beta$ & coefficients of $\log(1/q)$ and finite parts at small $q$ of the kite integrand (first two) and of the radial kernel of the kite (last two) & M Eq.~\eqref{eq:Az}, Eq.~\eqref{S-app-eq:kappa}\\
$p_{\rm loc}$ & frequency factor of the same-ball regions & M Eq.~\eqref{eq:largeq-block}\\
$\mathbb B$ & unit Fermi ball & M Sec.~\ref{sec:radial-kernels}\\
$c_{\rm hd}$, $c_1$, $A_1$ & coefficients of $r_s\log^2r_s$, $r_s\log r_s$ and $q_*^{-1}$ in $e_{\rm kite}$ & M Eqs.~\eqref{eq:Clog2}, \eqref{eq:A1-dispersion}\\
$\Sigma_x$, $\delta\mu_x$, $e_{\rm SE}$ & Fock self-energy, its value at $k_F$ (the fixed-density shift of the chemical potential), energy of the self-energy pair & M Eqs.~\eqref{eq:KT-sigma}, \eqref{eq:low-density-response-cancellation}\\
$\Pi^{(1)}$, $\Pi^{(1)}_{\rm vtx}$, $\Pi^{(1)}_{\rm SE}$ & first-order exchange polarization, its vertex and self-energy parts & Eqs.~\eqref{S-app-eq:Pix-norm}--\eqref{S-app-eq:combined-response}\\
$\mathcal K_{\Pi}[f]$, $R^{(1)}$ & closure of an insertion $\Pi$ with the kite normalization, small-$q$ limit of $-16\Pi^{(1)}/q_{\rm TF}^4$ & M Eqs.~\eqref{eq:KT-sigma}, \eqref{eq:R1}\\
$E_0$, $v_0$, $\omega_0$ & energy, interaction and frequency units (Fermi units), $E_0=\hbar^2k_F^2/m_e$, $v_0=\mathcal C_de^2/k_F^{d-1}$ ($4\pi e^2/k_F^2$ for $d=3$), $\omega_0=E_0/\hbar$ & M Eq.~\eqref{eq:unit-scales}\\
\midrule
soft line, hard line & screened line with momentum transfer of order $k_{\rm TF}=q_{\rm TF}k_F$ and frequency of order $v_F$ times this momentum transfer, line with momentum transfer of order $k_F$ & M Sec.~\ref{sec:completeness}\\
self-energy pair & the sum $\Pi^{(1)}_{\rm SE}$ of the two Fock self-energy insertions & M Sec.~\ref{sec:se-cancellation}\\
Fermi-surface factor & factor $\log(1/q)$ of an endpoint term $q\log^m(1/q)$, from the vertex or self-energy at the Fermi surface & M Sec.~\ref{sec:overview}\\
\bottomrule
\end{tabular}
\end{table}
\section{Conventions in \texorpdfstring{$d$}{d} dimensions and the long-wavelength Lindhard shape}\label{S-dsec:lindhard}
This section and the next seven support main-text
Secs.~\ref{sec:units} and \ref{sec:ring-d}, \ref{sec:Pi1-d} and \ref{sec:twod}.
They use the Fermi units of main-text Eq.~\eqref{eq:unit-scales} in
$d$ dimensions, with $v(q)=q^{1-d}$, and convert energies per electron to
Ry with $E_0=2/(\alpha_dr_s)^2$~Ry.

\emph{Constants.} Table~\ref{S-dtab:constants} collects the constants of
main-text Eqs.~\eqref{eq:alpha-d} and \eqref{eq:lindhard-convention}.
With $g_p=(\mathrm i\nu-\delta_{\bm p})^{-1}$ and the spin factor $2$, the factor $C$ of a
closed fermion loop in response units and the free response are
\[
 C=v_0k_F^d/E_0=\mathcal C_d/(k_Fa_B)=\mathcal C_d\alpha_dr_s,\qquad
 \Pi_0=2C\int\mathrm d^dp\,(2\pi)^{-d}\Delta_{\bm p}\operatorname{Re}g_p ,
\]
with $C=\pi^2q_{\rm TF}^2$ for $d=3$ [Eq.~\eqref{S-app-eq:definitions-u6}].
The value of $\Pi_0$ at $q\to0$, $\nu=0$ defines $q_{\rm TF}$,
\[
 \Pi_0\to-2C\mathcal A_d/(2\pi)^d,\qquad
 q_{\rm TF}^{d-1}=2C\mathcal A_d/(2\pi)^d=\mathcal B_d\alpha_dr_s .
\]

\begin{table}[h]
\centering\small
\caption{Constants of the $d$-dimensional gas for $d=2$ and $3$.}
\label{S-dtab:constants}
\begin{tabular}{@{}lll@{}}
\toprule
Quantity & $d=2$ & $d=3$\\
\midrule
$\mathcal A_d=2\pi^{d/2}/\Gamma(d/2)$ & $2\pi$ & $4\pi$\\
$\mathcal C_d=(4\pi)^{(d-1)/2}\Gamma(\tfrac{d-1}2)$ & $2\pi$ & $4\pi$\\
$\alpha_d=1/(k_Fa_Br_s)$ & $1/\sqrt2$ & $\alpha_L=[4/(9\pi)]^{1/3}$\\
$\mathcal B_d=2\mathcal A_d\mathcal C_d/(2\pi)^d$ & $2$ & $4/\pi$\\
$q_{\rm TF}^{d-1}=\mathcal B_d\alpha_dr_s$ & $q_{\rm TF}=\sqrt2\,r_s$ & $q_{\rm TF}^2=4\alpha_Lr_s/\pi$\\
$E_0$ (Ry) & $4/r_s^2$ & $32/(\pi^2q_{\rm TF}^4)=2/(\alpha_Lr_s)^2$\\
$\frac d{8\pi}E_0q_{\rm TF}^{d+1}$ (Ry) & $2\sqrt2\,r_s/\pi$ & $12/\pi^3$\\
\bottomrule
\end{tabular}
\end{table}

\emph{Long-wavelength limit.} For $q\to0$ at fixed $z=\nu/q$, with $\mu=\hat{\bm p}\cdot\hat{\bm q}$,
\[
 \Delta_{\bm p}\to\delta(1-|\bm p|)\,q\mu,\qquad \delta_{\bm p}\to q\mu,\qquad
 g_p\to b/q,\quad b=(\mathrm iz-\mu)^{-1},\quad \operatorname{Re}b=-\mu/(\mu^2+z^2),
\]
and therefore
\[
 \Pi_0\to2C\mathcal A_d(2\pi)^{-d}\langle\mu\operatorname{Re}b\rangle
 =-q_{\rm TF}^{d-1}\langle\mu^2/(\mu^2+z^2)\rangle,
\]
which is main-text Eq.~\eqref{eq:hz}. The average over $S^{d-1}$ has the weight
$(1-\mu^2)^{(d-3)/2}/\mathrm B(\tfrac12,\tfrac{d-1}2)$ on
$-1\le\mu\le1$, and Euler's integral gives
\[
 \begin{gathered}
 \langle(\mu^2+z^2)^{-1}\rangle=z^{-2}\,{}_2F_1(1,\tfrac12;\tfrac d2;-z^{-2}),\\
 \text{hence}\quad h_d=1-z^2\langle(\mu^2+z^2)^{-1}\rangle=1-{}_2F_1(1,\tfrac12;\tfrac d2;-z^{-2}).
 \end{gathered}
\]
For $d=2$, $\langle(\mu^2+z^2)^{-1}\rangle=1/(z\sqrt{1+z^2})$, and for $d=4$ the
hypergeometric function is elementary. These cases give main-text
Eq.~\eqref{eq:hd-examples}, and $\langle\mu^2\rangle=1/d$ gives
$h_d\simeq\langle\mu^2\rangle z^{-2}=1/(dz^2)$ at large $z$. Direct quadrature of the average for
$d=2,3,3.5,4,5$ agrees with the hypergeometric form to 18 digits.

\emph{The variable $r$.} For $d=2$ and $d=4$ the substitution
$r=\sqrt{1+z^2}-z$, which maps $0<z<\infty$ onto $1>r>0$, gives $h_4=r^2$ and
\begin{equation}
 \begin{gathered}
 z=\frac{1-r^2}{2r},\qquad \sqrt{1+z^2}=\frac{1+r^2}{2r},\qquad
 a\equiv\frac{z}{\sqrt{1+z^2}}=\frac{1-r^2}{1+r^2},\\
 h_2=\frac{2r^2}{1+r^2},\qquad |\mathrm dz|=\frac{1+r^2}{2r^2}\,\mathrm dr .
 \end{gathered}
 \label{S-deq:r-sub}
\end{equation}
Hence
\[
 h_2^3\,|\mathrm dz|=\frac{4r^4\,\mathrm dr}{(1+r^2)^2},\qquad
 h_4^{5/3}\,|\mathrm dz|=\tfrac12\bigl(r^{4/3}+r^{10/3}\bigr)\,\mathrm dr,
\]
and since $h_d$ is even, the ring coefficients contain the integrals
\begin{equation}
 \int_{-\infty}^\infty h_2^3\,\mathrm dz=8\int_0^1\frac{r^4\,\mathrm dr}{(1+r^2)^2}=10-3\pi,
 \qquad
 \int_{-\infty}^\infty h_4^{5/3}\,\mathrm dz=\int_0^1\bigl(r^{4/3}+r^{10/3}\bigr)\,\mathrm dr=\frac{60}{91},
 \label{S-deq:h-integrals}
\end{equation}
and in three dimensions $\int_{-\infty}^\infty h^2\,\mathrm dz=\tfrac{2\pi}3(1-\log2)$
[Eq.~\eqref{S-sm-slim-eq:GB}].

\emph{Small-$q$ expansion.} The Lindhard function $\mathcal L(q,\nu)$ is
invariant under $(\bm q,\nu)\to(-\bm q,-\nu)$ and even in $\nu$. Hence
$\mathcal L(q,qz)$ is even in $q$ at fixed $z$, in every dimension, and, with $g_3$ from Sec.~\ref{S-sm-slim-sec:residues},
\[
 \mathcal L(q,qz)=h_d(z)+q^2g_d(z)+O(q^4),\qquad g_3(z)=-1/[12(1+z^2)^2].
\]
In two dimensions the angular integral of the free response is elementary. At imaginary frequency
\[
 \mathcal L(q,qz)=1-(2/q)\operatorname{Re}\sqrt{(q/2+\mathrm iz)^2-1},
\]
with the branch that tends to $\mathrm i\sqrt{1+z^2}$ as $q\to0$. For $z\ge0$ the square root expands as
\[
 \operatorname{Re}\sqrt{(q/2+\mathrm iz)^2-1}
 =\frac{qz}{2\sqrt{1+z^2}}+\frac{q^3z}{16(1+z^2)^{5/2}}+O(q^5).
\]
This gives $h_2(z)$, a vanishing coefficient of $q$, and
\begin{equation}
 g_2(z)=-\frac{z}{8(1+z^2)^{5/2}},
 \label{S-deq:g2}
\end{equation}
again with a vanishing coefficient of $q^3$. A 40-digit evaluation of
the exact free response confirms \eqref{S-deq:g2} and the vanishing
coefficient of $q$ at $z=0.3$, $1$ and $3$, where
\[
 g_2(0.3)=-0.0302318708,\qquad g_2(1)=-1/(32\sqrt2),\qquad g_2(3)=-0.0011858541 .
\]
For $z=0$, $g_2=0$, as required by the flat static response of the
two-dimensional gas below $2k_F$ \cite{SI-Stern1967}.

\section{Ring series in \texorpdfstring{$d$}{d} dimensions}\label{S-dsec:ring}
\emph{Energy.} In units of $E_0k_F^d$ the RPA correlation energy per unit volume is
\[
 \frac12\int\mathrm d^dq\,(2\pi)^{-d}\int_{-\infty}^\infty\mathrm d\nu\,(2\pi)^{-1}[\log(1+c)-c].
\]
Multiplying by the volume $(2\pi)^d/(2V_d)$ per electron, with
$\mathrm d^dq=\mathcal A_dq^{d-1}\mathrm dq$, $\mathrm d\nu=q\,\mathrm dz$ and $\mathcal A_d=dV_d$, gives
\[
 \begin{gathered}
 e_{\rm RPA}=\frac{E_0\mathcal A_d}{8\pi V_d}\int_{-\infty}^\infty\mathrm dz\int_0^\infty\mathrm dq\,q^d\bigl[\log(1+c)-c\bigr],\qquad
 \frac{E_0\mathcal A_d}{8\pi V_d}=\frac{dE_0}{8\pi},\\
 \frac{dE_0}{8\pi}\Bigr|_{d=3}=\frac3{8\pi}\cdot\frac{32}{\pi^2q_{\rm TF}^4}=\frac{12}{\pi^3q_{\rm TF}^4}~\text{Ry},
 \end{gathered}
\]
The prefactor $dE_0/(8\pi)$ is that of main-text Eq.~\eqref{eq:ring-energy}, and its $d=3$ value is Eq.~(20) of
Ref.~\cite{SI-Endo1997}.

\emph{Resonant dimensions.} With $c=h_dt^{1-d}$ and $t=q/q_{\rm TF}$, the
$q$ integral at fixed $z$ is
\[
 q_{\rm TF}^{d+1}\int_0^\infty t^d\bigl[\log(1+h_dt^{1-d})-h_dt^{1-d}\bigr]\mathrm dt,
\]
cut off at $t\sim1/q_{\rm TF}$ by the Lindhard function. In the expansion
of $\log(1+h_dt^{1-d})$, the term of order $N_{\rm r}$, with $N_{\rm r}(d-1)=d+1$,
makes the integrand proportional to $t^{-1}$, and the window
$1\lesssim t\lesssim1/q_{\rm TF}$ gives
\[
 \begin{gathered}
 t^d\,\frac{(-1)^{N_{\rm r}+1}h_d^{N_{\rm r}}}{N_{\rm r}}\,t^{-N_{\rm r}(d-1)}
 =\frac{(-1)^{N_{\rm r}+1}h_d^{N_{\rm r}}}{N_{\rm r}\,t},\qquad
 \int_1^{1/q_{\rm TF}}\frac{\mathrm dt}t=-\log q_{\rm TF},\\
 q_{\rm TF}^{d+1}\int_0^\infty t^d\bigl[\dots\bigr]\mathrm dt\supset
 (-1)^{N_{\rm r}}(h_d^{N_{\rm r}}/N_{\rm r})\,q_{\rm TF}^{d+1}\log q_{\rm TF},
 \end{gathered}
\]
in agreement with Eq.~\eqref{S-sm-slim-eq:leading-residue} for $m=0$. With
$\log q_{\rm TF}=\log r_s/(d-1)+O(1)$ this is main-text
Eq.~\eqref{eq:ring-resonant}. The regions
$q\lesssim q_{\rm TF}$ and $q\gtrsim1$ are controlled as for $d=3$ in
Sec.~\ref{S-sm-slim-sec:residues}, with the two-dimensional Lindhard
function of Sec.~\ref{S-dsec:lindhard} for $d=2$. With
table~\ref{S-dtab:constants} and Eq.~\eqref{S-deq:h-integrals} the coefficients of $\log r_s$ are
\[
 d=3:\ \frac{12}{\pi^3}\cdot\frac14\int h^2\,\mathrm dz=\frac2{\pi^2}(1-\log2),\qquad
 d=2:\ -\frac{2\sqrt2\,r_s}{\pi}\cdot\frac13\int h_2^3\,\mathrm dz=-2\sqrt2\Bigl(\frac{10}{3\pi}-1\Bigr)r_s,
\]
which are main-text Eqs.~\eqref{eq:GB} and \eqref{eq:RK}.

\emph{Dimensions $d>3$.} The scale-free integral then
converges, and the substitution $u=h_dt^{1-d}$ gives
\begin{equation}
 \int_0^\infty t^d\bigl[\log(1+h_dt^{1-d})-h_dt^{1-d}\bigr]\mathrm dt
 =\frac{h_d^{N_{\rm r}}}{d-1}\int_0^\infty u^{-N_{\rm r}-1}\bigl[\log(1+u)-u\bigr]\mathrm du
 =\frac{h_d^{N_{\rm r}}}{d-1}\,\frac{\pi}{N_{\rm r}\sin(\pi N_{\rm r})},
 \label{S-deq:mellin-ring}
\end{equation}
for $1<N_{\rm r}<2$. The last integral follows by one integration by parts, whose boundary
terms vanish for $1<N_{\rm r}<2$, from
$\int_0^\infty u^{a-1}(1+u)^{-1}\mathrm du=\pi/\sin(\pi a)$ with $a=2-N_{\rm r}$,
\[
 \int_0^\infty u^{-N_{\rm r}-1}\bigl[\log(1+u)-u\bigr]\mathrm du
 =-\frac1{N_{\rm r}}\int_0^\infty\frac{u^{1-N_{\rm r}}\,\mathrm du}{1+u}
 =-\frac{\pi}{N_{\rm r}\sin\bigl(\pi(2-N_{\rm r})\bigr)}=\frac{\pi}{N_{\rm r}\sin(\pi N_{\rm r})}.
\]
With $(d-1)N_{\rm r}=d+1$ and $E_0=2/(\alpha_dr_s)^2$~Ry,
\[
 \begin{gathered}
 E_0q_{\rm TF}^{d+1}=\frac{2(\mathcal B_d\alpha_dr_s)^{N_{\rm r}}}{(\alpha_dr_s)^2}
 =2\alpha_d^{N_{\rm r}-2}\mathcal B_d^{N_{\rm r}}r_s^{N_{\rm r}-2}~\text{Ry},\\
 e_{\rm RPA}\simeq\frac d{8\pi}\,\frac{E_0q_{\rm TF}^{d+1}}{d-1}\,\frac{\pi}{N_{\rm r}\sin(\pi N_{\rm r})}\int_{-\infty}^\infty h_d^{N_{\rm r}}\,\mathrm dz
 =\frac{d\,\alpha_d^{N_{\rm r}-2}\mathcal B_d^{N_{\rm r}}}{4(d+1)\sin(\pi N_{\rm r})}\int_{-\infty}^\infty h_d^{N_{\rm r}}\,\mathrm dz\;r_s^{N_{\rm r}-2}~\text{Ry},
 \end{gathered}
\]
which gives main-text Eq.~\eqref{eq:Ad}. For $d=4$,
\[
 N_{\rm r}=5/3,\quad \mathcal A_4=2\pi^2,\quad \mathcal C_4=4\pi^2,\quad \mathcal B_4=1,\quad
 \alpha_4=2^{-5/4},\quad \sin(5\pi/3)=-\sqrt3/2,
\]
and with \eqref{S-deq:h-integrals}
\[
 A_4=\frac{4\cdot2^{5/12}}{20\cdot(-\sqrt3/2)}\cdot\frac{60}{91}~\text{Ry}=-24\cdot2^{5/12}/(91\sqrt3)~\text{Ry}.
\]

\emph{The limit $d\to3$.} Near $d=3$,
\[
 N_{\rm r}=2-(d-3)/(d-1),\qquad
 \sin(\pi N_{\rm r})=-\tfrac\pi2(d-3)+O((d-3)^2),
\]
and at $d=3$ the remaining factors of $A_d$ give
$3\mathcal B_3^2\int h^2\,\mathrm dz/16=(2/\pi)(1-\log2)$. Hence
\[
 (d-3)A_d\to\frac{(2/\pi)(1-\log2)}{-\pi/2}=-4(1-\log2)/\pi^2=-2A.
\] In Eq.~(21) of
Ref.~\cite{SI-Schlesier2020} the explicit $r_s^0$ term
$D\int R_{(D-1)/2}^2\,\mathrm du/[\pi^3(D-3)]$ has the opposite pole, since
$R_{(D-1)/2}$ reduces to $h$ at $D=3$. This term cancels the constant
$-2A/(d-3)$ (main-text Sec.~\ref{sec:ring-d}). Hence
the remaining part $\delta_D$ of their $r_s^0$ term must be regular
at $d=3$. Numerically
\[
 (d-3)A_d/2=-0.0639390,\ -0.0623257,\ -0.0621955\qquad (d=3.1,\ 3.01,\ 3.001),
\]
approaching $-A=-0.0621814$ linearly in $d-3$.

\emph{The limit $d\to\infty$.} For large $d$ the component $\mu$ is
Gaussian with variance $1/d$. Hence $h_d(y/\sqrt d)$ tends to a fixed
function of $y$, and with $x$ a standard normal variable
\[
 \begin{gathered}
 \int_{-\infty}^\infty h_d(y/\sqrt d)\,\mathrm dy\to
 \Bigl\langle\int_{-\infty}^\infty x^2(x^2+y^2)^{-1}\mathrm dy\Bigr\rangle_x=\pi\langle|x|\rangle=\sqrt{2\pi},\\
 \text{hence}\quad\int h_d^{N_{\rm r}}\,\mathrm dz\simeq\sqrt{2\pi/d}.
 \end{gathered}
\]
Stirling's formula gives $\mathcal B_d\simeq2\sqrt2/\sqrt{\pi d}$ and $\alpha_d^{-1}\simeq d/e$, and
$\sin(\pi N_{\rm r})\simeq-2\pi/d$. With these, \eqref{eq:Ad} gives
\[
 \begin{gathered}
 A_d\simeq\frac{d\,(d/e)\,\mathcal B_d\sqrt{2\pi/d}}{4d\,(-2\pi/d)}=-d/(2e\pi)~\text{Ry}.
 \end{gathered}
\]
Numerically $A_d/d$ is
$-0.052756$, $-0.057680$ and $-0.058216$ at $d=100$, $1000$ and $3000$,
approaching $-1/(2e\pi)=-0.058550$ with corrections of order $(\log d)/d$.

\emph{Comparison with Ref.~\cite{SI-Schlesier2020}.} Equation~(21) of
Ref.~\cite{SI-Schlesier2020}, with its series $\Sigma_D$, the functions
$R_{(D-1)/2}$ of its Eq.~(18) and $\alpha_D=1/\alpha_d$, gives the same
coefficient as \eqref{eq:Ad} to machine precision for every
$d$ tested ($3.001\le d\le3000$). A closed form of the series
$\Sigma_D$ agrees with its printed definition to 40 digits for
$d=4,5,6,9$. The values in Table~I of Ref.~\cite{SI-Schlesier2020}, in Hartree, are
smaller in magnitude than $A_d$ by the factors in main-text table~\ref{tab:Ad}.
The large-$d$ limit stated there, $-D/(2e\pi^2)$~Hartree, is smaller in magnitude than
$-d/(4e\pi)$~Hartree by a factor of $\pi/2$. Both differences appear in the published and
arXiv (v2) versions. Reference~\cite{SI-DucPham2023} states in its Ref.~68 that a small error in
the numerical evaluation of that table was corrected, but gives no
corrected values. Since main-text Eq.~\eqref{eq:Ad} reproduces Eq.~(21) of
Ref.~\cite{SI-Schlesier2020}, the difference lies in the numerical
evaluation, not in the formula.

\section{First-order exchange polarization at long wavelength}\label{S-dsec:response}
The steps of Sec.~\ref{S-app-sec:response} hold
in $d$ dimensions with $\mathrm d^3p\to\mathrm d^dp$, $(2\pi)^3\to(2\pi)^d$
and $C=\mathcal C_d\alpha_dr_s$. The Fock shift of the excitation energy
and the vertex share the factor
$\Delta_{\bm p}\Delta_{\bm k}Cv(\bm p-\bm k)$. The fixed-density
counterterm (the shift $\delta\mu_x$ of the chemical potential in
Sec.~\ref{S-app-sec:response}) removes the \emph{occupation-surface terms}, which come
from the change of the occupation numbers at the Fermi surface. Symmetrization then gives
\begin{equation}
 \Pi^{(1)}(q,\mathrm i\nu)=C\int\frac{\mathrm d^dp\,\mathrm d^dk}{(2\pi)^{2d}}\,
 Cv(\bm p-\bm k)\,\Delta_{\bm p}\Delta_{\bm k}\operatorname{Re}\bigl[(g_p-g_k)^2\bigr],
 \label{S-deq:combined}
\end{equation}
which for $d=3$ is Eq.~\eqref{S-app-eq:combined-response} and for $d=2$
is the combined kernel of Eq.~(11) of Ref.~\cite{SI-DasSarma2026} continued to
imaginary frequency.

\emph{Limit.} In the long-wavelength limit of Sec.~\ref{S-dsec:lindhard}, on the Fermi surface,
\[
 \Delta_{\bm p}\Delta_{\bm k}\to q^2\mu\mu',\qquad
 (g_p-g_k)^2\to\frac{(b-b')^2}{q^2},\qquad
 v(\bm p-\bm k)\to|\hat{\bm p}-\hat{\bm p}'|^{1-d},
\]
and therefore
\begin{equation}
 \Pi^{(1)}\to C^2\frac{\mathcal A_d^2}{(2\pi)^{2d}}
 \Bigl\langle\!\Bigl\langle\frac{\mu\mu'\operatorname{Re}[(b-b')^2]}{|\hat{\bm p}-\hat{\bm p}'|^{d-1}}\Bigr\rangle\!\Bigr\rangle
 =-q_{\rm TF}^{d-1}\alpha_dr_s\,R_d(z),
 \label{S-deq:Rd}
\end{equation}
with
\[
 q_{\rm TF}^{d-1}\alpha_dr_s=\frac{2C^2\mathcal A_d}{(2\pi)^d\mathcal C_d},\qquad
 \frac{C^2\mathcal A_d^2}{(2\pi)^{2d}}=\frac{\mathcal B_d}4\,q_{\rm TF}^{d-1}\alpha_dr_s,
\]
which gives the prefactor $-\mathcal B_d/4$ of main-text Eq.~\eqref{eq:Rd}.
Since $(b-b')^2=(\mu-\mu')^2b^2b'^2$ and
$|\hat{\bm p}-\hat{\bm p}'|^2\ge(\mu-\mu')^2$, the integrand obeys
\[
 \Bigl|\frac{\mu\mu'\operatorname{Re}[(b-b')^2]}{|\hat{\bm p}-\hat{\bm p}'|^{d-1}}\Bigr|
 \le|\hat{\bm p}-\hat{\bm p}'|^{3-d}|\mu\mu'|\,|b|^2|b'|^2\le
 |\hat{\bm p}-\hat{\bm p}'|^{3-d}z^{-4},
\]
which is integrable over pairs of directions in every dimension, uniformly
in $z\ge z_0>0$. The static limit $z\to0$ exists as a principal value in
$\mu$, since the terms singular at $\mu=0$ are odd in $\mu$.
Taken separately, the vertex and self-energy parts contain the
nonintegrable factor $|\hat{\bm p}-\hat{\bm p}'|^{1-d}$ near
$\hat{\bm p}=\hat{\bm p}'$. Cut off at angles of order $q$ at finite $q$,
it produces their Fermi-surface factors $\log(1/q)$.

\emph{Three dimensions.} At fixed $\mu,\mu'$ the average over the relative azimuth is
\[
 \Bigl\langle\frac1{|\hat{\bm p}-\hat{\bm p}'|^2}\Bigr\rangle
 =\Bigl\langle\frac1{2(1-\cos\Theta)}\Bigr\rangle=\frac1{2|\mu-\mu'|},
\]
and with $\mathcal B_3=4/\pi$
\begin{equation}
 R_3(z)=-\frac1{8\pi}\int_{-1}^1\mathrm d\mu\int_{-1}^1\mathrm d\mu'\,
 \mu\mu'|\mu-\mu'|\operatorname{Re}\frac1{(\mathrm iz-\mu)^2(\mathrm iz-\mu')^2}
 =\frac{R^{(1)}(z)}{4\pi}.
 \label{S-deq:R3}
\end{equation}
Quadrature on the triangle $\mu'<\mu$ gives
\[
 \begin{aligned}
 R_3(0.1)&=0.23049451763139286,&\qquad R_3(0.3)&=0.11645170899574007,\\
 R_3(1)&=0.011619772367581344,&\qquad R_3(3)&=0.0002451548733833656,
 \end{aligned}
\]
equal to $R^{(1)}/(4\pi)$ of main-text Eq.~\eqref{eq:R1} within
$2.5\times10^{-17}$. With
\[
 \int_{-\infty}^\infty hR_3\,\mathrm dz=\frac{\pi^2-12\log2}{12},\qquad
 \int_{-\infty}^\infty h^2g_3\,\mathrm dz=-\frac{\pi(\pi^2-6)}{288},
\]
the two parts of the Carr--Maradudin coefficient and their sum are
\[
 \frac{\alpha_L(\pi^2-12\log2)}{2\pi^3}+\frac{\alpha_L(\pi^2-6)}{12\pi^3}
 =0.0184584246606864\ldots\ {\rm Ry}.
\]
The sum agrees with the closed form
of the coefficient $C$ of main-text Eq.~\eqref{eq:jellium-expansion} to 60 digits.

\section{Closed form of \texorpdfstring{$R_2$}{R2} and the \texorpdfstring{$\log(1/q)$}{log(1/q)} terms of the vertex and self-energy parts in two dimensions}\label{S-dsec:R2}
\emph{Fourier evaluation.} For $d=2$, $\mathcal B_2=2$ and
$|\hat{\bm p}-\hat{\bm p}'|=2|\sin(\phi/2)|$, with $\theta$ and
$\theta'=\theta+\phi$ the polar angles of $\hat{\bm p}$ and
$\hat{\bm p}'$ relative to $\bm q$. Let
\[
 T(\phi)=\bigl\langle\cos\theta\cos\theta'\operatorname{Re}[(b-b')^2]\bigr\rangle_\theta
 =\sum_mT_m{\rm e}^{\mathrm im\phi},
 \qquad R_2=-\frac12\Bigl\langle\frac{T(\phi)}{2|\sin(\phi/2)|}\Bigr\rangle_\phi .
\]
The Coulomb factor has the Fourier coefficients
\[
 \Bigl\langle\frac{{\rm e}^{\mathrm im\phi}}{2|\sin\frac\phi2|}\Bigr\rangle_\phi=K_0-\frac2\pi O_{|m|},\qquad
 O_m=\sum_{j=1}^m\frac1{2j-1},
\]
with the divergent constant $K_0$. Since
$T(0)=\sum_mT_m=0$, $K_0$ drops out (with an angular cutoff $\epsilon$
the constant is finite and multiplies $T(0)=0$), and
$R_2=\pi^{-1}\sum_mO_{|m|}T_m$. With $s_z=\sqrt{1+z^2}$ and $r=s_z-z$,
\[
 b(\theta)=(\mathrm iz-\cos\theta)^{-1}=(\mathrm is_z)^{-1}\sum_m(-\mathrm ir)^{|m|}{\rm e}^{\mathrm im\theta}.
\]
The functions $F_0=\cos\theta$,
$F_1=\cos\theta\,b$ and $F_2=\cos\theta\,b^2$ have Fourier coefficients
even in the index,
\[
 F_{0,\pm1}=\tfrac12,\qquad
 F_{1,m}=\frac{(-\mathrm ir)^{m-1}(1-r^2)}{2\mathrm is_z}\quad(m\ge1),\qquad
 F_{2,\pm1}=\langle\cos^2\theta\,b^2\rangle .
\]
Then
$T_m=2\operatorname{Re}(F_{2,m}F_{0,m})-2\operatorname{Re}(F_{1,m}^2)$, where
the first term comes from the self-energy terms $b^2+b'^2$ and the second
from the vertex term $-2bb'$. With \eqref{S-deq:r-sub},
\[
 F_{1,m}^2=(-1)^mr^{2|m|}a^2,\qquad
 \operatorname{Re}F_{2,\pm1}=\langle\cos^2\theta\operatorname{Re}b^2\rangle=h_2+zh_2',
\]
hence, for $m\ne0$,
\[
 T_m=\delta_{|m|,1}(h_2+zh_2')-2(-1)^mr^{2|m|}a^2,\qquad
 R_2=\frac2\pi\Bigl[h_2+zh_2'-2a^2\sum_{m\ge1}O_m(-r^2)^m\Bigr].
\]
The generating function at $x=-r^2$ gives
\[
 \sum_{m\ge1}O_mx^m=\frac{\sqrt x\,\operatorname{artanh}\sqrt x}{1-x},\qquad
 \sum_{m\ge1}O_m(-r^2)^m=-\frac{r\arctan r}{1+r^2},
\]
and
\begin{equation}
 R_2=\frac2\pi\Bigl[h_2+zh_2'+\frac{2a^2r\arctan r}{1+r^2}\Bigr],\qquad
 h_2+zh_2'=h_2^2-\frac{2a^2r^2}{1+r^2},
 \label{S-deq:R2}
\end{equation}
which is main-text Eq.~\eqref{eq:R2}. The second equality follows from
\eqref{S-deq:r-sub} and
\[
 h_2'=-s_z^{-3},\qquad h_2+zh_2'=1-\frac z{s_z}-\frac z{s_z^3}.
\]

\emph{Checks.} With $\mu=\cos\theta$, $\mu'=\cos\theta'$,
$\sigma=(\theta+\theta')/2$ and $\delta=(\theta-\theta')/2$,
$(\mu-\mu')^2=4\sin^2\sigma\sin^2\delta$, so the angular integrand is smooth in
$(\sigma,\delta)$. Gauss--Legendre quadrature
in these variables, independent of the Fourier sum, gives
$R_2(1)=0.04977088289807023$, within $7\times10^{-18}$ of
\eqref{S-deq:R2}, and agrees with it to $10^{-13}$ at $z=0.3$ and $3$. At $z=0$,
$R_2=2/\pi$. At large $z$, \eqref{S-deq:R2} with
\[
 h_2=\frac1{2z^2}-\frac3{8z^4}+\dots,\qquad r=\frac1{2z}+\dots\quad\text{gives}\quad
 R_2=\frac5{12\pi z^4}-\frac{11}{20\pi z^6}+\frac{1283}{2240\pi z^8}+O(z^{-10}).
\]

\emph{Factor $\log(1/q)$ of the vertex part.} The vertex part alone is
\[
 -\frac12\Bigl\langle\!\Bigl\langle\frac{\mu\mu'\operatorname{Re}(-2bb')}{2|\sin\frac\phi2|}\Bigr\rangle\!\Bigr\rangle .
\]
With an angular cutoff $|\phi|>\epsilon$, the average of the Coulomb factor is
\[
 \Bigl\langle\frac1{2|\sin\frac\phi2|}\Bigr\rangle_{|\phi|>\epsilon}=\frac1\pi\log\frac1\epsilon+O(1).
\]
Near $\phi=0$,
$\mu\mu'\operatorname{Re}(bb')\to\cos^2\theta\operatorname{Re}b^2$, and therefore the
vertex part grows as
\[
 \frac1\pi\langle\cos^2\theta\operatorname{Re}b^2\rangle\log\frac1\epsilon=\frac1\pi(h_2+zh_2')\log\frac1\epsilon .
\]
The self-energy part has the same factor $\log(1/\epsilon)$ with the
opposite sign, since the sum of the two parts is finite.
At finite $q$ the shell width cuts off the angle at $\epsilon\sim q$, which
gives the coefficient of $\log(1/q)$ in main-text property (iii) of
Sec.~\ref{sec:R2}. Between $\epsilon=10^{-5}$ and $10^{-7}$, a quadrature with a symmetric
cutoff gives the slopes per unit of $\log(1/\epsilon)$
\[
 \begin{array}{lccc}
 z & 10^{-4} & 0.3 & 1\\
 \text{slope} & 0.318246 & 0.1429308 & -0.01930873\\
 \pi^{-1}(h_2+zh_2') & 0.318246 & 0.1429308 & -0.01930873
 \end{array}
\]
At $z=0$ the coefficient $1/\pi$ reproduces the terms
$\mp(r_s/\sqrt2\pi)\log(k_F/q)$ of Eqs.~(21) and (22) of
Ref.~\cite{SI-DasSarma2026}. In three dimensions the same argument gives
main-text Eq.~\eqref{eq:ward}, since $q_{\rm TF}^2\alpha_Lr_s/\pi=q_{\rm TF}^4/4$.

\emph{Isolated vertex in the ring chain.} Inserted alone into main-text
Eq.~\eqref{eq:Pi1-energy}, the vertex part has at small $q$ the kernel
$q\,q_{\rm TF}^{d-1}\alpha_dr_s\,\pi^{-1}(h_d+zh_d')\log(1/q)$.
With the tail coefficient $\gamma_{N_*}=(-h_d)^{N_*}$,
Eq.~\eqref{S-sm-slim-eq:leading-residue} with $m=1$ gives the triple-pole term
of the energy $e_{\rm vtx}$ of this series,
\begin{equation}
 e_{\rm vtx}\supset\frac d{8\pi}E_0\,\alpha_dr_s\,q_{\rm TF}^{d+1}
 \frac{(-1)^{N_*}}{2\pi(d-1)^2}\int_{-\infty}^\infty h_d^{N_*}(h_d+zh_d')\,\mathrm dz\;\log^2r_s .
 \label{S-deq:vtx}
\end{equation}
Integration by parts, with $zh_d^{N_*+1}\to0$ at both ends, gives
\[
 \int zh_d^{N_*}h_d'\,\mathrm dz=\int\frac{z\,\mathrm d(h_d^{N_*+1})}{N_*+1}=-\frac{\int h_d^{N_*+1}\,\mathrm dz}{N_*+1},\qquad
 \int h_d^{N_*}(h_d+zh_d')\,\mathrm dz=\frac{N_*}{N_*+1}\int h_d^{N_*+1}\,\mathrm dz .
\]
With table~\ref{S-dtab:constants},
\[
 \begin{aligned}
 d=2:&\quad\tfrac d{8\pi}E_0\alpha_2r_sq_{\rm TF}^3=\tfrac{2r_s^2}\pi\ {\rm Ry},&\quad
 e_{\rm vtx}&\supset\frac2{3\pi^2}(10-3\pi)\,r_s^2\log^2r_s,\\
 d=3:&\quad\tfrac3{8\pi}E_0\alpha_Lr_sq_{\rm TF}^4=\tfrac{12\alpha_Lr_s}{\pi^3},&\quad
 e_{\rm vtx}&\supset-\frac{\alpha_L(1-\log2)}{2\pi^3}\,r_s\log^2r_s=3c_{\rm hd}\,r_s\log^2r_s .
 \end{aligned}
\]
The first is main-text Eq.~\eqref{eq:vtx2d}. The second is the
chain-family coefficient of main-text Eq.~\eqref{eq:Clog2}, which checks
the normalization.

\section{The coefficient \texorpdfstring{$a_{2\ell}$}{a2l}}\label{S-dsec:a2l}
\emph{Residue.} For $d=2$ the kernel of main-text
Eq.~\eqref{eq:Pi1-energy} at fixed $z$ is
$q\cdot q_{\rm TF}(r_s/\sqrt2)R_2(z)$ at small $q$, and the scaling
function $1/(1+c)-1$ has the term $h_2^2t^{-2}$ ($N_*=2$). By
Eq.~\eqref{S-sm-slim-eq:leading-residue} with $m=0$ and, in its generic notation, $\varepsilon=q_{\rm TF}$,
$\kappa=q_{\rm TF}(r_s/\sqrt2)R_2$, $\gamma=h_2^2$ and $\beta=2$, the double pole
contributes to the $q$ integral
\[
 -q_{\rm TF}\frac{r_s}{\sqrt2}R_2h_2^2\,q_{\rm TF}^2\log q_{\rm TF}.
\]
With $\log q_{\rm TF}=\log r_s+O(1)$ and the
prefactor $(1/4\pi)E_0=1/(\pi r_s^2)$~Ry this gives main-text Eq.~\eqref{eq:a2l-integral},
\[
 e_{\Pi^{(1)}}\supset-\frac1{\pi r_s^2}\,q_{\rm TF}^3\frac{r_s}{\sqrt2}\int h_2^2R_2\,\mathrm dz\;\log r_s
 =-\frac2\pi\int h_2^2R_2\,\mathrm dz\;r_s^2\log r_s .
\]
The corrections to the limiting
forms are controlled as for $d=3$ in Sec.~\ref{S-sm-slim-sec:residues}.
By \eqref{S-deq:g2}, replacing $\mathcal L$ by $h_2$ changes $c$ by
$O(q_{\rm TF}q)$. The region $q>1$ contributes only at
$O(r_s)$ and $O(r_s^2)$, without $\log r_s$. The corrections of order $q^2\log q$ to $R_2$
(Sec.~\ref{S-dsec:finiteq}) enter at $O(r_s^4\log^2r_s)$.

\emph{Closed form.} In Hartree, with the first form of main-text
Eq.~\eqref{eq:R2} and \eqref{S-deq:r-sub},
\[
 \frac{a_{2\ell}}2=-\frac1\pi\int_{-\infty}^\infty h_2^2R_2\,\mathrm dz
 =-\frac8{\pi^2}\int_0^1\frac{r^2}{1+r^2}\Bigl[h_2^2+2a^2\,\frac{r\arctan r-r^2}{1+r^2}\Bigr]\mathrm dr,
\]
and therefore
\begin{equation}
 \begin{gathered}
 \frac{a_{2\ell}}2=\frac{-32\mathsf A+16\mathsf B-16\mathsf C}{\pi^2},\qquad
 \mathsf A=\int_0^1\frac{r^6\,\mathrm dr}{(1+r^2)^3},\\
 \mathsf B=\int_0^1\frac{r^4(1-r^2)^2\,\mathrm dr}{(1+r^2)^4},\qquad
 \mathsf C=\int_0^1\frac{r^3(1-r^2)^2\arctan r}{(1+r^2)^4}\,\mathrm dr .
 \end{gathered}
 \label{S-deq:ABC}
\end{equation}
The first two are integrals of rational functions,
\[
 \mathsf A=\tfrac32-\tfrac{15\pi}{32},\qquad \mathsf B=\tfrac{13}6-\tfrac{11\pi}{16}.
\]
For $\mathsf C$, with $u=1+r^2$, an antiderivative of the rational factor is
\[
 \Phi=\tfrac12\log u+\tfrac5{2u}-\tfrac2{u^2}+\tfrac2{3u^3},\qquad
 \Phi\big|_{r=1}=\tfrac12\log2+\tfrac56 .
\]
One integration by parts gives
\[
 \mathsf C=\frac\pi4\,\Phi\big|_{r=1}-\int_0^1\frac{\Phi\,\mathrm dr}{1+r^2},\qquad
 \int_0^1\frac{\Phi\,\mathrm dr}{1+r^2}=\frac12\int_0^1\frac{\log(1+r^2)}{1+r^2}\,\mathrm dr+\frac5{18}+\frac{17\pi}{96}.
\]
It leaves the integral
$\int_0^1\log(1+r^2)/(1+r^2)\,\mathrm dr=\tfrac\pi2\log2-G$, with $G$ Catalan's constant, and gives
\[
 \mathsf C=-\tfrac5{18}-\tfrac{\pi\log2}8+\tfrac\pi{32}+\tfrac G2 .
\]
Collecting terms gives main-text Eq.~\eqref{eq:a2l},
\[
 a_{2\ell}=\frac2{\pi^2}\bigl(-32\mathsf A+16\mathsf B-16\mathsf C\bigr)
 =\frac7\pi+\frac{4\log2}\pi-\frac{160}{9\pi^2}-\frac{16G}{\pi^2}.
\]
Multiprecision
quadrature of \eqref{S-deq:ABC} in $r$ agrees with it to $10^{-41}$
(40 digits) and to $3\times10^{-61}$ (60 digits), and quadrature of
\eqref{eq:a2l-integral} in $z$ to 100 digits. An integer-relation
search (PSLQ) on $\pi^2$ times this value over $\{1,\pi,G,\pi\log2\}$ returns the same
combination.

\section{Finite wave number in two dimensions}\label{S-dsec:finiteq}
\emph{Method.} At finite $q$ we evaluate \eqref{S-deq:combined} for $d=2$
directly. With $k_F=1$ and $\bm q$ along the $x$ axis,
$\Delta_{\bm p}\ne0$ only on the two crescents between the unit circle
and its translate. We use the coordinates
\[
 \rho=\rho_*(\varphi)+\bigl(1-\rho_*(\varphi)\bigr)t,\quad 0<t<1,\qquad
 \rho_*(\varphi)=-q\cos\varphi+(1-q^2\sin^2\varphi)^{1/2},
\]
with $\rho_*$ the polar radius of the translated circle. The measure is
signed and covers both signs of $\Delta_{\bm p}$. The Gauss--Legendre nodes
in $\varphi$ are split at $\cos\varphi=-q/2$, where the crescents close. The scaled
polarization at finite wave number is $R(q,z)=-\Pi^{(1)}/[q_{\rm TF}(r_s/\sqrt2)]$.

\emph{Results.} Table~\ref{S-dtab:finiteq} lists $R(q,z)-R_2(z)$. At
$z=0.3$ the differences decrease by factors of about $3.6$ when $q$ is
halved, as for a correction $q^2(a\log q+b)$ and not for an $O(q)$ term,
which would give a factor $2$. At $z=1$ they are below $10^{-5}$ in magnitude and change sign.
Refining the nodes changes the tabulated values by less than $10^{-13}$. At $z=0$ we compare
$\Delta R/R(0)$ with the values of $\Delta P/P(0)$ in Table~I of
Ref.~\cite{SI-DasSarma2026}, which gives the small-$x$ form
(its Eq.~(25) divided by $P(0)=-2/\pi^2$)
\[
 \frac{\Delta P}{P(0)}=\frac{x^2(\log x+\beta)}{96},\qquad \beta=6.6441,\qquad x=q/k_F .
\]
Our values and those of that table are
\[
 \begin{array}{lcccc}
 x & 0.1 & 0.2 & 0.3 & 0.4\\
 \Delta R/R(0) & 4.454\times10^{-4} & 2.020\times10^{-3} & 5.096\times10^{-3} & 9.625\times10^{-3}\\
 \Delta P/P(0) & 4.524\times10^{-4} & 2.104\times10^{-3} & 5.141\times10^{-3} & 9.699\times10^{-3}
 \end{array}
\]
We attribute the differences of
$1$--$4\%$ of $\Delta R$ to our quadrature, which is not adapted to
the singularity of the static integrand at coincident momenta at $z=0$.

\begin{table}[h]
\centering\small
\caption{Finite-$q$ first-order exchange polarization of the two-dimensional
gas: $R(q,z)-R_2(z)$, with $R_2(0.3)=0.3180676$ and
$R_2(1)=0.0497709$.}
\label{S-dtab:finiteq}
\begin{tabular}{@{}cccc@{}}
\toprule
$z$ & $q=0.08$ & $q=0.04$ & $q=0.02$\\
\midrule
$0.3$ & $-1.2186\times10^{-4}$ & $-3.3626\times10^{-5}$ & $-9.1990\times10^{-6}$\\
$1$ & $-5.1759\times10^{-6}$ & $+0.9888\times10^{-6}$ & $+0.8176\times10^{-6}$\\
\bottomrule
\end{tabular}
\end{table}

\section{Order counting and completeness}\label{S-dsec:counting}

\emph{One soft line.} Let an insertion of order $r_s^j$ have the
small-$q$ form
\[
 \Pi^{(j)}(q,\mathrm iqz)\simeq-q_{\rm TF}^{d-1}(\alpha_dr_s)^jq^mF(z).
\]
In main-text Eq.~\eqref{eq:Pi1-energy} its kernel is $q^{1+m}$, with a pole at $s=2+m$,
and the screening poles lie at $s=N(d-1)$. With
$E_0\propto r_s^{-2}$ and $q_{\rm TF}\propto r_s^{1/(d-1)}$, a coincidence gives
\[
 N(d-1)=2+m:\qquad E_0r_s^jq_{\rm TF}^{d+1+m}\log r_s\propto r_s^{\,j+(3-d+m)/(d-1)}\log r_s .
\]
The ring series ($j=0$, $m=0$, with its own kernel $q^d$,
Sec.~\ref{S-dsec:ring}) gives the same exponent $(3-d)/(d-1)$. The $q^2$
term of the Lindhard function is the case $j=0$, $m=2$. The exponent is
$j+m/2$ for $d=3$ and $j+1+m$ for $d=2$, and the orders listed in
main-text Sec.~\ref{sec:completeness} follow.

\emph{Several soft lines.} After the two-point loops (free bubbles) are
summed into the RPA line, the soft density fluctuations have nonlinear
vertices $\Gamma_l$ with $l\ge3$ legs. These are the closed fermion loops with three or
more density vertices, symmetrized over the vertex orderings. For soft momenta
$Q=\tau k_F$ and imaginary frequencies $\omega\sim v_FQ$, the symmetrized
loops remain bounded as $Q\to0$. Neumayr and Metzner showed this in two
dimensions \cite[Sec.~6, Eq.~(73)]{SI-NeumayrMetzner1998}. Kopper and
Magnen proved it with Ward identities for $d\le3$, for the
quadratic dispersion and with momenta and frequencies scaled to zero together
\cite{SI-KopperMagnen2001}. Their bound requires every partial sum of the
external frequencies of a loop to be at least a fixed fraction $\eta$ of
$v_FQ$, and its constant grows as a power of $1/\eta$. In
physical units it reads $\Gamma_l=O(g_Fk_F^{-2(l-2)})$, with $g_F$
the density of states at the Fermi level. Equation~(74) of Ref.~\cite{SI-NeumayrMetzner1998}
refers to the dynamical limit at fixed frequency and does not apply here.

A connected vacuum diagram with $V_l$ vertices of $l$ legs, $I$ soft
lines and $L$ soft loops has
\[
 L=I-\sum_lV_l+1,\qquad \sum_llV_l=2I .
\]
With each soft line of order $1/g_F$ and each loop integral of order
$k_FQ^{d+1}$, the energy per electron scales as
\[
 k_F^{-d}[k_FQ^{d+1}]^L[k_F^{2-d}]^I\prod_l[k_F^{d+2-2l}]^{V_l}
 =k_F^{-d+(d+2)(L-I+\sum_lV_l)}\tau^{L(d+1)}=k_F^2\tau^{L(d+1)}.
\]
With $\tau=q_{\rm TF}\propto r_s^{1/(d-1)}$ and $k_F^2\propto r_s^{-2}$,
\[
 k_F^2\tau^{L(d+1)}\propto r_s^{\,L(d+1)/(d-1)-2}:\qquad r_s^{3L-2}\ (d=2),\qquad r_s^{2L-2}\ (d=3).
\]
A $\log r_s$ term at these orders also requires a matching between soft
and hard regions. Without the loop
cancellation of Ref.~\cite{SI-NeumayrMetzner1998} (main-text
Sec.~\ref{sec:completeness}), each unsymmetrized three-point vertex
would be larger by $\tau^{-1}$ and each four-point vertex by
$\tau^{-2}$. The two-loop diagrams with two three-point
vertices or one four-point vertex would then scale as
\[
 k_F^2\tau^{2(d+1)-2}\propto r_s^{2/(d-1)}:\qquad r_s^2\ (d=2),\qquad r_s\ (d=3),
\]
the orders of $a_{2\ell}$ and of the Carr--Maradudin coefficient. The bound of
Refs.~\cite{SI-NeumayrMetzner1998,SI-KopperMagnen2001} excludes this enhancement.
Hypotheses (H1) and (H2) of main-text Sec.~\ref{sec:completeness}
concern what this counting does not control: the uniformity of the bound,
the degenerate regions (where the constant of the bound grows), the mixed
regions with hard subdiagrams, the absence of $\log r_s$ terms from the hard
regions, and the regularity of the higher-order polarization at fixed $z$.

\begin{proposition}[Completeness]\label{S-prop:complete}
Under hypotheses (H1) and (H2) of main-text Sec.~\ref{sec:completeness},
the coefficient of $r_s^2\log r_s$ in the correlation energy of the
paramagnetic two-dimensional electron gas is $a_{2\ell}$ of main-text
Eq.~\eqref{eq:a2l}.
\end{proposition}
\begin{proof}[Proof sketch]
By (H2) and the one-line counting, the only term of order $r_s^2\log r_s$
with one soft line is the residue of Sec.~\ref{S-dsec:a2l}. By the bound of
Refs.~\cite{SI-NeumayrMetzner1998,SI-KopperMagnen2001} and (H1), the diagrams
with $L\ge2$ soft loops are $O(r_s^4|\log r_s|^p)$, and
the hard regions give powers of $r_s$ without $\log r_s$. The
coefficient of $r_s^2\log r_s$ is therefore $a_{2\ell}$.
\end{proof}

\emph{Screened internal exchange line.} Replacing the bare exchange line in
\eqref{S-deq:combined} by the Thomas--Fermi form $2\pi/(|\bm p-\bm k|+\kappa_{\rm s})$ for both the vertex
and the self-energy parts changes $R_2(1)$. For the screening wave numbers
$\kappa_{\rm s}=10^{-1},\dots,10^{-5}$ (in units of $k_F$),
\[
 \frac{R_2(1;\kappa_{\rm s})-R_2(1)}{\kappa_{\rm s}}=-0.02951,\ -0.03180,\ -0.03212,\ -0.03216,\ -0.03217,
\]
approaching the value $-0.0321699$ obtained from the integrand at
$\kappa_{\rm s}\to0$. The change is thus $O(\kappa_{\rm s})$, without
$\log\kappa_{\rm s}$. The angles $\phi<\kappa_{\rm s}$
contribute $O(\kappa_{\rm s}^2)$, and the remainder after the linear
term is $O(\kappa_{\rm s}^2\log\kappa_{\rm s})$, from
$\kappa_{\rm s}\lesssim\phi\ll1$. With $\kappa_{\rm s}\propto r_s$, the
screening of the internal line shifts the $\Pi^{(1)}$ series by $O(r_s^3\log r_s)$.

\emph{Three dimensions.} The one-line terms $\Pi^{(1)}$ and $g_3$ alone
reproduce the complete coefficient $C$ of main-text
Eq.~\eqref{eq:jellium-expansion} (Sec.~\ref{S-dsec:response}), and the
counting places the two-loop diagrams at $r_s^{2L-2}=r_s^2$.

\section{Spin polarization in two dimensions}\label{S-dsec:spin}
For spin polarization $\zeta$ in two dimensions, with $r_s$ referring to the total density,
\[
 k_{F\sigma}=k_Fx_\sigma,\qquad x_\pm=(1\pm\zeta)^{1/2}.
\]
The density of states per spin is independent of $k_{F\sigma}$ in two dimensions. Hence the
free response of spin $\sigma$ is $h_2(z/x_\sigma)$ times its paramagnetic value per spin, and
the corresponding first-order exchange polarization, with coupling $e^2/(\hbar v_{F\sigma})$, is
$R_2(z/x_\sigma)/x_\sigma$.
The coefficients then scale as
\begin{equation}
 \frac{a_\ell(\zeta)}{a_\ell(0)}=\frac{\int[\sum_\sigma h_2(z/x_\sigma)]^3\,\mathrm dz}{8\int h_2^3\,\mathrm dz},\qquad
 \frac{a_{2\ell}(\zeta)}{a_{2\ell}(0)}=\frac{\int[\sum_\sigma R_2(z/x_\sigma)/x_\sigma][\sum_\sigma h_2(z/x_\sigma)]^2\,\mathrm dz}{8\int R_2h_2^2\,\mathrm dz}.
 \label{S-deq:spin}
\end{equation}

\emph{Checks.} For $\zeta=1$ only $x_+=\sqrt2$ remains, and the substitution $z\to\sqrt2z$ gives
\[
 \frac{a_\ell(1)}{a_\ell(0)}=\frac{\sqrt2\int h_2^3\,\mathrm dz}{8\int h_2^3\,\mathrm dz}=\frac1{4\sqrt2},\qquad
 \frac{a_{2\ell}(1)}{a_{2\ell}(0)}=\frac{\int R_2h_2^2\,\mathrm dz}{8\int R_2h_2^2\,\mathrm dz}=\frac18 .
\]
The first is the known ratio \cite{SI-RajagopalKimball1977,SI-LoosGill2DPRB,SI-ChesiGiuliani2007},
which checks \eqref{S-deq:spin}. Table~\ref{S-dtab:spin} lists intermediate values. At small $\zeta$,
\[
 a_{2\ell}(\zeta)/a_{2\ell}(0)\simeq1-0.0355\zeta^2,\qquad
 a_\ell(\zeta)/a_\ell(0)\simeq1-0.289\zeta^2.
\]

\begin{table}[h]
\centering\small
\caption{Spin dependence of the coefficients of the two $\log r_s$ terms of the
two-dimensional gas, \eqref{S-deq:spin}.}
\label{S-dtab:spin}
\begin{tabular}{@{}cccccccc@{}}
\toprule
$\zeta$ & $0$ & $0.2$ & $0.4$ & $0.6$ & $0.8$ & $0.9$ & $1$\\
\midrule
$a_\ell(\zeta)/a_\ell(0)$ & $1$ & $0.98826$ & $0.95076$ & $0.87865$ & $0.74375$ & $0.61946$ & $1/(4\sqrt2)$\\
$a_{2\ell}(\zeta)/a_{2\ell}(0)$ & $1$ & $0.99855$ & $0.99377$ & $0.98391$ & $0.96281$ & $0.93979$ & $1/8$\\
\bottomrule
\end{tabular}
\end{table}

\emph{Limit $\zeta\to1$.} The ratio for $a_\ell$ tends to $1/(4\sqrt2)$, with corrections of
order $(1-\zeta)^{1/2}$. However, the ratio for $a_{2\ell}$ does not tend to its value at
$\zeta=1$. In the region $z\lesssim x_-$, the minority response $R_2(z/x_-)/x_-$ is of order
$1/x_-$ on an interval of width $x_-$. With $z=x_-u$ and $h_2(0)=1$ it leaves a finite
contribution,
\[
 \int\frac{R_2(z/x_-)}{x_-}\Bigl[\sum_\sigma h_2(z/x_\sigma)\Bigr]^2\mathrm dz
 =\int R_2(u)\bigl[h_2(x_-u/x_+)+h_2(u)\bigr]^2\,\mathrm du
 \to\int R_2(u)[1+h_2(u)]^2\,\mathrm du .
\]
Hence
\[
 \frac{a_{2\ell}(\zeta\to1)}{a_{2\ell}(0)}=\frac18+\frac{\int R_2(1+h_2)^2\,\mathrm du}{8\int R_2h_2^2\,\mathrm du}=0.81649,
\]
a limit approached as $0.87373$, $0.82392$, $0.81727$ and $0.81657$ at
$1-\zeta=10^{-2}$, $10^{-4}$, $10^{-6}$ and $10^{-8}$.

\emph{Order of limits.} The minority spins have the coupling
$e^2/(\hbar v_{F\downarrow})\propto r_s(1-\zeta)^{-1/2}$. The expansion in $r_s$ at fixed
$\zeta$ is therefore uniform only for $r_s\ll(1-\zeta)^{1/2}$, and the limits $r_s\to0$ and
$\zeta\to1$ do not commute. In three dimensions, the minority density of states vanishes in
proportion to $k_{F\downarrow}$, and the same counting gives a contribution of order
$k_{F\downarrow}/k_F=(1-\zeta)^{1/3}$ that vanishes as $\zeta\to1$.
\section{Definitions and first-order exchange polarization}\label{S-app-sec:setting}
This section derives the normalization of the kite functional of main-text
Eq.~\eqref{eq:KD},
\[
 \mathcal K[f]=A_0\int\mathrm d^3q\int_0^\infty\frac{\mathrm d\nu}{\pi}\,f(q,\mathrm i\nu)X(q,\nu),\qquad
 A_0=\frac3{32\pi^5},
\]
and defines the vertex part and the fixed-density self-energy part of the first-order
exchange polarization.

\subsection{Units and kite normalization}\label{S-app-sec:units-kite}
\emph{Conventions.} For $N$ electrons in a volume $\mathcal V$, the density $n_e$ (not $n$,
which counts the interactions), the Bohr radius, $r_s$, the coupling-scaled Hamiltonian and the
Coulomb interaction are
\[
 n_e=N/\mathcal V,\quad a_B=\hbar^2/(m_ee^2),\quad 4\pi n_e(r_sa_B)^3/3=1,\quad
 \widehat H_\lambda=\widehat T+\lambda\widehat V,\quad v(q)=4\pi e^2/q^2.
\]
Relative to Fetter and Walecka \cite[Chap.~3, Sec.~9]{SI-FetterWalecka1971}, the Green function
is energy-normalized, with $G_{\rm FW}=\hbar G$ and $\Sigma_{\rm FW}=\Sigma/\hbar$. The internal
measure is $\mathrm d\mathfrak k=\mathrm d^3k\,\mathrm d(\hbar k^0)$, with $\mathrm d\mathfrak k/(2\pi)^4$
per independent integration. A diagram with $n$ interactions and $n_{\rm L}$ fermion loops
carries $\mathrm i^n(-1)^{n_{\rm L}}$.

\emph{Self-energy trace.} Coupling integration gives the ground-state energy as a self-energy
trace [Eq.~(5) of Ref.~\cite{SI-Benites2024}, BRM below]. BRM write the left side of that
equation as $\langle\Psi_0|\widehat V|\Psi_0\rangle$, but with the measure
$\mathrm d\lambda/\lambda$ the right side is the coupling-integrated interaction energy,
\begin{equation}
  E(1)-E(0)=\int_0^1\!\mathrm{d}\lambda\,\langle\widehat V\rangle_\lambda
  =-\frac{\mathrm{i}\mathcal V}{2}\int_0^1\frac{\mathrm{d}\lambda}{\lambda}
  \sum_\sigma\int\!\frac{\mathrm{d}\mathfrak k}{(2\pi)^4}\,\mathrm{e}^{\mathrm{i} k^0t_0}\,
  \Sigma_\sigma^\lambda(k)\,G_\sigma^\lambda(k),
  \label{S-app-eq:GM}
\end{equation}
with $E(\lambda)$ the ground-state energy of $\widehat H_\lambda$,
$\langle\cdot\rangle_\lambda$ its ground-state expectation value, $\sigma$ the spin index,
$t_0\to0^+$, $k=(\bm k,k^0)$, $\Sigma^\lambda_\sigma$ the proper self-energy and
$G^\lambda_\sigma$ the interacting propagator. With $k_{F,\sigma}$ the Fermi wave number of
spin $\sigma$, the free propagator is
\[
 G_{0,\sigma}(\bm k,k^0)=\frac{\Theta(k-k_{F,\sigma})}{\hbar k^0-\varepsilon_{\bm k}+\mathrm i0}
 +\frac{\Theta(k_{F,\sigma}-k)}{\hbar k^0-\varepsilon_{\bm k}-\mathrm i0},\qquad
 \varepsilon_{\bm k}=\frac{\hbar^2k^2}{2m_e}.
\]

\emph{Kite.} The proper second-order exchange self-energy with one line replaced by
$W_\lambda$ carries $\mathrm i^2=-1$. With the trace factor $-\mathrm i/2$ of \eqref{S-app-eq:GM}
the prefactor is $(-\mathrm i/2)\,\mathrm i^2=\mathrm i/2$, and the measure
$\mathrm d\lambda/\lambda$ absorbs the $\lambda$ of the bare line. With
$\mathrm d\mathfrak q_i=\mathrm d^3q_i\,\mathrm d(\hbar q_i^0)$, the kite contribution
$E_{\rm kite}$ to \eqref{S-app-eq:GM} is
\begin{equation}
  E_{\rm kite}=\frac{\mathrm{i}\mathcal V}{2}\int_0^1\!\mathrm{d}\lambda
  \int\!\frac{\mathrm{d}\mathfrak k\,\mathrm{d}\mathfrak q_1\,\mathrm{d}\mathfrak q_2}{(2\pi)^{12}}\,
  \mathrm{e}^{\mathrm{i} k^0t_0}\,v(\bm{k}-\bm{q}_1-\bm{q}_2)\,
    W_\lambda(q_1,q_1^0)\,\mathcal G_4(k,q_1,q_2),
  \label{S-app-eq:E2b-BRM}
\end{equation}
where the exchange loop is
\[
 \mathcal G_4=\sum_\sigma G_{0,\sigma}(q_1+q_2)\,G_{0,\sigma}(q_2)\,G_{0,\sigma}(k-q_1)\,G_{0,\sigma}(k).
\]
For a time-ordered line $f$ whose causal continuation is analytic in quadrants I and III of
the complex $q_1^0$ plane and bounded at large frequency, define
\begin{equation}
  \mathcal K[f]:=\frac{\mathrm{i}\mathcal V}{4N}
  \int\!\frac{\mathrm{d}\mathfrak k\,\mathrm{d}\mathfrak q_1\,\mathrm{d}\mathfrak q_2}{(2\pi)^{12}}\,
  \mathrm{e}^{\mathrm{i} k^0t_0}\,v(\bm{k}-\bm{q}_1-\bm{q}_2)\,f(q_1,q_1^0)\,\mathcal G_4(k,q_1,q_2),
  \label{S-app-eq:Kdef-BRM}
\end{equation}
\begin{equation}
  e_{\rm kite}=\frac{E_{\rm kite}}{N}=2\int_0^1\!\mathrm{d}\lambda\;\mathcal K[W_\lambda]=2\mathcal K[\overline W],
  \qquad
  \mathcal K[v]=e_{2x}.
  \label{S-app-eq:Kweight}
\end{equation}
The first relation is main-text Eq.~\eqref{eq:energy-definitions}. The factor $1/4$ in
\eqref{S-app-eq:Kdef-BRM} is the bare closed-diagram coefficient $2/(2!\,2^2)$.

\emph{Frequency integrals.} The $k^0$ and $q_2^0$ integrations each give a factor
\[
 \frac{\mathrm i\Delta}{\hbar q_1^0-\delta\pm\mathrm i0},\qquad
 \Delta=\Delta_{\bm p}(\bm q)=n_{\bm p}-n_{\bm p+\bm q},\qquad
 \delta=\delta_{\bm p}(\bm q)=\varepsilon_{\bm p+\bm q}-\varepsilon_{\bm p},\qquad
 n_{\bm p}=\Theta(k_{F,\sigma}-p),
\]
where $\bm q=\bm q_1$ and $\bm p$ stands for $\bm q_2$ or, after the shift
$\bm k\mapsto\bm k+\bm q_1$ below, for $\bm k$. The poles in $q_1^0$ lie in quadrants IV
($\Delta=1$) and II ($\Delta=-1$). Hence the contour rotation $q_1^0=\mathrm i\nu$ crosses no
singularity. The large arcs vanish by the $O(|q_1^0|^{-2})$ decay of the two denominators.
Combining the two halves of the imaginary axis gives
\[
 \int_{-\infty}^\infty\frac{\mathrm d\nu}{2\pi}\ \to\ \int_0^\infty\frac{\mathrm d\nu}{\pi}
 \quad\text{with the real part taken}.
\]
After the shift we obtain
\begin{equation}
  \begin{aligned}
  \mathcal K[f]={}&\frac{\mathcal V}{4N}\sum_\sigma
  \int\!\frac{\mathrm{d}^3q_1}{(2\pi)^{3}}\int_0^\infty\!\frac{\mathrm{d}(\hbar\nu)}{\pi}\,
  f(q_1,\mathrm{i}\nu)\\
  &\times\int\!\frac{\mathrm{d}^3k\,\mathrm{d}^3q_2}{(2\pi)^{6}}\,v(\bm{k}-\bm{q}_2)\,
  \operatorname{Re}\frac{\Delta_{\bm{q}_2}(\bm{q}_1)\,\Delta_{\bm{k}}(\bm{q}_1)}
  {[\mathrm{i}\hbar\nu-\delta_{\bm{q}_2}(\bm{q}_1)][\mathrm{i}\hbar\nu-\delta_{\bm{k}}(\bm{q}_1)]},
  \end{aligned}
  \label{S-app-eq:K-physunits}
\end{equation}
where the prefactor is positive because the phases (prefactor of $\mathcal K[f]$, the $k^0$
and $q_2^0$ integrations, $\mathrm dq_1^0=\mathrm i\,\mathrm d\nu$) combine as
$\mathrm i\cdot\mathrm i^2\cdot\mathrm i=+1$.

\emph{Prefactor.} After rescaling to the units of main-text Eq.~\eqref{eq:unit-scales}
($\bm q_1=\bm q$, $\bm q_2=\bm p$), the momentum measures supply $k_F^9$, the two interactions
$v_0^2$ and the frequency integral with two denominators $E_0^{-1}$. The prefactor $A_0$ of
main-text Eq.~\eqref{eq:KD} is
\begin{equation}
  \begin{aligned}
 A_0\,\mathrm{Ry}
 &=\frac{2}{4n_e}\frac{k_F^9}{(2\pi)^9}\frac{v_0^2}{E_0}
 =\frac{3m_e e^4}{64\pi^5\hbar^2}
 =\frac3{32\pi^5}\,\mathrm{Ry}\\
 &=2_{\rm Ha\to Ry}(\tfrac12)_{\rm diagram}2_{\rm spin}
 (\tfrac12)_{\rm norm}(4\pi)^2(2\pi)^{-9}(3\pi^2)\,\mathrm{Ry},
 \end{aligned}
 \label{S-app-eq:A0-dimensional}
\end{equation}
using
\[
 \frac1{n_e}=\frac{\mathcal V}N=\frac{3\pi^2}{k_F^3},\qquad
 v_0=\frac{4\pi e^2}{k_F^2},\qquad
 E_0=\frac{\hbar^2k_F^2}{m_e},\qquad
 \frac{m_ee^4}{\hbar^2}=2~\mathrm{Ry}.
\]
The factors $(\tfrac12)_{\rm diagram}$ from \eqref{S-app-eq:GM} and $(\tfrac12)_{\rm norm}$ from
\eqref{S-app-eq:E2b-BRM}--\eqref{S-app-eq:Kdef-BRM} give the bare coefficient $1/4$. All four
propagators of the exchange loop carry the same spin, and the spin sum gives $2_{\rm spin}$.
With the dimensionless occupation difference
$\Delta_{\bm p}(\bm q)=\Theta(1-|\bm p|)-\Theta(1-|\bm p+\bm q|)$ this gives main-text
Eq.~\eqref{eq:KD}. The known value $\mathcal K[v]=e_{2x}$ checks the normalization without
fitting (Sec.~\ref{S-sm-sec:verification}).

\emph{General diagram.} The same steps give main-text Eq.~\eqref{eq:general-functional} for
any closed diagram $D$ with one selected line $f$. The trace \eqref{S-app-eq:GM} and the factor
$\mathrm i^n(-1)^{n_{\rm L}}$ give the closed diagram, and its coefficient $1/S_D$ replaces
the bare kite coefficient $1/4$ of \eqref{S-app-eq:Kdef-BRM}. Carrying out all internal
frequency integrals except that of the selected line leaves $f$ times the time-ordered
insertion $\Pi_D(q,q^0)$. As for $\Delta=\pm1$ above, each of its time-ordered terms has
$q^0$-dependent denominators
\[
 \hbar q^0-E_I+\mathrm i0\quad\text{or}\quad\hbar q^0+E_I-\mathrm i0,
\]
with excitation energies $E_I>0$ of the intermediate states. Its poles therefore lie in
quadrants IV and II. The two halves of the imaginary axis combine to
$\int_0^\infty\mathrm d\nu/\pi$ of the real part, because $f$ is real and even in $\nu$ on the
imaginary axis and the real numerators give $\Pi_D(q,-\mathrm i\nu)=\Pi_D(q,\mathrm i\nu)^*$.
For the kite and the self-energy pair, every term has at least two $q^0$-dependent
denominators [\eqref{S-app-eq:K-physunits} and \eqref{S-app-eq:PiSigma}], and therefore the large arcs
vanish. For a general $D$ we assume this decay. Comparing \eqref{S-app-eq:K-physunits} with
\eqref{S-app-eq:Pix-norm} shows that the closed diagram per electron is, in dimensional form,
\[
 -\frac{\mathcal V}{N}\,\frac1{S_D}\int\frac{\mathrm d^3q}{(2\pi)^3}\int_0^\infty
 \frac{\mathrm d(\hbar\nu)}{\pi}\,f\,\Pi_D .
\]
In the units of Sec.~\ref{S-app-sec:response}, $\Pi_D$ carries one factor $C=\pi^2q_{\rm TF}^2$
for each independent loop momentum of the insertion, that is, one for each bare interaction of
$D$ including the selected line. With $\mathcal V/N=3\pi^2/k_F^3$ and
$E_0=32/(\pi^2q_{\rm TF}^4)$~Ry, the prefactor becomes
\[
 -\frac{3\pi^2}{k_F^3}\,\frac{k_F^3}{(2\pi)^3}\,\frac{E_0}{S_D}
 =-\frac{3E_0}{8\pi S_D}
 =-\frac{12}{S_D\pi^3q_{\rm TF}^4}~\mathrm{Ry}.
\]
This gives $Y_D=-12\,\Pi_D/(S_D\pi^3q_{\rm TF}^4)$ of main-text
Sec.~\ref{sec:linear-diagrams}. For the kite, $S_D=4$ and $\Pi_D=\Pi^{(1)}_{\rm vtx}$ of
\eqref{S-app-eq:Pix-norm} give main-text Eq.~\eqref{eq:KD},
\[
 Y_D=-\frac{3}{\pi^3q_{\rm TF}^4}\Bigl(-\frac{q_{\rm TF}^4}{32\pi^2}X\Bigr)=\frac3{32\pi^5}X=A_0X .
\]

\subsection{Vertex and Fock self-energy corrections}\label{S-app-sec:response}
This subsection supplies the responses that enter main-text Eqs.~\eqref{eq:KT-sigma} and
\eqref{eq:ward}. With \eqref{S-app-eq:Pix-norm} below, the prefactor of main-text
Eq.~\eqref{eq:KT-sigma} gives $\mathcal K_{\Pi^{(1)}_{\rm vtx}}=\mathcal K$ of main-text
Eq.~\eqref{eq:KD}. Responses are measured in $v_0^{-1}$, hence $\Pi_0=\chi_0=-q_{\rm TF}^2\mathcal L$. With the loop factor $C$ of
Sec.~\ref{S-dsec:lindhard} and the spin factor $2$, the free response is
\[
 \Pi_0(q,\mathrm i\nu)=2C\int\frac{\mathrm d^3p}{(2\pi)^3}\Delta_{\bm p}g_p,\qquad
 g_p=(\mathrm i\nu-\delta_{\bm p})^{-1}.
\]
For the first-order exchange polarization \cite{SI-Holas1979,SI-Qian2015} we use the dimensionless
definitions, with the argument $\bm q$ of $\Delta_{\bm p}(\bm q)$ and $\delta_{\bm p}(\bm q)$
suppressed,
\begin{equation}
\begin{gathered}
 \Delta_{\bm p}=\Delta_{\bm p}(\bm q),\quad \delta_{\bm p}=\bm p\cdot\bm q+q^2/2,\\
 q_{\rm TF}^2=a_{\rm TF}r_s,\quad a_{\rm TF}=4\alpha_L/\pi,\quad
 \alpha_L=(4/(9\pi))^{1/3},\\
 \Sigma_x(k)=-\frac{q_{\rm TF}^2}{2}f_L(k),\qquad
 f_L(k)=\frac12+\frac{1-k^2}{4k}\log\left|\frac{1+k}{1-k}\right|,\\
 C=\pi^2q_{\rm TF}^2,\qquad v(q)=q^{-2}.
\end{gathered}
 \label{S-app-eq:definitions-u6}
\end{equation}
Here $\Sigma_x$ is the Fock self-energy, and $\int_{|\bm k|<1}|\bm p-\bm k|^{-2}\,\mathrm d^3k=4\pi f_L(p)$
gives the form above,
\[
 \Sigma_x(\bm p)=-\int_{|\bm k|<1}\frac{\mathrm d^3k}{(2\pi)^3}\,Cv(\bm p-\bm k)
 =-\frac{4\pi C}{(2\pi)^3}\,f_L(p)=-\frac{q_{\rm TF}^2}2f_L(p).
\]
The vertex part $\Pi^{(1)}_{\rm vtx}(q,\mathrm i\nu)$ of the first-order exchange polarization
$\Pi^{(1)}$ (main-text Sec.~\ref{sec:se-cancellation}) is, with $X$ of main-text
Eq.~\eqref{eq:Xdef},
\begin{equation}
 \Pi_{\rm vtx}^{(1)}=-2C\int\frac{\mathrm{d}^3p\,\mathrm{d}^3k}{(2\pi)^6}
 C v(\bm p-\bm k)\Delta_{\bm p}\Delta_{\bm k}
 \operatorname{Re}\frac1{(\mathrm{i}\nu-\delta_{\bm p})(\mathrm{i}\nu-\delta_{\bm k})}
 =-\frac{q_{\rm TF}^4}{32\pi^2}X.
 \label{S-app-eq:Pix-norm}
\end{equation}

\emph{Self-energy part.} The Fock self-energy insertions shift $\delta_{\bm p}$ by
$\Sigma_x(\bm p+\bm q)-\Sigma_x(\bm p)$, which to first order gives
\[
 g_p\to g_p+[\Sigma_x(\bm p+\bm q)-\Sigma_x(\bm p)]\,g_p^2 .
\]
The vertex and self-energy corrections use the same bare exchange line,
$\delta\Sigma_x(\bm p)/\delta n_{\bm k}=-Cv(\bm p-\bm k)$. The fixed-density shift
$\delta\mu_x=\Sigma_x(1)=-q_{\rm TF}^2/4$ cancels the occupation-surface terms of the self-energy part
$\Pi^{(1)}_{\rm SE}(q,\mathrm i\nu)$ of $\Pi^{(1)}$, leaving
\begin{equation}
 \Pi_{\rm SE}^{(1)}=2C\int\frac{\mathrm{d}^3p}{(2\pi)^3}\,\Delta_{\bm p}\,\operatorname{Re}
 \frac{\Sigma_x(\bm p+\bm q)-\Sigma_x(\bm p)}
 {(\mathrm{i}\nu-\delta_{\bm p})^2},
 \label{S-app-eq:PiSigma}
\end{equation}
with
\begin{equation}
 \Sigma_x(\bm p+\bm q)-\Sigma_x(\bm p)
 =\int\frac{\mathrm d^3k}{(2\pi)^3}\,
 C v(\bm p-\bm k)\Delta_{\bm k}.
 \label{S-app-eq:definitions-u7}
\end{equation}

\emph{Sum.} Symmetrized in $\bm p\leftrightarrow\bm k$, the two parts read
\[
 \Pi^{(1)}_{\rm vtx}=C\int\frac{\mathrm d^3p\,\mathrm d^3k}{(2\pi)^6}\,Cv(\bm p-\bm k)\Delta_{\bm p}\Delta_{\bm k}\operatorname{Re}(-2g_pg_k),\qquad
 \Pi^{(1)}_{\rm SE}=C\int\cdots\operatorname{Re}\bigl(g_p^2+g_k^2\bigr).
\]
Adding the vertex and self-energy parts gives
\begin{equation}
 \Pi^{(1)}(q,\mathrm{i}\nu):=\Pi_{\rm vtx}^{(1)}+\Pi_{\rm SE}^{(1)}
 =C\int\frac{\mathrm d^3p\,\mathrm d^3k}{(2\pi)^6}\,
 C v(\bm p-\bm k)\Delta_{\bm p}\Delta_{\bm k}
 \operatorname{Re}\big[(g_p-g_k)^2\big].
 \label{S-app-eq:combined-response}
\end{equation}
This is Eq.~(4) of Ref.~\cite{SI-Qian2015}, symmetrized in $\bm p\leftrightarrow\bm k$ and
continued to imaginary frequency. The high-frequency $\nu^{-2}$ terms of the two parts cancel
(Sec.~\ref{S-app-sec:cancel}).

\section{High-density expansion and cancellation by the Fock self-energy insertions}\label{S-app-sec:highdensity}
We derive the $r_s\log^2r_s$ term of main-text Eq.~\eqref{eq:Clog2} and its
cancellation against the self-energy pair. The frequency-integrated
shell bounds are in Sec.~\ref{S-sm-sec:shell-bounds}.

\subsection{Fermi-surface origin and the integrated remainder}\label{S-app-sec:shell-origin}
This subsection derives main-text Eq.~\eqref{eq:Az} and, through
Eqs.~\eqref{S-app-eq:window}--\eqref{S-app-eq:highdensity-u2}, the
coefficient $c_{\rm hd}$ of main-text Eq.~\eqref{eq:Clog2}. It also gives
the small-$q$ form of $e_\Lambda$ stated in main-text Sec.~\ref{sec:radial-kernels}.
The three results are
\[
 X(q,qz)=a_x(z)\log\frac1q+b_x(z)+o(1),\quad
 e_\Lambda(q)=\kappa(\Lambda)\,q\log\frac1q+O(q),\quad
 e_{\rm kite}-e_{2x}\supset c_{\rm hd}\,r_s\log^2r_s ,
\]
with $a_x$, $b_x$ as in main-text Eq.~\eqref{eq:Az}.
The shells at the Fermi surface are integrated first, at fixed $z=\nu/q>0$.

\emph{Thin shells.} As $q\to0$ at fixed $z$,
\[
 \delta_{\bm p}=qc_p+O(q^2),\qquad c_p=\hat{\bm p}\cdot\hat{\bm q},\qquad \hat{\bm p}=\bm p/p .
\]
The two \emph{transition regions}, where $\Delta_{\bm p}=\pm1$, are at
small $q$ the shells of width $O(q)$ at the Fermi surface (main-text
Sec.~\ref{sec:radial-kernels}). Integrating $X$ of main-text
Eq.~\eqref{eq:Xdef} across them and over the
directions $\mathrm d\Omega_p$, $\mathrm d\Omega_k$
gives
\begin{equation}
\begin{aligned}
  X(q,qz)&=\int\!\mathrm{d}\Omega_p\,\mathrm{d}\Omega_k\,c_pc_k\,\mathcal R(c_p,c_k;z)\,
  \operatorname{Fp}\frac{1}{|\hat{\bm{p}}-\hat{\bm{k}}|^{2}}\\
  &\quad+2\pi\!\int\!\mathrm{d}\Omega_p\,c_p^{2}\mathcal R(c_p,c_p;z)
   \Big[\log\frac1q-\log|c_p|+\frac32\Big]\\
  &\quad+o(1),
\end{aligned}
  \label{S-app-eq:shell}
\end{equation}
where
\begin{equation}
  \mathcal R(a,b;z)=\frac{ab-z^{2}}{(ab-z^{2})^{2}+z^{2}(a+b)^{2}},
  \label{S-app-eq:Rdef}
\end{equation}
and $\operatorname{Fp}$ is the finite part. For an angular integrand $\Phi$,
\[
 \operatorname{Fp}\!\int\!\mathrm d\Omega_p\,\mathrm d\Omega_k\,\frac{\Phi(\hat{\bm p},\hat{\bm k})}{|\hat{\bm p}-\hat{\bm k}|^2}
 =\lim_{\epsilon\to0}\Big[\int_{|\hat{\bm p}-\hat{\bm k}|>\epsilon}\!\mathrm d\Omega_p\,\mathrm d\Omega_k\,
 \frac{\Phi(\hat{\bm p},\hat{\bm k})}{|\hat{\bm p}-\hat{\bm k}|^2}
 -2\pi\log\frac1\epsilon\int\!\mathrm d\Omega_p\,\Phi(\hat{\bm p},\hat{\bm p})\Big],
\]
with $\Phi=c_pc_k\mathcal R(c_p,c_k;z)$ here and $\Phi=c_pc_k\Lambda^2J(\Lambda;c_p,c_k)$ in
\eqref{S-app-eq:betaLambda} below. The subtracted term is $2\pi\log(1/\epsilon)$ times the
diagonal integral.

\emph{Derivation.} To derive \eqref{S-app-eq:shell}, write $\bm p=(1+qx)\hat{\bm p}$ and $\bm k=(1+qy)\hat{\bm k}$.
At fixed $z$ the occupation difference restricts $x$ to an interval of
length $|c_p|$. Each shell therefore contributes the radial measure $q|c_p|$
and the electron--hole factor $(\mathrm iqz-qc_p)^{-1}$. The factors $q$ cancel,
and with the signs $\Delta_{\bm p}=\operatorname{sgn}c_p$ of the transition regions the real part gives
\[
 \Delta_{\bm p}\Delta_{\bm k}\,(q|c_p|)(q|c_k|)\operatorname{Re}\frac1{(\mathrm iqz-qc_p)(\mathrm iqz-qc_k)}
 =c_pc_k\,\mathcal R(c_p,c_k;z).
\]
For $\hat{\bm p}\ne\hat{\bm k}$ the Coulomb factor tends to
$|\hat{\bm p}-\hat{\bm k}|^{-2}$, which gives the first line. Near
$\hat{\bm p}=\hat{\bm k}$, the separation of the two momenta has the tangential component
$\bm d$ and the normal component $q(x-y)$. With $x=|c_p|u$, $y=|c_p|v$,
\[
 \int_{|\bm d|<\epsilon}\frac{\mathrm d^2d}{d^2+q^2(x-y)^2}=2\pi\log\frac{\epsilon}{q|x-y|}+o(1),\qquad
 \int_0^1\!\!\int_0^1\log\frac1{q|x-y|}\,\mathrm du\,\mathrm dv=\log\frac1q-\log|c_p|+\frac32,
\]
since $\int_0^1\!\int_0^1\log|u-v|\,\mathrm du\,\mathrm dv=-3/2$. This
average over the two shells gives the second line.
Section~\ref{S-app-app:shell-integrated} justifies these steps after the
frequency integration.

\emph{Radial kernel.} For fixed $\Lambda$, integrating the pointwise expansion
\eqref{S-app-eq:shell} term by term gives the radial kernel
$e_\Lambda(q)$ of main-text Eq.~\eqref{eq:XLambda} in the form
\begin{equation}
\begin{gathered}
  e_\Lambda(q)=\kappa(\Lambda)q\log\frac1q+\beta(\Lambda)q+o(q),\\
  \kappa(\Lambda)=\frac{3}{2\pi^{2}}\,\Lambda\!\int_0^1\!\frac{c_p^{2}\,\mathrm{d} c_p}{(c_p+\Lambda)^{2}}
  =\frac{3\Lambda}{2\pi^{2}}\Big[1-2\Lambda\log\Big(1+\tfrac1\Lambda\Big)
   +\frac{\Lambda}{1+\Lambda}\Big],
\end{gathered}
  \label{S-app-eq:kappa}
\end{equation}
with
\[
 \kappa(\Lambda)\sim\frac{3\Lambda}{2\pi^2}\quad(\Lambda\to0),\qquad
 \kappa(\Lambda)\sim\frac1{2\pi^2\Lambda}\quad(\Lambda\to\infty).
\]
The coefficient $\kappa$ follows from the diagonal line of
\eqref{S-app-eq:shell}, with $\mathcal R(c,c;z)=(c^2-z^2)/(c^2+z^2)^2$.
The integrand is even in $c_p$, hence
$\int\mathrm d\Omega_p=4\pi\int_0^1\mathrm d|c_p|$, and
\eqref{S-app-app-eq:kappa-int} applies with $c=|c_p|$.
With the prefactor $4\pi A_0$, the Lorentzian weight of main-text
Eq.~\eqref{eq:XLambda} and $\mathrm d\nu=q\,\mathrm dz$,
\[
 \kappa=4\pi A_0\cdot2\pi\!\int\!\mathrm d\Omega_p\,c_p^2\!
 \int_0^\infty\frac{\mathrm dz}{\pi}
 \frac{\Lambda^2(c_p^2-z^2)}{(\Lambda^2+z^2)(c_p^2+z^2)^2}
 =16\pi^3A_0\Lambda\int_0^1\frac{c^2\,\mathrm dc}{(c+\Lambda)^2},\qquad
 16\pi^3A_0=\frac3{2\pi^2}.
\]
The finite part is
\begin{equation}
\begin{aligned}
  \beta(\Lambda)&=4\pi A_0\!\int\!\mathrm{d}\Omega_p\mathrm{d}\Omega_k\,c_pc_k\,
  \Lambda^{2}J(\Lambda;c_p,c_k)\,\operatorname{Fp}|\hat{\bm{p}}-\hat{\bm{k}}|^{-2}\\
  &\quad+8\pi^{2}A_0\!\int\!\mathrm{d}\Omega_p\,c_p^{2}\,
  \Lambda^{2}J(\Lambda;c_p,c_p)\big[\tfrac32-\log|c_p|\big],
\end{aligned}
\label{S-app-eq:betaLambda}
\end{equation}
with $J$ from \eqref{S-app-app-eq:J-definition}. It follows by inserting
\eqref{S-app-eq:shell} into
main-text Eq.~\eqref{eq:XLambda} with $\nu=qz$ and using
\[
 (ab-z^2)^2+z^2(a+b)^2=(a^2+z^2)(b^2+z^2),\qquad
 \int_0^\infty\frac{\mathrm dz}\pi\,\frac{\Lambda^2}{\Lambda^2+z^2}\,\mathcal R(a,b;z)=\Lambda^2J(\Lambda;a,b).
\]
The first relation turns $\mathcal R$ into the integrand of \eqref{S-app-app-eq:J-definition}, which gives the second.
The bounds
of Sec.~\ref{S-sm-sec:shell-bounds} control the integrated remainder of
\eqref{S-app-eq:kappa} only to $O(q)$. For finite $\Lambda$, only the form
$e_\Lambda(q)=\kappa(\Lambda)q\log(1/q)+O(q)$ is therefore proved, and
$\beta(\Lambda)$ is the term-by-term coefficient. The constants in these bounds do not
depend on $\Lambda$ (Sec.~\ref{S-sm-sec:shell-bounds}). The same form therefore holds at
$\Lambda=\infty$, where $\kappa(\infty)=0$. The main text uses only the form with $O(q)$.

\emph{Window.} The coupling average of the screened line, main-text
Eqs.~\eqref{eq:Dlambda} and \eqref{eq:rho}, is
\begin{equation}
\begin{gathered}
 2\overline W=2\int_0^1\frac{\lambda v}{1+\lambda c}\,\mathrm d\lambda=v\rho(c),\qquad
 \rho(c)=\frac{2[c-\log(1+c)]}{c^2},\qquad c=\frac{q_{\rm TF}^2 \mathcal L(q,\nu)}{q^2},\\
 \mathcal L(0,0)=1,\qquad \rho(c)-1=-\tfrac23c+O(c^2).
\end{gathered}
 \label{S-app-eq:highdensity-u1}
\end{equation}
In the window $q_{\rm TF}\ll q\ll1$, the screening ratio $c$ is small.
By main-text Eqs.~\eqref{eq:bare-value} and
\eqref{eq:energy-definitions}, and with
$\rho-1\simeq-\tfrac23c$ in main-text Eq.~\eqref{eq:KD},
\[
 e_{\rm kite}-e_{2x}=\mathcal K[v(\rho-1)],\qquad
 v(\rho-1)\simeq-\frac23\,\frac{q_{\rm TF}^2\mathcal L(q,\nu)}{q^4},\qquad
 \mathrm d^3q=4\pi q^2\,\mathrm dq,
\]
and therefore
\begin{equation}
  e_{\rm kite}-e_{2x}=-\frac{8\pi A_0q_{\rm TF}^2}{3}\int_{q_{\rm TF}}^{1}\frac{\mathrm{d} q}{q^{2}}
  \int_0^\infty\frac{\mathrm{d}\nu}{\pi}\,\mathcal L(q,\nu)\,X(q,\nu)+O(q_{\rm TF}^2\log(1/q_{\rm TF}^2)),
  \label{S-app-eq:window}
\end{equation}
where the neglected regions are bounded as follows. For $q>1$ the bound uses
\eqref{S-app-eq:spectralconvergence-u2}, and for $q<q_{\rm TF}$ it uses
Sec.~\ref{S-sm-sec:shell-bounds} and Eq.~\eqref{S-app-eq:spectralconvergence-u2},
\[
\begin{aligned}
 q>1:&\quad |\rho-1|\le\tfrac23c,\quad \int_0^\infty|X|\,\mathrm d\nu=O(q^{-2})
 \ \Rightarrow\ O(q_{\rm TF}^2),\\
 q<q_{\rm TF}:&\quad |\rho-1|\le1,\quad \int|X(q,qz)|\,\mathrm dz=O\bigl(\log(e/q)\bigr)
 \ \Rightarrow\ O\bigl(q_{\rm TF}^2\log(1/q_{\rm TF}^2)\bigr).
\end{aligned}
\]
By the bound $\int|X(q,qz)|\,\mathrm dz=O(\log(e/q))$, the $O(c^2)$ error of $\rho-1\simeq-\tfrac23c$
inside the window also contributes $O(q_{\rm TF}^2\log(1/q_{\rm TF}^2))$.

\emph{Coefficient $a_x$.} Within the window, $X(q,qz)$ has the expansion of main-text
Eq.~\eqref{eq:Az}. The
second line of \eqref{S-app-eq:shell},
$\mathrm d\Omega_p=2\pi\,\mathrm dc_p$ and
\eqref{S-app-eq:highdensity-u4} give
\begin{equation}
  a_x(z)=2\pi\!\int\!\mathrm{d}\Omega_p\,c_p^{2}\,\frac{c_p^{2}-z^{2}}{(c_p^{2}+z^{2})^{2}}
  =8\pi^{2}\Big[1-2z\arctan\frac1z+\frac{z^{2}}{1+z^{2}}\Big]
  =8\pi^{2}\big[h(z)+zh'(z)\big],
  \label{S-app-eq:Az}
\end{equation}
and therefore $\int_0^\infty h(z)X(q,qz)\,\mathrm dz/\pi=\mathcal A\log(1/q)+O(1)$ by
\eqref{S-app-app-eq:integrated-shell}. The replacement of $h$ by
$\mathcal L(q,qz)$ costs only $O(q\log(e/q))$ by
\eqref{S-app-eq:highdensity-u3}. The integral of main-text
Eq.~\eqref{eq:Az} is
\begin{equation}
  \mathcal A=\int_0^\infty\frac{\mathrm{d} z}{\pi}\,h(z)a_x(z)
  =8\pi\!\int_0^\infty\!\big[h^{2}+\tfrac12 z\,(h^{2})'\big]\mathrm{d} z
  =4\pi\!\int_0^\infty\!h(z)^{2}\,\mathrm{d} z=\frac{4\pi^{2}}{3}(1-\log2),
  \label{S-app-eq:calA}
\end{equation}
where the third equality integrates $\tfrac12z(h^2)'$ by parts, since
$zh^2\to0$ at both ends ($h\simeq1/(3z^2)$ at large $z$). The last equality
uses $\int_0^\infty h(z)^2\,\mathrm dz=\pi(1-\log2)/3$
\cite{SI-GellMannBrueckner}. With
$h(z)=\int_0^1a^2\,\mathrm da/(a^2+z^2)$ and
$\int_0^\infty\mathrm dz/[(z^2+a^2)(z^2+b^2)]=\pi/[2ab(a+b)]$,
\[
 \int_0^\infty h^2\,\mathrm dz=\frac\pi2\int_0^1\!\!\int_0^1\frac{ab}{a+b}\,\mathrm da\,\mathrm db
 =\frac\pi2\cdot\frac23(1-\log2).
\]

\emph{Coefficient $c_{\rm hd}$.} Insert
$\int_0^\infty\mathcal LX\,\mathrm d\nu/\pi=q[\mathcal A\log(1/q)+O(1)]$
into \eqref{S-app-eq:window}, with $q_{\rm TF}^2=a_{\rm TF}r_s$ and
$\log^2(q_{\rm TF}^2)=\log^2r_s+O(|\log r_s|)$. The $\log^2$ term is
\[
 e_{\rm kite}-e_{2x}\supset-\frac{8\pi A_0q_{\rm TF}^2}{3}\,\mathcal A\int_{q_{\rm TF}}^1\frac{\log(1/q)}q\,\mathrm dq
 \supset c_{\rm hd}\,r_s\log^2r_s,
\]
and main-text Eq.~\eqref{eq:Clog2} follows from
\begin{equation}
 \int_{q_{\rm TF}}^1\frac{\log(1/q)}q\,\mathrm dq=\tfrac18\log^2(q_{\rm TF}^2),
 \qquad c_{\rm hd}=-\frac{\pi A_0a_{\rm TF}}3\mathcal A
 =-\frac{\alpha_L(1-\log2)}{6\pi^3}.
 \label{S-app-eq:highdensity-u2}
\end{equation}
\subsection{Cancellation by the Fock self-energy insertions}\label{S-app-sec:cancel}
This subsection derives the self-energy factor $\log(1/q)$ of the kernel in
main-text Eq.~\eqref{eq:ward}, its cancellation after the frequency integration,
the tail condition of Lemma~\ref{S-lem:fsum}, and the
high-density result for $e_{\rm kite}+e_{\rm SE}$ of main-text
Sec.~\ref{sec:se-cancellation},
\[
 e_{\rm kite}+e_{\rm SE}=e_{2x}+O(r_s|\log r_s|).
\]

\emph{Shell expansion of $\Pi^{(1)}_{\rm SE}$.} The cusp
of the Fock self-energy $\Sigma_x=-q_{\rm TF}^2f_L/2$ at the Fermi surface is
\[
 f_L(1+u)-f_L(1)=-\frac u2\log\frac2{|u|}+O(u^2|\log|u||),\qquad
 p=1+q\eta_{\rm sh},\quad
 |\bm p+\bm q|-1=q(\eta_{\rm sh}+c_p)+O(q^2),
\]
which gives
\begin{equation}
 \Sigma_x(\bm p+\bm q)-\Sigma_x(\bm p)
 =\frac{q_{\rm TF}^2 q}{4}\left[c_p\log\frac1q+\mathcal H_{\rm sh}(\eta_{\rm sh},c_p)\right]
 +O(q_{\rm TF}^2 q^2|\log q|),
 \label{S-app-eq:fock-shell}
\end{equation}
where $\mathcal H_{\rm sh}$ denotes the part without
$\log(1/q)$ and is bounded on the shell. Integration across the
signed shell supplies the factor $qc_p$, and
\begin{equation}
 \int_{-1}^1\frac{c_p^2(c_p^2-z^2)}{(c_p^2+z^2)^2}\,\mathrm{d}c_p
 =2[h(z)+zh'(z)],\qquad z=\nu/q.
 \label{S-app-eq:highdensity-u4}
\end{equation}
With $2C/(2\pi)^3=q_{\rm TF}^2/(4\pi)$, the shell factor $qc_p$,
$(\mathrm i\nu-\delta_{\bm p})^{-2}=q^{-2}(\mathrm iz-c_p)^{-2}$ and
$\operatorname{Re}(\mathrm iz-c_p)^{-2}=(c_p^2-z^2)/(c_p^2+z^2)^2$,
Eqs.~\eqref{S-app-eq:PiSigma}, \eqref{S-app-eq:fock-shell} and
\eqref{S-app-eq:highdensity-u4} give
\[
\begin{aligned}
 \Pi^{(1)}_{\rm SE}(q,\mathrm iqz)
 &=\frac{q_{\rm TF}^2}{4\pi}\int\!\mathrm d\Omega_p\,qc_p\cdot\frac{q_{\rm TF}^2q\,c_p}4\log\frac1q
 \cdot\frac{c_p^2-z^2}{q^2(c_p^2+z^2)^2}+O(q_{\rm TF}^4)\\
 &=\tfrac14q_{\rm TF}^4[h+zh']\log(1/q)+O(q_{\rm TF}^4).
\end{aligned}
\]
The two parts therefore have terms in $\log(1/q)$ of opposite sign, main-text
Eq.~\eqref{eq:ward}.

\emph{Frequency integration.} The cancellation survives the frequency integration. For $a=|\delta_{\bm p}|/q$,
\begin{equation}
 \int_0^\infty h(z)\operatorname{Re}(\mathrm iz-a)^{-2}\frac{\mathrm dz}{\pi}
 =\tfrac12[\log(1+1/a)-(1+a)^{-1}],
 \label{S-app-eq:highdensity-u5}
\end{equation}
which, multiplied by the shell factor $a^2$, is the diagonal weight
\[
 M_h(a,a)=\int_0^\infty\frac{\mathrm dz}{\pi}\,h(z)\,a^2\mathcal R(a,a;z)
 =\int_0^1M_\Lambda(a,a)\,\mathrm d\Lambda,
\]
with $M_\Lambda$ of
\eqref{S-app-app-eq:weighted-shell} and $M_h$ defined after it. This diagonal weight
carries the vertex factor $\log(1/q)$ in \eqref{S-app-app-eq:integrated-shell}. The two
frequency-integrated terms in $\log(1/q)$ therefore cancel direction by direction.

\emph{Bounded part.} The part
$\tfrac14q_{\rm TF}^2q\,\mathcal H_{\rm sh}$ of \eqref{S-app-eq:fock-shell}
is $O(q_{\rm TF}^2qa[1+|\log a|])$ on the shell. By
\eqref{S-app-eq:highdensity-u5} with $\mathrm d\nu=q\,\mathrm dz$ and
$(\mathrm i\nu-\delta)^{-2}=q^{-2}(\mathrm iz-a)^{-2}$, and with the radial
width $qa$ and the factor
$2C=2\pi^2q_{\rm TF}^2$ of \eqref{S-app-eq:PiSigma},
\[
\begin{gathered}
 \int_0^\infty\frac{\mathrm d\nu}{\pi}\,h\operatorname{Re}(\mathrm i\nu-\delta)^{-2}=O\bigl(q^{-1}[1+|\log a|]\bigr),\\
 2C\cdot qa\cdot O\bigl(q_{\rm TF}^2qa[1+|\log a|]\bigr)\cdot O\bigl(q^{-1}[1+|\log a|]\bigr)
 =O\bigl(q_{\rm TF}^4\,q\,a^2[1+|\log a|]^2\bigr)
\end{gathered}
\]
in $\int_0^\infty h\,\Pi^{(1)}_{\rm SE}\,\mathrm d\nu/\pi$, that is,
$O(q\,a^2[1+|\log a|]^2)$ in the units of the kite integrand. This bound is
integrable over the grazing directions $a\to0$, perpendicular to the
momentum transfer, where the shell width vanishes.

\emph{High-frequency tails.} At large $\nu$,
\[
 \operatorname{Re}\frac1{(\mathrm i\nu-\delta_{\bm p})(\mathrm i\nu-\delta_{\bm k})}\simeq
 \operatorname{Re}\frac1{(\mathrm i\nu-\delta_{\bm p})^{2}}\simeq-\frac1{\nu^2}.
\]
Then \eqref{S-app-eq:Pix-norm} gives
$\Pi^{(1)}_{\rm vtx}\simeq q_{\rm TF}^4Z/(32\pi^2\nu^2)$, with $Z$ defined
below, and \eqref{S-app-eq:PiSigma} with \eqref{S-app-eq:definitions-u6} gives
a tail proportional to $\int\mathrm d^3p\,\Delta_{\bm p}[f_L(|\bm p+\bm q|)-f_L(p)]$.
With
$I_{\mathbb B}(p)=\int_{\mathbb B}|\bm p-\bm k|^{-2}\mathrm d^3k=4\pi f_L(p)$, the shift $\bm k\to\bm k+\bm q$ gives
\[
 I_{\mathbb B}(|\bm p+\bm q|)=\int_{|\bm k+\bm q|<1}\frac{\mathrm d^3k}{|\bm p-\bm k|^2},\qquad
 I_{\mathbb B}(|\bm p+\bm q|)-I_{\mathbb B}(p)=-\int\frac{\Delta_{\bm k}\,\mathrm d^3k}{|\bm p-\bm k|^2},
\]
and therefore
\begin{equation}
 \int\mathrm{d}^3p\,\Delta_{\bm p}[f_L(|\bm p+\bm q|)-f_L(p)]
 =-\frac{Z(q)}{4\pi},\qquad
 Z(q)=\iint\frac{\Delta_{\bm p}\Delta_{\bm k}}{|\bm p-\bm k|^2}\,\mathrm d^3p\,\mathrm d^3k,
 \label{S-app-eq:highdensity-u6}
\end{equation}
which gives $\Pi^{(1)}_{\rm SE}\simeq-q_{\rm TF}^4Z/(32\pi^2\nu^2)$, and the
$\nu^{-2}$ terms cancel for every $q$. This is the tail condition of
Lemma~\ref{S-lem:fsum}.

\emph{Energies.} With $\mathcal K_{\Pi^{(1)}_{\rm SE}}$ of main-text Eq.~\eqref{eq:KT-sigma},
\[
 e_{\rm kite}=2\!\int_0^1\!\mathcal K[W_\lambda]\,\mathrm d\lambda,\qquad
 e_{\rm SE}=2\!\int_0^1\!\mathcal K_{\Pi_{\rm SE}^{(1)}}[W_\lambda]\,\mathrm d\lambda,\qquad
 W_\lambda-\lambda v\simeq-\lambda^2vc,\qquad
 \int_0^12\lambda^2\,\mathrm d\lambda=\tfrac23 .
\]
The coupling weight $2$, integrated against the factor $\lambda^2$ of the first-order screening
term, gives the same $2/3$ for both energies. The
coefficients of $r_s\log^2r_s$ in $e_{\rm kite}$ and in $e_{\rm SE}$
are therefore $c_{\rm hd}$ and $-c_{\rm hd}$. The $r_s\log r_s$ terms require
the finite shell parts \cite{SI-LoosGill,SI-CarrMaradudin,SI-LoosGill2011}. For
$e_{\rm kite}$ the coefficient is $c_1=0.01150$~Ry
\cite{SI-Imoto2026b}. For $e_{\rm SE}$ and for the sum it is not
evaluated here. The window estimate \eqref{S-app-eq:window} applies to
$e_{\rm SE}$ with the absolute bound \eqref{S-app-eq:selfenergy-L1} in
place of the shell bound of $X$, and the bounded part of
\eqref{S-app-eq:fock-shell} is estimated above. Both $r_s\log^2r_s$
coefficients therefore come with remainders $O(r_s|\log r_s|)$, and
$\mathcal K_{\Pi_{\rm SE}^{(1)}}[v]=0$ [Eq.~\eqref{S-app-eq:selfenergyreduction-u7}, Sec.~\ref{S-app-sec:sigma-reduction}]
gives $e_{\rm kite}+e_{\rm SE}=e_{2x}+O(r_s|\log r_s|)$.

\subsection{Relation to the vertex of the \texorpdfstring{$GW\Gamma$}{GWGamma} scheme}\label{S-sm-sec:gwgamma}
The $GW\Gamma$ scheme of Takada \cite{SI-Takada2001,SI-MaebashiTakada2011}
constructs the density vertex as the ratio
\begin{equation}
 \Gamma_{\rm WI}(p,p+q)=\frac{G(p+q)^{-1}-G(p)^{-1}}{G_0(p+q)^{-1}-G_0(p)^{-1}}
 \label{S-sm-eq:gamma-wi}
\end{equation}
at four-momenta $p$ and $p+q$, with $q=(\bm q,\mathrm i\nu)$ on the
imaginary axis. The Ward identity
constrains the density and current vertices together. Hence
\eqref{S-sm-eq:gamma-wi} is consistent with it but is not the
unique density vertex at finite $q$ and $\nu$. With
$G_0(p+q)^{-1}-G_0(p)^{-1}=\mathrm i\nu-\delta_{\bm p}$,
\[
 G(p)G(p+q)\Gamma_{\rm WI}(p,p+q)=\frac{G(p)-G(p+q)}{\mathrm i\nu-\delta_{\bm p}},
\]
so a polarization built with $\Gamma_{\rm WI}$ depends on the
self-energy only through the momentum distribution. The static Fock
self-energy with the shift $\delta\mu_x$ leaves this
distribution unchanged. The first-order exchange
correction built with $\Gamma_{\rm WI}$ therefore vanishes identically. The
first-order vertex of $\Gamma_{\rm WI}$ and the split of the symmetrized exact first-order vertex are
\[
 -\frac{\Delta\Sigma_x(\bm p)}{\mathrm i\nu-\delta_{\bm p}}\ \to\ -\Delta\Sigma_x\,g_p^2,\qquad
 -2g_pg_k=-(g_p^2+g_k^2)+(g_p-g_k)^2 .
\]
The first cancels the self-energy insertion $+\Delta\Sigma_xg_p^2$. After symmetrization
this Ward part is $-(g_p^2+g_k^2)$ in main-text
Eq.~\eqref{eq:response-square-reminder}, and the remainder of the
exact first-order vertex is $(g_p-g_k)^2$. The terms in $\log(1/q)$ of main-text
Eq.~\eqref{eq:ward} are the two halves of the vanishing part. The
finite limit $R^{(1)}(z)$ of $-16\Pi^{(1)}(q,\mathrm iqz)/q_{\rm TF}^4$
as $q\to0$ [main-text Eq.~\eqref{eq:R1}] belongs to the remainder, which the
$GW\Gamma$ scheme represents by an electron--hole interaction $\bar I(q)$.

\subsection{Reference states and related work}\label{S-sm-sec:references}
In the uniform
gas, a Kohn--Sham reference with a local exchange--correlation potential
shifts all orbital energies by a constant, whereas the Fock self-energy insertion
$\Sigma_x(k)-\delta\mu_x$ depends on $k$. Other SOSEX variants
\cite{SI-Jansen2010,SI-Angyan2011,SI-BatesFurche2013,SI-Ren2013,SI-HesselmannGorling2010,SI-Gorling2019,SI-Hellgren2018,SI-MaggioKresse2016}
differ in the screened lines and the weights of the diagrams. Whether they retain
the $r_s\log^2r_s$ and $r_s^{-1/4}$ terms depends on their frequency kernels. Singles
(single-excitation) corrections \cite{SI-Ren2013} vanish on a Kohn--Sham reference by
momentum conservation. Hartree--Fock orbital energies contain the insertions to all orders,
together with the singularity $\propto u\log|u|$ of the Hartree--Fock
dispersion at the Fermi surface, where $k=k_F(1+u)$. Yasuhara \emph{et al.}
\cite{SI-Yasuhara1988}, whose exchange correction to the RPA starts from the Hartree--Fock
approximation, observed a related cancellation in the electron liquid. In
the contribution of first order in the dynamically screened interaction
to the spin-parallel structure factor, two self-energy terms and one
vertex term cancel at small momentum transfer. These authors argued that the
self-energy part should be treated together with the other terms of the
same order.

\section{Correspondence with the BRM representation}\label{S-sm-app-app:brm}
This section transforms main-text Eqs.~\eqref{eq:KD} and
\eqref{eq:energy-definitions} into Eqs.~(28)--(32) of
Ref.~\cite{SI-Benites2024}, starting from
\[
 e_{\rm kite}-e_{2x}=2\int_0^1\mathrm d\lambda\,\mathcal K[W_\lambda-\lambda v],\qquad
 \mathcal K[f]=A_0\int\mathrm d^3q\int_0^\infty\frac{\mathrm d\nu}{\pi}\,f(q,\mathrm i\nu)X(q,\nu),
\]
since $2\int_0^1\mathcal K[\lambda v]\,\mathrm d\lambda=\mathcal K[v]=e_{2x}$.
The comparison fixes the normalization and the coupling weight in the BRM representation.

\emph{Sectors.} The occupied-to-empty transitions form one of the two transition
regions of Sec.~\ref{S-app-sec:shell-origin},
\begin{equation}
 \mathcal C_{\bm q}=\{\bm p:|\bm p|\le1,\ |\bm p+\bm q|>1\},\qquad
 \Delta_{\bm p}(\bm q)=1,\quad\delta_{\bm p}(\bm q)>0.
 \label{S-sm-app-eq:couplingtransfer-u2}
\end{equation}
The reverse transitions map onto this region,
\[
 \Delta=-1:\quad\bm p=-\bm p'-\bm q,\quad\bm p'\in\mathcal C_{\bm q},\qquad
 \delta_{\bm p}(\bm q)=-\delta_{\bm p'}(\bm q).
\]
Reversing the signs of both excitation energies leaves the real part in main-text
Eq.~\eqref{eq:Xdef} unchanged. Hence, of the four sectors,
\[
 (\Delta_{\bm p},\Delta_{\bm k})=(1,1),\,(-1,-1):\ \text{equal},\qquad
 (1,-1),\,(-1,1):\ \text{equal after exchanging }\bm p\text{ and }\bm k,
\]
and their sum is
\begin{equation}
\begin{gathered}
  X(q,\nu)=2\int_{\mathcal C_{\bm q}^2}\!\mathrm{d}^3p\,\mathrm{d}^3k\left[
  \frac{\operatorname{Re}[(\mathrm{i}\nu-a)(\mathrm{i}\nu-b)]^{-1}}{|\bm{p}-\bm{k}|^{2}}
  -\frac{\operatorname{Re}[(\mathrm{i}\nu-a)(\mathrm{i}\nu+b)]^{-1}}{|\bm{p}+\bm{k}+\bm{q}|^{2}}\right],\\
  a=\delta_{\bm{p}}(\bm{q}),\qquad b=\delta_{\bm{k}}(\bm{q}).
\end{gathered}
  \label{S-sm-app-eq:sectors}
\end{equation}
Here $|\bm p+\bm k+\bm q|^2$ is the Coulomb factor
$|\bm p-\bm k|^2$ of the mixed sectors after the map is applied
to $\bm k$.

\emph{Frequency integral.} Partial fractions and
$\operatorname{Re}(\mathrm i\nu-a)^{-1}=-a/(\nu^2+a^2)$ give sums of Lorentzians,
\[
 \operatorname{Re}\frac1{(\mathrm i\nu-a)(\mathrm i\nu\mp b)}
 =\frac1{a\mp b}\Bigl[-\frac a{\nu^2+a^2}\pm\frac b{\nu^2+b^2}\Bigr],
\]
with $a=b$ as a continuous limit. For $a>0$, $\nu=a\tan u$ absorbs each
Lorentzian into the measure (first identity, for any $\phi$ for
which the integral converges). For a line $f(q,\mathrm i\nu)$ as in main-text
Eq.~\eqref{eq:KD}, the second relation defines $\mathsf H(a)$ at
fixed $q$:
\begin{equation}
 \int_0^\infty\frac{a\phi(\nu)}{\nu^2+a^2}\,\mathrm d\nu
 =\int_0^{\pi/2}\phi(a\tan u)\,\mathrm du,\qquad
 \mathsf H(a)=\frac1\pi\int_0^{\pi/2}f(q,\mathrm i a\tan u)\,\mathrm du.
 \label{S-sm-app-eq:couplingtransfer-u3}
\end{equation}
Applied to \eqref{S-sm-app-eq:sectors}, these relations give
\begin{equation}
  \int_0^\infty\frac{\mathrm{d}\nu}{\pi}\,f(q,\mathrm{i}\nu)\,X(q,\nu)
  =2\int_{\mathcal C_{\bm q}^2}\!\mathrm{d}^3p\,\mathrm{d}^3k\left[
  \frac{\mathsf H(b)-\mathsf H(a)}{(a-b)\,|\bm{p}-\bm{k}|^{2}}
  +\frac{\mathsf H(a)+\mathsf H(b)}{(a+b)\,|\bm{p}+\bm{k}+\bm{q}|^{2}}\right].
  \label{S-sm-app-eq:tan}
\end{equation}

\emph{Spin scaling.} For spin polarization $\zeta=(n_\uparrow-n_\downarrow)/n_e$,
the species $\sigma=\pm1$ have Fermi momenta $x_\sigma k_F$ with
$x_\sigma=(1+\sigma\zeta)^{1/3}$. For $|\zeta|<1$, rescaling the
momenta by $x_\sigma$ and the frequency by $x_\sigma^2$ maps each
spin contribution to a unit Fermi ball, with a factor $x_\sigma^3$
from the measures and denominators. The fully spin-polarized case follows by continuity.
In BRM notation, $\alpha=\alpha_L$ and $g(q,\nu)=2\mathcal L(q,\nu)$. Below, $q$ and
$\nu$ are the rescaled variables for spin $\sigma$, and
$\sigma'$ runs over the screening responses of both species.
Species $\sigma'$, with density of states at the Fermi surface
proportional to $x_{\sigma'}$, contributes to $\Pi_0$
\[
 -\tfrac12q_{\rm TF}^2x_{\sigma'}\mathcal L\bigl(q/x_{\sigma'},\nu/x_{\sigma'}^2\bigr)
 \ \text{(unscaled)}\quad\longrightarrow\quad
 -\tfrac12q_{\rm TF}^2x_{\sigma'}\mathcal L(\theta q,\theta^2\nu)
 \ \text{(rescaled)}.
\]
In the rescaled variables,
\begin{equation}
 \epsilon_\lambda=1+\frac{\lambda\alpha_L r_sS_\sigma}{\pi x_\sigma^2q^2},\qquad
 S_\sigma=\sum_{\sigma'}x_{\sigma'}g(\theta q,\theta^2\nu),\qquad
 \theta=x_\sigma/x_{\sigma'}.
 \label{S-sm-app-eq:couplingtransfer-u4}
\end{equation}
Here $\theta$ and $S_\sigma$ are local to this section.
With $1/\epsilon_\lambda-1=-(\epsilon_\lambda-1)/\epsilon_\lambda$, the screening correction
to the selected line, with its spin scaling factor, is
\begin{equation}
  \lambda x_\sigma^{3}\,\frac{1}{q^{2}}\Big(\frac{1}{\epsilon_\lambda}-1\Big)
  =-\frac{\lambda^2\alpha_L r_s}{\pi q^{4}}\,f^\sigma_\lambda(q,\nu),
  \qquad
  f^\sigma_\lambda=\frac{x_\sigma^{3}S_\sigma}{x_\sigma^{2}+\lambda\alpha_L r_sS_\sigma/(\pi q^{2})}.
 \label{S-sm-app-app-eq:spin-line}
\end{equation}

\emph{Energy.} For $e_{\rm kite}-e_{2x}$ the line in Eq.~\eqref{S-sm-app-eq:tan} is $W_\lambda-\lambda v$.
Each spin contribution has normalization $A_0/2$, the BRM energy uses the coupling
integration $2\,\mathrm d\lambda$, and the two coupling factors of $W_\lambda-\lambda v$
give the factor $\lambda^2$:
\begin{equation}
\begin{split}
  e_{\rm kite}-e_{2x}={}&\frac{3\alpha_L r_s}{16\pi^{7}}\sum_\sigma\int_0^1\!\mathrm{d}\lambda\,\lambda^{2}
  \int\!\frac{\mathrm{d}^3q}{q^{4}}\int_{\mathcal C_{\bm q}^2}\!\mathrm{d}^3p\,\mathrm{d}^3k\int_0^{\pi/2}\!\mathrm{d} u\\
  &\times\Bigg[\frac{f^\sigma_\lambda(q,a\tan u)-f^\sigma_\lambda(q,b\tan u)}
              {(a-b)\,|\bm{p}-\bm{k}|^{2}}\\
  &\qquad-\frac{f^\sigma_\lambda(q,a\tan u)+f^\sigma_\lambda(q,b\tan u)}
        {(a+b)\,|\bm{p}+\bm{k}+\bm{q}|^{2}}\Bigg].
\end{split}
 \label{S-sm-app-app-eq:brm-final}
\end{equation}
The coupling, spin, sector, $\mathsf H$ and line factors give the prefactor,
\[
 2\cdot\frac{A_0}2\cdot2\cdot\frac1\pi\cdot\frac{\alpha_Lr_s}\pi
 =\frac{2A_0\alpha_L r_s}{\pi^2}=\frac{3\alpha_L r_s}{16\pi^7},\qquad A_0=\frac3{32\pi^5},
\]
and the minus sign of the line reverses both terms in the bracket.
The relabeling
$(\bm q,\bm p,\bm k)\mapsto(\bm q_1,\bm q_2,\bm k)$ turns our
$a$ and $b$ into the excitation energies $a(\bm q_1,\bm q_2)$ and
$a(\bm q_1,\bm k)$ of BRM and completes the correspondence.
\section{Diagram coefficients, coupling weights and energy expressions}\label{S-sm-sec:diagram-weights}\label{S-app-sec:coupling-transfer}
This section derives the closed-diagram coefficients of the chain family
of main-text Sec.~\ref{sec:diagram-coefficients}, the
weights assigned by the coupling average and the relations between the energy expressions.

\subsection{Closed-diagram coefficients of the chain family}\label{S-sm-sec:kite-chain}

The closed-diagram coefficient of a completed diagram is its
Wick-contraction coefficient in the connected closed-diagram expansion
with fixed $G_0$ \cite[Chap.~3]{SI-FetterWalecka1971}. The chain family
contains the bare kite and every diagram with a
nonempty chain of free bubbles on one interaction, the other interaction staying
bare (figure~\ref{S-fig:coefficients}). Each unlabeled diagram appears
once.

\begin{figure}[!ht]
\centering

\includegraphics[width=0.9\textwidth]{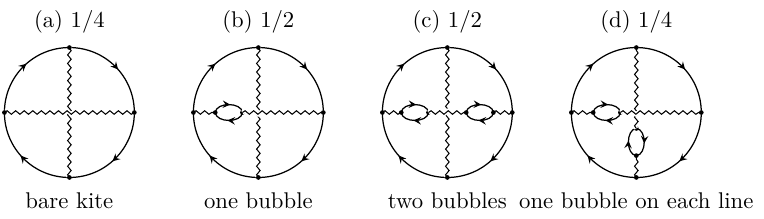}

\caption{Closed-diagram coefficients in the connected closed-diagram expansion.
(a) Bare exchange kite, coefficient $1/4$.
(b,c) One and two polarization bubbles on one interaction line,
coefficient $1/2$ each.
(d) One bubble on each line, coefficient $1/4$, outside the chain family $e_{\rm closed}$.
Solid arrows are $G_0$ and wavy lines are bare interactions
$v$. A crossing without a dot is not a vertex.
The coefficients multiply complete diagrams, not individual bubbles,
and refer to $e_{\rm closed}$, not to the BRM coupling average.}
\label{S-fig:coefficients}
\end{figure}

A vertex is an
interaction end with one incoming and one outgoing fermion line.
A symmetry is a bijection of the vertices preserving every directed
fermion connection, every unordered interaction pair and any fixed
physical label. It determines the map of every line.

\begin{proposition}[Chain family]
\label{S-prop:kite-chain}
In the connected closed-diagram expansion with fixed $G_0$, the bare
kite has coefficient $1/4$, whereas every diagram with a chain of
$m\geq1$ polarization bubbles on one interaction has coefficient $1/2$.
With $\mathcal K[f]$ retaining the bare coefficient $1/4$, the energy
per electron $e_{\rm closed}$ of this family is
\begin{equation}
 e_{\rm closed}=\mathcal K[v]+2\mathcal K[W_{\lambda=1}-v]
   =2\mathcal K[W_{\lambda=1}]-\mathcal K[v].
 \label{S-eq:unmarked-kite-family}
\end{equation}
\end{proposition}

\begin{proof}
\emph{Bare kite.} Number the core vertices along the fermion cycle $1\to2\to3\to4\to1$,
with interactions $\{1,3\}$ and $\{2,4\}$. The image of vertex $1$ fixes the other three,
and all four choices preserve the interaction pairs,
\[
 \mathsf r_k:\ i\mapsto i+k\ (\mathrm{mod}\ 4),\ k=0,\dots,3,\qquad
 \mathsf r_0,\mathsf r_2:\ \{1,3\}\mapsto\{1,3\},\qquad
 \mathsf r_1,\mathsf r_3:\ \{1,3\}\mapsto\{2,4\}.
\]
The bare kite therefore has four symmetries and coefficient $1/4$.
In general, the $2^nn!$ relabelings of interaction
labels and ends produce each distinct contraction as many times
as the diagram has symmetries \cite[Sec.~3.2]{SI-LinLindsey2021}. With $N_{\rm sym}$ symmetries,
\[
 \text{coefficient}=\frac{\#\,\text{contractions}}{2^nn!}=\frac{2^nn!/N_{\rm sym}}{2^nn!}=\frac1{N_{\rm sym}}.
\]

\emph{Chain.} Insert a nonempty chain on $\{1,3\}$. The four-vertex cycle remains
uniquely identifiable, since every added cycle has two vertices. A
symmetry of the completed graph therefore induces a core symmetry preserving
$\{1,3\}$, and only $\mathsf r_0$ (the identity) and $\mathsf r_2$ ($1\leftrightarrow3$,
$2\leftrightarrow4$) qualify. Number the bubbles $j=1,\dots,m$ along the
chain from vertex 1, and let $a_j$ ($b_j$) be the vertex of bubble
$j$ whose interaction leads toward vertex 1 (vertex 3). Then
\[
 \begin{gathered}
 \text{chain: }\{1,a_1\},\quad\{b_j,a_{j+1}\}\ (1\le j<m),\quad\{b_m,3\},\qquad
 \text{bare: }\{2,4\},\\
 \mathsf r_2\text{ on the chain: }a_j\mapsto b_{m+1-j},\quad b_j\mapsto a_{m+1-j}.
 \end{gathered}
\]
Each core symmetry extends uniquely to the chain. The
second reverses the bubble order and exchanges the two
vertices of each bubble.
The completed graph therefore has two symmetries and coefficient
$1/2$ for every $m\ge1$. Expanding $W_{\lambda=1}-v$ into its
nonempty chains and retaining the bare kite once gives
\eqref{S-eq:unmarked-kite-family}.
\end{proof}

The subtracted bare kite produces the
constant that survives in $e_{\rm closed}$ at low density.

Table~\ref{S-app-tab:chain-counts} checks the coefficients by enumerating
the Wick contractions as permutations of the
$2n$ interaction ends with one four-cycle for the kite loop and
$n-2$ two-cycles for the bubbles. For one bubble, choose the bare interaction, the ends of
the other two interactions that meet the loop, and the loop direction,
\[
 3\times2^2\times2=24,\qquad \frac{24}{2^3\,3!}=\frac12 .
\]

These coefficient identities hold in the formal expansion with fixed
$G_0$ at the free chemical potential $\mu_0(n_e)$. The individual
chain terms diverge at small momentum transfer, so the energies are evaluated
from the RPA-resummed integrals. The fixed-density energy also
contains chemical-potential counterterm diagrams, which lie outside
$e_{\rm closed}$.

\subsection{Coupling-constant integration and coupling weights}\label{S-sm-sec:marking}
This subsection derives the fraction of the closed-diagram coefficient that the coupling
average assigns to a chain of $m$ bubbles ($j$ in the main text), stated after main-text
Eq.~\eqref{eq:closed-family},
\[
 \frac{\text{BRM coefficient}}{\text{closed-diagram coefficient}}=\frac1{m+2},\qquad m\ge1.
\]

\emph{Marking.} The BRM energy expression is a coupling-integrated self-energy
trace \cite{SI-Benites2024} (Sec.~\ref{S-app-sec:setting}). For diagrams containing a chain,
its coefficients differ from those of $e_{\rm closed}$
because the two expressions mark different interactions in the Hellmann--Feynman relation
\begin{equation}
 E(1)-E(0)=\int_0^1\langle\widehat V\rangle_\lambda\,\mathrm d\lambda.
 \label{S-eq:hellmann-feynman-reminder}
\end{equation}
Here $E(\lambda)$ is the ground-state energy of
$\widehat H_\lambda=\widehat T+\lambda\widehat V$
(Sec.~\ref{S-app-sec:units-kite}), and
$\mathrm dE/\mathrm d\lambda=\langle\widehat V\rangle_\lambda$
holds for the normalized ground state (Hellmann--Feynman theorem \cite{SI-Feynman1939}).
A closed diagram with $n$ interactions is proportional to $\lambda^n$.
Differentiation marks one interaction, and summing over the $n$
markings and integrating over the coupling restores the coefficient,
\[
 \frac{\mathrm d}{\mathrm d\lambda}\lambda^n=n\,\lambda^{n-1},\qquad
 n\int_0^1\lambda^{n-1}\,\mathrm d\lambda=1 .
\]
An energy reconstructed from a response or self-energy expression counts only
the interactions that this expression marks.

\emph{Chain series.} The screened line expands as
\begin{equation}
 W_\lambda=\sum_{m\geq0}\lambda^{m+1}d_m,
 \qquad d_m=v(\Pi_0v)^m,
 \qquad K_m=\mathcal K[d_m],
 \label{S-eq:chain-series}
\end{equation}
where $m$ counts the bubbles, and the kite has $m+2$ interactions.
The two energy expressions
weight these contributions as
\begin{equation}
 e_{\rm kite}=2\int_0^1\mathcal K[W_\lambda]\,\mathrm d\lambda
   =\sum_{m\geq0}\frac{2K_m}{m+2},
 \qquad
 e_{\rm closed}=K_0+2\sum_{m\geq1}K_m.
 \label{S-eq:brm-family-weights}
\end{equation}
Both count $K_0=e_{2x}$ once. For $m\ge1$ the BRM coefficient is
$1/(m+2)$ times the closed-diagram coefficient $2$,
\[
 \frac{2\int_0^1\lambda^{m+1}\,\mathrm d\lambda}{2}=\frac{2/(m+2)}{2}=\frac1{m+2} .
\]

\emph{Restored coupling.} To trace this ratio, the coupling is restored on every
interaction in the family energy,
\begin{equation}
 \mathcal E_{\rm closed}(\lambda)
 =2\lambda\mathcal K[W_\lambda]-\lambda^2\mathcal K[v],
 \qquad \mathcal E_{\rm closed}(1)=e_{\rm closed}.
 \label{S-eq:family-coupling}
\end{equation}
With $G_0$ and $\Pi_0$ fixed,
$\partial_\lambda W_\lambda=v/(1-\lambda v\Pi_0)^2$, so
\begin{equation}
 \mathcal E_{\rm closed}'(\lambda)
 =2\mathcal K\!\left[
 W_\lambda+\frac{\lambda v}{(1-\lambda v\Pi_0)^2}-\lambda v\right].
 \label{S-eq:family-coupling-derivative}
\end{equation}
The explicit $\lambda$ marks the bare kite line and $\partial_\lambda
W_\lambda$ marks each interaction inside the screened line, so all
$m+2$ positions are counted. Per $m$-bubble term ($m\ge1$), with arrows for the sum over $m\ge1$,
\[
 \begin{aligned}
 \text{all }m+2\text{ marks}&:\ 2(m+2)\lambda^{m+1}K_m,\\
 \text{two kite marks}&:\ 4\lambda^{m+1}K_m\ \to\ 4\mathcal K[W_\lambda-\lambda v],\\
 \text{BRM integrand}&:\ 2\lambda^{m+1}K_m\ \to\ 2\mathcal K[W_\lambda-\lambda v],
 \end{aligned}
 \qquad
 \frac{2}{2(m+2)}=\frac12\cdot\frac2{m+2}=\frac1{m+2}.
\]
The BRM integrand is half of the two-mark term, because $\mathcal K$ retains the bare
coefficient $1/4$ in place of the closed-diagram coefficient $1/2$. It also marks none of the
added interactions. The two factors $1/2$ and $2/(m+2)$ give the ratio $1/(m+2)$ in
\eqref{S-eq:brm-family-weights}. Since
$F=\mathcal K[W_{\lambda=1}]$, $e_{\rm kite}$ and $e_{\rm closed}$ are
all expressed through the same
functional $\mathcal K$, they obey main-text Eq.~\eqref{eq:density-transfer}.

\subsection{Both kite lines screened}\label{S-sm-sec:two-line}
This subsection supports the remark after main-text
Eq.~\eqref{eq:A1-dispersion} on the constant $-e_{2x}$ of
$e_{\rm closed}$ at low density. Let
$\mathcal K_2[f,g]$ be the kite with its two interactions replaced by
$f$ and $g$ and the bare coefficient $1/4$ retained. Then
$\mathcal K_2$ is bilinear and symmetric, and
$\mathcal K[f]=\mathcal K_2[f,v]$. A diagram with $m_1$ and $m_2$ bubbles on
the two lines appears in $\mathcal K_2[W_{\lambda=1},W_{\lambda=1}]$
once for $m_1=m_2$ and twice for $m_1\neq m_2$. The proof of Proposition~\ref{S-prop:kite-chain} gives the two symmetries for $m_1\neq m_2$. For $m_1=m_2$, the core rotation
$1\to2\to3\to4\to1$, which exchanges the two lines, also extends to
the chains and gives four symmetries.
This counting reproduces the closed-diagram coefficients without subtraction:
\[
 \begin{array}{lccc}
  & \text{occurrences} & \text{symmetries} & \text{occurrences}\times\tfrac14=1/\text{symmetries}\\
 m_1\neq m_2: & 2 & 2 & 2\times\tfrac14=\tfrac12\\
 m_1=m_2: & 1 & 4 & 1\times\tfrac14=\tfrac14
 \end{array}
\]
Instances of both cases are checked by enumeration in Sec.~\ref{S-sm-sec:verification}
(table~\ref{S-app-tab:chain-counts} for $m_2=0$,
figure~\ref{S-fig:coefficients}(d) for $m_1=m_2=1$). By bilinearity and symmetry of $\mathcal K_2$,
\[
 \mathcal K_2[W_{\lambda=1},W_{\lambda=1}]
 =\mathcal K_2[v,v]+2\mathcal K_2[W_{\lambda=1}-v,\,v]+\mathcal K_2[W_{\lambda=1}-v,\,W_{\lambda=1}-v],
\]
and with Eq.~\eqref{S-eq:unmarked-kite-family},
\begin{equation}
 \mathcal K_2[W_{\lambda=1},W_{\lambda=1}]
 =e_{\rm closed}+\mathcal K_2[W_{\lambda=1}-v,\,W_{\lambda=1}-v].
 \label{S-eq:two-line-family}
\end{equation}
The chain family is the part of degree at most one in
$W_{\lambda=1}-v$. Its constant $-e_{2x}$ is
compensated by the two-chain terms if
$\mathcal K_2[W_{\lambda=1},W_{\lambda=1}]$ vanishes at low density
together with $F$. This vanishing is expected
but is not established here. The spectral decomposition of the bilinear
$\mathcal K_2$ is a double
integral over the phase velocities of two Lorentzian lines, so the one-dimensional reductions of
main-text Sec.~\ref{sec:spectral} do not apply directly.

\begin{table}[ht]
\centering
\caption{Wick counts of the chain family at interaction order $n$.
Dividing by $2^n n!$ gives the closed-diagram coefficient.}
\label{S-app-tab:chain-counts}
\begin{tabular}{cccc}
\toprule
$n$ & Bubbles & Contractions & Coefficient \\
\midrule
2 & 0 & 2 & $1/4$ \\
3 & 1 & 24 & $1/2$ \\
4 & 2 & 192 & $1/2$ \\
5 & 3 & 1920 & $1/2$ \\
6 & 4 & 23040 & $1/2$ \\
\bottomrule
\end{tabular}
\end{table}

\subsection{Coupling averages and density scaling}\label{S-sm-sec:coupling-averages}
With $W_\lambda(r_s)=\lambda W_{\lambda=1}(\lambda r_s)$ [main-text
Sec.~\ref{sec:linear-diagrams}, after Eq.~\eqref{eq:general-average}]
and the formal expansion \eqref{S-eq:chain-series}, the three energy expressions of main-text
Sec.~\ref{sec:diagram-coefficients} weight the $m$-bubble term as
in table~\ref{S-app-tab:energy-prescriptions} ($K_0=e_{2x}$ is counted once in each).
\begin{table}[ht]
\centering
\caption{Factors of $K_m$ in the three energy expressions.}
\label{S-app-tab:energy-prescriptions}
\begin{tabular}{clc}
\toprule
Quantity & Definition & Factor of $K_m$, $m\ge1$ \\
\midrule
$F$ & $\mathcal K[W_{\lambda=1}]$ & $1$ \\
$e_{\rm kite}$ & $2\int_0^1\mathcal K[W_\lambda]\,\mathrm d\lambda$ & $2/(m+2)$ \\
$e_{\rm closed}$ & $2\mathcal K[W_{\lambda=1}]-\mathcal K[v]$ & $2$ \\
\bottomrule
\end{tabular}
\end{table}

\emph{Direction of transfer.} The asymptotic forms of $F$ are derived directly at full coupling. At
high density they follow from the Thomas--Fermi tail $-h/t^2$ of the full-coupling
line (main-text table~\ref{tab:screening}) and the window estimate
\eqref{S-app-eq:window}, in which $|c/(1+c)|\le1$ replaces
$|\rho-1|\le1$. At low density they follow from
the bounds of Sec.~\ref{S-sm-sec:kite-remainders}, which hold at each
fixed $\lambda$ with explicit factors $\lambda^{-1/2}$, $\lambda^{-3/4}$
and $\lambda^{-1}$ and therefore at $\lambda=1$. The moment rule of
Sec.~\ref{S-sm-slim-sec:moment} carries asymptotic terms only
from the full-coupling energy to the coupling average, since
an asymptotic remainder is not in general preserved by differentiation.
A term of $F$ with constant coefficient $\varphi$ appears in $e_{\rm kite}$ as
\[
 \varphi r_s^p\ \longmapsto\ 2\int_0^1\lambda\,\varphi\,(\lambda r_s)^p\,\mathrm d\lambda=\frac{2\varphi r_s^p}{p+2}\qquad(p>-2),
\]
and a term $\varphi r_s^p\log^2r_s$ as stated
at the end of Sec.~\ref{S-sm-slim-sec:moment}.
The full-coupling tail is $\tfrac32$ times the averaged tail that determines $c_{\rm hd}$ in
\eqref{S-app-eq:highdensity-u2}, and main-text Eq.~\eqref{eq:closed-family} gives $e_{\rm closed}$,
\[
 -\frac h{t^2}=\frac32\Bigl(-\frac{2h}{3t^2}\Bigr),\qquad e_{\rm closed}=2F-e_{2x}.
\]
These relations give
the $F$ and $e_{\rm closed}$ entries of
main-text Eq.~\eqref{eq:Clog2} and table~\ref{app-tab:static-dynamic}.

\emph{Low density.} At leading order $F(\lambda r_s)\simeq\lambda^{-1/4}F(r_s)$, and the first
relation of main-text Eq.~\eqref{eq:density-transfer} gives
\[
 e_{\rm kite}=2\int_0^1\lambda^{3/4}\,\mathrm d\lambda\,F=\tfrac87F,\qquad F=\tfrac78e_{\rm kite}.
\]
The value
$q_*^{-1}=[3\pi/(16\alpha_L)]^{1/4}r_s^{-1/4}=1.0311348\ldots\,r_s^{-1/4}$
expresses $A_1/q_*$ in powers of $r_s$. At low density,
\begin{equation}
 \begin{aligned}
 e_{\rm kite}(r_s)&=-0.2213782293\ldots\,r_s^{-1/4}+o(r_s^{-1/4}),\\
 F(r_s)&=-0.1937059506\ldots\,r_s^{-1/4}+o(r_s^{-1/4}),\\
 e_{\rm closed}(r_s)&=-e_{2x}-0.38741190129\ldots\,r_s^{-1/4}+o(r_s^{-1/4}).
 \end{aligned}
 \label{S-app-eq:low-density-family-transfer}
\end{equation}

\section{Spectral representation and convergence}\label{S-app-sec:spectral}
This section derives the spectral representation of main-text
Eq.~\eqref{eq:block-superposition} and the absolute bounds that justify interchanging the spectral and energy
integrations in main-text Eq.~\eqref{eq:common-blocks}.
\subsection{Loss spectrum and Lorentzian lines}\label{S-app-sec:loss-spectrum}
With $\mathrm i\nu\to\omega+\mathrm i0$ for real frequency $\omega$, the imaginary-frequency
Lorentzian line of main-text Eq.~\eqref{eq:block-superposition} continues analytically to the
retarded line
\begin{equation}
 b_\Lambda(q,\omega)=\frac{\Lambda^2}{(q\Lambda)^2-(\omega+\mathrm i0)^2},
 \qquad \omega_{\rm poles}=\pm q\Lambda-\mathrm i0.
 \label{S-app-eq:spectralconvergence-u1}
\end{equation}
The loss
spectrum, defined for $\omega>0$, gives the screened interaction as
\begin{equation}
 \begin{gathered}
  B_{W,\lambda}(q,\omega)=-\frac1\pi\operatorname{Im}W_\lambda^R(q,\omega)\ge0,\\
  W_\lambda(q,\mathrm{i}\nu)=\lambda v(q)
  -\int_0^\infty\frac{2\omega B_{W,\lambda}(q,\omega)}{\nu^{2}+\omega^{2}}\,\mathrm{d}\omega.
 \end{gathered}
  \label{S-app-eq:herglotz}
\end{equation}
Here $W_\lambda^R$ and $\Pi_0^R$ are the retarded functions, obtained by
$\mathrm i\nu\to\omega+\mathrm i0$ as in \eqref{S-app-eq:spectralconvergence-u1}.
The second equation is the Kramers--Kronig relation at fixed $q>0$
\cite{SI-Nussenzveig1972}. It applies for two reasons. First, for a stable ground state
$W^R_\lambda-\lambda v$ is analytic in the upper half of the complex $\omega$ plane and
vanishes at large $|\omega|$ there, since $\Pi_0^R(q,\omega)$ decays as $|\omega|^{-2}$ in the
whole upper half plane (closed form of the Lindhard function). Second, $W_\lambda$ is real in
the time domain, hence
\[
 \operatorname{Im}W^R_\lambda(q,-\omega)=-\operatorname{Im}W^R_\lambda(q,\omega).
\]
The inequality $B_{W,\lambda}\ge0$ is the passivity condition of main-text
Sec.~\ref{sec:spectral-rep}.

\emph{Phase velocity.} With $\omega=\Lambda q$, $\mathrm d\omega=q\,\mathrm d\Lambda$ and
$b_\Lambda(q,\mathrm i\nu)=\Lambda^2/(q^2\Lambda^2+\nu^2)$,
\[
 \frac{2\omega B_{W,\lambda}(q,\omega)\,\mathrm d\omega}{\nu^2+\omega^2}
 =\frac{2q^2\Lambda\,B_{W,\lambda}(q,q\Lambda)\,\mathrm d\Lambda}{\nu^2+q^2\Lambda^2}
 =\frac{2q^2}{\Lambda}B_{W,\lambda}(q,q\Lambda)\,b_\Lambda(q,\mathrm i\nu)\,\mathrm d\Lambda,
\]
and therefore
\begin{equation}
 \begin{gathered}
  W_\lambda(q,\mathrm{i}\nu)=\lambda v(q)-\int_0^\infty\mathrm{d}\Lambda\;a_\lambda(q,\Lambda)\,b_\Lambda(q,\mathrm{i}\nu),\\
  a_\lambda(q,\Lambda)=\frac{2q^2}{\Lambda}B_{W,\lambda}(q,q\Lambda)\ge0 .
 \end{gathered}
  \label{S-app-eq:block-superposition}
\end{equation}

\emph{Functionals.} Linearity in the selected line gives, for
$T=\Pi_{\rm vtx}^{(1)},\Pi_{\rm SE}^{(1)},\Pi^{(1)}$,
\begin{equation}
 \mathcal K_T[W_\lambda]=\lambda\mathcal K_T[v]
 -\int_0^\infty\mathrm d\Lambda\,
 \mathcal K_T[a_\lambda(q,\Lambda)\,b_\Lambda(q,\mathrm i\nu)].
 \label{S-app-eq:common-blocks}
\end{equation}
The functional acts on the product $a_\lambda(q,\Lambda)b_\Lambda(q,\mathrm i\nu)$
as a line at fixed $\Lambda$. Here $\mathcal K_T[f]$ denotes main-text Eq.~\eqref{eq:KT-sigma} with
$\Pi^{(1)}_{\rm SE}$ replaced by $T$, hence
\[
 \mathcal K_{\Pi^{(1)}_{\rm vtx}}=\mathcal K\ \ \text{(Sec.~\ref{S-app-sec:response})},\qquad
 \mathcal K_{\Pi^{(1)}}=\mathcal K_{\Pi^{(1)}_{\rm vtx}}+\mathcal K_{\Pi^{(1)}_{\rm SE}}.
\]
The spectral measure includes any plasmon point masses, and the momentum-dependent weight remains
inside $\mathcal K_T$ (main-text Sec.~\ref{sec:spectral-rep}).

\emph{Interchange.} The interchange of the spectral and energy integrations follows by the
argument of Sec.~\ref{S-app-app:block-generality} from the bounds of
Secs.~\ref{S-app-app:shell-integrated} (small momentum transfer) and
\ref{S-app-app:selfenergy-blocks} and from a large-momentum estimate. This estimate uses
\[
 |\delta_{\bm p}|,|\delta_{\bm k}|\ge q^2/2-q\quad(q>2),\qquad
 \int_0^\infty\frac{\mathrm d\nu}{\nu^2+\delta^2}=\frac\pi{2|\delta|},
\]
the first from the end of Sec.~\ref{S-app-sec:static-bounds}, with the Coulomb
integrals of the four regions from Sec.~\ref{S-app-sec:largeq-kernel}. For every $q>0$,
\[
 \bigl|(\mathrm i\nu-\delta_{\bm p})(\mathrm i\nu-\delta_{\bm k})\bigr|^{-1}
 \le\tfrac12\bigl[(\nu^2+\delta_{\bm p}^2)^{-1}+(\nu^2+\delta_{\bm k}^2)^{-1}\bigr],\qquad
 \int_{\mathbb B}\frac{\mathrm d^3k}{|\bm p-\bm k|^2}=4\pi f_L(p)\le4\pi.
\]
The latter is applied to $\mathbb B$ and to its translate $\mathbb B-\bm q$, which together contain
the two transition regions of $\bm k$. By the free response of Sec.~\ref{S-app-sec:response} at
$\nu=0$, $\int\mathrm d^3p\,|\Delta_{\bm p}|/|\delta_{\bm p}|=4\pi\mathcal L(q,0)=4\pi f_L(q/2)$.
The two estimates, the symmetry $\bm p\leftrightarrow\bm k$ and $4\pi+4\pi=8\pi$ for the
$\bm k$ integral give
\[
\begin{aligned}
 \int_0^\infty|X|\,\mathrm d\nu
 &\le\frac\pi2\int\mathrm d^3p\,\frac{|\Delta_{\bm p}|}{|\delta_{\bm p}|}
 \int_{\mathbb B\cup(\mathbb B-\bm q)}\frac{\mathrm d^3k}{|\bm p-\bm k|^2}
 \le4\pi^2\int\mathrm d^3p\,\frac{|\Delta_{\bm p}|}{|\delta_{\bm p}|}\\
 &=4\pi^2\cdot4\pi\mathcal L(q,0)=16\pi^3f_L(q/2)\le16\pi^3.
\end{aligned}
\]
At large $q$, $f_L(q/2)\simeq4/(3q^2)$ gives the entry $O(q^{-2})$ below. At small
$q$, this bound is only $O(1)$. The entry $O(q\log(e/q))$ follows instead
from the absolute frequency bound $\int|X(q,qz)|\,\mathrm dz=O(\log(e/q))$
of Sec.~\ref{S-app-app:shell-integrated}, with $\mathrm d\nu=q\,\mathrm dz$.
The lower bound on $\int_0^\infty a_\lambda b_\Lambda\,\mathrm d\Lambda$ holds because
$a_\lambda\ge0$ and $b_\Lambda>0$, and the upper bound because
$W_\lambda=\lambda/(q^2+\lambda q_{\rm TF}^2\mathcal L)\ge0$. Below,
$\log(e/q)=1+\log(1/q)$, and the last relation includes the
angular integration of $\mathrm d^3q$.
For every $q>0$,
\begin{equation}
 0\le\int_0^\infty a_\lambda b_\Lambda\,\mathrm d\Lambda\le\lambda v,\qquad
 \int_0^\infty|X(q,\nu)|\,\mathrm d\nu=
 \begin{cases}O(q\log(e/q)),&q\to0,\\O(q^{-2}),&q\to\infty,\end{cases}
 \qquad \mathrm d^3q\,v=4\pi\,\mathrm dq.
 \label{S-app-eq:spectralconvergence-u2}
\end{equation}
\subsection{Elementary frequency integrals}\label{S-app-app:Jkernel}
This subsection derives the two-factor frequency integral cited
in main-text Sec.~\ref{sec:radial-general} and its static limit,
main-text Eq.~\eqref{eq:Wdef}. Let $a$ and $b$ be the signed excitation
energies and $\omega>0$ the Lorentzian width in $\nu$ ($\omega=q\Lambda$
for the line $b_\Lambda$). For $a,b\ne0$,
\begin{equation}
  J(\omega;a,b)=\frac1\pi\int_0^\infty\frac{ab-\nu^2}
  {(a^2+\nu^2)(b^2+\nu^2)(\nu^2+\omega^2)}\,\mathrm{d}\nu .
 \label{S-app-app-eq:J-definition}
\end{equation}

\emph{Partial fractions.} With $\alpha=|a|$, $\beta=|b|$, distinct positive widths, and $c$
running over the three widths,
\begin{align}
  \frac{1}{(\alpha^2+\nu^2)(\beta^2+\nu^2)(\omega^2+\nu^2)}
  &=\sum_{c\in\{\alpha,\beta,\omega\}}\frac{r_c}{c^2+\nu^2},\label{S-app-app-eq:J-partial}\\
  r_c&=\prod_{c'\ne c}\frac{1}{c'^2-c^2},
 \label{S-app-app-eq:J-residues}
\end{align}
and we denote the denominator on the left by
\begin{equation}
  \mathsf N(\nu):=(\alpha^2+\nu^2)(\beta^2+\nu^2)(\omega^2+\nu^2).
  \label{S-app-eq:Ddef}
\end{equation}

\emph{Moments.} The coefficient of $\nu^{-2}$ in \eqref{S-app-app-eq:J-partial} at large $\nu$
gives $\sum_cr_c=0$, hence
\[
 \frac{\nu^2}{\mathsf N(\nu)}=\sum_cr_c\Bigl(1-\frac{c^2}{c^2+\nu^2}\Bigr)=-\sum_c\frac{c^2r_c}{c^2+\nu^2},\qquad
 \int_0^\infty\frac{\mathrm d\nu}{c^2+\nu^2}=\frac\pi{2c}.
\]
Term-by-term integration gives
\begin{align}
  \frac1\pi\!\int_0^\infty\!\frac{\mathrm{d}\nu}{\mathsf N(\nu)}
  &=\frac{\alpha+\beta+\omega}{2\alpha\beta\omega(\alpha+\beta)(\beta+\omega)(\omega+\alpha)},\label{S-app-app-eq:J-moment0}\\
  \frac1\pi\!\int_0^\infty\!\frac{\nu^2\,\mathrm{d}\nu}{\mathsf N(\nu)}
  &=\frac{1}{2(\alpha+\beta)(\beta+\omega)(\omega+\alpha)} .
 \label{S-app-app-eq:J-moment2}
\end{align}
Equation~\eqref{S-app-app-eq:J-definition} equals $ab$ times
\eqref{S-app-app-eq:J-moment0} minus \eqref{S-app-app-eq:J-moment2}. With $ab=\pm\alpha\beta$,
\[
 J(\omega;a,b)=\frac{\pm(\alpha+\beta+\omega)-\omega}{2\omega(\alpha+\beta)(\beta+\omega)(\omega+\alpha)},
\]
which is
\begin{equation}
  J(\omega;a,b)=
  \begin{cases}
    \dfrac{1}{2\omega(\alpha+\omega)(\beta+\omega)}, & ab>0,\\[10pt]
    -\dfrac{\alpha+\beta+2\omega}{2\omega(\alpha+\beta)(\alpha+\omega)(\beta+\omega)}, & ab<0.
  \end{cases}
 \label{S-app-app-eq:J-closed}
\end{equation}
The closed form is continuous when positive widths coincide,
although the partial-fraction coefficients diverge there.

\emph{Static limit.} Since
\[
 \operatorname{Re}\frac1{(\mathrm i\nu-a)(\mathrm i\nu-b)}=\frac{ab-\nu^2}{(a^2+\nu^2)(b^2+\nu^2)},
\]
the static integral of main-text Eq.~\eqref{eq:Wdef} is the limit $\omega\to\infty$ of
$\omega^2J(\omega;a,b)$. By \eqref{S-app-app-eq:J-closed},
\[
 \lim_{\omega\to\infty}\omega^2J(\omega;a,b)=0\quad(ab>0),\qquad
 \lim_{\omega\to\infty}\omega^2J(\omega;a,b)=-\frac1{\alpha+\beta}\quad(ab<0).
\]
\subsection{Absolute bounds for the self-energy pair}\label{S-app-app:selfenergy-blocks}
These bounds justify the spectral interchange for the self-energy pair
and give the remainders of $e_{\rm SE}$ used in Secs.~\ref{S-app-sec:cancel}
and \ref{S-app-sec:sigma-lowdensity} for main-text
Sec.~\ref{sec:se-cancellation}. For $\delta\ne0$,
\begin{equation}
 \int_0^\infty\!\left|\operatorname{Re}(\mathrm i\nu-\delta)^{-2}\right|\,
 \mathrm d\nu=\frac1{|\delta|}.
 \label{S-app-eq:spectralconvergence-u3}
\end{equation}
This follows from
\[
 \operatorname{Re}(\mathrm i\nu-\delta)^{-2}=\frac{\delta^2-\nu^2}{(\nu^2+\delta^2)^2}
 =\partial_\nu\frac{\nu}{\nu^2+\delta^2},\qquad
 \frac{\nu}{\nu^2+\delta^2}\Big|_{\nu=|\delta|}=\frac1{2|\delta|},
\]
since the integrand changes sign only at $\nu=|\delta|$.
Inserting \eqref{S-app-eq:PiSigma}, with $2C=2\pi^2q_{\rm TF}^2$, and
\eqref{S-app-eq:spectralconvergence-u3} gives
\begin{equation}
 \int_0^\infty |\Pi_{\rm SE}^{(1)}(q,\mathrm i\nu)|\,\mathrm d\nu
 \le 2\pi^2q_{\rm TF}^2\int\!\frac{\mathrm d^3p}{(2\pi)^3}\,
 |\Delta_{\bm p}|\,
 \frac{|\Sigma_x(\bm p+\bm q)-\Sigma_x(\bm p)|}{|\delta_{\bm p}|}.
 \label{S-app-eq:spectralconvergence-u4}
\end{equation}

\emph{Small momentum transfer.} For $0<q<1/2$, we parameterize the
$\Delta_{\bm p}=1$ transition region by the distances of the two momenta from the Fermi surface,
\[
 r=|\bm p|=1-\xi,\qquad r'=|\bm p+\bm q|=1+\zeta,\qquad \xi,\zeta\ge0,\qquad s=\xi+\zeta\le q,
\]
with $s$ a distance, not the Mellin variable. Then
\begin{equation}
 \mathrm d^3p=\frac{2\pi}{q}rr'\,\mathrm d\xi\,\mathrm d\zeta,
 \qquad
 |\delta_{\bm p}|=\frac{s(r+r')}{2}.
 \label{S-app-eq:spectralconvergence-u5}
\end{equation}
In the bounds below, $C_{\rm b}$ is independent of $q$ and $r_s$
(Notation). Near $r=1$, differentiating $f_L$ in Eq.~\eqref{S-app-eq:definitions-u6} gives
$|\Sigma_x'(r)|\le C_{\rm b}q_{\rm TF}^2[1+|\log|r-1||]$. Integrating from $r$ to $r'$,
\[
 |\Sigma_x(r')-\Sigma_x(r)|\le C_{\rm b}q_{\rm TF}^2\bigl[2s+s\log(1/s)+\xi\log(s/\xi)+\zeta\log(s/\zeta)\bigr],\qquad
 r+r'\ge2-q.
\]
Combining this with $\xi\log(s/\xi)+\zeta\log(s/\zeta)\le s\log2$ gives
\begin{equation}
 \frac{|\Sigma_x(r')-\Sigma_x(r)|}{|\delta_{\bm p}|}
 \le C_{\rm b}q_{\rm TF}^2[1+\log(1/s)].
 \label{S-app-eq:spectralconvergence-u6}
\end{equation}
Including both transition regions and the response normalization gives
\[
 \int_0^\infty|\Pi_{\rm SE}^{(1)}(q,\mathrm i\nu)|\,\mathrm d\nu
 \le C_{\rm b}q_{\rm TF}^4q^{-1}\int_0^q s[1+\log(1/s)]\,\mathrm ds
 =C_{\rm b}q_{\rm TF}^4\,\frac q4\bigl[3+2\log(1/q)\bigr].
\]
For $q>3$, uniformly in $q$, every $\bm p$ in the two transition
regions satisfies
\[
 |\delta_{\bm p}|\ge q^2/2-q\ge3/2,\qquad |\Sigma_x|\le C_{\rm b}q_{\rm TF}^2,
\]
and the regions are already disjoint for $q>2$.
The small- and large-momentum estimates give
\begin{equation}
 \int_0^\infty |\Pi_{\rm SE}^{(1)}(q,\mathrm i\nu)|\,\mathrm d\nu
 =\begin{cases}
 O\!\left(q_{\rm TF}^4 q\log(e/q)\right),&q\to0,\\
 O(q_{\rm TF}^4 q^{-2}),&q\to\infty.
 \end{cases}
 \label{S-app-eq:selfenergy-L1}
\end{equation}

On compact momentum intervals away from zero, the only remaining singularity is the integrable $\log(1/s)$ at $s=0$, including at
$q=2$. The complete bound is therefore integrable against
$\mathrm d^3q\,v=4\pi\,\mathrm dq$. Combining it with the vertex bound
gives the corresponding result for $\Pi^{(1)}$.

\subsection{General spectral interchange and weighted remainders}\label{S-app-app:block-generality}

The same argument applies to any functional linear in one selected
line, provided its bare-line absolute integral is finite:
\begin{equation}
 \mathcal J[f]=\int \mathfrak f(x)f(Q(x))\,\mathrm d\mu(x),\qquad
 \int |\mathfrak f(x)|v(q(x))\,\mathrm d\mu(x)<\infty.
 \label{S-app-eq:general-line-majorant}
\end{equation}
Here $x$ collects all integration variables, $Q=(q,\mathrm i\nu)$
is the momentum transfer carried by the selected line, and $q(x)$ is its
magnitude. The measure $\mathrm d\mu(x)$ collects the integration domains and measures, and
$\mathfrak f$ includes the closed-diagram coefficient and the fixed unselected lines. The choice
\[
 x=(\bm q,\nu),\qquad \mathrm d\mu=\mathrm d^3q\,\mathrm d\nu,\qquad \mathfrak f=Y_D/\pi
\]
gives the condition stated after main-text Eq.~\eqref{eq:common-blocks}.

\emph{Majorant.} Consider a line with a nonnegative weight $a(q,\Lambda)$, such as
$a_\lambda/\lambda$ of \eqref{S-app-eq:block-superposition}, and assume
\[
 f=v-\int_0^\infty a(q,\Lambda)b_\Lambda\,\mathrm d\Lambda,\qquad
 0\le\int_0^\infty ab_\Lambda\,\mathrm d\Lambda\le v.
\]
Tonelli's theorem gives
\begin{equation}
 \int\mathrm d\Lambda\int |\mathfrak f(x)|a(q(x),\Lambda)
 b_\Lambda(Q(x))\,\mathrm d\mu(x)
 \le \int |\mathfrak f(x)|v(q(x))\,\mathrm d\mu(x).
 \label{S-app-eq:spectral-majorant}
\end{equation}
By \eqref{S-app-eq:spectral-majorant}, the integrand
$\mathfrak f\,a\,b_\Lambda$ is absolutely integrable over $(x,\Lambda)$,
and Fubini's theorem gives
\begin{equation}
 \mathcal J[f]=\mathcal J[v]-\int\mathrm d\Lambda\,
 \mathcal J[a(q,\Lambda)\,b_\Lambda].
 \label{S-app-eq:general-line-spectral}
\end{equation}
Here $\mathcal J$ acts at fixed $\Lambda$, as in \eqref{S-app-eq:common-blocks}. The argument also
covers positive spectral measures with point masses.

\emph{Weighted remainders.} Inserting an asymptotic expansion into the resulting integrals also
requires a bound on the integrated remainder. Write an integrand as
\[
 K_\eta=K_\eta^{\rm asy}+R_\eta,\qquad\eta\to0,
\]
with $K_\eta^{\rm asy}$ the retained terms of its asymptotic expansion and $R_\eta$ the remainder,
and let $\mathrm d\rho_\eta$ include the momentum, spectral and coupling measures.
With a nonnegative majorant $M$ of the remainder, an error scale $r(\eta)$
and a fixed $\eta_0>0$, a sufficient condition for the remainder bound is
\begin{equation}
 |R_\eta|\le r(\eta)M,\qquad
 \sup_{0<\eta<\eta_0}\int M\,\mathrm d\rho_\eta<\infty
 \quad\Longrightarrow\quad
 \left|\int R_\eta\,\mathrm d\rho_\eta\right|=O(r(\eta)).
 \label{S-app-eq:weighted-block-remainder}
\end{equation}
This criterion can be applied separately to each momentum region, as in
Sec.~\ref{S-sm-sec:kite-remainders}. The condition
$r(\eta)=o(\sigma(\eta))$ makes the integrated remainder smaller than
the retained scale $\sigma(\eta)$, so the expansion is controlled through
that order.
\section{Technical conditions and proofs for Sec.~\ref{sec:setting}}\label{S-sm-slim-sec2}
This section collects the hypotheses and auxiliary steps behind main-text
Secs.~\ref{sec:screening-functions} and \ref{sec:mellin-convolution}.

\subsection{Statement and proof of the endpoint--residue correspondence}\label{S-sm-slim-sec:prop1}
The correspondence concerns the Mellin convolution of the main text and the Mellin
transforms of its two factors,
\[
 I(\varepsilon)=\int_0^\infty k(q)\,g(q/\varepsilon)\,\mathrm dq,\qquad
 \widetilde k(1-s)=\int_0^\infty q^{-s}k(q)\,\mathrm dq,\qquad
 \widetilde g(s)=\int_0^\infty t^{s-1}g(t)\,\mathrm dt .
\]
We follow L\'opez \cite{SI-Lopez2008}. His method needs neither the
Mellin--Parseval formula \eqref{eq:mellin-parseval} nor the decay of
the transforms along vertical lines. The contour shift of
Refs.~\cite{SI-BleisteinHandelsman,SI-FlajoletGourdonDumas1995} rests on these two conditions.
In the applications the decay $h(z)\simeq1/(3z^2)$ at large $z$ supplies the
integrability required in condition (iii) below.

\begin{proposition}[Endpoint--residue correspondence]\label{S-prop:endpoint-residue}
Suppose the following. (i) The integral \eqref{eq:mellin-convolution}
converges absolutely, and $\widetilde k(1-s)$ and $\widetilde g(s)$
converge absolutely in a nonempty common strip $c_-<{\rm Re}\,s<c_+$.
The second part of (i) does not follow from the first.
(ii) $k$ admits finite endpoint expansions in terms of the form
$x^{a}\log^{m}x$ at $0$ and at $\infty$, and $g$ admits one in terms of the form
$x^{a}$ at the endpoint that generates the expansion ($t\to\infty$ for
$\varepsilon\to0$ and $t\to0$ for $\varepsilon\to\infty$), with
remainders of strictly lower order. For $\varepsilon\to0$ this means that for
some $\sigma>c_+$, positive constants $K$ and $G$, and $p\ge0$,
\begin{equation}
\begin{aligned}
 k(q)&=\sum_{1+\alpha<\sigma}\sum_m\kappa_{\alpha,m}\,q^{\alpha}\log^m(1/q)+k_\sigma(q),
 &|k_\sigma(q)|&\le K\,q^{\sigma-1}\bigl(1+\log(1/q)\bigr)^{p}\quad(0<q\le1),\\
 g(t)&=\sum_{\beta<\sigma}\gamma_\beta\,t^{-\beta}+g_\sigma(t),
 &|g_\sigma(t)|&\le G\,t^{-\sigma}\quad(t\ge1),
\end{aligned}
\label{S-sm-slim-eq:endpoint-conditions}
\end{equation}
where the sums run over finitely many real exponents $\alpha$ (with
integers $m\ge0$) and $\beta$,
$\kappa_{\alpha,m}$ and $\gamma_\beta$ are their coefficients, and
$k_\sigma$, $g_\sigma$ are the remainders. The substitutions $q\to1/q$ and $t\to1/t$ map
$\varepsilon\to\infty$ onto this case, with the poles to the left of the strip
and the residues of opposite sign. (iii) When $k$ and $g$ depend on a parameter,
such as $z$, that is integrated after the residues are taken,
(i) and (ii) hold for almost every value of the parameter, with
integrable coefficients and remainder constants.
A small-$q$
term $q^{\alpha}\log^m(1/q)$ of $k$ gives
$\widetilde k(1-s)$ a pole of order $m+1$ at $s=1+\alpha$, a large-$q$
term $q^{-\alpha'}$ a simple pole at $s=1-\alpha'$. A large-$t$ term
$t^{-\beta}$ of $g$ gives $\widetilde g$ a simple pole at $s=\beta$, a
small-$t$ term $t^{\beta'}$ one at $s=-\beta'$. Then, as
$\varepsilon\to0$, $I(\varepsilon)$ is minus the sum of the residues of
$\widetilde k(1-s)\widetilde g(s)\varepsilon^s$ at the poles to the
right of the strip, in increasing order, and as $\varepsilon\to\infty$
the sum of the residues at the poles to its left, in decreasing order.
A pole of total order $N$ at $s_0$ contributes $\varepsilon^{s_0}$
times a polynomial of degree $N-1$ in $\log\varepsilon$. The remainder
after the poles up to the line ${\rm Re}\,s=\sigma$ is bounded by
$\varepsilon^{\sigma}$ times a power of $1+|\log\varepsilon|$. It is
bounded by a constant times $\varepsilon^{\sigma}$ if the endpoint
bounds extend beyond that line in the direction of the expansion with
no further pole: to some $\sigma'>\sigma$ with no pole in
$\sigma\le{\rm Re}\,s<\sigma'$ for $\varepsilon\to0$, and to some
$\sigma'<\sigma$ with no pole in $\sigma'<{\rm Re}\,s\le\sigma$ for
$\varepsilon\to\infty$.
\end{proposition}

\begin{proof}
We prove the case $\varepsilon\to0$ in the form
\[
 I(\varepsilon)=-\sum_{s_0<\sigma}\operatorname*{Res}_{s=s_0}\bigl[\widetilde k(1-s)\,\widetilde g(s)\,\varepsilon^s\bigr]+R_\sigma(\varepsilon),
\]
with the sum over the poles $s_0$ to the right of the strip and the remainder $R_\sigma$ estimated at the end of the proof.
Under (ii), $\widetilde k(1-s)$ and $\widetilde g(s)$ continue meromorphically to
$c_-<{\rm Re}\,s<\sigma$, with the stated poles [Eqs.~(9) and (10) of
Ref.~\cite{SI-Lopez2008}]. L\'opez subtracts the endpoint terms of the two factors
one at a time, in the order of their poles. Each subtraction leaves
a convergent integral equal to one term of the expansion [Theorem~1 of
Ref.~\cite{SI-Lopez2008}].

\emph{One factor.} When only one factor has a term at a pole $s_0$, that term contributes
\[
 -\operatorname*{Res}_{s=s_0}\widetilde k(1-s)\,\widetilde g(s)\,\varepsilon^s,\qquad
 q^{\alpha}\log^m(1/q)=(-\partial_\alpha)^m\,q^{\alpha}.
\]
The second relation reduces a term $q^{\alpha}\log^m(1/q)$ to the pure-power case by $m$
differentiations with respect to $\alpha$, which may be taken under the integral inside the strip.

\emph{Both factors.} When both factors have terms at $s_0$, L\'opez inserts a regulator
$q^{-z}$ and takes the limit $z\to0$ by dominated convergence. For two simple poles this
limit gives the term $-a_kb_j\log x$ of his Eq.~(19), with $x=\varepsilon$, and the same
regulator applies to poles of higher order. The two regulated terms are minus the residues of
$F_z(s)=\widetilde k(1-s-z)\widetilde g(s)\varepsilon^s$ at $s_0-z$ and at $s_0$. With
$\mathcal C$ a small circle about $s_0$,
\[
 -\Bigl[\operatorname*{Res}_{s=s_0-z}+\operatorname*{Res}_{s=s_0}\Bigr]F_z(s)
 =-\frac1{2\pi\mathrm i}\oint_{\mathcal C}F_z(s)\,\mathrm ds
 \ \xrightarrow{\ z\to0\ }\ -\operatorname*{Res}_{s=s_0}\widetilde k(1-s)\,\widetilde g(s)\,\varepsilon^s .
\]

\emph{Remainder.} After the term of the last pole is subtracted,
\[
 R_\sigma(\varepsilon)=\int_0^\infty k_\sigma(q)\,g_\sigma(q/\varepsilon)\,\mathrm dq,
\]
where $k_\sigma$ and $g_\sigma$ extend to the whole half-line by the subtractions of
\eqref{S-sm-slim-eq:endpoint-conditions}. On $q<\varepsilon$, where $t=q/\varepsilon<1$, and on $q>1$,
\[
 |g_\sigma(t)|\le|g(t)|+\sum_{\beta<\sigma}|\gamma_\beta|\,t^{-\beta}\quad(t<1),\qquad
 |k_\sigma(q)|\le|k(q)|+\sum|\kappa_{\alpha,m}|\,q^{\alpha}\log^mq\quad(q>1).
\]
Condition (i), $\sigma>c_+$ and the strict inequalities $\beta<\sigma$ and $1+\alpha<\sigma$
make the contributions of these two ranges $O(\varepsilon^\sigma[1+\log(1/\varepsilon)]^p)$ and
$O(\varepsilon^\sigma)$, respectively. The range $\varepsilon<q<1$, where both bounds of
\eqref{S-sm-slim-eq:endpoint-conditions} hold, gives one more power of $1+\log(1/\varepsilon)$,
\[
 \int_\varepsilon^1|k_\sigma(q)\,g_\sigma(q/\varepsilon)|\,\mathrm dq
 \le KG\,\varepsilon^\sigma\int_\varepsilon^1\bigl(1+\log(1/q)\bigr)^p\,\frac{\mathrm dq}q
 \le\frac{KG}{p+1}\,\varepsilon^\sigma[1+\log(1/\varepsilon)]^{p+1}.
\]
Splitting $R_\sigma(\varepsilon)$ at $q=\varepsilon$ and at $q=1$ therefore gives,
with $C_{\rm b}$ independent of $\varepsilon$,
\[
 |R_\sigma(\varepsilon)|\le C_{\rm b}\,\varepsilon^{\sigma}[1+\log(1/\varepsilon)]^{p+1}\qquad(0<\varepsilon\le1).
\]

\emph{Sharper bound.} If condition (ii) also holds for some $\sigma'$ with
\[
 \varepsilon\to0:\ \sigma'>\sigma,\ \text{no pole in }\sigma\le{\rm Re}\,s<\sigma',\qquad
 \varepsilon\to\infty:\ \sigma'<\sigma,\ \text{no pole in }\sigma'<{\rm Re}\,s\le\sigma,
\]
the same estimate with $\sigma'$ gives $R_\sigma=O(\varepsilon^{\sigma})$. The condition for
$\varepsilon\to\infty$ follows because the inversion reverses the direction of the expansion.
These bounds are analogous to Theorem~2 and Eqs.~(27)--(32) of Ref.~\cite{SI-Lopez2008}.
\end{proof}

\subsection{Proofs of Corollaries~\ref{cor:exponents} and \ref{cor:logs}}\label{S-sm-slim-sec:corollaries}
\begin{proof}[Proof of Corollary~\ref{cor:exponents}]
Both statements apply Proposition~\ref{S-prop:endpoint-residue} with the
scaling functions of main-text Sec.~\ref{sec:screening-functions}.
Below, $P_j$ denotes a polynomial of degree $j$.

(1) At high density $\varepsilon=q_{\rm TF}\to0$, and the poles to the right of the strip
contribute. The first pole of $\widetilde k(1-s)$ there and the poles of $\widetilde g(s)$
from the Thomas--Fermi tail $t^{-N(d-1)}$, $N\ge1$, of the scaling function are
\[
 \widetilde k(1-s):\ s=1+\alpha_0,\ \text{order }m+1\ [\text{Eq.~\eqref{S-sm-slim-eq:laurent}}],\qquad
 \widetilde g(s):\ s=N(d-1),\ \text{simple}.
\]
The leading term is the residue at the smaller of $1+\alpha_0$ and $d-1$ and is of order
$q_{\rm TF}^{\min(1+\alpha_0,d-1)}$, up to powers of $\log q_{\rm TF}$.
For $d=3$ and $\alpha_0=1$ the two poles coincide,
\[
 1+\alpha_0=d-1=2:\qquad\text{total order }m+2,\qquad
 \text{residue}=q_{\rm TF}^2\,P_{m+1}(\log q_{\rm TF}).
\]

(2) At low density $\varepsilon=q_*\to\infty$, and the poles to the left of the strip
contribute. The large-$q$ term $q^{-\alpha}$ of the kernel and the large-width term of the scaling
function, proportional to $t^2$, give simple poles,
\[
 q^{-\alpha}\ \to\ s=1-\alpha,\qquad t^2\ \to\ s=-2 .
\]
The leading term is the residue at $\max(1-\alpha,-2)$ and is of order
$q_*^{\max(1-\alpha,-2)}=q_*^{-\min(\alpha-1,2)}$.
For $\alpha=3$ the two poles coincide, and the
double pole gives the factor $\log q_*$. For $\alpha>3$ the leading pole
is at $s=-2$, and Sec.~\ref{S-sm-slim-sec:largewidth} shows that its
residue is main-text Eq.~\eqref{eq:fsum-pole} under the conditions
stated in the corollary. These conditions also justify
the separation of the large-$q$ kernel that the width requires at low density
(main-text Sec.~\ref{sec:mellin-convolution}).
\end{proof}

\begin{proof}[Proof of Corollary~\ref{cor:logs}]
The density enters through the prefactor of $Y_D=r_s^{\,n-2}\hat Y_D$
and through $q_{\rm TF}$ in the scaling function (main-text
Sec.~\ref{sec:linear-diagrams}). The kernel term $q\log^m(1/q)$ gives
$\widetilde k(1-s)$ a pole of order $m+1$ at $s=2$. With $P_j$ as above, its residue is
\[
 \begin{aligned}
 2&\ne N(d-1)\ \text{for every }N:&&q_{\rm TF}^2\,P_m(\log q_{\rm TF}),\\
 2&=N_*(d-1):&&q_{\rm TF}^2\,P_{m+1}(\log q_{\rm TF}),
 \end{aligned}
\]
since the screening pole $N=N_*$ adds one order in the second case. With
$q_{\rm TF}^{d-1}=a_{\rm TF}r_s$,
\[
 q_{\rm TF}^2\propto r_s^{2/(d-1)},\qquad \log q_{\rm TF}=\log r_s/(d-1)+O(1),
\]
and with the prefactor these residues give the terms stated in the corollary.
The poles $s=N(d-1)<2$ are simple, because the kernel transform is regular
there, and give $q_{\rm TF}^{N(d-1)}\propto r_s^N$. For the density-independent
kernel $4\pi q\hat Y_D(q,qz)$, with $\kappa_D$ and $\gamma_{N_*}$
of Corollary~\ref{cor:logs},
Eq.~\eqref{S-sm-slim-eq:leading-residue} gives
\[
 q_{\rm TF}^2\log^{m+1}q_{\rm TF}\cdot\frac{(-1)^{m+1}\kappa_D(z)\,\gamma_{N_*}(z)}{m+1},
\]
and the coefficient of the corollary is therefore proportional to
$\int_0^\infty\kappa_D\gamma_{N_*}\,\mathrm dz/\pi$.
Replacing $\mathcal L(q,qz)$ by $h(z)$ in the scaling function changes $c$ by
\[
 O\bigl(q_{\rm TF}^2/[q(1+z^2)]\bigr)\quad(q<1,\ d=3)\ [\text{Eq.~\eqref{S-app-eq:highdensity-u3}}],\qquad
 O(q_{\rm TF}q)\quad(d=2)\ [\text{Eq.~\eqref{S-deq:g2}}].
\]
Since the scaling functions of
main-text table~\ref{tab:screening} have bounded derivatives in $c$, this
replacement changes the screening correction by $O(q_{\rm TF}^2)$, which lies within the remainder
$O(q_{\rm TF}^2|\log q_{\rm TF}|^m)$, as in Sec.~\ref{S-sm-slim-sec:residues}.
For $q>1$, where $c=O(q_{\rm TF}^2/q^2)$, both forms of $g$ contribute
$O(q_{\rm TF}^2)$. If the factors $\log(1/q)$ of the kernels cancel in a sum
of diagrams, the summed kernel has a smaller $m$, which lowers the order of its pole at
$s=2$ and the power of $\log q_{\rm TF}$ in the resonant term.
\end{proof}

\subsection{Expansion of the radial kernel in \texorpdfstring{$1/\Lambda$}{1/Lambda} and the low-density pole}\label{S-sm-slim-sec:largewidth}
The second statement of Corollary~\ref{cor:exponents} uses the
expansion of the radial kernel \eqref{eq:radial-kernel} in $1/\Lambda$.
If $Y_D$ has a high-frequency tail $A_D(q)/\nu^2$, then, formally at fixed
internal momenta, the radial kernel of the Lorentzian line of width $\Lambda$
differs from that of the bare line by
\begin{equation}
 e_{D,\Lambda}(q)-e_{D,\infty}(q)=-\frac{2\pi A_D(q)}{q\Lambda}
 +O\bigl((q\Lambda)^{-2}\bigr),\qquad\Lambda\to\infty.
 \label{S-sm-slim-eq:large-width}
\end{equation}
\emph{Derivation.} By \eqref{eq:radial-kernel},
\[
 e_{D,\Lambda}-e_{D,\infty}
 =-4\pi\int_0^\infty\frac{\mathrm d\nu}{\pi}
 \frac{\nu^2/q^2}{\Lambda^2+\nu^2/q^2}\,Y_D(q,\nu).
\]
For $\Lambda\to\infty$ the factor is small unless $\nu\gtrsim q\Lambda$,
where $Y_D\simeq A_D/\nu^2$. Inserting this tail, and bounding the rest by
$\nu^2/(q^2\Lambda^2+\nu^2)\le\nu^2/(q\Lambda)^2$,
\[
 \begin{aligned}
 \text{tail}&=-\frac{4\pi A_D}{q^2}\int_0^\infty\frac{\mathrm d\nu}{\pi}\frac1{\Lambda^2+\nu^2/q^2}
 =-\frac{2\pi A_D}{q\Lambda},\\
 |\text{remainder}|&\le\frac4{(q\Lambda)^2}\int_0^\infty\nu^2\Bigl|Y_D-\frac{A_D}{\nu^2}\Bigr|\,\mathrm d\nu .
 \end{aligned}
\]
This bound is finite when
\[
 \int_0^\infty|Y_D|\,\mathrm d\nu<\infty,\qquad
 Y_D-A_D/\nu^2=O(\nu^{-4})\quad(\nu\to\infty).
\]
For the first-order exchange integrands the first condition follows from the
absolute bounds of Secs.~\ref{S-app-sec:loss-spectrum} and
\ref{S-app-app:selfenergy-blocks}, and the second holds because their frequency
factors are even in $\nu$. For the kite, \eqref{eq:Xdef} at large $\nu$ gives
$Y_D=A_0X\simeq-A_0Z(q)/\nu^2$ with $Z$ of \eqref{S-app-eq:highdensity-u6},
and the self-energy pair has the opposite tail (Sec.~\ref{S-app-sec:cancel}).

\emph{Residue at $s=-2$.} With the single Lorentzian
\eqref{eq:pole-dispersion}, and up to the contribution of $\varsigma_\lambda$
bounded in the next paragraph,
\[
 \mathcal J_D[W_\lambda]=\lambda\int_0^\infty\bigl[(1-\theta_\lambda)e_{D,\infty}
 +\theta_\lambda(e_{D,\infty}-e_{D,\Lambda_\lambda})\bigr]\,\mathrm dq .
\]
By the leading term of \eqref{S-sm-slim-eq:large-width}, the second term in the
brackets gives $\int_0^\infty k(q)\,g(q/q_*)\,\mathrm dq$ with the kernel and
the scaling function
\[
 k(q)=\frac{2\pi A_D(q)}{q^2},\qquad
 g(t)=\frac{\theta_\lambda q}{\Lambda_\lambda}=\frac{2\lambda t^2}{(t^4+\lambda)^{3/2}}
 \simeq\frac{2t^2}{\sqrt\lambda}\quad(t\to0).
\]
The scaling function depends on $q$ only through $t=q/q_*$. Its small-$t$ term
gives the simple pole, and $\widetilde k(1-s)$ is finite there by the first
condition of the corollary,
\[
 \widetilde g(s)\simeq\frac{2/\sqrt\lambda}{s+2},\qquad
 \widetilde k(1-s)\big|_{s=-2}=2\pi\int_0^\infty A_D\,\mathrm dq .
\]
By Proposition~\ref{S-prop:endpoint-residue} with $\varepsilon=q_*\to\infty$,
this pole contributes, with the prefactor $\lambda$,
\[
 \lambda\operatorname*{Res}_{s=-2}\bigl[\widetilde k(1-s)\widetilde g(s)q_*^{s}\bigr]
 =\lambda\cdot\frac{4\pi}{\lambda^{1/2}q_*^2}\int_0^\infty A_D\,\mathrm dq
 =\frac{4\pi\lambda^{1/2}}{q_*^2}\int_0^\infty A_D\,\mathrm dq ,
\]
which at $\lambda=1$ is main-text Eq.~\eqref{eq:fsum-pole}. The factor
$1-\theta_\lambda$ and the remainder of \eqref{S-sm-slim-eq:large-width}, which
carries $(q/\Lambda_\lambda)^2$, are of order $t^4$ at small $t$,
\[
 1-\theta_\lambda=\frac{t^4}{t^4+\lambda},\qquad
 \Bigl(\frac q{\Lambda_\lambda}\Bigr)^2\simeq\frac{4t^4}\lambda,
\]
two powers beyond the pole term, and are controlled by the remainder condition
of the corollary. The term with $1-\theta_\lambda$ is $o(q_*^{-2})$ also for
$q\gtrsim q_*$ when the bare-line kernel $e_{D,\infty}$ decays faster than
$q^{-3}$. This holds for the first-order exchange diagrams, since the bare-line
radial kernel of the kite decays as $q^{-4}$ (main-text
Sec.~\ref{sec:mellin-convolution} and Sec.~\ref{S-app-sec:largeq-kernel}), and
that of the self-energy pair vanishes for every $q$ by
Eq.~\eqref{S-app-eq:selfenergyreduction-u7}.

\emph{The remainder $\varsigma_\lambda$.} The pole at $s=-2$ collects
$\int A_D\,\mathrm dq$ from all $q\ll q_*$, including $q\lesssim1$, where
the bound on $\varsigma_\lambda$ in main-text
Eq.~\eqref{eq:pole-dispersion} is not small. Fix $q_0\ge3$. For
$q\le q_0$ we compare the screened line with its single-Lorentzian form
directly, using the representation of $\mathcal L$ behind
Eq.~\eqref{S-app-eq:lowdensity-u13}, valid for every $q$. The splitting
\[
 \frac1{\delta_{\bm p}^2+\nu^2}=\frac1{\nu^2}-\frac{\delta_{\bm p}^2}{\nu^2(\delta_{\bm p}^2+\nu^2)},
\]
the free $f$-sum identity of Sec.~\ref{S-app-sec:sigma-lowdensity} and
$|\delta_{\bm p}|\le q+q^2/2$ give, with $\omega_p^2=q_{\rm TF}^2/3$ of main-text
Eq.~\eqref{eq:qstar},
\[
 \mathcal L(q,\nu)=\frac{q^2}{3\nu^2}\Bigl[1+O\Bigl(\frac{q^2(1+q)^2}{\nu^2}\Bigr)\Bigr],\qquad
 c=\frac{\omega_p^2}{\nu^2}\Bigl[1+O\Bigl(\frac{q^2(1+q)^2}{\nu^2}\Bigr)\Bigr].
\]
The single-Lorentzian line of main-text Eq.~\eqref{eq:pole-dispersion} is
\[
 \frac{\lambda v(\nu^2+q^4/4)}{\nu^2+q^4/4+\lambda\omega_p^2}
 =\frac{\lambda v\nu^2}{\nu^2+\lambda\omega_p^2}+v\,O\Bigl(\frac{q^4}{\omega_p^2}\Bigr),
\]
and the correction lies within the bound below. For $\nu>q(1+q/2)$ one has
$\delta_{\bm p}^2<\nu^2$ on the transition regions, and therefore $c$ lies between
$\omega_p^2/(2\nu^2)$ and $\omega_p^2/\nu^2$. Since
$\lambda c/(1+\lambda c)^2\le\min(\lambda c,1/(\lambda c))$, the two lines
then differ by $v\,O(q^2(1+q)^2/\omega_p^2)$. For $\nu\le q(1+q/2)$ both lines
are $v\,O(q^2(1+q)^2/q_{\rm TF}^2)$, the exact line by
Eq.~\eqref{S-app-eq:lowdensity-u13} and the Lorentzian one because it is at most
$v(\nu^2+q^4/4)/\omega_p^2$. If
\[
 \int_0^{q_0}q^2(1+q)^2\int_0^\infty|Y_D|\,\mathrm d\nu\,\mathrm dq<\infty
\]
(for the first-order exchange integrands by the bounds of
Secs.~\ref{S-app-sec:loss-spectrum} and \ref{S-app-app:selfenergy-blocks}), both
differences, integrated against $|Y_D|$, contribute
$O(q_{\rm TF}^{-2})=O(q_*^{-4})$ from $q\le q_0$, with constants depending on
$q_0$. For $q>q_0$ we use a bound
\[
 |Y_D(q,\nu)|\le\frac{a_D(q)}{\nu^2+q^4},\qquad
 \int_{q_0}^\infty a_Dq^{-2}\,\mathrm dq<\infty .
\]
For the first-order exchange integrands this bound holds with $a_D$ bounded for
$q\ge3$, where $q^2/2-q\ge q^2/6$, by Eq.~\eqref{S-app-eq:Xudecay} and the bound
on $\Pi^{(1)}_{\rm SE}$ used in Sec.~\ref{S-app-sec:sigma-lowdensity}. With
$u=\nu/q^2$, $t=q/q_*$ and $\hat c$ of main-text Eq.~\eqref{eq:pole-dispersion},
\[
 \int_0^\infty\frac{\mathrm du}{1+u^2}\,\frac{\lambda\hat c}{(1+\lambda\hat c)^2}
 \le\frac\pi2\,\frac{t^2}{\sqrt\lambda},\qquad
 q^{-4}t^2=q^{-2}q_*^{-2}.
\]
In this range the error bound therefore contributes at most a constant times
$q_*^{-2}\int_{q_0}^\infty a_Dq^{-2}\,\mathrm dq$, and the integral vanishes as
$q_0\to\infty$. Hence $\varsigma_\lambda$ does not change the coefficient of
$q_*^{-2}$, that is, the residue at $s=-2$.

\emph{Conditions and the kite.} The conditions of the corollary are that
$\int_0^\infty A_D\,\mathrm dq$ converges and that the remainder in
\eqref{S-sm-slim-eq:large-width}, including regions with large internal momentum
transfers, is integrable uniformly in $q_*$ after the internal integrations.
For the kite the first condition fails. Here $A_D(q)=-A_0Z(q)$. Each of the two
same-ball regions of Sec.~\ref{S-app-sec:largeq-kernel} gives $4\pi^2$ and the
different-ball regions are $O(q^{-2})$, hence, as $q\to\infty$,
\[
 Z(q)=2\cdot4\pi^2+O(q^{-2})\to8\pi^2,\qquad
 A_D\to-\frac3{4\pi^3},\qquad
 -\frac{2\pi A_D}{q\Lambda}\to\frac3{2\pi^2q\Lambda}.
\]
This is the term of order $1/\Lambda$ of the $q^{-2}$ kernel
\eqref{eq:largeq-block}, and the leading low-density term of the kite is the
kernel pole at $s=-1$ ($\alpha=2$ in the corollary).

\subsection{Cancellation of the universal low-density residue}\label{S-sm-slim-sec:lemma}
\begin{lemma}[Cancellation of the universal low-density residue]\label{S-lem:fsum}
Let $\mathcal S$ be a set of diagrams closed with the same screened line and the
same coupling weight, and let the summed integrand
$Y_{\mathcal S}=\sum_{D\in\mathcal S}Y_D$ have the high-frequency tail
$Y_{\mathcal S}\simeq A_{\mathcal S}(q)/\nu^2$. If $A_{\mathcal S}(q)=0$ for every $q$ and the summed
radial kernel satisfies the remainder condition of
main-text Corollary~\ref{cor:exponents}, the pole at $s=-2$ is absent from
$\sum_{D\in\mathcal S}\mathcal J_D[W]$. The condition $A_{\mathcal S}=0$ holds for the
first-order exchange corrections to the proper polarization at fixed
density.
\end{lemma}
\begin{proof}
The residue \eqref{eq:fsum-pole} of the pole at $s=-2$ comes from the
term $-2\pi A/(q\Lambda)$ of the expansion \eqref{S-sm-slim-eq:large-width} in
$1/\Lambda$ and is linear in the tail $A$ of the integrand. The summed integrand
$Y_{\mathcal S}=\sum_{D\in\mathcal S}Y_D$ of the lemma has the radial kernel
$e_{\mathcal S,\Lambda}$ and the tail $A_{\mathcal S}$, with
\[
 e_{\mathcal S,\Lambda}=\sum_{D\in\mathcal S}e_{D,\Lambda},\qquad
 A_{\mathcal S}=\sum_{D\in\mathcal S}A_D,\qquad
 -\frac{2\pi A_{\mathcal S}}{q\Lambda}=0\quad\text{when }A_{\mathcal S}=0.
\]
The remainder condition of Corollary~\ref{cor:exponents}, imposed on the
summed integrand, then leaves no contribution at $s=-2$. For the first-order
exchange corrections at fixed density, $A_{\mathcal S}=0$ for every $q$ by
\eqref{S-app-eq:highdensity-u6} (Sec.~\ref{S-app-sec:cancel}).
\end{proof}

\emph{Argument for complete sets of corrections of higher order.} With
$\hat\rho_{\bm q}=\sum_j{\rm e}^{-\mathrm i\bm q\cdot\bm r_j}$ the
coefficient of $\nu^{-2}$ in the exact density response is
$-\langle[[\hat\rho_{\bm q},\hat H],\hat\rho_{-\bm q}]\rangle/\hbar^2$
per unit volume (the $f$-sum rule \cite[Sec.~2.2]{SI-PinesNozieres1966}). The
interaction commutes with the density operator, and in physical units and in
the present sign convention
\[
 [[\hat\rho_{\bm q},\hat T],\hat\rho_{-\bm q}]=\hat N\hbar^2q^2/m_e,\qquad
 \chi(q,\mathrm i\nu)=-\frac{n_eq^2}{m_e\nu^2}+\cdots.
\]
The coefficient therefore equals its noninteracting value for every
$N$-electron state and every coupling. The same holds for the proper
polarization, since
\[
 \Pi_\lambda=\frac{\chi_\lambda}{1+\lambda v\chi_\lambda},\qquad
 \lambda v\chi_\lambda=O(\nu^{-2}).
\]
The argument needs two further assumptions. If the coupling expansion of the exact
fixed-density response coincides, order by order, with the sum of the $G_0$
diagrams and their chemical-potential counterterms, and if the large-$\nu$ limit
may be taken term by term, the tails of the corrections of each order sum to
zero. Neither assumption is proved here beyond first order. At second and
higher order, individual diagrams with bare internal lines can be infrared
divergent, and the statement therefore applies to the regularized or resummed
sum, not diagram by diagram.

\subsection{Moment rule for coupling averages}\label{S-sm-slim-sec:moment}
Let $M^{(j)}(p)=\mathrm d^jM/\mathrm dp^j$ be the derivatives of the Mellin multiplier $M(p)$
of main-text Sec.~\ref{sec:mellin-convolution}. A term
$\varphi r_s^p\log^mr_s$ of the full-coupling energy $F_D$ with constant coefficient $\varphi$
then appears in the average \eqref{eq:general-average} as
\begin{equation}
 \varphi r_s^p\sum_{j=0}^m\binom mjM^{(j)}(p)\log^{m-j}r_s .
 \label{S-sm-slim-eq:moment-rule}
\end{equation}
This follows by expanding $\log^m(\lambda r_s)$ binomially in
$\log r_s$ and $\log\lambda$,
\[
 \begin{gathered}
 \int_0^1w_D(\lambda)\,\varphi(\lambda r_s)^p\log^m(\lambda r_s)\,\mathrm d\lambda
 =\varphi r_s^p\sum_{j=0}^m\binom mj\log^{m-j}r_s
 \int_0^1w_D(\lambda)\lambda^p\log^j\lambda\,\mathrm d\lambda,\\
 \int_0^1w_D(\lambda)\lambda^p\log^j\lambda\,\mathrm d\lambda=M^{(j)}(p).
 \end{gathered}
\]
The rule requires that the moments converge and may be differentiated under the
integral (for $w_D=2\lambda$, $p>-2$). For $w_D=2\lambda$ and $m=2$,
\[
 M(p)=\frac2{p+2}:\qquad
 \varphi r_s^p\bigl[M\log^2r_s+2M'\log r_s+M''\bigr]
 =2\varphi r_s^p\Bigl[\frac{\log^2r_s}{p+2}-\frac{2\log r_s}{(p+2)^2}+\frac2{(p+2)^3}\Bigr].
\]

\subsection{Residues at the high-density pole}\label{S-sm-slim-sec:residues}
This subsection gives the residue computations behind main-text
Secs.~\ref{sec:benchmarks}, \ref{sec:Pi1-d} and \ref{sec:highdensity}.
Let the kernel and the scaling function behave as
\[
 k(q)\simeq\kappa\,q^{\alpha_0}\log^m(1/q)\quad(q\to0),\qquad
 g(t)\simeq\gamma\,t^{-\beta}\quad(t\to\infty),\qquad
 1+\alpha_0=\beta,
\]
where $\kappa$ and $\beta$ are generic (not $\kappa(\Lambda)$ and
$\beta(\Lambda)$ of Eq.~\eqref{S-app-eq:kappa}). Then, near $s=\beta$, the
leading Laurent terms are
\begin{equation}
 \widetilde k(1-s)\simeq\frac{m!\,\kappa}{(\beta-s)^{m+1}},\qquad
 \widetilde g(s)\simeq\frac{\gamma}{\beta-s},
 \label{S-sm-slim-eq:laurent}
\end{equation}
from
\[
 \int_0^1q^{-s}q^{\alpha_0}\log^m(1/q)\,\mathrm dq=\frac{m!}{(1+\alpha_0-s)^{m+1}},\qquad
 \int_1^\infty t^{s-1-\beta}\,\mathrm dt=\frac1{\beta-s}.
\]
\emph{Leading residue.} The pole at $s=\beta$ has total order $m+2$. With the
expansion of $\varepsilon^s$ about $s=\beta$,
\[
 \begin{gathered}
 \widetilde k(1-s)\widetilde g(s)\varepsilon^s\simeq
 \frac{(-1)^mm!\,\kappa\gamma\,\varepsilon^s}{(s-\beta)^{m+2}},\qquad
 \varepsilon^s=\varepsilon^\beta\sum_{n\ge0}\frac{(s-\beta)^n\log^n\varepsilon}{n!},\\
 \operatorname*{Res}_{s=\beta}\frac{\varepsilon^s}{(s-\beta)^{m+2}}
 =\frac{\varepsilon^\beta\log^{m+1}\varepsilon}{(m+1)!}.
 \end{gathered}
\]
For the integral $I(\varepsilon)$ of main-text Eq.~\eqref{eq:mellin-convolution}
this gives, as $\varepsilon\to0$,
\begin{equation}
 I(\varepsilon)=-\operatorname*{Res}_{s=\beta}
 \frac{(-1)^{m}m!\,\kappa\gamma\,\varepsilon^s}{(s-\beta)^{m+2}}+\dots
 =\frac{(-1)^{m+1}\kappa\gamma}{m+1}\,\varepsilon^{\beta}\log^{m+1}\varepsilon
 +O\bigl(\varepsilon^{\beta}|\log\varepsilon|^{m}\bigr).
 \label{S-sm-slim-eq:leading-residue}
\end{equation}
The same leading term follows from the window integral
\[
 \int_\varepsilon^1\kappa q^{\alpha_0}\log^m(1/q)\,\gamma(\varepsilon/q)^{\beta}\,\mathrm dq
 =\kappa\gamma\varepsilon^\beta\int_\varepsilon^1\frac{\log^m(1/q)}q\,\mathrm dq
 =\frac{\kappa\gamma\varepsilon^\beta\log^{m+1}(1/\varepsilon)}{m+1}+\dots.
\]

\emph{Gell-Mann--Brueckner.} The kernel and scaling function of main-text
Eq.~\eqref{eq:ring-energy} for $d=3$, whose prefactor is
$\frac3{8\pi}E_0=12/(\pi^3q_{\rm TF}^4)$~Ry, are at fixed $z$
\[
 \begin{gathered}
 k=q^3\quad(\kappa,\alpha_0,m)=(1,3,0),\\
 g(t)=\log\Bigl(1+\frac h{t^2}\Bigr)-\frac h{t^2}\simeq-\frac{h^2}{2t^4}\quad(t\to\infty),\qquad(\gamma,\beta)=\bigl(-\tfrac12h^2,4\bigr).
 \end{gathered}
\]
Then \eqref{S-sm-slim-eq:leading-residue} gives
$-\kappa\gamma\,q_{\rm TF}^4\log q_{\rm TF}=\tfrac12h^2q_{\rm TF}^4\log q_{\rm TF}$.
Replacing $\mathcal L(q,qz)$ by $h(z)$ for
$q<1$ changes $c$ as in Eq.~\eqref{S-app-eq:highdensity-u3}, and for $q>1$ the $z$
integral of $c^2$ follows from the large-$q$ form of $\mathcal L$ in
Eq.~\eqref{S-app-eq:Glargeq},
\[
 \begin{aligned}
 q<1:&\quad \Delta c=O\Bigl(\frac{q_{\rm TF}^2}{q(1+z^2)}\Bigr),\qquad
 \bigl|\partial_c[\log(1+c)-c]\bigr|\le c,\\
 q>1:&\quad 0\le c-\log(1+c)\le\frac{c^2}2,\qquad
 \int c^2\,\mathrm dz=O\bigl(q_{\rm TF}^4q^{-7}\bigr).
 \end{aligned}
\]
Hence the replacement for $q<1$ changes $e_{\rm RPA}$ only at $O(1)$, and the region $q>1$ also
contributes only at $O(1)$. With the prefactor
$12/(\pi^3q_{\rm TF}^4)$ and $\log q_{\rm TF}=\tfrac12\log r_s+O(1)$,
\begin{equation}
 e_{\rm RPA}=\frac{3}{\pi^3}\int_{-\infty}^\infty h^2\,\mathrm dz\,\log r_s+O(1),
 \qquad \int_{-\infty}^\infty h^2\,\mathrm dz=\frac{2\pi}{3}(1-\log2),
 \label{S-sm-slim-eq:GB}
\end{equation}
which is main-text Eq.~\eqref{eq:GB}.

\emph{Closed-diagram coefficients.} Write $c=c_{\rm f}+\delta c$ with
$c_{\rm f}=-v\Pi_0$ and $\delta c=-v\Pi^{(1)}$. Then
\[
 \log(1+c)-c=\sum_{j\ge2}\frac{(-1)^{j+1}}j\,c^j,\qquad
 c^j=c_{\rm f}^{\,j}+j\,c_{\rm f}^{\,j-1}\delta c+O(\delta c^2),
\]
and the coefficient of $c_{\rm f}^{\,j-1}\delta c$ is $j$ times that of
$c_{\rm f}^{\,j}$. A ring of $j$ lines with one insertion therefore receives $j$
times the coefficient $1/(2j)$ of the ring of $j$ free bubbles, that is,
$\frac12$. A ring with one insertion $\Pi^{(1)}_{\rm vtx}$ keeps only the
reflection that reverses its direction, and therefore its closed-diagram coefficient,
counted as in Sec.~\ref{S-sm-sec:kite-chain}, is $\frac12$. The two insertions of
$\Pi^{(1)}_{\rm SE}$ close to the same diagram, as for $S_D=4$ in
main-text Sec.~\ref{sec:se-example}. The first-order change of
main-text Eq.~\eqref{eq:ring-energy} is therefore the closed-diagram
series of Refs.~\cite{SI-DuBois1959,SI-CarrMaradudin}. A coupling average with
$\delta\Pi_\lambda=\lambda\Pi^{(1)}$, main-text Eq.~\eqref{eq:acfd-linear},
would weight the $t^{-2}$ tail by
\[
 \int_0^1\lambda\bigl[(1+\lambda c)^{-2}-1\bigr]\,\mathrm d\lambda
 \simeq-2c\int_0^1\lambda^2\,\mathrm d\lambda=-\tfrac23c
 \qquad\text{in place of }-c.
\]
This weight multiplies the coefficient of $r_s\log r_s$ by the
Mellin multiplier $\frac23$.

\emph{First-order exchange polarization in the ring chain.} The first-order change
of $\log(1+c)-c$ (main-text Sec.~\ref{sec:DuBois}) and its scaling function are
\[
 -v\Pi^{(1)}\Bigl[\frac1{1+c}-1\Bigr],\qquad
 \frac1{1+c}-1=-\frac c{1+c}\simeq-\frac h{t^2}\qquad(\gamma=-h,\ \beta=2).
\]
With $\Pi^{(1)}(q,\mathrm iqz)\to-(q_{\rm TF}^4/16)R^{(1)}(z)$
\cite{SI-DuBois1959} and \cite[Eqs.~(40) and (41)]{SI-Endo1997}, where
$R^{(1)}$ is given in main-text Eq.~\eqref{eq:R1}, the kernel is
$q^3\cdot q_{\rm TF}^4R^{(1)}/(16q^2)$ ($\alpha_0=1$, $m=0$). Then
\eqref{S-sm-slim-eq:leading-residue} with the prefactor
$12/(\pi^3q_{\rm TF}^4)$ gives
\[
 \begin{gathered}
 \frac{12}{\pi^3q_{\rm TF}^4}\int\mathrm dz\,\frac{q_{\rm TF}^4R^{(1)}}{16}\,h\;q_{\rm TF}^2\log q_{\rm TF}
 =\frac3{4\pi^3}\int hR^{(1)}\,\mathrm dz\;q_{\rm TF}^2\log q_{\rm TF},\\
 q_{\rm TF}^2\log q_{\rm TF}=\frac{2\alpha_L}\pi\,r_s\log r_s+O(r_s),
 \end{gathered}
\]
which is main-text Eq.~\eqref{eq:DuBois}. The remainder $O(r_s)$ also contains
the corrections to the limiting forms. Replacing $\mathcal L(q,qz)$ by $h(z)$ for
$q<1$, estimated as for the ring series above, and the region $q>1$, bounded using
$|c/(1+c)|\le c$ and the absolute frequency bounds of
Sec.~\ref{S-app-app:selfenergy-blocks}, both give corrections
$O(q_{\rm TF}^2)=O(r_s)$. The correction to the limit
$-(q_{\rm TF}^4/16)R^{(1)}(z)$ of $\Pi^{(1)}(q,\mathrm iqz)$ also enters at
$O(r_s)$ if it vanishes as a positive power of $q$ with a bound
integrable in $z$. This condition is assumed, not proved. The order of the
remainder agrees with the next term, $-0.021\,r_s$~Ry, of
Ref.~\cite{SI-CarrMaradudin} and Ref.~\cite[Eq.~(5)]{SI-Endo1999}.

\emph{Lindhard correction in the ring series} (main-text
Sec.~\ref{sec:DuBois}). Expanding Eq.~\eqref{S-app-app-eq:lindhard-shell}
at fixed $z$ gives
\[
 \mathcal L(q,qz)=h(z)+q^2g_3(z)+O(q^4),\qquad g_3(z)=-\frac1{12(1+z^2)^2}.
\]
The resulting addition to the integrand of main-text Eq.~\eqref{eq:ring-energy}
has the kernel parameters and the term of $-c/(1+c)$
\[
 \begin{gathered}
 (\kappa,\alpha_0,m)=(q_{\rm TF}^2g_3,3,0),\qquad
 (\gamma,\beta)=(h^2,4),\\
 -\kappa\gamma\,q_{\rm TF}^4\log q_{\rm TF}=-h^2g_3\,q_{\rm TF}^6\log q_{\rm TF},
 \end{gathered}
\]
the last by Eq.~\eqref{S-sm-slim-eq:leading-residue}. With the prefactor
$12/(\pi^3q_{\rm TF}^4)$,
$q_{\rm TF}^2\log q_{\rm TF}=(2\alpha_L/\pi)\,r_s\log r_s+O(r_s)$ and
$\int_{-\infty}^{\infty}h^2g_3\,\mathrm dz=-\pi(\pi^2-6)/288$,
\[
 \begin{aligned}
 -\frac{12}{\pi^3}\int_{-\infty}^{\infty}h^2g_3\,\mathrm dz\;q_{\rm TF}^2\log q_{\rm TF}
 &=\frac{12}{\pi^3}\cdot\frac{\pi(\pi^2-6)}{288}\cdot\frac{2\alpha_L}\pi\,r_s\log r_s+O(r_s)\\
 &=\frac{\alpha_L(\pi^2-6)}{12\pi^3}\,r_s\log r_s=0.00542\,r_s\log r_s+O(r_s) .
 \end{aligned}
\]
The next term of the expansion, of order $q^4$, has its kernel pole at $s=6$,
which meets the third screening pole. The term quadratic in
$q_{\rm TF}^2g_3$ carries a further factor $q_{\rm TF}^2$. Both therefore
contribute only at $r_s^2\log r_s$ and beyond. The region $q>1$ is
bounded as for the ring series.

\emph{Kite.} For $e_{\rm kite}$ at fixed $z$, main-text
Eq.~\eqref{eq:hd-convolution} has
\[
 \begin{gathered}
 (\kappa,\alpha_0,m)=(a_x(z),1,1),\qquad
 (\gamma,\beta)=\bigl(-\tfrac23h(z),2\bigr),\\
 \tfrac12\kappa\gamma\,q_{\rm TF}^2\log^2q_{\rm TF}
 =-\tfrac13a_xh\,q_{\rm TF}^2\log^2q_{\rm TF},
 \end{gathered}
\]
the last by \eqref{S-sm-slim-eq:leading-residue}. Integrating with
$4\pi A_0\int_0^\infty\mathrm dz/\pi$ and \eqref{S-app-eq:calA}, and using
$q_{\rm TF}^2\log^2q_{\rm TF}=(\alpha_L/\pi)\,r_s\log^2r_s+O(r_s|\log r_s|)$,
gives $c_{\rm hd}$ of \eqref{S-app-eq:highdensity-u2}. For the full-coupling energy $F$ the tail
is $\gamma=-h$, which multiplies the coefficient by $3/2$.

\section{Radial integral of the Fock self-energy corrections}\label{S-app-sec:sigma-reduction}
\subsection{Frequency integration and radial variables}\label{S-app-app:sigma-reduction}
\emph{Frequency integral.} With $\operatorname{Re}(\mathrm i\nu-\delta)^{-1}=-\delta/(\nu^2+\delta^2)$
and $\partial_\delta(\mathrm i\nu-\delta)^{-1}=(\mathrm i\nu-\delta)^{-2}$, the frequency
identity \eqref{eq:sigma-frequency-main} follows by differentiating the single-pole integral in $\delta$,
\[
 \begin{aligned}
 \int_0^\infty\frac{\mathrm d\nu}{\pi}\frac{\Omega^2}{\Omega^2+\nu^2}
 \operatorname{Re}(\mathrm i\nu-\delta)^{-1}
 &=-\frac{\Omega^2\delta}{\pi}\int_0^\infty\frac{\mathrm d\nu}{(\Omega^2+\nu^2)(\nu^2+\delta^2)}\\
 &=-\frac{\operatorname{sgn}(\delta)\,\Omega}{2(\Omega+|\delta|)}
 \ \xrightarrow{\ \partial_\delta\ }\ \frac{\Omega}{2(\Omega+|\delta|)^2}.
 \end{aligned}
\]

\emph{Radial variables.} On the transition region $\Delta_{\bm p}=1$, set $r=|\bm p|<1$ and
$s=|\bm p+\bm q|>1$ (here $s$ is a radius, not the Mellin variable). The map
$\bm p\mapsto-\bm p-\bm q$ takes the other transition region onto this one, with an equal
contribution. With $\delta_{\bm p}$ of main-text Eq.~\eqref{eq:Xdef} and $\Sigma_x$ of
Eq.~\eqref{S-app-eq:definitions-u6},
\begin{equation}
 \delta_{\bm p}=\frac{s^2-r^2}{2},\qquad
 \Sigma_x(s)-\Sigma_x(r)=\frac{q_{\rm TF}^2}{2}[f_L(r)-f_L(s)],
 \label{S-app-eq:selfenergyreduction-u2}
\end{equation}
and, with $\theta$ the angle between $\bm p$ and $\bm q$,
\[
 s^2=r^2+q^2+2rq\cos\theta,\qquad
 \mathrm d^3p=2\pi r^2\,\mathrm dr\,\mathrm d\!\cos\theta=\frac{2\pi}{q}\,rs\,\mathrm dr\,\mathrm ds,\qquad
 \mathrm d^3q=4\pi q^2\,\mathrm dq,
\]
give
\begin{equation}
 \frac{\mathrm d^3q\,\mathrm d^3p}{(2\pi)^3}
 =\frac{r\,s\,q}{\pi}\,\mathrm dr\,\mathrm ds\,\mathrm dq,
 \qquad s-r<q<s+r.
 \label{S-app-eq:selfenergyreduction-u3}
\end{equation}

\emph{Insertion.} Insert $f=b_\Lambda=q^{-2}\Omega^2/(\Omega^2+\nu^2)$ with $\Omega=\Lambda q$
and \eqref{S-app-eq:PiSigma} into \eqref{eq:KT-sigma}. By \eqref{eq:sigma-frequency-main},
with $\delta_{\bm p}>0$ on $\Delta_{\bm p}=1$, a factor $2$ for the two regions and $2C=2\pi^2q_{\rm TF}^2$,
\[
 \begin{aligned}
 \int_0^\infty\frac{\mathrm d\nu}{\pi}\,b_\Lambda\operatorname{Re}(\mathrm i\nu-\delta_{\bm p})^{-2}
 &=\frac{\Lambda}{2q(\Lambda q+\delta_{\bm p})^2},\\
 \mathcal K_{\Pi_{\rm SE}^{(1)}}[b_\Lambda]
 &=-\frac3{\pi^3q_{\rm TF}^4}\cdot2\cdot2\pi^2q_{\rm TF}^2
 \int_{\Delta_{\bm p}=1}\frac{\mathrm d^3q\,\mathrm d^3p}{(2\pi)^3}\,
 [\Sigma_x(s)-\Sigma_x(r)]\,\frac{\Lambda}{2q(\Lambda q+\delta_{\bm p})^2}.
 \end{aligned}
\]
In the radial variables the prefactor is
\[
 -\frac3{\pi^3q_{\rm TF}^4}\cdot4\pi^2q_{\rm TF}^2\cdot\frac{q_{\rm TF}^2}2\cdot\frac1\pi\cdot\frac\Lambda2=-\frac{3\Lambda}{\pi^2}.
\]
The inner $q$ integral is elementary, and the factorization
$\Lambda(s\mp r)+\tfrac12(s^2-r^2)=(s\mp r)[\Lambda+\tfrac12(s\pm r)]$ gives its closed form,
\[
 \int_{s-r}^{s+r}\frac{\mathrm dq}{[\Lambda q+(s^2-r^2)/2]^2}
 =\frac1\Lambda\Bigl[\frac1{\Lambda(s-r)+(s^2-r^2)/2}-\frac1{\Lambda(s+r)+(s^2-r^2)/2}\Bigr].
\]
These give the form below, whose last line is the inner $q$ integral,
\begin{equation}
 \begin{aligned}
 \mathcal K_{\Pi_{\rm SE}^{(1)}}[b_\Lambda]
 &=-\frac{3\Lambda}{\pi^2}\int_0^1 r\,\mathrm dr
 \int_1^\infty s\,\mathrm ds\,[f_L(r)-f_L(s)]\\
 &\qquad\times\int_{s-r}^{s+r}
 \frac{\mathrm dq}{[\Lambda q+(s^2-r^2)/2]^2},\\
 \int_{s-r}^{s+r}
 \frac{\mathrm dq}{[\Lambda q+(s^2-r^2)/2]^2}
 &=\frac{2r}{(s^2-r^2)[(\Lambda+s/2)^2-r^2/4]},
 \end{aligned}
 \label{S-app-eq:selfenergyreduction-u1}
\end{equation}
and hence
\begin{equation}
 \mathcal K_{\Pi_{\rm SE}^{(1)}}[b_\Lambda]
 =-\frac{6\Lambda}{\pi^2}\int_0^1\!r^2\,\mathrm dr
 \int_1^\infty\!s\,\mathrm ds\,
 \frac{f_L(r)-f_L(s)}
 {(s^2-r^2)[(\Lambda+s/2)^2-r^2/4]}.
 \label{S-app-eq:sigma-double}
\end{equation}
Before the $q$ integration, the integrand of the first line of
\eqref{S-app-eq:selfenergyreduction-u1} is the radial kernel of the
self-energy pair (Sec.~\ref{S-app-sec:sigma-lowdensity}).
Section~\ref{S-sm-sec:verification} evaluates \eqref{S-app-eq:sigma-double} numerically.

\emph{Sign.} Apart from the prefactor $-6\Lambda/\pi^2$, the integrand of
\eqref{S-app-eq:sigma-double} is positive because $f_L$ is decreasing,
\[
 f_L(r)-f_L(s)>0,\qquad (\Lambda+s/2)^2>r^2/4\quad(r<1<s)
 \qquad\Longrightarrow\qquad
 \mathcal K_{\Pi^{(1)}_{\rm SE}}[b_\Lambda]<0\quad(\Lambda>0).
\]
The sign already holds at fixed momentum transfer $q$ of the line. Apart from the prefactor
$-3\Lambda/\pi^2$, the $r$, $s$ integrand in the first line of
\eqref{S-app-eq:selfenergyreduction-u1} is positive at fixed $q$. This pointwise sign is
needed because the weight $a_\lambda(q,\Lambda)$ of main-text
Eq.~\eqref{eq:block-superposition} depends on $q$. With
$\mathcal K_{\Pi^{(1)}_{\rm SE}}[v]=0$ [Eq.~\eqref{S-app-eq:selfenergyreduction-u7}] and
$a_\lambda\ge0$, it gives
\[
 \mathcal K_{\Pi^{(1)}_{\rm SE}}[W_\lambda]
 =\lambda\,\mathcal K_{\Pi^{(1)}_{\rm SE}}[v]
 -\mathcal K_{\Pi^{(1)}_{\rm SE}}\Bigl[\int_0^\infty a_\lambda b_\Lambda\,\mathrm d\Lambda\Bigr]>0,
\]
consistent with $e_{\rm SE}=-A_1/q_*>0$ in Eq.~\eqref{S-app-eq:sigma-lowdensity-cancellation}.

\subsection{Bare limit}
The bare-line energy of the self-energy pair vanishes because, for $\delta\neq0$,
\begin{equation}
 \int_0^\infty\frac{\mathrm d\nu}{\pi}
 \operatorname{Re}(\mathrm i\nu-\delta)^{-2}=0.
 \label{S-app-eq:selfenergyreduction-u7}
\end{equation}
This holds because (see after Eq.~\eqref{S-app-eq:spectralconvergence-u3})
\[
 \operatorname{Re}(\mathrm i\nu-\delta)^{-2}=\frac{\delta^2-\nu^2}{(\nu^2+\delta^2)^2}
 =\partial_\nu\frac{\nu}{\nu^2+\delta^2},\qquad
 \frac{\nu}{\nu^2+\delta^2}\Big|_{\nu=0}=0=\lim_{\nu\to\infty}\frac{\nu}{\nu^2+\delta^2}.
\]
Integrating over frequency first, which the absolute bounds of
Sec.~\ref{S-app-app:selfenergy-blocks} permit, gives $\mathcal K_{\Pi_{\rm SE}^{(1)}}[v]=0$.
Since $0\le b_\Lambda\le v$, dominated convergence gives
$\mathcal K_{\Pi_{\rm SE}^{(1)}}[b_\Lambda]\to0$ as $\Lambda\to\infty$, in agreement with
\eqref{S-app-eq:sigma-double}.

\section{Frequency-integrated shell bounds}\label{S-app-app:shell-integrated}\label{S-sm-sec:shell-bounds}
This section gives the bounds behind the frequency-integrated form of main-text Eq.~\eqref{eq:Az},
\[
 \int_0^\infty\frac{\mathrm dz}{\pi}\,h(z)X(q,qz)=\mathcal A\log(1/q)+O(1),
\]
with constants independent of $q$ and $\Lambda$, and the remainder $O(r_s|\log r_s|)$ of
main-text Eq.~\eqref{eq:Clog2}. The frequency integral is taken first. It turns the
electron--hole factors into a weight with bounded derivatives.

\emph{Weights.} With $a=\delta_{\bm p}/q$, $b=\delta_{\bm k}/q$, $\alpha=|a|$, $\beta=|b|$,
\begin{equation}
 M_\Lambda(a,b):=ab\Lambda^2J(\Lambda;a,b)=
 \begin{cases}
 \dfrac{\alpha\beta\Lambda}{2(\Lambda+\alpha)(\Lambda+\beta)},&ab>0,\\[6pt]
 \dfrac{\alpha\beta}{2(\alpha+\beta)}
 \left(\dfrac{\Lambda}{\Lambda+\alpha}+\dfrac{\Lambda}{\Lambda+\beta}\right),&ab<0.
 \end{cases}
 \label{S-app-app-eq:weighted-shell}
\end{equation}
With $\mathcal R$ of \eqref{S-app-eq:Rdef} and
$h(z)=1-z\arctan(1/z)=\int_0^1\Lambda^2\,\mathrm d\Lambda/(\Lambda^2+z^2)$ [main-text Eq.~\eqref{eq:hz}],
\[
 M_\Lambda(a,b)=\int_0^\infty\frac{\mathrm dz}{\pi}\frac{\Lambda^2}{\Lambda^2+z^2}\,ab\,\mathcal R(a,b;z),\qquad
 M_h(a,b)=\int_0^\infty\frac{\mathrm dz}{\pi}\,h(z)\,ab\,\mathcal R(a,b;z)=\int_0^1M_\Lambda\,\mathrm d\Lambda .
\]
Extended by continuity, both weights vanish on the axes. From $\alpha\beta/(\alpha+\beta)\le\alpha$ and
$\alpha\Lambda/(\alpha+\Lambda)^2\le1/4$, uniformly in $\Lambda>0$,
\[
 \begin{gathered}
 |\partial_aM|,\ |\partial_bM|\le\tfrac98\ \text{(one-sided)},\qquad
 |M(a,b)|\le\tfrac98\min(|a|,|b|),\\
 0\le M(c,c)\le\tfrac18|c|,\qquad M=M_\Lambda,M_h .
 \end{gathered}
\]
Since $(\mathrm iqz-\delta_{\bm p})^{-1}=q^{-1}(\mathrm iz-a_p)^{-1}$,
$\operatorname{Re}[(\mathrm iz-a)(\mathrm iz-b)]^{-1}=\mathcal R(a,b;z)$ and
$\Delta_{\bm p}=\operatorname{sgn}a_p$ on the transition regions, Eq.~\eqref{S-app-app-eq:weighted-shell} and the
representations of $M_\Lambda$, $M_h$ above give, for either weight $w=\Lambda^2/(\Lambda^2+z^2)$ or $w=h$,
\[
 \int_0^\infty\frac{\mathrm dz}{\pi}\,w(z)X(q,qz)
 =\frac1{q^2}\int\mathrm d^3p\,\mathrm d^3k\,
 \frac{|\Delta_{\bm p}\Delta_{\bm k}|\,M_w(a_p,a_k)}{|a_p|\,|a_k|\,|\bm p-\bm k|^2},\qquad a_p=\delta_{\bm p}/q,\ a_k=\delta_{\bm k}/q .
\]

\emph{End caps.} At fixed $\bm q$ and $a=\bm p\cdot\hat{\bm q}+q/2$,
\[
 \Delta_{\bm p}=1:\ 1-(a+q/2)^2<p_\perp^2<1-(a-q/2)^2,\qquad
 \mathrm d^3p=2\pi q\,m_q(a)\,\mathrm da,
\]
with $m_q$ of \eqref{S-app-app-eq:lindhard-shell}, and $\bm p\mapsto-\bm p-\bm q$ maps it onto
$\Delta_{\bm p}=-1$. The thin-shell parameterization of
Sec.~\ref{S-app-sec:shell-origin} ($a\to c_p=\hat{\bm p}\cdot\hat{\bm q}$, density $2\pi q\,c_p$
on $0<c_p<1$) agrees with $2\pi q\,m_q(a)$ for $a\le1-q/2$. They differ only in the
\emph{end cap} $1-q/2<a\le1+q/2$, the polar cap about $\pm\hat{\bm q}$. There, for $q\le1$
and $|\bm p|\le2$,
\[
\begin{aligned}
 &\text{volume: }2\pi q\int_{1-q/2}^{1+q/2}m_q\,\mathrm da=\pi q^2[1+O(q)],\qquad
 \text{weight: }\frac{|M|}{|a_p||a_k|}\le\frac9{8|a_p|}\le\frac94\quad(a_p\ge\tfrac12),\\
 &\text{Coulomb: }\int_{1-q}^{1+q}\!\rho^2\,\mathrm d\rho\int\frac{\mathrm d\Omega_{\bm k}}{|\bm p-\rho\hat{\bm k}|^2}
 =\int_{1-q}^{1+q}\frac{2\pi\rho}{|\bm p|}\log\frac{\rho+|\bm p|}{|\rho-|\bm p||}\,\mathrm d\rho=O\bigl(q\log(e/q)\bigr),\\
 &\text{end caps: }\tfrac94\,q^{-2}\cdot\pi q^2[1+O(q)]\cdot O\bigl(q\log(e/q)\bigr)=O\bigl(q\log(e/q)\bigr),
\end{aligned}
\]
for either sign of $\Delta_{\bm p}$ and either momentum. The thin-shell end cap has a volume of the
same order. With the chord distance
$d=|\hat{\bm p}-\hat{\bm k}|$, angular measure $2\pi d\,\mathrm dd$ and the other shell of
radial width $O(q)$, the thin-shell Coulomb factors of \eqref{S-app-eq:shell} obey the same bound,
\[
 d^{-2}\ \text{on}\ d>\epsilon:\ O\bigl(q\log(1/\epsilon)\bigr),\qquad
 [d^2+q^2(x-y)^2]^{-1}\ \text{on the disk}:\ O\bigl(q\log(e/q)\bigr),
\]
the second after the radial average of Sec.~\ref{S-app-sec:shell-origin}.

\emph{Diagonal.} Subtracting the coincident-direction value leaves $O(d)$, and $\mathrm d\Omega=\mathrm d\Omega_p$. With
$C_{\rm b}$ independent of $q$ and $\Lambda$ (Notation) and the radial average of
Sec.~\ref{S-app-sec:shell-origin} [Eq.~\eqref{S-app-eq:shell}],
\[
 \int_0^1 O(d)\,\frac{2\pi d\,\mathrm dd}{d^2}\le C_{\rm b}\int_0^1\mathrm dd,\qquad
 \text{radial average}=\log\frac1q-\log|c|+\frac32\quad(c=c_p).
\]
The diagonal bound $M(c,c)\le|c|/8$ controls the grazing directions, where the width of the
transition regions vanishes.

\emph{Curvature correction.} The first line of Eq.~\eqref{S-app-eq:shell} uses $d^{-2}$, and its second
line uses the flat form $[d^2+q^2(x-y)^2]^{-1}$ on the disk. With $\bm p=(1+qx)\hat{\bm p}$, $\bm k=(1+qy)\hat{\bm k}$,
$|x|,|y|\le\frac98$ on the transition regions and $q\le\frac14$,
\[
 \begin{gathered}
 |\bm p-\bm k|^2=q^2(x-y)^2+(1+qx)(1+qy)\,d^2,\\
 |(1+qx)(1+qy)-1|\le3q,\qquad (1+qx)(1+qy)\ge\tfrac12,\qquad |\bm p-\bm k|^2\ge\tfrac12d^2,
 \end{gathered}
\]
and therefore
\[
 \begin{gathered}
 \Bigl||\bm p-\bm k|^{-2}-d^{-2}\Bigr|\le\frac{2q^2(x-y)^2}{d^4}+\frac{6q}{d^2}\le\frac{24\,q}{d^3}\quad(q\le d\le2),\\
 \Bigl||\bm p-\bm k|^{-2}-[d^2+q^2(x-y)^2]^{-1}\Bigr|\le\frac{6\,q}{d^2}.
 \end{gathered}
\]
After the radial averages the weight is at most $\frac98$, and the correction is
\[
\begin{aligned}
 \epsilon<d\le2:&\quad C_{\rm b}\,q\int_\epsilon^2\frac{\mathrm dd}{d^2}=O(q/\epsilon)\quad\text{(first bound)},\\
 q<d<\epsilon:&\quad C_{\rm b}\,q\int_q^\epsilon\frac{\mathrm dd}{d}=O\bigl(q\log(\epsilon/q)\bigr)\quad\text{(second bound, disk)},\\
 d<q:&\quad \int_0^qd\,\mathrm dd\iint\mathrm dx\,\mathrm dy\,[q^2(x-y)^2+d^2/2]^{-1}=O(1).
\end{aligned}
\]
For $d<q$ both the exact and the approximate Coulomb factors are controlled by the radial
average of Sec.~\ref{S-app-sec:shell-origin}.
Replacing $a_p$ by $c_p$, with $|a_p-c_p|\le q$,
\[
 |\delta M|\le\tfrac98\bigl(|a_p-c_p|+|a_k-c_k|\bigr)=O(q),\qquad
 \frac{\delta\int\mathrm dx/|a_p|}{\int\mathrm dx/|a_p|}=O(q/|c_p|),
\]
with the radial integral over the exact interval. Together these changes alter the weight by $O(q)$
uniformly. For $|c_p|\lesssim q$ the bound $|M|\le\frac98\min(|a|,|b|)$ is used instead. These bounds give
\begin{equation}
 \begin{gathered}
 \int_0^\infty\frac{\mathrm{d}z}{\pi}\,h(z)X(q,qz)
 =2\pi\log(1/q)\int\mathrm{d}\Omega\,M_h(c,c)+O(1),\\
 2\pi\int\mathrm d\Omega\,M_h(c,c)=\int_0^\infty\frac{\mathrm dz}{\pi}\,h\,a_x=\mathcal A,
 \end{gathered}
 \label{S-app-app-eq:integrated-shell}
\end{equation}
with $a_x$ of \eqref{S-app-eq:Az} and $\mathcal A$ of \eqref{S-app-eq:calA}.

\emph{Absolute frequency bound.} The spectral interchange also needs, by Cauchy--Schwarz,
\[
 \int_0^\infty\frac{\mathrm dz}{\pi}\,|ab\mathcal R(a,b;z)|
 \le\frac{|ab|}\pi\Bigl(\int_0^\infty\frac{\mathrm dz}{z^2+a^2}\Bigr)^{1/2}\Bigl(\int_0^\infty\frac{\mathrm dz}{z^2+b^2}\Bigr)^{1/2}
 =\frac{\sqrt{|ab|}}2 .
\]
Each shell has radial width $q|a|$, so the shell integrand is $ab\mathcal R$, and the Coulomb
integral between the two shells gives
\[
 \int_0^\infty|X(q,qz)|\,\mathrm dz=O\bigl(\log(e/q)\bigr).
\]

\emph{Lindhard function.} For $0<q<1$,
$\mathcal L(q,qz)=\frac1{4\pi}\int\mathrm d^3p\,\Delta_{\bm p}\delta_{\bm p}/(\delta_{\bm p}^2+q^2z^2)$
(Sec.~\ref{S-app-sec:response}). The region $\Delta_{\bm p}=-1$ contributes as much as the region
$\Delta_{\bm p}=1$, by $\bm p\mapsto-\bm p-\bm q$, and $\mathrm d^3p=2\pi q\,m_q(a)\,\mathrm da$ there, so
\begin{equation}
 \begin{gathered}
 \mathcal L(q,qz)=\int_0^{1+q/2}\frac{a m_q(a)}{a^2+z^2}\,\mathrm{d}a,
 \\
 m_q(a)=
 \begin{cases}
 a,&0\le a\le1-q/2,\\
 [1-(a-q/2)^2]/(2q),&1-q/2<a\le1+q/2.
 \end{cases}
 \end{gathered}
 \label{S-app-app-eq:lindhard-shell}
\end{equation}
With $h(z)=\int_0^1a^2\,\mathrm da/(a^2+z^2)$ (main-text Sec.~\ref{sec:units}) and
$\mathbf1_{a<1}$ the indicator of $a<1$,
\[
 \mathcal L(q,qz)-h(z)=\int_{1-q/2}^{1+q/2}\bigl[m_q(a)-a\mathbf1_{a<1}\bigr]\frac{a\,\mathrm da}{a^2+z^2},\qquad
 \frac a{a^2+z^2}\le\frac{C_{\rm b}}{1+z^2}\quad\bigl(\tfrac12\le a\le\tfrac32\bigr),
\]
and therefore
\begin{equation}
 \int|m_q(a)-a\mathbf1_{a<1}|\,\mathrm da=q/4,\qquad
 |\mathcal L(q,qz)-h(z)|\le\frac{C_{\rm b}q}{1+z^2}.
 \label{S-app-eq:highdensity-u3}
\end{equation}
With the absolute frequency bound, replacing $h$ by $\mathcal L$ changes the
frequency-integrated shell expansion by
\[
 \Bigl|\int_0^\infty\frac{\mathrm dz}{\pi}\,[\mathcal L(q,qz)-h(z)]\,X(q,qz)\Bigr|
 \le\sup_z|\mathcal L-h|\int_0^\infty\frac{\mathrm dz}{\pi}\,|X(q,qz)|
 =O\bigl(q\log(e/q)\bigr),
\]
which in the window integral \eqref{S-app-eq:window} costs
$q_{\rm TF}^2\int_{q_{\rm TF}}^1(\mathrm dq/q)\,O(q\log(e/q))=O(q_{\rm TF}^2)=O(r_s)$.

\emph{Two formulas used elsewhere.} The signed occupation integral $Z(q)$ of
\eqref{S-app-eq:highdensity-u6} has the Fourier representation
\begin{equation}
 Z(q)=\frac{2\pi^2}{(2\pi)^3}\int\mathrm{d}^3t\,
 \frac{|\widehat\chi_{\mathbb B}(t)|^2}{t}\,2[1-\cos(qt_z)],\qquad
 \widehat\chi_{\mathbb B}(t)=4\pi\frac{\sin t-t\cos t}{t^3},
 \label{S-app-app-eq:ZFourier}
\end{equation}
where $\bm t$ is conjugate to $\bm p$ (not the scaled momentum $t=q/\varepsilon$),
$t=|\bm t|$, $t_z=\bm t\cdot\hat{\bm q}$ and $\chi_{\mathbb B}$ is the indicator of the
unit Fermi ball. It follows from the convolution theorem, where $\to$ denotes the Fourier transform,
\[
 \begin{gathered}
 \Delta_{\bm p}=\chi_{\mathbb B}(\bm p)-\chi_{\mathbb B}(\bm p+\bm q)\ \to\
 \widehat\chi_{\mathbb B}(t)\bigl(1-\mathrm e^{\mathrm i\bm t\cdot\bm q}\bigr),\qquad
 |1-\mathrm e^{\mathrm i\bm t\cdot\bm q}|^2=2[1-\cos(qt_z)],\\
 \int\mathrm d^3x\,{\rm e}^{-\mathrm i\bm t\cdot\bm x}|\bm x|^{-2}=\frac{2\pi^2}t .
 \end{gathered}
\]
The frequency integral behind $\kappa(\Lambda)$ in \eqref{S-app-eq:kappa}, with $c$ the
diagonal excitation energy $c_p$ there, is
\begin{equation}
 \int_0^\infty\frac{\Lambda^2(c^2-z^2)\,\mathrm{d}z}
 {(\Lambda^2+z^2)(c^2+z^2)^2}
 =\frac{\pi\Lambda}{2(\Lambda+c)^2}=\pi\Lambda^2J(\Lambda;c,c),\qquad c>0,
 \label{S-app-app-eq:kappa-int}
\end{equation}
with $J$ of \eqref{S-app-app-eq:J-closed}.
\section{Low-density expansion with dynamical screening}\label{S-app-sec:dispersion}
This section derives the large-momentum kernel \eqref{eq:largeq-block}, the
screening decomposition \eqref{eq:pole-dispersion} and the coefficient
\eqref{eq:A1-dispersion}, and repeats the calculation for the self-energy pair,
\[
 e_{\rm kite}=\frac{A_1}{q_*}+O(q_*^{-3/2}),\qquad
 e_{\rm kite}+e_{\rm SE}=O(q_*^{-3/2}).
\]
The radial kernel $e_\Lambda(q)$ is defined in \eqref{eq:XLambda}. As in the main text,
the relevant momentum transfers are $q\sim q_*\gg1$.

\subsection{Large-momentum expansion of the radial kernel \texorpdfstring{$e_\Lambda$}{e(Lambda)}}\label{S-app-sec:largeq-kernel}
\emph{Coulomb integrals.} For $q>2$ each same-ball region of main-text
Sec.~\ref{sec:radial-kernels} uses
$\int_{\mathbb B}|\bm p-\bm k|^{-2}\mathrm d^3k=4\pi f_L(p)$ and
\begin{equation}
 \int_0^1r^2f_L(r)\,\mathrm dr
 =\frac16+\frac14\int_0^1r(1-r^2)\log\frac{1+r}{1-r}\,\mathrm dr
 =\frac16+\frac1{12}=\frac14,
 \label{S-app-eq:fL-moment}
\end{equation}
hence
\[
 \begin{aligned}
 &\text{same ball: }\iint_{\mathbb B^2}\frac{\mathrm d^3p\,\mathrm d^3k}{|\bm p-\bm k|^2}=16\pi^2\cdot\frac14=4\pi^2,\\
 &\text{different balls: }|\bm p-\bm k|\ge q-2,\qquad
 \iint\frac{\mathrm d^3p\,\mathrm d^3k}{|\bm p-\bm k|^2}\le\Bigl(\frac{4\pi}3\Bigr)^2\frac1{(q-2)^2}.
 \end{aligned}
\]

\emph{Frequency integral.} With $\nu=q^2u$, $\Lambda=qw_\Lambda$,
$\delta=\pm q^2/2+O(q)$ and $p_{\rm loc}(w)=2\pi w/(1+2w)^2$, a same-ball region at
$\delta_{\bm p}=\delta_{\bm k}=\pm q^2/2$ has
\[
 q^4\operatorname{Re}\frac1{(\mathrm i\nu-\delta_{\bm p})(\mathrm i\nu-\delta_{\bm k})}
 =\frac{1/4-u^2}{(u^2+1/4)^2},\qquad
 \frac{\Lambda^2}{\Lambda^2+\nu^2/q^2}=\frac{w_\Lambda^2}{w_\Lambda^2+u^2}.
\]
The frequency integral follows by differentiating an elementary Lorentzian integral:
\begin{equation}
\begin{gathered}
 \frac{1/4-u^2}{(u^2+1/4)^2}
 =-\left.\partial_b\frac{b}{u^2+b^2}\right|_{b=1/2},\qquad
 \int_0^\infty\frac{\mathrm du}{(u^2+w_\Lambda^2)(u^2+b^2)}
 =\frac\pi{2w_\Lambda b(w_\Lambda+b)},\\
 \int_0^\infty\frac{w_\Lambda^2(1/4-u^2)}{(u^2+w_\Lambda^2)(u^2+1/4)^2}\,\mathrm du=p_{\rm loc}(w_\Lambda),
 \qquad 32\pi^2A_0=3/\pi^3.
\end{gathered}
 \label{S-app-eq:lowdensity-u2}
\end{equation}
The third identity is
$-w_\Lambda^2\,\partial_b[\pi/(2w_\Lambda(w_\Lambda+b))]_{b=1/2}=\pi w_\Lambda/[2(w_\Lambda+1/2)^2]=p_{\rm loc}(w_\Lambda)$.
The prefactor $3/(\pi^3q^2)$ of \eqref{eq:largeq-block} collects the factor
$4\pi A_0$ of \eqref{eq:XLambda}, the Coulomb integrals of the two same-ball
regions and $\mathrm d\nu=q^2\,\mathrm du$,
\[
 e_\Lambda(q)\simeq4\pi A_0\cdot2\cdot4\pi^2\int_0^\infty\frac{q^2\,\mathrm du}{\pi}\,
 \frac{w_\Lambda^2}{w_\Lambda^2+u^2}\,\frac{1/4-u^2}{q^4(u^2+1/4)^2}
 =\frac{32\pi^2A_0}{q^2}\,p_{\rm loc}(w_\Lambda)=\frac{3}{\pi^3q^2}\,p_{\rm loc}(w_\Lambda).
\]

\emph{Remainders.} The $O(q)$ correction to $\delta$ gives $O(q^{-3}/(1+w_\Lambda))$.
The different-ball regions and the bare term $e_\infty$ are $O(q^{-4})$. These
estimates give \eqref{eq:largeq-block} uniformly for $w_\Lambda$ bounded away from zero.

\subsection{Large-momentum approximation of RPA screening}\label{S-app-sec:largeq-screening}
\emph{Lindhard function.} Let $q_*^4=4q_{\rm TF}^2/3$ as in main-text
Eq.~\eqref{eq:qstar}, $u=\nu/q^2$ and
$\hat c=(q_*/q)^4/(1+4u^2)$, the leading large-$q$ screening ratio of main-text Eq.~\eqref{eq:pole-dispersion}.
The Lindhard function has the large-momentum expansion
\begin{equation}
  \mathcal L(q,q^2u)=\frac{4}{3q^{2}(1+4u^{2})}\Big[1+\frac{g_4(u)}{q^{2}}+O(q^{-4})\Big],
  \label{S-app-eq:Glargeq}
\end{equation}
with $g_4(u)=-4(12u^2-1)/[5(1+4u^2)^2]$. It is obtained by expanding the $\log$
function in powers of $1/q$ at fixed $u$ and integrating over the radius $k$ of the
occupied momentum in the Lindhard function of Ref.~\cite{SI-FetterWalecka1971} at
$\omega=\mathrm i\nu$,
\begin{equation}
 \mathcal L(q,q^2u)=\frac{1}{2q}\int_0^1 k\log
 \frac{\varpi+k/q+k^2/q^2}{\varpi-k/q+k^2/q^2}\,\mathrm{d}k,
 \qquad \varpi=u^2+\frac14.
 \label{S-app-eq:lowdensity-u3}
\end{equation}
At leading order
\[
 \begin{gathered}
 \log\frac{\varpi+k/q+k^2/q^2}{\varpi-k/q+k^2/q^2}=\frac{2k}{q\varpi}+O(q^{-3}),\\
 \mathcal L(q,q^2u)=\frac1{3q^2\varpi}+O(q^{-4})=\frac4{3q^2(1+4u^2)}+O(q^{-4}).
 \end{gathered}
\]
The expansion \eqref{S-app-eq:Glargeq} is uniform in $u$ because $\varpi\ge1/4$.

\emph{Single Lorentzian.} Since $c=-v\Pi_0=\hat c[1+O(q^{-2})]$,
\begin{equation}
  \frac1{1+\lambda\hat c}=\frac{1+4u^{2}}{1+4u^{2}+\lambda(q_*/q)^{4}}
  =1-\frac{\lambda q_*^{4}/4}{\nu^{2}+\omega_\lambda(q)^{2}},
  \qquad \omega_\lambda(q)^{2}=\frac{q^{4}+\lambda q_*^{4}}{4},
 \label{S-app-eq:lowdensity-u1}
\end{equation}
and, with the remainder $\varsigma_\lambda=1/(1+\lambda c)-1/(1+\lambda\hat c)$ of
the single-Lorentzian form, the inverse dielectric function is
\begin{equation}
\begin{gathered}
  \frac{1}{\epsilon_\lambda(q,\mathrm{i}\nu)}
  =1-\theta_\lambda(q)\,\frac{\Lambda_\lambda(q)^{2}}{\Lambda_\lambda(q)^{2}+z^{2}}
  +\varsigma_\lambda(q,u),\\
  \varsigma_\lambda=O\!\left(\frac{1}{q^2}\frac{\lambda\hat c}{(1+\lambda\hat c)^2}\right),
  \qquad
  \theta_\lambda(q)=\frac{\lambda q_*^{4}}{q^{4}+\lambda q_*^{4}},\\
  \Lambda_\lambda(q)=\frac q2\sqrt{1+\lambda\,(q_*/q)^{4}}.
\end{gathered}
  \label{S-app-eq:pole-dispersion}
\end{equation}
Multiplying by $\lambda v$ gives \eqref{eq:pole-dispersion} with
\begin{equation}
 \varsigma_\lambda=-\frac{\lambda(c-\hat c)}{(1+\lambda c)(1+\lambda\hat c)},
 \label{S-app-eq:lowdensity-u4}
\end{equation}
and the single-Lorentzian form of the retarded line has poles at
$\omega=\pm\omega_\lambda(q)-\mathrm i0$. The two Lorentzian forms agree,
\[
 \frac{\lambda q_*^4/4}{\nu^2+\omega_\lambda^2}=\theta_\lambda\,\frac{\Lambda_\lambda^2}{\Lambda_\lambda^2+z^2},\qquad
 \omega_\lambda=q\Lambda_\lambda,\quad z=\nu/q.
\]
At fixed $\lambda>0$,
\[
 q\gg q_*:\ \ \omega_\lambda\sim q^2/2,\quad \theta_\lambda\to0,\qquad
 1\ll q\ll q_*:\ \ \omega_\lambda\sim\sqrt{\lambda q_{\rm TF}^2/3},\quad \theta_\lambda\to1.
\]

\subsection{Low-density coefficient and remainder estimates}\label{S-app-sec:lowdensity}
Inserting main-text Eq.~\eqref{eq:pole-dispersion} into
$e_{\rm kite}=2\int_0^1\mathcal K[W_\lambda]\,\mathrm d\lambda$ and collecting in
the remainder the terms bounded in Sec.~\ref{S-sm-sec:kite-remainders} gives (main
text, before Eq.~\eqref{eq:A1-dispersion})
\begin{equation}
\begin{aligned}
 e_{\rm kite}(r_s)&=-2\int_0^1\lambda\,\mathrm{d}\lambda\int_0^\infty\mathrm{d}q\,
 \theta_\lambda(q)[e_{\Lambda_\lambda(q)}(q)-e_\infty(q)]+O(q_*^{-2})\\
 &=\frac{A_1}{q_*}+O(q_*^{-3/2}),
\end{aligned}
  \label{S-app-eq:A1-split}
\end{equation}
with $A_1$ of \eqref{eq:A1-dispersion}. Its closed form follows from three substitutions.

\emph{First substitution.} With $x=q/q_*$ and $w_{\Lambda,\lambda}(x)$ of main-text
Eq.~\eqref{eq:A1-dispersion},
\[
 x=\lambda^{1/4}y:\qquad
 w_{\Lambda,\lambda}=\tfrac12\sqrt{1+y^{-4}},\qquad
 \frac{\lambda}{x^4+\lambda}=\frac1{y^4+1},\qquad
 \frac{\mathrm dx}{x^2}=\lambda^{-1/4}\frac{\mathrm dy}{y^2}.
\]
The coupling integral separates as $\int_0^1\lambda^{3/4}\,\mathrm d\lambda=4/7$, and
\begin{equation}
 A_1=-\frac{48}{7\pi^2}\int_0^\infty
 \frac{\mathrm dy}{2\sqrt{y^4+1}\,\bigl(y^2+\sqrt{y^4+1}\bigr)^2}.
 \label{S-app-eq:A1-y}
\end{equation}

\emph{Second and third substitutions.} The second is $y^2=\tau/\sqrt{1-\tau^2}$ with
$0<\tau<1$. The third, $\sigma=2\tau/(1+\tau)$, removes the factor $(1+\tau)$:
\[
\begin{gathered}
 y^2=\frac{\tau}{\sqrt{1-\tau^2}}:\quad
 \sqrt{y^4+1}=\frac1{\sqrt{1-\tau^2}},\quad
 y^2+\sqrt{y^4+1}=\frac{1+\tau}{\sqrt{1-\tau^2}},\quad
 \mathrm dy=\frac{\mathrm d\tau}{2\tau^{1/2}(1-\tau^2)^{5/4}},\\
 \sigma=\frac{2\tau}{1+\tau}:\quad
 1+\tau=\frac2{2-\sigma},\quad
 1-\tau=\frac{2(1-\sigma)}{2-\sigma},\quad
 \mathrm d\tau=\frac{2\,\mathrm d\sigma}{(2-\sigma)^2},
\end{gathered}
\]
hence
\[
\begin{aligned}
 \int_0^\infty\frac{\mathrm dy}{2\sqrt{y^4+1}\,(y^2+\sqrt{y^4+1})^2}
 &=\frac14\int_0^1\frac{(1-\tau)^{1/4}\,\mathrm d\tau}{\tau^{1/2}(1+\tau)^{7/4}}\\
 &=\frac1{4\sqrt2}\int_0^1\sigma^{-1/2}(1-\sigma)^{1/4}\,\mathrm d\sigma
 =\frac{\mathrm B(1/2,5/4)}{4\sqrt2},
\end{aligned}
\]
where $\mathrm B$ is the Euler beta function. Since $\Gamma(5/4)=\frac14\Gamma(1/4)$ and
$\Gamma(7/4)=\frac34\Gamma(3/4)$, one has
$\mathrm B(1/2,5/4)=\sqrt\pi\,\Gamma(1/4)/[3\Gamma(3/4)]$, and
\begin{equation}
 A_1=-\frac{12}{7\sqrt2\,\pi^2}\,\mathrm B\!\left(\tfrac12,\tfrac54\right)
 =-\frac{4}{7\sqrt2\,\pi^{3/2}}\frac{\Gamma(1/4)}{\Gamma(3/4)}
 =-0.2146937609\ldots.
 \label{S-app-eq:A1-closed}
\end{equation}

\emph{Static line.} In a same-ball region, where $a=b=q^2/2$, the static frequency
integral $\mathcal W(a,b)$ of main-text Eq.~\eqref{eq:Wdef} and the Lorentzian one
of main-text Eq.~\eqref{eq:XLambda} for one pair of excitation energies, with $J$
from \eqref{S-app-app-eq:J-closed}, are
\begin{equation}
 \mathcal W(a,b)=0,\qquad
 (\Lambda q)^2J(\Lambda q;a,b)=\frac{p_{\rm loc}(w_\Lambda)}{\pi q^2}>0.
 \label{S-app-eq:lowdensity-u10}
\end{equation}
The static kernel thus has no $q^{-2}$ term, and the dynamical kernel has the
$q^{-2}$ term ($q_*^{-1}\propto r_s^{-1/4}$). For the static Lindhard line
$v/[1+c(q,0)]$, Eq.~\eqref{S-app-eq:Glargeq} at $u=0$ and main-text
Eq.~\eqref{eq:qstar} give, for $q\gg1$ and $t=q/q_*$,
\[
 c(q,0)=\frac{q_{\rm TF}^2}{q^2}\mathcal L(q,0)\simeq\frac{4q_{\rm TF}^2}{3q^4}=t^{-4},\qquad
 \frac{v}{1+c(q,0)}\simeq v\,\frac{t^4}{1+t^4}.
\]
The region $q\le1$, where the line is bounded by $2q^2v/q_{\rm TF}^2$, and the relative
corrections $O(q^{-2})$ to $c(q,0)$ contribute $O(q_*^{-4})$. With $e_\infty=O(q)$
for $q\le1$ (Sec.~\ref{S-app-sec:shell-origin}) and $e_\infty=O(q^{-4})$ for $q\ge1$
(Sec.~\ref{S-app-sec:largeq-kernel}), the kite energy of the leading-order line is
\[
 \int_0^\infty e_\infty(q)\,\frac{t^4}{1+t^4}\,\mathrm dq:\qquad
 q<q_*:\ O\Bigl(q_*^{-4}\int_0^{q_*}q^4e_\infty\,\mathrm dq\Bigr)=O(q_*^{-3}),\qquad
 q>q_*:\ O(q_*^{-3}).
\]
The static line therefore gives at most $O(q_*^{-3})\propto r_s^{-3/4}$, smaller than
the $q_*^{-1}\propto r_s^{-1/4}$ of the dynamically screened line.

\subsection{Self-energy pair and cancellation of the leading low-density term}
\label{S-app-sec:sigma-lowdensity}
This subsection derives main-text Eq.~\eqref{eq:low-density-response-cancellation}
and the large-momentum self-energy kernel stated in main-text
Sec.~\ref{sec:se-cancellation}. For
$e_{\rm SE}=2\int_0^1\mathcal K_{\Pi_{\rm SE}^{(1)}}[W_\lambda]\,\mathrm d\lambda$,
inserting \eqref{S-app-eq:pole-dispersion} multiplied by $\lambda v$ and using
$\mathcal K_{\Pi_{\rm SE}^{(1)}}[v]=0$ [Eq.~\eqref{S-app-eq:selfenergyreduction-u7}, Sec.~\ref{S-app-sec:sigma-reduction}]
gives
\begin{equation}
 e_{\rm SE}=-2\int_0^1\lambda\,\mathrm d\lambda
 \int_0^\infty\theta_\lambda(q)e_{\Sigma,\Lambda_\lambda(q)}(q)\,\mathrm dq
 +O(q_*^{-2}),
 \label{S-app-eq:sigma-lowdensity-split}
\end{equation}
where $e_{\Sigma,\Lambda}(q)$ is the radial kernel of the self-energy pair, defined by
$\mathcal K_{\Pi_{\rm SE}^{(1)}}[b_\Lambda]=\int_0^\infty e_{\Sigma,\Lambda}(q)\,\mathrm dq$
as in main-text Eq.~\eqref{eq:radial-kernel}. Before the $q$ integration,
\eqref{S-app-eq:selfenergyreduction-u1} gives
\[
 e_{\Sigma,\Lambda}(q)=-\frac{3\Lambda}{\pi^2}\int_{\max(0,1-q)}^1r\,\mathrm dr
 \int_{\max(1,q-r)}^{q+r}s\,\mathrm ds\,
 \frac{f_L(r)-f_L(s)}{[\Lambda q+(s^2-r^2)/2]^2}.
\]

\emph{Large momentum.} For $q>2$ the lower limit of $s$ is $q-r$, and
\[
 f_L(s)=O(q^{-2}),\qquad
 \frac{s^2-r^2}2=\frac{q^2}2+q(s-q)+\frac{(s-q)^2-r^2}2,\qquad
 \int_{q-r}^{q+r}s\,\mathrm ds=2qr.
\]
The $O(q)$ term $q(s-q)$ of the denominator is odd in $s-q$ on the symmetric interval.
At first order it therefore contributes only through the odd part $s-q$ of the measure
$s\,\mathrm ds=[q+(s-q)]\,\mathrm ds$, at relative order $q^{-2}$. With
$\Lambda=qw_\Lambda$, $\Lambda q+q^2/2=q^2(w_\Lambda+1/2)$ and \eqref{S-app-eq:fL-moment},
the leading term is
\[
 -\frac{3qw_\Lambda}{\pi^2}\int_0^1\frac{r f_L(r)\,2qr\,\mathrm dr}{q^4(w_\Lambda+1/2)^2}
 =-\frac{6w_\Lambda}{\pi^2}\cdot\frac14\cdot\frac4{q^2(1+2w_\Lambda)^2}
 =-\frac{3p_{\rm loc}(w_\Lambda)}{\pi^3q^2},
\]
and therefore (checked numerically in Sec.~\ref{S-sm-sec:verification})
\begin{equation}
 e_{\Sigma,\Lambda_\lambda}(q)
 =-\frac3{\pi^3q^2}p_{\rm loc}(w_{\Lambda,\lambda})
 +O\!\left(\frac{w_{\Lambda,\lambda}}{q^4(1+w_{\Lambda,\lambda})^2}\right),
 \qquad w_{\Lambda,\lambda}=\Lambda_\lambda/q\ge\tfrac12.
 \label{S-app-eq:lowdensity-u11}
\end{equation}

\emph{Low momentum.} Bound $\Omega/(\Omega+\delta)^2\le1/\Omega$ in
\eqref{S-app-eq:selfenergyreduction-u1}, with $\Omega_\lambda=\Lambda_\lambda q=\omega_\lambda$
the pole frequency of \eqref{S-app-eq:lowdensity-u1}. With $0\le f_L(r)-f_L(s)\le f_L(r)$,
\[
 M_\Sigma(q)=\frac{2\pi}{q}\int_{\max(0,1-q)}^1r\,\mathrm dr
 \int_{\max(1,q-r)}^{q+r}s\,[f_L(r)-f_L(s)]\,\mathrm ds
 \le4\pi\int_0^1r^2f_L(r)\,\mathrm dr=\pi,
\]
and
\begin{equation}
 |e_{\Sigma,\Lambda_\lambda}(q)|
 \le\frac{3M_\Sigma(q)}{2\pi^3\Omega_\lambda},\qquad
 \Omega_\lambda\ge\tfrac12\sqrt\lambda\,q_*^2,
 \label{S-app-eq:lowdensity-u12}
\end{equation}
and with $q_1=q_*^{1/2}$ the contributions are
\[
\begin{aligned}
 &q<q_1:\ O\Bigl(\frac{q_1}{\sqrt\lambda\,q_*^2}\Bigr),\qquad
 \text{error of the large-$q$ kernel, }q>q_1:\ O(q_1^{-3}),\\
 &\text{extending \eqref{S-app-eq:lowdensity-u11} to }q=0:\ O\Bigl(\frac{q_1}{\sqrt\lambda\,q_*^2}\Bigr),
\end{aligned}
\]
all $O(q_*^{-3/2})$ at fixed $\lambda$. The factor $\lambda^{-1/2}$ in these bounds is
integrable against $2\lambda\,\mathrm d\lambda$.

\emph{Remainder and fixed interval.} The remainder $\varsigma_\lambda$ and the fixed
momentum interval $0<q\le q_0$ contribute $O(q_*^{-2})$ to
\eqref{S-app-eq:sigma-lowdensity-split} when bounded as for the kite in
Sec.~\ref{S-sm-sec:kite-remainders}, with the cutoff $q_0\gg1$ chosen there. For the
self-energy pair these bounds use the large-$q$ estimate that follows from
\eqref{S-app-eq:PiSigma},
\[
\begin{gathered}
 \frac{|\Pi_{\rm SE}^{(1)}(q,\mathrm iq^2u)|}{q_{\rm TF}^4}=O\Bigl(\frac{q^{-4}}{1+u^2}\Bigr),
 \quad\text{since}\\
 |\Sigma_x(\bm p+\bm q)-\Sigma_x(\bm p)|\le\frac{q_{\rm TF}^2}2,\qquad
 |\mathrm i\nu-\delta_{\bm p}|^{-2}\le\frac1{\nu^2+(q^2/2-q)^2}\quad(q>2).
\end{gathered}
\]
They also use the lower bound on $\mathcal L$, derived below with $Q_0=(1+q_0/2)^2$,
\begin{equation}
 \mathcal L(q,\nu)\ge\frac{q^2}{3(\nu^2+Q_0q^2)},\qquad
 \frac{W_\lambda}{\lambda v}\le
 \frac{\nu^2+Q_0q^2}{\nu^2+Q_0q^2+\lambda q_{\rm TF}^2/3}
 \quad(0<q\le q_0),
 \label{S-app-eq:lowdensity-u13}
\end{equation}
together with $|\operatorname{Re}(\mathrm i\nu-\delta)^{-2}|\le(\nu^2+\delta^2)^{-1}$
in \eqref{S-app-eq:PiSigma} and the bounds of Sec.~\ref{S-app-app:selfenergy-blocks}.

\emph{Lower bound on $\mathcal L$.} Eq.~\eqref{S-app-eq:lowdensity-u13} uses the
free-response representation of Sec.~\ref{S-app-app:shell-integrated}, valid for every $q>0$,
\[
 \mathcal L(q,\nu)=(4\pi)^{-1}\int\mathrm d^3p\,\frac{\Delta_{\bm p}\delta_{\bm p}}{\delta_{\bm p}^2+\nu^2}.
\]
Its integrand is nonnegative because $\Delta_{\bm p}=\operatorname{sgn}\delta_{\bm p}$ on
the transition regions. There, for $q\le q_0$,
\[
 |\delta_{\bm p}|\le q+q^2/2\le q(1+q_0/2),\qquad
 \delta_{\bm p}^2+\nu^2\le\nu^2+Q_0q^2.
\]
Shifting $\bm p\to\bm p-\bm q$ in the term with $n_{\bm p+\bm q}$ gives the free
$f$-sum identity, valid for every $q$,
\[
 \int\mathrm d^3p\,\Delta_{\bm p}\delta_{\bm p}=\frac12\int\mathrm d^3p\,n_{\bm p}\bigl[(\bm p+\bm q)^2+(\bm p-\bm q)^2-2p^2\bigr]=4\pi q^2/3,
\]
and hence the first inequality. Inserting it into
$W_\lambda/(\lambda v)=[1+\lambda q_{\rm TF}^2\mathcal L/q^2]^{-1}$ gives the second.

\emph{Cancellation.} Inserting \eqref{S-app-eq:lowdensity-u11} into
\eqref{S-app-eq:sigma-lowdensity-split} gives \eqref{S-app-eq:A1-split} with the
large-momentum kernel \eqref{eq:largeq-block} of opposite sign, and the substitutions
of \eqref{eq:A1-dispersion} then give
\begin{equation}
 e_{\rm SE}=-\frac{A_1}{q_*}+O(q_*^{-3/2}),\qquad
 e_{\rm kite}+e_{\rm SE}=O(q_*^{-3/2}).
 \label{S-app-eq:sigma-lowdensity-cancellation}
\end{equation}
The bounds above hold at each fixed $\lambda$, hence at full coupling $\lambda=1$. Then
\[
 F_\Sigma=\mathcal K_{\Pi_{\rm SE}^{(1)}}[W_{\lambda=1}]=-\frac{7A_1}{8q_*}+O(q_*^{-3/2}),
\]
since at $\lambda=1$ the factor $2\int_0^1\lambda^{3/4}\mathrm d\lambda=8/7$ of $A_1$ is
absent, as for $F=\tfrac78e_{\rm kite}$ in Sec.~\ref{S-sm-sec:coupling-averages}. Hence
$F+F_\Sigma$ also has no $q_*^{-1}$ term. The first nonzero term after the cancellation
is not determined here.

\section{Remainder estimates for the low-density expansion}\label{S-sm-sec:lowdensity-remainders}
The bounds of this section prove the remainder $O(q_*^{-3/2})$ in main-text
Eqs.~\eqref{eq:A1-dispersion} and
\eqref{eq:low-density-response-cancellation}, including the limit
$\lambda\to0$ of the coupling integral.

\subsection{Fixed-frequency limit and the approach to the bare line}\label{S-app-sec:static-bounds}
For every $q>0$, with
$|\delta_{\bm p}|,|\delta_{\bm k}|\le q^2/2+q$ on the transition regions,
the radial kernel $e_\Lambda(q)$ of main-text Eq.~\eqref{eq:XLambda}
and its bare-line limit $e_\infty(q)$ obey
\begin{equation}
\begin{gathered}
 e_\Lambda(q)-e_\infty(q)
 =O\!\left(\bar Z(q)\left[\frac1{\Lambda q}
                  +\frac{1+q^2}{\Lambda^2q^2}\right]\right),\\
 \bar Z(q)=\int\mathrm{d}^3p\,\mathrm{d}^3k\,
 \frac{|\Delta_{\bm p}(\bm q)\Delta_{\bm k}(\bm q)|}{|\bm p-\bm k|^2}
 =O(q^2\log(e/q))\ (q\to0),\qquad \bar Z\to8\pi^2\ (q\to\infty),
\end{gathered}
\label{S-app-eq:static-approach}
\end{equation}
with an absolute constant in the $O$ symbol.

\emph{Pair bound.} For each pair of unscaled excitation energies
$a=\delta_{\bm p}$, $b=\delta_{\bm k}$, the closed form
\eqref{S-app-app-eq:J-closed} and the static frequency integral $\mathcal W$ of
main-text Eq.~\eqref{eq:Wdef} give
\[
 |\Omega^2J(\Omega;a,b)-\mathcal W(a,b)|\le\frac1{2\Omega}+\frac{\min(|a|,|b|)}{\Omega^2},
 \qquad\Omega=\Lambda q .
\]
Integrating this bound with the factor
$|\Delta_{\bm p}\Delta_{\bm k}|/|\bm p-\bm k|^2$ gives the factor $\bar Z$,
\[
 |e_\Lambda(q)-e_\infty(q)|\le4\pi A_0\,\bar Z(q)\Bigl[\frac1{2\Lambda q}+\frac{q^2/2+q}{\Lambda^2q^2}\Bigr]
 \le4\pi A_0\,\bar Z(q)\Bigl[\frac1{\Lambda q}+\frac{1+q^2}{\Lambda^2q^2}\Bigr].
\]

\emph{Limits of $\bar Z$.} The small-$q$ form follows from the shell volume $O(q)$ of
either transition region and the shell Coulomb bound of
Sec.~\ref{S-app-app:shell-integrated}. For $q>2$ the
two same-ball regions contribute $2\cdot4\pi^2$ (Sec.~\ref{S-app-sec:largeq-kernel}).
Hence
\[
 \begin{gathered}
 \bar Z(q)\le O(q)\cdot O\bigl(q\log(e/q)\bigr)\quad(q\to0),\\
 \bar Z(q)=2\cdot4\pi^2\ \text{(same ball)}+O(q^{-2})\ \text{(different balls)}\quad(q>2).
 \end{gathered}
\]

\emph{Fixed frequency.} For fixed $\Omega_0>0$, the choice
$\Lambda=\Omega_0/q$ keeps the width in $\nu$ fixed as $q\to0$. Then $\Lambda q=\Omega_0$, and
\eqref{S-app-eq:static-approach} gives
\begin{equation}
 e_{\Omega_0/q}(q)=e_\infty(q)+O_{\Omega_0}(q^2\log(e/q)),\qquad \Omega_0>0,\quad q\to0.
 \label{S-app-eq:plasmon-block}
\end{equation}
A spectral component of fixed frequency, such as the
small-$q$ plasmon with $\Lambda\simeq\omega_p/q$ (main-text
Sec.~\ref{sec:spectral-rep}), thus acts like the bare line at small $q$.

\emph{Large momentum.} For $q\gg1$, with the scaled frequency $u=\nu/q^2$ of
Sec.~\ref{S-app-sec:largeq-screening}, main-text Eq.~\eqref{eq:Xdef} gives
the decay bound
\begin{equation}
 X(q,q^2u)=O\!\left(\frac1{q^4(1+u^2)}\right),
 \qquad q\gg1,\quad u\ge0.
 \label{S-app-eq:Xudecay}
\end{equation}
For $q>2$ every excitation energy on the transition regions obeys
$|\delta|\ge q^2/2-q$. The product of the two frequency factors therefore obeys
\[
 \Bigl|\frac1{(\mathrm i\nu-\delta_{\bm p})(\mathrm i\nu-\delta_{\bm k})}\Bigr|
 \le\frac1{\nu^2+(q^2/2-q)^2}=\frac1{q^4[u^2+(1/2-1/q)^2]},
\]
and the Coulomb integrals of the four regions are bounded as in Sec.~\ref{S-app-sec:largeq-kernel}.

\subsection{Remainder bounds for the coupling-averaged kite}\label{S-sm-sec:kite-remainders}
To justify \eqref{S-app-eq:A1-split}, we insert
\eqref{S-app-eq:pole-dispersion} into \eqref{S-app-eq:Kweight} and keep the
error of the single-Lorentzian form as a separate term:
\begin{equation}
 e_{\rm kite}=e_{2x}-2\int_0^1\lambda\,\mathrm{d}\lambda
 \int_0^\infty\theta_\lambda(q)\,e_{\Lambda_\lambda(q)}(q)\,\mathrm{d}q
 +2\int_0^1\lambda\,\mathrm{d}\lambda\int_0^\infty
 \mathcal P_\lambda(q)\,\mathrm{d}q.
 \label{S-app-eq:lowdensity-u6}
\end{equation}
Here
\[
 \mathcal P_\lambda(q)=4\pi A_0q^2\int_0^\infty v(q)\,\varsigma_\lambda(q,\nu/q^2)\,X(q,\nu)\,\frac{\mathrm d\nu}\pi
\]
is the radial energy contribution of the error of the single-Lorentzian
form, without the coupling factor $\lambda$. Three errors are bounded in turn:
$\mathcal P_\lambda$, the bare remainder, and the replacement of
$e_{\Lambda_\lambda}-e_\infty$ by its large-momentum form.

\emph{Single-Lorentzian error, $q\gg1$.} We use $u$ and $\hat c$ of
Sec.~\ref{S-app-sec:largeq-screening} and the local variable
$\rho=\lambda(q_*/q)^4$ (not the function $\rho(c)$ of
Eq.~\eqref{S-app-eq:highdensity-u1}). With $(1+4u^2)/(1+u^2)\le4$,
\[
\begin{aligned}
 \lambda\hat c&=\frac\rho{1+4u^2},\\
 \int_0^\infty\frac{\lambda\hat c\,\mathrm du}{(1+\lambda\hat c)^2(1+u^2)}
 &\le\int_0^\infty\frac{4\rho\,\mathrm du}{(1+\rho+4u^2)^2}
 =\frac\pi2\,\frac\rho{(1+\rho)^{3/2}}
 =\frac\pi2\,\theta_\lambda\sqrt{1-\theta_\lambda}.
\end{aligned}
\]
This estimate, \eqref{S-app-eq:Xudecay} and the bound on
$\varsigma_\lambda$ in \eqref{S-app-eq:pole-dispersion} give
\begin{equation}
 \mathcal P_\lambda(q)=O\!\left(q^{-4}\,\theta_\lambda\sqrt{1-\theta_\lambda}\,\right),
 \qquad q\gg1.
 \label{S-app-eq:lowdensity-u7}
\end{equation}
For a fixed cutoff $q_0\gg1$,
\[
 \theta_\lambda\sqrt{1-\theta_\lambda}\le\rho^{-1/2}=\frac{q^2}{\sqrt\lambda\,q_*^2},\qquad
 \int_{q_0}^\infty\frac{\mathrm dq}{q^4}\,\frac{q^2}{\sqrt\lambda\,q_*^2}=\frac1{\sqrt\lambda\,q_0q_*^2},
\]
and integration over $\lambda$ with the weight $2\lambda\,\mathrm d\lambda$ gives $O(q_*^{-2})$.

\emph{Single-Lorentzian error, $q\le q_0$.} We use the positive lower bound
\eqref{S-app-eq:lowdensity-u13} on $\mathcal L$, the absolute shell bounds
of Sec.~\ref{S-app-app:shell-integrated}, and the large-frequency bound
from main-text Eq.~\eqref{eq:Xdef} and \eqref{S-app-eq:static-approach}.
With $\Xi=\lambda q_{\rm TF}^2/3$ and the ratio $W_\lambda/(\lambda v)$ of
\eqref{S-app-eq:lowdensity-u13},
\[
 |X(q,\nu)|\le\frac{\bar Z(q)}{\nu^2},\qquad
 \frac{W_\lambda}{\lambda v}\le\frac{\nu^2+Q_0q^2}{\Xi},\qquad
 \frac{W_\lambda}{\lambda v}\le1.
\]
The part $Q_0q^2/\Xi$ multiplies the frequency integral of $|X|$, bounded by
\eqref{S-app-eq:spectralconvergence-u2}. The part $\nu^2/\Xi$ for $\nu<\Xi^{1/2}$ and the
ratio bound $1$ for $\nu>\Xi^{1/2}$ use the large-frequency bound:
\[
 \frac{Q_0q^2}{\Xi}\int_0^\infty|X|\,\mathrm d\nu=O(\Xi^{-1}),\qquad
 \int_0^{\Xi^{1/2}}\frac{\nu^2}{\Xi}\,\frac{\bar Z}{\nu^2}\,\mathrm d\nu=\frac{\bar Z}{\Xi^{1/2}},\qquad
 \int_{\Xi^{1/2}}^\infty\frac{\bar Z}{\nu^2}\,\mathrm d\nu=\frac{\bar Z}{\Xi^{1/2}}.
\]
Each of the last two parts is therefore $O(\Xi^{-1/2})$.
Integrated over $0<q\le q_0$, these estimates bound the contribution of the screened line, with its
factor $\lambda$ removed, by
$O((\lambda q_{\rm TF}^2)^{-1/2}+(\lambda q_{\rm TF}^2)^{-1})$, and the
single-Lorentzian form obeys the same bound. Their difference therefore
contributes $O(q_*^{-2})$ after integration with $2\lambda\,\mathrm d\lambda$.

\emph{Bare remainder.} Write $e_{\Lambda_\lambda}=e_\infty+
(e_{\Lambda_\lambda}-e_\infty)$ and combine the $e_\infty$ term
with $e_{2x}$. The remaining bare contribution is
\begin{equation}
 2\int_0^1\lambda\,\mathrm d\lambda\int_0^\infty
 (1-\theta_\lambda)e_\infty\,\mathrm dq=O(q_*^{-3}),\qquad e_\infty=O(q^{-4}).
 \label{S-app-eq:lowdensity-u8}
\end{equation}
Here $e_\infty=O(q^{-4})$ by Sec.~\ref{S-app-sec:largeq-kernel} and
$\int_0^\infty|e_\infty|\,\mathrm dq<\infty$ by
\eqref{S-app-eq:spectralconvergence-u2}. Splitting the $q$ integral at $\lambda^{1/4}q_*$,
\[
\begin{aligned}
 1-\theta_\lambda&\le\min\{1,q^4/(\lambda q_*^4)\},\\
 \int_0^\infty(1-\theta_\lambda)|e_\infty|\,\mathrm dq
 &\le\int_0^{\lambda^{1/4}q_*}\frac{q^4|e_\infty|}{\lambda q_*^4}\,\mathrm dq
 +\int_{\lambda^{1/4}q_*}^\infty|e_\infty|\,\mathrm dq
 =O(\lambda^{-3/4}q_*^{-3}).
\end{aligned}
\]

\emph{Large-momentum kernel.} Finally, replace the subtracted kernel $e_{\Lambda_\lambda}-e_\infty$ by
its large-momentum limit $3p_{\rm loc}(w_{\Lambda,\lambda})/(\pi^3q^2)$
of main-text Eq.~\eqref{eq:largeq-block}, and bound the
resulting error separately below and above $q_1=q_*^{1/2}$. For $q<q_1$, take
\eqref{S-app-eq:static-approach} with $\Lambda q=\omega_\lambda$, the pole frequency of
Eq.~\eqref{S-app-eq:lowdensity-u1}, and $p_{\rm loc}(w)\le\pi/(2w)$. Then
\[
 e_{\Lambda_\lambda}-e_\infty=O\bigl(\omega_\lambda^{-1}+(1+q^2)\,\omega_\lambda^{-2}\bigr),\qquad
 \frac{3p_{\rm loc}(w_{\Lambda,\lambda})}{\pi^3q^2}\le\frac3{2\pi^2\omega_\lambda},
\]
and integration over this interval gives
\[
 O\bigl(\lambda^{-1/2}q_1/q_*^2+\lambda^{-1}q_1^3/q_*^4\bigr),\qquad
 q_1^3/q_*^4=q_*^{-5/2}<q_1/q_*^2 .
\]
For $q>q_1$, integrate the remainder $O(q^{-3}/(1+w_{\Lambda,\lambda})+q^{-4})$ of
main-text Eq.~\eqref{eq:largeq-block} with
$w_{\Lambda,\lambda}=\Lambda_\lambda/q$ of
Eq.~\eqref{S-app-eq:lowdensity-u11}. Since
$w_{\Lambda,\lambda}\simeq\tfrac12\sqrt\lambda\,(q_*/q)^2$ for
$q<\lambda^{1/4}q_*$,
\[
 \int_{q_1}^\infty\frac{\mathrm dq}{q^3(1+w_{\Lambda,\lambda})}=O(\lambda^{-1/2}q_*^{-2}\log q_*),\qquad
 \int_{q_1}^\infty\frac{\mathrm dq}{q^4}=O(q_1^{-3}).
\]
With the factors $\lambda^{-1/2}$
and $\lambda^{-1}$ suppressed, the bounds and integrated errors are as follows.
\begin{equation}
\begin{array}{c|c|c}
\text{region}&\text{bound}&\text{integrated error}\\
q<q_1&\omega_\lambda\ge\sqrt\lambda q_*^2/2&O(q_1/q_*^2)\\
q>q_1&w_{\Lambda,\lambda}\ge1/2&O(q_*^{-2}\log q_*)+O(q_1^{-3})
\end{array}
 \label{S-app-eq:lowdensity-u9}
\end{equation}

\emph{Coupling integral.} The suppressed factors and the factor $\lambda^{-3/4}$ below
\eqref{S-app-eq:lowdensity-u8} are integrable against
$2\lambda\,\mathrm d\lambda$, and the estimates therefore survive the
limit $\lambda\to0$ of the coupling integral.
With $q=q_*x$ and $\lambda>0$,
\[
 w_{\Lambda,\lambda}\ge\frac{\sqrt\lambda}{2x^2},\qquad
 p_{\rm loc}(w_{\Lambda,\lambda})\le\frac{\pi x^2}{\sqrt\lambda},\qquad
 \frac1{q_*}\int_0^{x_1}\frac{\mathrm dx}{x^2}\,p_{\rm loc}(w_{\Lambda,\lambda})
 \le\frac{\pi x_1}{\sqrt\lambda\,q_*},\qquad x_1=q_*^{-1/2}.
\]
Integration with the weight $2\lambda\,\mathrm d\lambda$
shows that extending the $x$ integral of main-text
Eq.~\eqref{eq:A1-dispersion} from $x_1$
to zero adds $O(q_*^{-3/2})$. The error terms of $e_{\rm kite}$ are
\[
\begin{aligned}
 \text{single-Lorentzian error }\mathcal P_\lambda,\ q\ge q_0\text{ and }q\le q_0\text{:}&\quad O(q_*^{-2}),\\
 \text{bare remainder:}&\quad O(q_*^{-3}),\\
 \text{large-momentum kernel, }q<q_1\text{ and }q>q_1\text{:}&\quad O(q_*^{-3/2}),\\
 \text{extension to }0<x<x_1\text{:}&\quad O(q_*^{-3/2}).
\end{aligned}
\]
This completes the proof of the remainders in
\eqref{S-app-eq:A1-split}. The exponent of the remainder bounds the
omitted terms but does not identify the next nonzero power in the density expansion.

\section{Numerical checks}\label{S-sm-sec:verification}
The checks are independent calculations, mostly quadratures, in double precision
unless a multiprecision working precision is stated. Each entry gives the largest deviation found.

\emph{Frequency integrals and normalization.} The closed forms of $J(\omega;a,b)$ of
Sec.~\ref{S-app-app:Jkernel} (equal-sign and opposite-sign arguments, coincident limit)
agree with direct positive-frequency quadrature to $2.2\times10^{-16}$. The closed form
of the Lindhard function $\mathcal L$ of Ref.~\cite{SI-FetterWalecka1971} agrees with direct
three-dimensional quadrature of its representation in Sec.~\ref{S-app-app:shell-integrated},
valid for every $q>0$ (Sec.~\ref{S-app-sec:sigma-lowdensity}), to $10^{-25}$ at
$(q,\nu)=(0.7,0.4)$ and $(3,2)$, with a working precision of 30 digits.
A symbolic evaluation of the kite prefactor and two further checks give
\[
 \begin{gathered}
 A_0=\tfrac3{32\pi^5},\qquad
 \textstyle\int_0^\infty h^2\,\mathrm dz=0.32133552109539704\ \ \text{vs.}\ \ \pi(1-\log2)/3,\\
 c_{\rm hd}=-0.000859445718977859,
 \end{gathered}
\]
the last from its analytic expression.
Quadrature of main-text
Eqs.~\eqref{eq:GB} and \eqref{eq:DuBois} (main-text Secs.~\ref{sec:benchmarks}
and \ref{sec:Pi1-d}) at a working precision of 50 significant digits gives the values below,
equal to the closed forms to 40 digits,
\[
 \begin{gathered}
 \frac3{\pi^3}\int_{-\infty}^{\infty}h^2\,\mathrm dz=0.0621813817\ldots=\frac{2(1-\log2)}{\pi^2},\\
 \text{residue integral}=0.0130393528\ldots=\frac{\alpha_L(\pi^2-12\log2)}{2\pi^3}.
 \end{gathered}
\]
For the direct evaluation of main-text Eq.~\eqref{eq:KD} with $f=v$, the bare-line
frequency integral, main-text Eq.~\eqref{eq:Wdef}, leaves only transitions of opposite
sign. The map $\bm k\mapsto-\bm k-\bm q$
then gives, with $\mathcal C_{\bm q}$ of Eq.~\eqref{S-sm-app-eq:couplingtransfer-u2},
\[
 \mathcal K[v]=2A_0\int\mathrm d^3q\,q^{-2}\int_{\mathcal C_{\bm q}}\mathrm d^3p\int_{\mathcal C_{\bm q}}\mathrm d^3k\,[|\bm\ell|^2\,\bm q\cdot\bm\ell]^{-1},\qquad
 \bm\ell=\bm p+\bm k+\bm q .
\]
The Monte Carlo integration fixes the axis of $\bm q$ and draws $\bm p$ uniformly in
$\mathcal C_{\bm q}$. It draws $\bm k$ from an equal mixture of the uniform density on
$\mathcal C_{\bm q}$ and a density proportional to $|\bm\ell|^{-2}$, which absorbs the
Coulomb factor. Both densities enter the sample weights. Four independent runs of $10^8$
samples give, with the standard error from the spread between runs,
\[
 \mathcal K[v]=0.04836(2)\qquad\text{vs.}\qquad e_{2x}=0.0483583\ \ \text{[main-text Eq.~\eqref{eq:bare-value}]}.
\]
This tests the prefactor $A_0$ and the signs of
Eqs.~\eqref{S-app-eq:GM}--\eqref{S-app-eq:A0-dimensional} independently of
the BRM correspondence of Sec.~\ref{S-sm-app-app:brm}.

\emph{Radial reduction.} Two evaluations of the double radial integral
\eqref{S-app-eq:sigma-double} for $\mathcal K_{\Pi_{\rm SE}^{(1)}}[b_\Lambda]$ agree to
relative differences below $1.5\times10^{-14}$ at $\Lambda=0.5$, $2$ and $10$. The first
uses adaptive quadrature. The second reduces the integral independently to one dimension,
with rational and $\log$ functions in the integrand, and evaluates it by quadrature at a
working precision of 45 significant digits. A third evaluation in the original coordinates
$(q,p,\theta)$ with composite Gauss--Legendre quadrature converges to the same values as
the order increases from 48 to 128 (differences $5.6\times10^{-8}$ to $2.3\times10^{-10}$).
The large-momentum kernel \eqref{S-app-eq:lowdensity-u11} of the self-energy pair was tested
through the scaled error
\[
 q^4\bigl[e_{\Sigma,\Lambda}+3p_{\rm loc}/(\pi^3q^2)\bigr],\qquad
 w_\Lambda=0.5,\ 2,\ 10\ \ \text{(fixed)},\qquad q=8,\ 32,\ 128 .
\]
It approaches a finite value over these momenta, consistent with the $q^{-4}$ remainder.
The factor $p_{\rm loc}$ agrees with an independent frequency integral to $2.2\times10^{-16}$.

\emph{Low-density coefficient.} The two-dimensional $(\lambda,x)$ integral of main-text
Eq.~\eqref{eq:A1-dispersion}, without the beta-function substitution, gives
$A_1=-0.21469376098909768$, the value of the gamma-function formula. An evaluation after
the coupling substitution, at a working precision of 45 significant digits, agrees. The
moment $1/4$ and the self-energy coefficient $-A_1$ are reproduced as
\[
 \int_0^1r^2f_L\,\mathrm dr\to0.25000000000000006,\qquad
 -A_1\to0.21469376098909773 .
\]
Hence the sum of the kite and self-energy coefficients is below $6\times10^{-17}$. This
checks the relative normalization of the cancellation, not the next surviving power.

\emph{Diagram counts and coupling identities.} A direct enumeration of Wick contractions
at orders $n=2$--$6$, and an independent brute-force enumeration at $n=2$--$4$, give the
coefficients of table~\ref{S-app-tab:chain-counts} and $1/4$ for the two-chain graph with
one bubble on each line at $n=4$ (96 contractions). We checked the density-scaling
identity, main-text Eq.~\eqref{eq:density-transfer}, and the moment rules for a term
$r_s^{-p}$ or $r_s^{-p}\log r_s$ of $F$,
\[
 2\int_0^1\lambda^{1-p}\,\mathrm d\lambda=2/(2-p),\qquad
 2\int_0^1\lambda^{1-p}\log\lambda\,\mathrm d\lambda=-2/(2-p)^2,
\]
on a scalar test function at $r_s=10^{-3}$--$10^3$ and $p=1/4,1/2,3/4$
(largest discrepancy $8.5\times10^{-12}$).

\emph{Direct high-density values of the kite.}
None of the checks above evaluates the screened kite energy at finite density.
Table~\ref{S-sm-tab:hd-direct} lists the direct values of $e_{\rm kite}$ at $r_s\le0.1$
used in main-text Sec.~\ref{sec:highdensity}, from the public repository
\cite{SI-ImotoCode}, with $e_{2x}=48.35832$~mRy. Each value uses 64 replicates, and different
$r_s$ use disjoint seeds. The five values are therefore statistically independent.
Weighted fits ($1/\sigma^2$) of
\[
 y=e_{\rm kite}-e_{2x}\quad\text{to}\quad c_{\rm fit}r_s\log^2r_s+c_{1,\rm fit}r_s\log r_s+c_2r_s
\]
give $c_{\rm fit}=-0.00091(5)$~Ry for all five points and $-0.00092(12)$~Ry for the
four points with $r_s\le0.03$. The weighted sum of squared residuals for the five points
is $0.02$, and $380$ without the $r_s\log^2r_s$ term. With $c_{\rm fit}=c_{\rm hd}$ fixed,
\[
 c_{1,\rm fit}=0.01125(6),\ 0.01134(12),\ 0.01143(25)\quad\text{for}\quad r_s\le0.1,\ 0.03,\ 0.01,
\]
against $c_1=0.01150$ \cite{SI-Imoto2026b}.
\begin{table}[ht]
\centering\small
\caption{Direct values of the coupling-averaged kite energy at high density,
in mRy, with the standard error $\sigma$ over replicates \cite{SI-ImotoCode}.
}
\label{S-sm-tab:hd-direct}
\begin{tabular}{@{}rrrr@{}}
\toprule
$r_s$ & $e_{\rm kite}$ & $e_{\rm kite}-e_{2x}$ & $\sigma$\\
\midrule
$0.001$ & $48.22031$ & $-0.13801$ & $0.00073$\\
$0.003$ & $48.01855$ & $-0.33977$ & $0.00100$\\
$0.01$ & $47.47046$ & $-0.88786$ & $0.00196$\\
$0.03$ & $46.29632$ & $-2.06200$ & $0.00326$\\
$0.1$ & $43.43311$ & $-4.92520$ & $0.00614$\\
\bottomrule
\end{tabular}
\end{table}

\end{document}